\documentclass[11pt]{article}
\usepackage{enumerate}
\usepackage[OT1]{fontenc}
\usepackage{amsmath,amssymb, amsthm}
\usepackage[square,sort,comma,numbers]{natbib}
\usepackage[usenames]{color}
\usepackage{dsfont}
\usepackage{pgfplots}
\usepackage{multirow}
\usepackage{rotating}
\usepackage{enumerate}
\usepackage{esvect}
\usepackage{multirow}
\usepackage{tikz}
\usetikzlibrary{patterns}
\usetikzlibrary{arrows}

\usepackage{times}
\usepackage{setspace}
\usepackage{amsfonts}       
\usepackage{nicefrac}       
\usepackage{microtype}      
\usepackage{mathtools}
\usepackage{fullpage}

\usepackage{float}

\usepackage{latexsym}
\usepackage{amssymb}
\usepackage{amsmath}
\usepackage{amsthm}
\usepackage{booktabs}
\usepackage{graphicx}
\usepackage{color}
\usepackage{subfigure}

\newtheorem{theorem}{Theorem}
\newtheorem{lemma}[theorem]{Lemma}
\newtheorem{corollary}[theorem]{Corollary}
\newtheorem{proposition}[theorem]{Proposition}

\newtheorem{definition}[theorem]{Definition}

\newcommand{\BibTeX}{B\kern-.05em{\sc i\kern-.025em b}\kern-.08em\TeX}

\theoremstyle{definition}

\usepackage{thmtools}
\usepackage{thm-restate}

\usepackage{booktabs}

\usepackage{arydshln}
\usepackage{pgfplots}

\newcommand{\NN}{\mathbb{N}}

\newcommand{\bigO}{\mathcal{O}}

\usepackage{mathtools}

\newcommand{\sminst}{$\mathcal{I}=(U,W,\mathcal{P})$}

\usepackage{amsmath}

\newtheoremstyle{probenv}%
{0pt}
{0em}
{\hangindent=\parindent}
{}
{\itshape}
{{\normalfont:}}
{.5em}
{}

\theoremstyle{probenv}

\usepackage[inline]{enumitem}
\usepackage{amsfonts}

\usepackage{xargs} 

\newcommandx{\problemdef}[6][3=Input,5=Question]{
	\begingroup
	\par\noindent\nopagebreak[4]
	\colorbox{gray!17!white}{\textsc{#1}}\nopagebreak[4]\nopagebreak[4]
	\par\noindent\hangindent=\parindent\textbf{#3}:  #4\nopagebreak[4]
	\par\noindent\hangindent=\parindent\textbf{#5}:  #6
	\par\smallskip
	\endgroup
}

\usepackage{todonotes}
\DeclareMathOperator{\rk}{rk}
\usepackage{cleveref}

\pgfplotsset{compat = 1.3}
\allowdisplaybreaks

\title{Worst- and Average-Case Robustness of Stable Matchings:\\ (Counting) Complexity and Experiments}
\author{
	Kimon Boehmer \\ Universit\'e Paris-Saclay \\ France \and 
	Niclas Boehmer \\ Hasso Plattner Institute\\ University of Potsdam \\ Germany
}

\begin{document}
	
	\maketitle

\begin{abstract}
	Focusing on the bipartite \textsc{Stable Marriage} problem, we investigate  different robustness measures related to stable matchings. 
We analyze the computational complexity of computing them and analyze their behavior in extensive experiments  on synthetic instances. 
For instance,
we examine whether a stable matching is guaranteed to remain stable if a given number of adversarial swaps in the agent's preferences are performed and the probability of stability when applying swaps uniformly at random. 
Our results reveal that stable matchings in our synthetic data are highly unrobust to adversarial swaps, whereas the average-case view presents a more nuanced and informative~picture. 
\end{abstract}

	\begin{table*}[t]
	\centering
	\small
	\resizebox{\textwidth}{!}{\begin{tabular}{r|cc|cc|cc|cc}
			& \multicolumn{2}{c|}{Matching} & \multicolumn{2}{c|}{Matching (bps)}  & \multicolumn{2}{|c}{Pair} & \multicolumn{2}{|c}{Agent} \\[1pt]
			& decision & counting & decision & counting & decision & counting & decision & counting \\
			\midrule
			\textsc{Swap} & P (Ob. \ref{s:1})  & \#P-h. (Th. \ref{s:6}) & NP-h. (Th. \ref{s:2})  & $-$ & NP-h. (Pr. \ref{s:3}) & $-$ & NP-h. (Pr. \ref{s:5}) & $-$ \\[1pt] 
			\midrule 
			\textsc{Delete}  & P (Ob. \ref{s:1}) & \#P-h. $(\dagger)$  & NP-h. (Th. \ref{s:2}) & $-$ & ? & \#P-h. (Th. \ref{s:7}) & P (Ob. \ref{s:4}) & \#P-h. (Th. \ref{s:7}) \\ 
	\end{tabular}}
	\caption{Overview of our results. ``Matching (bps)'' stands for the problem to create a given number of blocking pairs for a given matching using a limited number of changes. The result marked with $(\dagger)$ follows from Proposition 6 of \protect\citet{DBLP:conf/sagt/BoehmerBHN20}. }\label{ov:overview}
\end{table*}

\section{Introduction}
In two-sided stable matching problems, there are two sets of agents with each agent having preferences over the agents from the other set. 
The goal is to find a stable matching of agents from one side to agents from the other side, i.e., a matching where no pair of agents prefer each other to their current partner.  
Since their introduction by \citet{gale1962college}, such problems have been studied extensively in economics and computer science (see the survey of \citet{DBLP:books/ws/Manlove13}) and many real-world applications including online dating \cite{hitsch2010matching} and the assignment of students to schools \cite{abdulkadirouglu2005boston} or of children to daycare places \cite{kennes2011daycare} have been identified.  
Importantly, in many applications, agents remain matched for a longer period of time, over which their preferences might change as they learn more about their current match and other
options (as witnessed, e.g., by students
switching colleges or roommates). 
This observation already motivated different lines of research, for instance, the study of \emph{finding} robust stable matchings, that are, matchings which remain stable even if some changes are performed \cite{DBLP:journals/teco/ChenSS21,DBLP:conf/esa/MaiV18,DBLP:journals/corr/abs-2304-02590}. 

We contribute  a new perspective on robustness and stable matchings. Instead of computing different types of robust stable matchings, we focus on \emph{quantifying} the robustness of a given stable matching. 
In addition to making matching markets more transparent and predictable, we see multiple possible use cases for our robustness measures especially for market organizers. 
They can for instance be used to make an informed decision between different proposed stable matchings or more generally might serve as an additional criterion to decide between different matching algorithms. 
Moreover, if a matching is detected to be very non-robust, the organizers of the matching market can initiate some countermeasures to prevent the possibly high costs attached to changing a once-implemented matching to reestablish stability. For instance, they can check whether preferences were elicited correctly or ask agents to reevaluate their preferences after assisting them with further guidance or information. 

A first measure of the robustness of a stable matching is via the study of destructive bribery \cite{DBLP:conf/atal/ShiryaevYE13}:  
The idea is to compute the minimum number of changes (of a certain type) that need to be applied to the instance such that the matching becomes unstable.
This measure gives us a worst-case guarantee on the stability of our matching.
However, this measure disregards that changes in the real world will not be adversarial.
This is why we additionally study the robustness of stable matchings to random noise:
The idea is to compute the probability that the given matching remains stable if we apply a given number $k$ of changes uniformly at random. 
Observing how quickly this probability decreases with increasing $k$ gives us an estimate for the average-case robustness of the matching. 
While these two approaches are conceptually different, they are computationally closely related: Computing the probabilities in our second measure can be done by first counting the number of instances where $k$ changes have been applied and the given matching is stable and then dividing this number by the total number of instances at this distance. 
Consequently, computing our average-case measure requires solving the counting analog of the computational problem underlying our worst-case measure. 

\paragraph{Our Contributions.} 
We contribute a novel average-case perspective to the study of robustness and stable matchings and further explore the previously mentioned worst-case analog. 
In addition to studying the robustness of matchings, we also pioneer the study of the robustness of different local configurations such as whether an agent pair can be included in a stable matching or whether an agent is assigned a partner.
We focus on measuring robustness against swaps in the agent's preferences or the deletion of agents.

In \Cref{com:dec}, we analyze the computational complexity of eight decision problems resulting from our two change types and four goals (see \Cref{ov:overview} for an overview). 
The resulting complexity picture is mixed, as the complexity decisively depends on the type of change and the assessed object. 
Subsequently, in \Cref{com:count}, for all variants for which we did not show hardness in the decision setting, we present involved reductions showing the respective counting problem to be \#P-hard.\footnote{\#P is the counting analog of NP. One consequence of this is that a polynomial-time algorithm for a \#P-hard problem implies a polynomial-time algorithm for all problems in NP.} The \#P-hardness of the counting problems stand in contrast with simple polynomial-time algorithms for the decision variants and require intricate reductions needing new ideas (as counting problems have been rarely studied in the matching under preferences literature).
We conclude the section by presenting some approximation algorithms for the counting problems.

In \Cref{sub:exp}, customizing the large and diverse synthetic dataset introduced by \citet{DBLP:journals/corr/abs-2208-04041}, we perform extensive experiments related to the stability of stable matchings and pairs in case swaps in the agent's preferences are performed, thereby adding a new empirical component to the almost purely theoretical literature on robustness in stable matchings.
Among others, we find that in our dataset almost all stable matchings can be made unstable by performing a single swap. 
To measure matching's robustness to random swaps, motivated by our intractability results from \Cref{com:dec,com:count}, we explore a sampling-based approach building upon the popular Mallows noise model \cite{mal:j:mallows}. 
It turns out that stable matchings generally have a remarkably low robustness to random noise, yet the degree of the non-robustness significantly varies between instances; even within a single instance, different initially stable matchings may have a noticeably different average-case robustness. For instance, matchings produced by the popular Gale-Shapley algorithm tend to be less robust than so-called summed-rank minimizing stable matchings.
We also present a heuristic to measure matching's average-case robustness which builds upon the number of pairs that are close to being blocking, and demonstrate that it is of excellent quality in practice. 
Lastly, we observe that stable pairs are generally much more robust than stable matchings.  

The (full) proofs of all statements and many experimental details are deferred to the appendix. The code for our experiments can be found at github.com/n-boehmer/robustness-of-matchings. 

\paragraph{Related Work.}\label{sec:RW}
Previous work on the robustness of stable matchings predominantly focused on computing matchings that are guaranteed to remain stable even if a given number of changes are performed \cite{DBLP:journals/teco/ChenSS21,DBLP:conf/esa/MaiV18,DBLP:conf/fsttcs/GangamMRV22,DBLP:journals/corr/abs-2304-02590} or for which stability can be easily reestablished by changing only a few pairs in the matching \cite{DBLP:conf/ijcai/Genc0OS17,DBLP:conf/aaai/Genc0OS17,DBLP:journals/tcs/GencSSO19,DBLP:conf/fsttcs/GajulapalliLMV20}.  In contrast, our work proposes and analyzes different ways to \emph{quantify} the robustness of a given stable matching. As a result, while our paper shares a similar motivation, it is technically quite different. 

From a technical perspective, closest to our work are the papers of \citet{DBLP:journals/tcs/0001BFGHMR22,DBLP:journals/algorithmica/AzizBGHMR20}  and \citet{DBLP:conf/sagt/BoehmerBHN20}. 
Related to the worst-case robustness of matchings, \citet{DBLP:conf/sagt/BoehmerBHN20} initiated the study of constructive bribery problems in the stable matching literature, i.e., problems asking whether a matching can be made stable by performing a given number of changes. 
\citet{DBLP:journals/tcs/EibenGNRWY23} and \citet{DBLP:journals/corr/abs-2204-13485} extended their studies by analyzing related problems connected to ensuring  the existence of a stable matching with some desired property. 
Notably, all three of these papers focused purely on decision problems and did not analyze matching's robustness. 
Related to the average-case robustness of matchings, \citet{DBLP:journals/tcs/0001BFGHMR22,DBLP:journals/algorithmica/AzizBGHMR20} analyzed different problems occurring when agent's preferences are uncertain. 
For instance, they consider the problem of computing the probability that a given matching is stable assuming that each agent provides a probability distribution over preference lists or the full preference profile is drawn from some probability distribution. 
Our \textsc{\#Matching-Swap-Robustness} problem~(see~\Cref{ch:prelims}) is related to a special case of the latter problem, where we draw a preference profile uniformly at random from the set of all preference profiles at a given distance from some profile. 
However, a formal polynomial-time reduction cannot be established, as there are exponentially many profiles in the support of this distribution. 

In addition to finding robust stable matchings, there are also multiple other lines of work motivated by the observation that agent's preferences change over time. 
For instance, in the study of dynamic stable matchings the goal is usually to adapt classic stability notions to dynamic settings~\cite{DBLP:conf/atal/BoehmerN21,ALO20,baccara2020optimal,DBLP:journals/geb/DamianoL05,DBLP:journals/corr/abs-1906.11391,DBLP:journals/corr/abs-2007-03794}.
Moreover, various works have studied problems related to minimally adapting stable matchings to reestablish stability after changes occurred  
\cite{DBLP:conf/fsttcs/GajulapalliLMV20,Feigenbaum17,DBLP:conf/icalp/BhattacharyaHHK15,DBLP:conf/approx/KanadeLM16,uschanged,DBLP:conf/aaai/BredereckCKLN20}.

While the idea of using bribery for robustness is new in the context of matching under preferences, there are numerous works on bribery for quantifying robustness in the voting \cite{DBLP:journals/jair/FaliszewskiHH09,DBLP:reference/choice/FaliszewskiR16,DBLP:conf/atal/ShiryaevYE13,DBLP:journals/sncs/BaumeisterH23,DBLP:conf/ijcai/BoehmerBFN21,DBLP:conf/eaamo/BoehmerBFN22}, tournament \cite{DBLP:conf/atal/Doring023,DBLP:journals/ai/BrillSS22} and group identification \cite{DBLP:journals/aamas/BoehmerBKL23} literature . 

\section{Preliminaries}
\label{ch:prelims}

For a given set $S$, let $\mathcal{L}(S)$ denote the set of all strict and complete orders over elements from $S$. 

An instance $\mathcal{I}=(U,W,\mathcal{P})$ of \textsc{Stable Marriage} (SM) is defined by a set  $U=\{m_1,\dots,m_n\}$ of \emph{men} and a set $W=\{w_1,\dots,w_m\}$ of \emph{women}. Each man $m\in U$ is associated with a preference list $\succ_u\in \mathcal{L}(W)$ over women, and each woman $w\in W$ is associated with a preference list $\succ_w\in \mathcal{L}(U)$ over men. The preferences are collected in a \emph{preference profile} $\mathcal{P}=\{\succ_u\in \mathcal{L}(W)\mid u\in U \}\cup \{\succ_w\in \mathcal{L}(U)\mid w\in W\}$. 
We refer to $A=U\cup W$ as the set of \emph{agents}. For $a\in A$, we call $\succ_a$ the preference list of $a$ and say that $a$ \emph{prefers} agent $x$ over $y$ if $x \succ_a y$. We sometimes write $\succ_a$ as $a: a_1 \succ_a a_2 \succ_a a_3 \succ_a \dots$, where the ``$\dots$'' at the end mean that all other agents of opposite gender follow in some arbitrary order.  We write $\rk_b(a)$ for the position that agent $a$ has in the \emph{preference list} $\succ_b$ of agent $b$, i.e., the number of agents $b$ prefers to $a$ plus one. For an agent $a$ and two agents $b$ and $c$ appearing in $\succ_a$, we define the \emph{distance} between $b$ and $c$ (in the preferences of $a$) as $|\rk_a(b)-\rk_a(c)|$.

A \emph{matching} $M$ is a set of pairs $\{u,w\}$ with $ u\in U$ and $w \in W$ such that each agent appears in at most one pair. If an agent is contained in a pair  in $M$, they are \emph{assigned}; otherwise, they are \emph{unassigned}. For a matching $M$ and an agent $a$, $M(a)$ is the agent $a$ is matched to in $M$ if $a$ is assigned; otherwise, we set $M(a):=\bot$. A matching is \emph{complete} if no agent is unassigned. A pair $\{u,w\}$ \emph{blocks} a matching $M$ if 
\begin{enumerate*}[label=(\roman*)]
	\item $u$ prefers $w$ to $M(u)$ or is unassigned, and 
	\item $w$ prefers $m$ to $M(w)$ or is unassigned.
\end{enumerate*}
If a matching does not admit a blocking pair, it is \emph{stable}.
If a pair $\{u,w\}\in U\times W$ appears in some stable matching (in $\mathcal{I}$), then $\{u,w\}$ is a \emph{stable pair} (in $\mathcal{I}$). 
Similarly, if an agent $a\in A$ is assigned in some stable matching (in $\mathcal{I}$), then $a$ is a \emph{stable agent} (in $\mathcal{I}$).
Note that by the Rural Hospitals Theorem \cite{roth1986allocation}, the set of assigned agents is the same in all stable matchings in~$\mathcal{I}$. 

A \emph{swap} operation swaps two neighboring agents in the preference list of one agent.
The \emph{swap distance} between two instances $\mathcal{I}=(U,W,\mathcal{P})$ and $\mathcal{I}'=(U,W,\mathcal{P}')$ is the minimum number of swaps that are needed to transform $\mathcal{P}$ into~$\mathcal{P}'$.
A \emph{delete} operation deletes an agent from the agent set and from the preferences of all other agents.
Instance $\mathcal{I}=(U,W,\mathcal{P})$ is at \emph{deletion distance} $d$ from $\mathcal{I}'=(U',W',\mathcal{P}')$  if $\mathcal{I}'$ can be obtained from $\mathcal{I}$ by deleting $d$ agents. 

We now define our computational problems. In the name of our problems, we first specify the object whose stability we want to measure (i.e., matchings, pairs, or agents) and then the action that we allow (i.e., swaps or deletions).

\newcommand{\dpairexists}{\textsc{Destructive-Pair-$\mathcal{X}$}}
\newcommand{\dexactexists}{\textsc{Destructive-Exact-Exists-$\mathcal{X}$}}
\newcommand{\cagentassigned}{\textsc{Constructive-Agent-Assigned-$\mathcal{X}$}}
\newcommand{\dagentassigned}{\textsc{Destructive-Agent-Assigned-$\mathcal{X}$}}

\problemdef{Matching/Pair/Agent-Swap-Robustness}{dpe}
{An SM instance \sminst, a budget $\ell \in \NN$, and a matching $M$/pair $\{m,w\}\in U\times W$/agent $a\in U\cup W$.}
{Is there an SM instance  $\mathcal{I}'=(U,W,\mathcal{P}')$ at swap distance at most $\ell$ from $\mathcal{I}$ such that $M$/$\{m,w\}$/$a$ is not stable in $\mathcal{I}'$?}
\noindent \textsc{Matching/Pair/Agent-Delete-Robustness} is defined analogously by replacing swap distance with deletion distance. However, for the matching variant we only require that $M\cap (U'\times W')$ is not stable in the instance $\mathcal{I}'=(U',W',\mathcal{P}')$ resulting from the deletion, for the pair variant we require that neither $m$ nor $w$ get deleted, and for the agent variant that $a$ is not deleted.
For all defined decision problems $\mathcal{X}$, in the analogous counting problem $\#\mathcal{X}$, we ask for the number of instances at distance exactly\footnote{The exact constraint here is only for presentation purposes. In fact, the exact and at most variants of the problem can be Turing reduced to each other.} $\ell$ fulfilling the specified~property.

These computational problems can be used to quantify the robustness of matchings, agent pairs, and whether agents are assigned.
For instance, \textsc{Matching-Swap-Robustness} allows for computing the minimum number of swaps that are needed to make a given matching unstable
(in fact, it is easy to see that any stable matching and pair can be made unstable by using at most $n-1$ swaps; see \Cref{app:prelims}).
Similarly, \textsc{\#Matching-Swap-Robustness} can be used to compute the probability that a given matching is stable in case $k$ swaps are performed uniformly at random, by taking the answer to the counting problem and dividing it by the number of instances at swap distance $k$. The latter quantity can be computed in polynomial time using a dynamic program (see \Cref{app:prelims}).

Arguably, maintaing ``perfect'' stability even after preferences change is a quite strong requirement due to the binary nature of this criterion.
This motivates us to quantify the ``degree of instability'' after changes are performed. 
We consider the number of pairs that block a matching as a measure for this. 
To compute this, we face the problem of calculating the (maximum) number of pairs that block a matching when a certain number of changes are performed: 

\problemdef{Blocking Pairs-Swap [Delete]-Robustness}{dpe}
{An SM instance \sminst, a budget $\ell \in \NN$, a matching $M$, and an integer $b\in \NN$.}
{Is there an SM instance  $\mathcal{I}'=(U',W',\mathcal{P}')$ at swap [deletion] distance at most $\ell$ from $\mathcal{I}$ such that $M$ [$M\cap (U'\times W')$] is blocked by at least $b$ pairs in $\mathcal{I}'$?}

Solving the counting version of this problem would allow us to compute the expected ``degree of instability'' as the expected number of pairs blocking a matching after a given number of changes are performed uniformly at random.

\section{Complexity of Decision Variants}\label{com:dec}
We analyze the complexity of our decision problems starting with matchings, then pairs, and lastly agents. 

\paragraph{Stable Matchings.}

Making a given matching unstable is algorithmically straightforward.
For swap, it suffices to iterate over all pairs of agents $\{u,w\}\in U\times W$ that are currently not matched to each other, compute the minimum number of swaps to make this pair blocking (by swapping down $M(u)$ after $w$ in $\succ_u$ and  $M(w)$ after $u$ in $\succ_w$), and return the minimum. 
For delete, we can always make a matching unstable by deleting one woman and one man, as their partners form blocking pairs with each other. Deleting one agent is sufficient in case not all agents are matched to their top choice.
\begin{restatable}{observation}{one}
	\label{s:1}
	\textsc{Matching-Swap/Delete-Robustness} can be solved in $\mathcal{O}((n+m)^2)$ time.
\end{restatable}

If we want to modify the instance to create not only one but a certain number of blocking pairs, then the problem becomes NP-hard for both swap and delete. 
The reason for this is that there can be synergy effects when creating two blocking pairs: Making a pair blocking might become cheaper after we have already made another pair blocking by swapping down an agent's partner in their preferences. This effect allows us to devise reductions from \textsc{Clique} and \textsc{Independent Set}, respectively. In fact, the hardness both holds for the case where we ask for exactly and for at least $b$ blocking pairs.
\begin{restatable}{theorem}{two}
	\label{s:2}
	\textsc{Blocking Pairs-Swap/Delete-Robustness} is NP-complete.
\end{restatable}

\paragraph{Stable Pairs.}
Turning to stable pairs, the simple algorithms for stable matchings can no longer be applied, as there is no unique straightforward way to make a pair unstable. 
In fact, by reducing from the NP-hard constructive problem to make a given pair stable \cite{DBLP:conf/sagt/BoehmerBHN20}, we establish the NP-hardness of \textsc{Pair-Swap-Robustness}.
\begin{restatable}{proposition}{three}
	\label{s:3}\label{thm:dps}
	\textsc{Pair-Swap-Robustness} is NP-complete.
\end{restatable}
Notably, the analogous \textsc{Constructive-Exists-Delete} problem of making a given pair stable by deleting some agents was shown to be polynomial-time solvable by  \citet{DBLP:conf/sagt/BoehmerBHN20}. However, there seems to be no easy possibility to extend their algorithm to solve our \textsc{Pair-Delete-Robustness} problem. 
Settling the problem's complexity remains an intriguing open question. 

\paragraph{Stable Agents.}
Turning to the problem of making a given agent unstable, we observe the first difference between swap and delete. For delete, it is possible to reduce the problem to the polynomial-time solvable \textsc{Constructive-Exists-Delete} problem:
Assuming that a woman $w^*$ should be made unstable, we add a man $m^*$ that ranks $w^*$ first followed by all other women in an arbitrary ordering and add $m^*$ at the end of the preferences of all women. It is easy to see that stable matchings where $w^*$ is unmatched in the original instance correspond to matchings including $\{m^*,w^*\}$ in the constructed instance. Accordingly, it suffices to compute the number of deletions needed to make $\{w^*,m^*\}$ a stable~pair:
\begin{restatable}{observation}{four} 
	\label{s:4}
	\textsc{Agent-Delete-Robustness} can be solved in $\mathcal{O}(n\cdot m)$ time.
\end{restatable}

For swap, we establish hardness by reducing from one of the NP-hard constructive bribery problems studied by \citet{DBLP:conf/sagt/BoehmerBHN20}:
\begin{restatable}{proposition}{five}
	\label{s:5}
	\textsc{Agent-Swap-Robustness} is NP-complete.
\end{restatable}

\section{Complexity of Counting Variants}\label{com:count}
For all decision problems for which we have proven NP-hardness, their counting variants are naturally also computationally expensive to solve (in particular, at least as hard as the decision variants). 
We prove the \#P-hardness of all remaining counting problems in this section, starting with the robustness of stable matchings to swaps in the preferences.  Recall that we have observed in Observation~\ref{s:1}  a very simple algorithm for the decision problem, yet the counting version turns out to be computationally intractable:  
\begin{restatable}{theorem}{six}
	\label{s:6}
	\textsc{\#Matching-Swap-Robustness} is \#P-hard.
\end{restatable}
\begin{proof}[Proof Sketch]
	We reduce from the \#P-hard \textsc{\#Bipartite 2-SAT With No Negations} problem \cite{provan1983complexity}, where we are given a set $V=U \cup Z$ of variables  and a set $C \subseteq U \times Z$ of clauses, i.e., the formula does not contain any negative literals.
	The task is to count the number of truth assignments of the variables from $V$ such that each clause contains a fulfilled literal.  
	
	In our reduction, for each positive and negative literal, we add a literal man and a literal woman that are matched to each other in the given matching. Moreover, we add a set of $|C|$ dummy men and women. The agent set is
	\begin{align*}
	U'=\{m_v,m_{\bar{v}} \mid v \in V\} \cup \{m_i^d \mid i \in [|C|]\}\\
	W'=\{w_v,w_{\bar{v}}, \mid v \in V\} \cup \{w_i^d \mid i \in [|C|]\}
	\end{align*} 
	and the designated matching is
	\begin{align*}
	M:=&\{\{m_v,w_v\},\{m_{\bar{v}},w_{\bar{v}}\} \mid v \in V\}\\
	& \cup \{\{m_i^d,w_i^d\} \mid i \in [|C|]\}.
	\end{align*} 
	
	We now describe the preferences of the agents, focusing on agents corresponding to variables from $U$ first. 
	We construct the preferences so that matching $M$ is not stable in the initial instance:
	Specifically, for each variable $u\in U$, $m_u$ and $w_{\bar{u}}$ form a blocking pair, which implies that one of the two needs to modify their preferences to resolve the pair. If we modify the preferences of $m_u$, then this means that we set $u$ to true, while modifying $w_{\bar{u}}$ implies setting $u$ to false. 
	To ensure that the induced variable assignment fulfills all clauses, for each clause $\{p,q\}\in C$, we let $\{m_p,w_q\}$ form a blocking pair, implying that one of the involved literal agents needs to modify its preferences.  
	Combining these ideas, the preferences are as follows. 
	Consider a variable $u \in U$ and let $p_1,\dots,p_{c(u)} \in Z$ be all variables that appear together with $u$ in a clause. The preferences for the corresponding agents are as follows:
	\begin{align*}
	& m_u: w_{\bar{u}} \succ  w_{p_1} \succ \dots \succ w_{p_{c(u)}} \succ  \\
	& \qquad \qquad \qquad \qquad \qquad w_1^d \succ \dots \succ w_{|C|-c(u)}^d \succ w_u \succ \dots \\
	& w_{\bar{u}}: m_u \succ m_1^d \succ \dots \succ m_{|C|}^d \succ m_{\bar{u}} \succ \dots & \\
	& w_u: m_u \succ \dots \qquad \qquad m_{\bar{u}}: w_{\bar{u}} \succ \dots 	
	\end{align*}
	For $z \in Z$, the preferences are constructed analogously where the roles of men and women are reversed.	
	The dummy men have preferences $m_i^d: w_i^d \succ \dots$ and the dummy women have preferences $w_i^d: m_i^d \succ \dots$ for $i \in [|C|]$. 
	
	Now, a swap budget of $\ell:=|V|(|C|+1)$ is exactly sufficient to resolve the blocking pair for each variable  by modifying the preferences of one of the two literal agents by executing $|C|+1$ swaps to swap their partner in $M$ in the first position. Doing so, we also need to resolve the blocking pairs induced by the clauses, which ensures that the induced variable assignment satisfies all clauses. 
	However, this is not sufficient to establish a one-to-one correspondence between the  satisfying assignments for the given \textsc{2-SAT} formula and preference profiles at swap distance exactly $\ell$ from $\mathcal{I}'$ where $M$ is stable. In fact, if a literal is not contained in any clause, there are two ways to modify the corresponding agent within the given budget by either swapping the designated partner in $M$ up or the literal agent forming  a blocking pair down. 
	This results in a difficult-to-quantify blow-up in the number of solutions.
	To account for this, in the full construction and proof of correctness presented in \Cref{app:th41}, we alter the construction by adding a second copy of the instance. 
\end{proof}

It remains to consider the problem of counting the number of agent sets whose deletion makes a given agent or pair stable. In the decision world, we proved that the former problem is polynomial-time solvable while we were unable to settle the complexity of the latter problem.
In our most involved construction, we prove that both variants are \#P-hard in their counting version:
\begin{restatable}{theorem}{seven}
	\label{s:7}
	\textsc{\#Pair/Agent-Delete-Robustness} is \#P-hard.
\end{restatable}
\begin{proof}[Proof Sketch (Agent)]
	For an instance $\mathcal{I}$, we denote by \textsc{\#AD}$(\mathcal{I})$ the number of solutions for this instance. To show \#P-hardness for \textsc{\#Agent-Delete}, we will give a Turing reduction from the \#P-hard \textsc{\#Edge Cover} problem, where we are given a graph $G=(V,E)$ and an integer $k \in \NN$ and want to compute the number of subsets of edges $E' \subseteq E$ with $|E'|=k$ and $\bigcup_{e \in E'}{e}=V$ \cite{bubley1997graph}. In this proof, we let $n:=|V|$, $m:=|E|$, and $N:=m+1$. 
	
	We start by defining the \textsc{\#Agent-Delete} instances that we will give as input to our oracle. To this end, given a graph $G$, let $\mathcal{I}_{i,j}^G$ be the \textsc{\#Agent-Delete} instance constructed as follows. We set the deletion budget to $j+i \cdot N$. The agent set consists of \emph{edge men} $\{m_e \mid e \in E\}$, \emph{vertex men} $\{m_v^{q} \mid v \in V, q \in [N]\}$, \emph{extra men} $\{m_p^* \mid p \in [i]\}$ and \emph{vertex women} $\{w_v \mid v \in V\}$. The designated agent is $m_i^*$.
	
	The preferences of an edge men $m_e$ with $e=\{u,v\} \in E$ are:
	$
	m_e: w_u \succ w_v \succ \dots
	$. 
	Further, for some $v \in V$, the preferences of vertex women $w_v$ are 
	\begin{align*}
	w_v: m_v^1 \succ \dots \succ m_v^N \succ m_{e_1}& \succ \dots \succ m_{e_{\text{deg}(v)}}\\ &\succ m_1^* \succ \dots \succ m_i^* \succ \dots,
	\end{align*}
	where $e_1,\dots,e_{\text{deg}(v)}$ are the edges incident to $v$. 
	The extra men have arbitrary preferences, and the vertex men $m_v^q$ rank the corresponding vertex woman $w_v$ in the first position. 
	
	To give an intuition for how the solutions to this instance are connected to edge covers of $G$, consider the case $i=1$ with designated agent $m_1^*$ and budget $j+N$. Assume that we are interested in edge covers of size $m-j$.  Given a size-$j$ edge subset $E'\subseteq E$, we delete all edge men that do not correspond to an edge from $E'$. 
	If $E'$ is an edge cover, then $m_1^*$ cannot be made stable by deleting $N$ additional agents: For each vertex woman $w_v$, there is one edge man that was not deleted plus the $N$ corresponding vertex men that $w_v$ prefers to $m_1^*$.
	If, however, $E'$ is not an edge cover, then there is a vertex $v\in V$ such that all edge men corresponding to incident edges are deleted. After deleting all vertex men $m_v^q$ for $q \in [N]$, $w_v$ can be matched to our designated agent $m_1^*$ in a stable matching, since all vertex men and edge men that $w_v$ prefers to $m_1^*$ were deleted. 
	Thus, edge sets that are not edge covers in $G$ correspond to solutions to our constructed instance of the above-described form.
	However, one problem with this idea is that the correspondence is not one-to-one, as in case an edge set $E'$ does not cover multiple vertices, we have multiple ways to ensure that $m_1^*$ gets matched. 
	To deal with this issue, for each non-edge cover $E'$, we need to know how many vertices are not covered to bound the number of corresponding solutions. 
	More formally, we can establish a connection between solutions of the constructed instance and \emph{$i$-vertex-isolating sets}, that are edge sets $E' \subseteq E$ such that after deleting all edges in $E'$ there are exactly $i$ vertices that have no incident edges.
	The crucial ingredient of our Turing reduction is the following lemma whose proof appears in the appendix:
	\begin{lemma}
		Let $G=(V,E)$ be a graph, $i \in [n]$, and $j \in [m]$. It holds that:
		\begin{equation*}
		\text{\textsc{\#AD}}(\mathcal{I}_{i,j}^G)=\sum_{j'=0}^j{\sum_{i'=i}^{n}({\binom{(N+1)(n-i)}{j-j'} \cdot \binom{i'}{i}|\mathcal{E}_{i'}^{j'}|}}) + D_{i,j}, 
		\end{equation*}
		where $\mathcal{E}_i^j$ is the set of all $i$-vertex-isolating sets of size $j$ in $G$ and $D_{i,j}$ is the number of solutions to $\mathcal{I}_{i,j}^G$ that delete at least one extra man.
	\end{lemma}
	After showing that $D_{i,j}$ can be computed via a dynamic program that queries the oracle for \textsc{\#Agent-Delete} with altered versions of  $\mathcal{I}^G_{i,j}$, we can use the above lemma to compute the number of $i$-vertex isolating sets using another dynamic program with oracle calls. From this, the number of size-$k$ edge covers can be computed easily. 
\end{proof}

In light of the computational hardness results from this and the previous section, we seek an algorithm to approximate the probability that a matching/pair is stable after $k$ random changes are performed. 
It turns out that using Hoeffding's inequality, we can establish the effectiveness of a simple Monte-Carlo algorithm that samples profiles at distance $k$ uniformly at random and records the fraction of these profiles in which the matching/pair is stable: 
\begin{restatable}{proposition}{nine}
	Given $\varepsilon,\delta>0$,  $\ell\in \mathbb{N}$, and SM instance $\mathcal{I}$, there is a polynomial-time algorithm that computes for a given matching $M$ an estimate $p$ of the  probability that $M$ is stable at profiles at swap  (or deletion) distance $\ell$ so that $p\in [p^*-\varepsilon,p^*+\varepsilon]$ with probability $1-\delta$, where $p^*$ is the correct probability. The statement also applies to pairs and agents. 
\end{restatable}
Additionally, for the Matching-Swap setting, i.e. for counting the number of preference profiles at some swap distance from a given SM instance such that a given matching is unstable, we can find an $n^2-n$-approximation by counting all profiles where a specific pair is blocking, and summing over all pairs (by that possibly overcounting the number of profiles by a factor of  $n^2-n$). We can estimate the average factor by which we overcount and obtain a fully polynomial-time randomized approximation scheme (FPRAS):
\begin{restatable}{theorem}{ten}
There is a FPRAS for the problem of counting the number of preference profiles at swap distance exactly $\ell$ from $\mathcal{I}$ where $M$ is unstable. 
\end{restatable}
\section{Experiments}\label{sub:exp}
In this section, we analyze the robustness of stable matchings (\Cref{sec:matchings}) and pairs (\Cref{sec:pairs}) against random or adversarial swaps in a diverse set of synthetic instances.\footnote{To maintain focus, we do not consider delete operations, as we believe that swaps are the more common changes in practice. Moreover, we do not consider stable agents, as this would require the generation of instances where the two sides have a different size, which would make the setup more complicated.}
In our experiments, we sometimes measure the linear correlation between two quantities using the \emph{Pearson Correlation Coefficient} (PCC) which is $1$ in case of a perfect  positive linear correlation, $-1$ in case of a perfect negative linear correlation and $0$ if there is no linear correlation. 

\paragraph{Computing Stability Probabilities.}
Note that the sampling algorithm described at the end of \Cref{com:count} for approximating the probability that a given matching or pair remains stable after $k$ changes are performed requires sampling profiles at some given swap distance; a process that tends to be time-consuming already for $30$ agents \cite{DBLP:conf/ijcai/BoehmerBFN21}. 
This is why, 
following the works of \citet{DBLP:journals/sncs/BaumeisterH23} and  \citet{DBLP:conf/eaamo/BoehmerBFN22} in the context of elections, we make use of the popular Mallows model \cite{mal:j:mallows} for adding noise to preference lists. 
The Mallows model is parameterized by a central preference list $\succ\in \mathcal{L}(A)$ and a dispersion parameter $\phi\in [0,1]$, and samples a preference list $\succ'\in \mathcal{L}(A)$ with probability proportional to  $\phi^{\kappa(\succ,\succ')}$, where $\kappa(\succ,\succ')$ is the swap distance between $\succ$ and $\succ'$. Note that the dispersion parameter controls the level of noise added to the central preference list: For $\phi=0$ only the central preference list is sampled, whereas for $\phi=1$ all lists are sampled with the same probability.
However, as argued by \citet{DBLP:journals/corr/abs-2105-07815} and \citet{icml} the connection between the value of $\phi$ and the expected number of swaps applied to the preference list is non-linear.
Thus, to make our results easier to interpret, we use the normalized Mallows model introduced by \citet{DBLP:journals/corr/abs-2105-07815}: Here, we specify a normalized dispersion parameter norm-$\phi\in [0,1]$, which is internally converted to a value of the dispersion parameter $\phi$ such that the expected swap distance between a sampled preference list and the central one in the resulting Mallows model is a  norm-$\phi$ fraction of the maximum possible one. 

To measure the robustness of an SM instance $\mathcal{I}=(U,W,\mathcal{P})$ in our experiments, we fix a value of the norm-$\phi$ parameter, and for each $a \in A$, we draw a preference list $\succ'_a$ from the Mallows model with this norm-$\phi$ value and $\succ_a$ as the central preference list. We refer to the \emph{stability probability} of a matching or pair at norm-$\phi$ as the probability of the matching or pair being stable when executing this procedure.\footnote{In \Cref{app:Mal}, we prove that our counting problems can be reduced to computing the stability probability under the described model, implying that the latter task is also computationally intractable.} To approximate the stability probabilities in the resulting instance $\mathcal{I}':=(U,W,\mathcal{P}':=\{\succ_a' \mid a \in A\})$, we record whether a certain matching (or pair) is stable in $\mathcal{I}'$  and repeat this process 1000 times to get a Monte-Carlo-style approximation.
Similar in spirit to the works of \citet{DBLP:conf/ijcai/BoehmerBFN21,DBLP:conf/eaamo/BoehmerBFN22}, we sometimes assess the robustness of a matching or pair by the $50\%$-(stability)-threshold, which is the smallest examined value of norm-$\phi$ for which the estimated probability of being stable drops below $50\%$. 

\paragraph{Dataset.} 
Due to the lack of publicly available real-world data, we use the large diverse synthetic dataset created by \citet{DBLP:journals/corr/abs-2208-04041} focusing on instances with $50$ men and women. 
To generate the data, they used $10$ different synthetic models (see \Cref{app:data} for descriptions) including the Impartial Culture (IC) model, where each agent samples its preference list uniformly at random, and Euclidean models, where agents are uniformly at random sampled points in the Euclidean space and rank each other depending on the distance between their points. 
Further, their dataset contains three ``extreme'' instances: \begin{enumerate*}[label=(\roman*)]
	\item The \emph{Identity} instance, where all agents have the same preferences,
	\item the \emph{Mutual Agreement} instance, where $u$ ranks $w$ in position $i$ if $w$ ranks $u$ in position $i$, and
	\item the \emph{Mutual Disagreement} instance, where $u$ ranks $w$ in position $i$ if $w$ ranks $u$ in position~$50-i$. 
\end{enumerate*}

As we find that in instances from the dataset of \citet{DBLP:journals/corr/abs-2208-04041}, the worst-and average-case robustness of stable matchings is quite low, we add a \emph{Robust} extreme instance and instances sampled form a new \emph{Mallows-Robust} model. 
In the \emph{Robust} instance, for $i\in [50]$, the preference list of man $u_i$ is 
$
u_i: w_i\succ w_{i+1} \succ \dots \succ w_n \succ w_1 \succ \dots \succ w_{i-1}
$ and the preference list of woman $w_i$ is 
$
w_i: u_i \succ u_{i+1} \succ \dots \succ u_n \succ u_1 \succ \dots \succ u_{i-1}
$. In this instance, the matching $\{\{u_i,w_i\} \mid i \in [n]\}$ can only be made unstable by executing $50$ swaps.
The Mallows-Robust model is parameterized by the normalized dispersion parameter norm-$\phi\in [0,1]$. Each agent generates their preferences by making a sample from the Mallows model with their preference list from the Robust instance as the  central preference list and parameter norm-$\phi$. We sample $20$ instances for each norm-$\phi\in \{0.2,0.4,0.6,0.8\}$.
In total, our dataset~contains~$544$~instances.

\subsection{Robustness of Stable Matchings} \label{sec:matchings}\label{sub:exp1}
In this section, we analyze the robustness of men-optimal and other types of stable matchings as well as a simple heuristic for matching's average-case robustness.
\paragraph{Men-Optimal Matching.}
We start by considering the robustness of men-optimal matchings computed by the popular Gale-Shapley algorithm, which turn out to be very non-robust to adversarial swaps: In almost all instances from our dataset, the men-optimal matching can be made unstable by a single swap, implying that the worst-case robustness is seemingly not capable to meaningfully distinguish the robustness of matchings in practice. 

Turning to the robustness against random noise, \Cref{fig:distplot} depicts the distribution of the $50\%$-thresholds. In most instances, the $50\%$-threshold of the men-optimal matching is between $0.002$ and $0.003$. This corresponds to performing, on average, between $1.225$ and $1.8375$ swaps in each preference list. Considering that most swaps do not involve an agent's partner and thus do not influence the stability of the matching, this number is remarkably low. The explanation for this behavior is that in our instances there are in fact many agent pairs that only need one swap to become blocking. 
In \Cref{app:men}, we also examine how the $50\%$-threshold depends on the model from which the instance was sampled. Among others, we find that in instances sampled from the IC and Mallows-Robust model the $50\%$-threshold tends to be larger, whereas for instances sampled from Euclidean models it is lower. 
These findings highlight that while matchings are also in general quite non-robust against random noise, the picture is more nuanced than for adversarial noise.

\begin{figure*}
	\centering
	\begin{minipage}[t]{.48\textwidth}
		\centering
		\resizebox{0.7\textwidth}{!}{\begin{tikzpicture}[scale=0.8]
\begin{axis}[
  ymin=0, ymax=220,
  xlabel=norm-$\phi$,
  ylabel=number of instances,
  symbolic x coords={0.001,0.002,0.003,0.004,0.005,0.006,0.007,0.008,0.009,0.01,$>$0.01},
  xtick=data,
  xticklabels={0.001,,0.003,,0.005,,0.007,,0.009,,$>$0.01},
    x tick label style={rotate=90,anchor=east},
    nodes near coords align={vertical},
    x label style={at={(axis description cs:0.5,-0.28)},anchor=north},every tick label/.append style={font=\LARGE}, 
    label style={font=\LARGE },legend style={
    font=\huge
    }
]
\addplot[ybar, nodes near coords, fill=blue] 
    coordinates {
    (0.001,3) 
    (0.002,179) 
    (0.003,188) 
    (0.004,92)
    (0.005,37)
    (0.006,20)
    (0.007,6)
    (0.008,2)
    (0.009,4)
    (0.01,0)
    ($>$0.01,13)
 };

\end{axis}
\end{tikzpicture}} \vspace*{-0.5cm}
		\caption{Distribution of the 50\%-threshold of men-optimal matching.}
		\label{fig:distplot}
	\end{minipage}%
	\hfill 
	\begin{minipage}[t]{.48\textwidth}
		\centering
		\resizebox{0.7\textwidth}{!}{\input{devplots1}}\vspace*{-0.5cm}
		\caption{Stability probability of men-optimal matchings in three different instances for varying levels of noise.}
		\label{fig:fine1}
	\end{minipage}
\end{figure*}

Taking a closer look, in \Cref{fig:fine1} we depict how the stability probability behaves when we increase the added noise (aka.\ the chosen norm-$\phi$ parameter) for men-optimal matchings in three of the extreme instances, which indeed behave extreme here:
The Identity instance is the least robust instance in our dataset (independent of the norm-$\phi$ value). Its stability probability drops down very quickly and is already below $50\%$ at norm-$\phi=0.0006$ (which corresponds to making an expected number of $36$ swaps in the full instance) and below $10\%$ at norm-$\phi=0.002$.
This can be easily explained by the fact that in this instance as soon as an agent swaps down its current partner in its preferences a blocking pair is formed.
In contrast, the Mutual Agreement instance is the most robust instance that is not sampled from the Mallows-Robust model. Its stability probability drops down slower, as in this matching all agents are matched to their most preferred agent.
Lastly, the Robust instance has the highest stability probability in all our instances and we see that only at norm-$\phi=0.2$ does the matching's probability of being stable drop below $99\%$.
Moreover, we observe that for all instances (as for the three depicted ones) the stability probability of the men-optimal matching decreases monotonically (up to sampling errors) with increasing norm-$\phi$.  

\paragraph{Robustness Heuristic.}
Next, we examine a heuristic to estimate the $50\%$-threshold of a matching, which can be used as a fast approximation and sheds further light on what makes a matching robust to random noise. 
For this, for a matching $M$ and a man-woman pair $\{m,w\}$, we let $\beta(m,w)$ be  the minimum number of swaps needed to make $\{m,w\}$ blocking and $\#\beta(M,k)$ the number of man-woman pairs $\{m,w\}\in U\times W$ with $\beta(m,w)=k$. Intuitively, if $\#\beta(M,1)$ is large, then the probability that we create a blocking pair after making a few random swaps is high, as there are many swaps that when executed immediately create a blocking pair. Looking at the number $\#\beta(M,2)$ of pairs for which two swaps are needed, those are slightly less likely to become blocking (as both necessary swaps would need to be performed). If we increase the $k$ further, the likelihood of these pairs becoming blocking decreases exponentially. Accordingly, we derive the blocking pair proximity, which is small if matchings are robust:
\begin{definition}\label{def:bpdist}
	For some $d\in [n]$, we define the \emph{blocking pair proximity} of a matching $M$ as 
	$
	\pi(M):=\log_{n}{\sum_{k=1}^{d}n^{d-k} \cdot \#\beta(M,k)}.
	$
\end{definition}
To avoid large numbers, in our experiments, we set $d=5$, implying that we only examine pairs that need at most $5$ swaps to become blocking (those pairs have the strongest effect on the blocking pair proximity values). 
We find that the blocking pair proximity is indeed a very good indicator for the $50\%$-threshold of a matching: For the men-optimal matchings from our dataset the correlation between the two measures is with $-0.965$ very strong \citep{schober2018correlation}. 
Consequently, in practice, the average-case robustness of a stable matching seems to boil down to how far agent pairs are away from being blocking and can be well approximated using the much faster-to-compute blocking pair proximity measure. 

\begin{figure*}
	\centering
	\begin{minipage}[t]{.48\textwidth}
		\centering
		\resizebox{0.7\textwidth}{!}{\input{devplots3}}\vspace*{-0.5cm}
		\caption{Stability probability of three different matchings in instance sampled from IC model.}
		\label{fig:fine2}
	\end{minipage}%
	\hfill
	\begin{minipage}[t]{.48\textwidth}
		\centering
		\resizebox{0.7\textwidth}{!}{
\begin{tikzpicture}[scale=0.8]

\definecolor{color0}{rgb}{0.12156862745098,0.466666666666667,0.705882352941177}

\begin{axis}[
tick align=outside,
tick pos=left,
x grid style={white!69.0196078431373!black},
xlabel={50\%-threshold},
xmin=0, xmax=0.01,
xtick style={color=black},
y grid style={white!69.0196078431373!black},
ylabel={\#bps at norm-$\phi$ = 0.1},
ymin=0, ymax=60,
ytick style={color=black},every tick label/.append style={font=\huge}, 
label style={font=\huge },legend style={
	font=\huge
},
]
\addplot [semithick, color0, mark=*, mark size=2.5, mark options={solid}, only marks]
table {%
0.001 52.04
0.0115 15.02
0.0005 93.04
0.2926 0
0.0162 5.31
0.2036 0.03
0.0062 9.43
0.2257 0.01
0.0453 1.43
0.0088 9.1
0.2351 0
0.0035 21.46
0.0041 17.35
0.0135 4.68
0.0048 13.09
0.2265 0.06
0.0028 16.1
0.0171 5.13
0.2133 0.03
0.0049 14.02
0.0087 6.69
0.2341 0.01
0.2185 0
0.0083 4.98
0.0052 14.12
0.0082 10.42
0.004 12.23
0.0033 16.86
0.0063 12.76
0.0071 10.24
0.0031 19.92
0.0075 12.01
0.0045 13.44
0.0038 23.55
0.0041 14.49
0.0029 23.01
0.0026 22.54
0.0058 11.11
0.0065 13.15
0.0053 14.75
0.0051 12.76
0.0054 13.69
0.0042 17.68
0.0025 22.64
0.0028 21.19
0.0039 17.04
0.0028 19.89
0.005 14.44
0.0034 18.68
0.0035 19.14
0.0024 16.26
0.0028 19.91
0.0033 19.94
0.0047 21.35
0.0026 23.62
0.0033 17.89
0.0033 21.51
0.0032 23.51
0.0027 20.21
0.0046 18.83
0.0034 19.53
0.0032 18.45
0.0041 18.07
0.0051 16.55
0.0035 18.18
0.0037 20.5
0.0051 18.24
0.0039 16.65
0.0026 24.23
0.0039 17.29
0.0024 18
0.0051 13.87
0.0024 21.33
0.0047 18.43
0.0043 17.03
0.003 22.48
0.0035 19.62
0.0022 27.18
0.0043 19.79
0.004 17.54
0.0034 21.79
0.0043 16.07
0.0023 26.47
0.0024 25.1
0.0023 23.89
0.0037 17.63
0.0038 18.24
0.0034 20.24
0.0052 14.63
0.0048 21.24
0.0027 20.72
0.0022 21.73
0.0025 23.18
0.0034 20.3
0.0027 18.51
0.0037 19.18
0.0028 23.32
0.0031 19.44
0.0038 18.06
0.0062 14.34
0.0037 20.64
0.0033 16.74
0.0038 17.86
0.0034 18.51
0.0028 24.56
0.0023 24.46
0.0022 26.87
0.0029 26.07
0.0027 27.09
0.0022 27.47
0.0037 19.5
0.0019 26.49
0.0022 28.57
0.0028 26.07
0.0023 24.84
0.0022 27.56
0.0033 20.8
0.0031 21.71
0.0028 22.34
0.0037 25.82
0.0024 27.19
0.0052 20.65
0.0054 21.29
0.0041 20.65
0.002 25.07
0.0026 22.69
0.0026 24.41
0.0023 24.66
0.0025 24.35
0.003 26.12
0.0035 23.43
0.0025 25.51
0.0018 28.91
0.0027 25.28
0.0022 26.42
0.0026 26.61
0.0022 26.51
0.0035 21.55
0.0032 24.72
0.0018 26.35
0.0038 24.69
0.0025 23.07
0.0032 21.3
0.0031 24.18
0.0029 32.66
0.0022 36.8
0.0027 31.22
0.0026 32.95
0.0018 38.86
0.0022 36.29
0.0038 31.84
0.0026 32.1
0.0024 35.86
0.0023 36.12
0.0029 32.38
0.0023 34.64
0.0019 33.62
0.0025 32.9
0.0023 33.2
0.0037 29.74
0.0019 38.49
0.0023 32.82
0.0029 29.61
0.0023 35.24
0.003 23.35
0.0031 20.72
0.0034 21.68
0.005 17.14
0.0025 24.42
0.0037 18.89
0.0056 20.09
0.0044 16.04
0.003 21.44
0.0036 19.95
0.0047 18.12
0.0039 21.43
0.0064 18.08
0.0028 21.79
0.0028 22.63
0.0037 23.91
0.0039 19.7
0.0024 23.73
0.0025 21.27
0.0057 18.96
0.0045 16.62
0.004 20.07
0.0028 25.69
0.0043 17.47
0.0035 16.54
0.0044 17.32
0.0051 17.92
0.0047 17.58
0.0052 15.83
0.006 14.75
0.0032 19.82
0.004 16.27
0.0031 22.9
0.0034 20.89
0.0028 21
0.0035 17.71
0.0049 17.63
0.0038 18.82
0.0035 19.02
0.0038 17.89
0.0033 20.14
0.0053 15.96
0.0034 21.16
0.0043 15.38
0.0031 22.26
0.0045 16.6
0.0043 19.35
0.0024 21.54
0.0028 23.4
0.0048 17.33
0.0035 15.58
0.0031 20.48
0.0035 18.61
0.0035 18.67
0.0034 20.7
0.0061 16.21
0.0047 14.88
0.0028 14.86
0.005 16.07
0.0028 17.95
0.0013 41.32
0.0015 33.49
0.0011 43.69
0.0014 39.82
0.0014 40.74
0.0012 33.14
0.0011 40.71
0.0017 33.45
0.0012 40.91
0.0013 40.52
0.0015 37.06
0.0011 46.58
0.0012 44.02
0.0012 38.02
0.0013 34.32
0.0012 40.75
0.001 47.22
0.0012 34.84
0.0011 44.88
0.0015 33.58
0.0015 39.35
0.0015 34.56
0.002 31.61
0.0014 37.03
0.0014 31.72
0.0013 36.45
0.0015 33.32
0.0019 30.84
0.0014 35.54
0.0016 37.08
0.0012 34.9
0.0018 31.83
0.0014 35.97
0.0022 28.57
0.0016 31.3
0.0015 35.92
0.0013 35.88
0.0018 32.48
0.0012 33.6
0.0015 39.07
0.0013 41.55
0.0021 39.19
0.0014 42.45
0.0014 41.28
0.0016 38.83
0.0013 46.24
0.0016 43.08
0.0017 39.7
0.0013 44.69
0.0014 44.71
0.0016 39.63
0.0013 43.03
0.0014 46.31
0.0015 44.07
0.0019 36.4
0.0019 40.22
0.0015 42.26
0.0013 44.47
0.0016 42.29
0.0021 42.99
0.0018 37.62
0.0017 37.89
0.0015 37.25
0.0016 38.94
0.0019 35
0.0016 37.69
0.0016 37.69
0.0018 35.07
0.0016 36.11
0.0016 38.6
0.0018 35.54
0.0017 38.64
0.002 33.43
0.0017 37.7
0.0018 34.34
0.0015 39.5
0.0014 40.17
0.0019 35.17
0.0014 40.11
0.002 33.26
0.0023 30.32
0.0023 29.7
0.0017 35.65
0.0023 31
0.002 30.46
0.0018 29.92
0.0023 32.1
0.0019 32
0.0026 30.21
0.0015 34.42
0.0015 31.47
0.0018 32.37
0.002 34.8
0.002 31.67
0.002 33.8
0.0016 33.69
0.0017 33.34
0.0016 32.14
0.0028 31.47
0.0023 31.92
0.0017 32.62
0.0015 34.73
0.0017 34.84
0.0014 35.72
0.0015 36.72
0.0016 33.64
0.0016 33.82
0.0019 33.48
0.0024 32.34
0.0018 31.65
0.0016 31.77
0.0018 33.2
0.0023 32.39
0.0018 33.41
0.0021 33.83
0.0019 34.86
0.0019 33.81
0.0017 33.28
0.0018 32.65
0.0017 34.35
0.002 30.2
0.002 34.53
0.002 32.75
0.0027 30.57
0.0014 34.7
0.0027 28.51
0.0025 31.09
0.0018 36.26
0.0017 33.82
0.0015 29.97
0.0016 36.87
0.0015 33.79
0.0019 31.85
0.0018 32.71
0.0022 31.2
0.0019 36.24
0.0019 29.44
0.0022 29.62
0.0028 26.74
0.0018 33.31
0.0024 25.77
0.0029 25.94
0.0016 29.61
0.0018 30.31
0.003 26.68
0.0026 28.23
0.0025 28.8
0.0018 30.23
0.0015 33.29
0.0027 30.2
0.002 29.67
0.0021 31.77
0.0022 29.05
0.0021 34.06
0.0031 24.56
0.0019 33.44
0.002 29.41
0.0021 31.2
0.0041 20.6
0.0023 26.1
0.0021 27.74
0.0028 20.44
0.0038 19.63
0.0028 20.63
0.0028 26.4
0.0016 29.03
0.0033 23.8
0.0028 25.69
0.0024 24.81
0.0023 28.06
0.0026 25.84
0.0027 26.01
0.003 22.7
0.0027 24.86
0.0035 21.27
0.0028 21.91
0.0027 26.33
0.0026 24.69
0.0031 22.39
0.0029 23.41
0.0019 28.83
0.0014 39.17
0.0024 25.98
0.0017 32.31
0.0018 32.97
0.0021 31.63
0.0019 31.49
0.0021 30.31
0.0013 33.35
0.0017 32.44
0.0025 29.12
0.0022 29.82
0.0016 33.38
0.0021 29.06
0.0019 32.65
0.0017 25.22
0.0021 32.33
0.0024 27.86
0.0022 25.96
0.0017 28.65
0.0022 25.69
0.0019 29.91
0.0018 30.15
0.0027 24.6
0.0029 26.3
0.0027 26.53
0.0023 29.32
0.0032 24.63
0.0029 28.35
0.0022 27.19
0.0019 27.77
0.0026 24.16
0.0039 21.94
0.0019 26.1
0.002 32.17
0.0019 26.59
0.0019 25.54
0.0022 28.37
0.0023 25.47
0.0019 26.56
0.0022 34.99
0.0023 31.28
0.0025 34.56
0.002 38.55
0.0018 37.12
0.0016 41
0.002 40.07
0.0023 38.01
0.002 36.62
0.0021 37.72
0.003 32.81
0.0024 35.83
0.0019 35.47
0.0022 33.52
0.0024 32.3
0.0029 31.83
0.0019 34.92
0.0026 33.51
0.0021 34.91
0.0021 32.8
0.0041 24.81
0.0028 27.65
0.0034 23
0.0032 19.96
0.0029 27.42
0.0027 25.57
0.0044 18.01
0.0029 27.09
0.0021 31.98
0.0038 19.96
0.0025 27.39
0.0028 26.08
0.0021 26.5
0.0057 16.96
0.0035 24.79
0.0022 32.39
0.0029 28.29
0.0038 20.67
0.0026 23.48
0.0031 19.04
0.005 15.86
0.0031 21.35
0.0034 23.16
0.0036 16.89
0.0035 20.18
0.0032 19.48
0.0023 25.09
0.0036 18.89
0.0033 18.05
0.006 15.25
0.0044 17.71
0.0044 19.67
0.003 22.19
0.0023 22.16
0.0026 19
0.002 24.76
0.0036 15.86
0.0046 17.18
0.0056 20.08
0.0025 22.95
0.002 39.56
0.0022 37.96
0.0013 43.43
0.002 43.3
0.0017 44.75
0.0018 42.52
0.0016 43.94
0.002 41.91
0.0018 43.86
0.0021 40.18
0.0015 46.47
0.0018 41.2
0.0018 39.93
0.0015 42.72
0.0021 43.06
0.0016 42.93
0.0022 39.85
0.002 46.34
0.0015 46.47
0.0018 44.67
0.0027 36.67
0.0027 30.71
0.0031 29.11
0.0026 30.59
0.0024 28.5
0.0024 30.23
0.0027 29.19
0.0023 34.16
0.0026 27.66
0.0025 32.77
0.0021 34.05
0.0027 26.74
0.0023 29.69
0.002 33.03
0.003 34
0.0022 30.5
0.0031 27.92
0.0025 26.38
0.0027 28.18
0.0025 35.64
};
\end{axis}

\end{tikzpicture}}\vspace*{-0.5cm}
		\caption{Correlation between 50\%-threshold and average number of blocking pairs of men-optimal matching at norm-$\phi=0.1$.}
		\label{fig:corr-bps-fif}
	\end{minipage}
\end{figure*}

\paragraph{Other Stable Matchings.}
Stable matchings are often not unique, which motivates the question of whether different initially stable matchings have a different robustness. 
To this end, we consider two alternative types of stable matchings. First, the \emph{summed-rank minimizing stable matching} \cite{knuth1976marriages}, which minimizes the average rank of an agent's partner in their preference list and can be computed in polynomial time \cite{VANDEVATE1989147}. 
Second, the so-called \emph{robust stable matching}, which is a stable matching minimizing the expression from \Cref{def:bpdist} (see \Cref{app:oth} for a description of an ILP to compute this matching).
Note that in $176$ of our instances the three matchings are identical. 

In terms of their worst-case robustness, the three types of stable matching behave almost identically in all instances. 
In contrast, for the average-case robustness, there is a larger difference: In case the three matchings differ from each other, oftentimes the robust and summed-rank minimizing matching have a clearly higher $50\%$-threshold than the men-optimal matching. In particular, for 49 instances, the $50\%$-threshold of the robust and summed-rank minimizing stable matching is more than twice as high as for the men-optimal one; however, there are also few instances where the men-optimal stable matching slightly outperforms the summed-rank minimizing one. 
The average difference between the $50\%$-threshold of the robust and summed-rank minimizing matching is with $0.0004$ less pronounced. 
Nevertheless, within one instance the stability probability of all three matchings typically behaves quite similarly when adding noise (see \Cref{fig:fine2} for an example). 
All in all, we can conclude that from a robustness perspective it is recommendable to use the summed-rank minimizing stable matching instead of the men-optimal one (the small gain from using the robust stable matching instead arguably does not justify the increase in computation time).

\paragraph{Average Number of Blocking Pairs.}
As discussed in \Cref{ch:prelims}, enforcing that an initially stable matching remains stable after changes have been performed might be viewed as a quite strict requirement. 
This motivated the study of the \textsc{Blocking Pairs-Robustness} problems and motivates us in this section to analyze the expected number of pairs by which a given matching is blocked when changes are performed. 
For this, we again make use of the Mallows model and compute the average number of pairs blocking a given matching when applying the Mallows model with a given norm-$\phi$ parameter to all preference lists. 
We find that the average number of pairs blocking a matching is highly correlated with the matching's $50\%$-threshold: In \Cref{fig:corr-bps-fif}, each point corresponds to the men-optimal matching in one of our instances with the $x$-axis showing the matching's $50\%$-threshold and the $y$-axis showing the average number of pairs blocking the matching at norm-$\phi=0.1$ (plots for other values show similar trends).
We see a clear correlation between the two measures, but also clear differences on the instance-level, e.g., a matching with $50\%$-threshold $0.02$ might be blocked by between $23$ and $44$ pairs on average.
The general connection between the measures indicates that the $50\%$-threshold remains informative beyond the focus on binary stability, yet the average number of blocking pairs allows for a more nuanced picture. 
We refer to \Cref{app:bp} for further details, e.g., we find that within one instance, the average number of blocking pairs typically grows linearly with increasing value of norm-$\phi$ (a clearly different behavior compared to the stability probability of matchings).

\begin{figure}[t]
	\centering
	\begin{minipage}[t]{.48\textwidth}
		\centering
		\resizebox{0.7\textwidth}{!}{\begin{tikzpicture}[scale=0.8]
\begin{axis}[
  ymin=0, ymax=150,
    xlabel=norm-$\phi$,
    ylabel=number of instances,
     symbolic x coords={0.02,0.04,0.06,0.08,0.1,0.12,0.14,0.16,0.18,0.2,$>$0.2},
    xtick=data,
    xticklabels={0.02,,0.06,,0.1,,0.14,,0.18,,$>$0.2},
    x tick label style={rotate=90,anchor=east},
    nodes near coords align={vertical},
    x label style={at={(axis description cs:0.5,-0.28)},anchor=north},every tick label/.append style={font=\LARGE}, 
    label style={font=\LARGE },legend style={
    	font=\LARGE
    }
]
\addplot[ybar, nodes near coords, fill=blue] 
    coordinates {
    (0.02,2) 
    (0.04,15) 
    (0.06,105) 
    (0.08,133)
    (0.1,69)
    (0.12,69)
    (0.14,72)
    (0.16,33)
    (0.18,10)
    (0.2,6)
    ($>$0.2,30)
 };

\end{axis}

\end{tikzpicture}}\vspace*{-0.5cm}
		\caption{Distribution of the average $50\%$-threshold of stable pairs in each instance (rounded up to a multiple of $0.02$).}
		\label{fig:avg-pairs-dist}
	\end{minipage}%
	\hfill
	\begin{minipage}[t]{.48\textwidth}
		\centering
		\resizebox{0.7\textwidth}{!}{
\begin{tikzpicture}[scale=0.6,every plot/.append style={line width=2.2pt}]

\definecolor{color0}{rgb}{0.12156862745098,0.466666666666667,0.705882352941177}
\definecolor{color1}{rgb}{1,0.498039215686275,0.0549019607843137}
\definecolor{color2}{rgb}{0.172549019607843,0.627450980392157,0.172549019607843}
\definecolor{burgundy}{rgb}{0.5, 0.0, 0.13}

\definecolor{camel}{rgb}{0.76, 0.6, 0.42}

\definecolor{darkolivegreen}{rgb}{0.33, 0.42, 0.18}

\begin{axis}[
tick align=outside,
tick pos=left,
x grid style={white!69.0196078431373!black},
xlabel={norm-$\phi$},
xmin=-0.0495, xmax=1.0395,
xtick style={color=black},
y grid style={white!69.0196078431373!black},
ylabel={stable pair probability},
ymin=-0.029, ymax=1.049,
ytick style={color=black},every tick label/.append style={font=\huge}, 
label style={font=\huge },
yticklabels={,0\%,20\%,40\%,60\%,80\%,100\%}
]

\addplot [semithick, burgundy]
table {%
0 1
0.005 1
0.01 0.9998
0.015 0.9988
0.02 0.9968
0.025 0.9908
0.03 0.984
0.035 0.974
0.04 0.9672
0.045 0.96
0.05 0.952
0.055 0.947
0.06 0.9392
0.065 0.937
0.07 0.9216
0.075 0.923
0.08 0.9144
0.085 0.8996
0.09 0.891
0.095 0.884
0.1 0.885
0.105 0.8784
0.11 0.8652
0.115 0.866
0.12 0.8592
0.125 0.861
0.13 0.837
0.135 0.8432
0.14 0.847
0.145 0.8374
0.15 0.8176
0.155 0.8224
0.16 0.816
0.165 0.8124
0.17 0.7962
0.175 0.7892
0.18 0.7864
0.185 0.7674
0.19 0.7844
0.195 0.7602
0.2 0.7522 
0.205 0.741
0.21 0.7372
0.215 0.7262
0.22 0.714
0.225 0.7202
0.23 0.7008
0.235 0.68
0.24 0.6892
0.245 0.678
0.25 0.6774
0.255 0.6626
0.26 0.6614
0.265 0.6524
0.27 0.6292
0.275 0.6288
0.28 0.6244
0.285 0.6106
0.29 0.6108
0.295 0.5954
0.3 0.6008
0.305 0.5794
0.31 0.5642
0.315 0.5622
0.32 0.5704
0.325 0.562
0.33 0.542
0.335 0.5382
0.34 0.5392
0.345 0.5206
0.35 0.5224
0.355 0.5144
0.36 0.509
0.365 0.5086
0.37 0.491
0.375 0.4848
0.38 0.478
0.385 0.4648
0.39 0.4678
0.395 0.4686
0.4 0.4574
0.405 0.4372
0.41 0.4422
0.415 0.4356
0.42 0.4296
0.425 0.4322
0.43 0.4232
0.435 0.4334
0.44 0.415
0.445 0.4014
0.45 0.401
0.455 0.3904
0.46 0.38
0.465 0.368
0.47 0.3612
0.475 0.369
0.48 0.365
0.485 0.3688
0.49 0.356
0.495 0.3476
0.5 0.3512
0.505 0.3462
0.51 0.3232
0.515 0.3384
0.52 0.3244
0.525 0.3174
0.53 0.3064
0.535 0.3048
0.54 0.3166
0.545 0.3062
0.55 0.2884
0.555 0.2898
0.56 0.2832
0.565 0.283
0.57 0.2862
0.575 0.2762
0.58 0.275
0.585 0.262
0.59 0.2496
0.595 0.2514
0.6 0.2538
0.605 0.2518
0.61 0.2528
0.615 0.234
0.62 0.2098
0.625 0.2296
0.63 0.2246
0.65 0.2122
0.67 0.199
0.690000000000001 0.1788
0.710000000000001 0.1534
0.730000000000001 0.1458
0.750000000000001 0.1414
0.770000000000001 0.13
0.790000000000001 0.1066
0.810000000000001 0.0998
0.830000000000001 0.0918
0.850000000000001 0.0826
0.870000000000001 0.0762
0.890000000000001 0.0684
0.910000000000001 0.058
0.930000000000001 0.0514
0.950000000000001 0.0498
0.970000000000001 0.0454
0.990000000000001 0.04
};
\addplot [semithick, camel]
table {%
0 1
0.005 0.8746
0.01 0.7554
0.015 0.6654
0.02 0.5964
0.025 0.5286
0.03 0.4864
0.035 0.446
0.04 0.4322
0.045 0.4016
0.05 0.3966
0.055 0.3862
0.06 0.3802
0.065 0.3774
0.07 0.3758
0.075 0.3678
0.08 0.3528
0.085 0.3544
0.09 0.3452
0.095 0.3476
0.1 0.3484
0.105 0.349
0.11 0.3328
0.115 0.3156
0.12 0.3192
0.125 0.3142
0.13 0.3094
0.135 0.3108
0.14 0.3084
0.145 0.2894
0.15 0.292
0.155 0.2802
0.16 0.2914
0.165 0.2836
0.17 0.2738
0.175 0.255
0.18 0.2688
0.185 0.259
0.19 0.2538
0.195 0.245
0.2 0.2368
0.205 0.245
0.21 0.229
0.215 0.223
0.22 0.2264
0.225 0.2238
0.23 0.2166
0.235 0.2064
0.24 0.1938
0.245 0.1944
0.25 0.2
0.255 0.2006
0.26 0.1832
0.265 0.1782
0.27 0.1672
0.275 0.1702
0.28 0.1656
0.285 0.171
0.29 0.1714
0.295 0.1598
0.3 0.159
0.305 0.145
0.31 0.1484
0.315 0.1494
0.32 0.1342
0.325 0.1442
0.33 0.1356
0.335 0.1406
0.34 0.135
0.345 0.1394
0.35 0.1296
0.355 0.1216
0.36 0.1186
0.365 0.115
0.37 0.108
0.375 0.114
0.38 0.1132
0.385 0.1028
0.39 0.096
0.395 0.097
0.4 0.1114
0.405 0.099
0.41 0.1004
0.415 0.0982
0.42 0.0972
0.425 0.0948
0.43 0.0932
0.435 0.0922
0.44 0.0904
0.445 0.0956
0.45 0.0918
0.455 0.08
0.46 0.0924
0.465 0.086
0.47 0.0894
0.475 0.0824
0.48 0.0836
0.485 0.0836
0.49 0.0842
0.495 0.0756
0.5 0.0802
0.505 0.0806
0.51 0.0764
0.515 0.0746
0.52 0.0752
0.525 0.0746
0.53 0.0696
0.535 0.0774
0.54 0.075
0.545 0.0702
0.55 0.065
0.555 0.0698
0.56 0.0686
0.565 0.0666
0.57 0.0676
0.575 0.0684
0.58 0.0654
0.585 0.06
0.59 0.068
0.595 0.0592
0.6 0.0606
0.605 0.0742
0.61 0.0612
0.615 0.0606
0.62 0.072
0.625 0.063
0.63 0.0628
0.65 0.0574
0.67 0.057
0.690000000000001 0.0572
0.710000000000001 0.0548
0.730000000000001 0.0586
0.750000000000001 0.0564
0.770000000000001 0.05
0.790000000000001 0.053
0.810000000000001 0.0482
0.830000000000001 0.0556
0.850000000000001 0.0464
0.870000000000001 0.0454
0.890000000000001 0.047
0.910000000000001 0.0452
0.930000000000001 0.0442
0.950000000000001 0.0414
0.970000000000001 0.0362
0.990000000000001 0.0402
};
\addplot [semithick, darkolivegreen]
table {%
0 1
0.005 0.778
0.01 0.585
0.015 0.46
0.02 0.3678
0.025 0.2932
0.03 0.242
0.035 0.191
0.04 0.1612
0.045 0.14
0.05 0.1316
0.055 0.1098
0.06 0.0992
0.065 0.0972
0.07 0.0828
0.075 0.0798
0.08 0.0702
0.085 0.0716
0.09 0.0622
0.095 0.061
0.1 0.0626
0.105 0.0584
0.11 0.0548
0.115 0.0572
0.12 0.052
0.125 0.0526
0.13 0.0524
0.135 0.0494
0.14 0.0474
0.145 0.055
0.15 0.048
0.155 0.047
0.16 0.0516
0.165 0.0444
0.17 0.046
0.175 0.0462
0.18 0.0472
0.185 0.0478
0.19 0.0462
0.195 0.047
0.2 0.0454
0.205 0.0426
0.21 0.0466
0.215 0.0474
0.22 0.0418
0.225 0.049
0.23 0.0458
0.235 0.042
0.24 0.0424
0.245 0.0414
0.25 0.0468
0.255 0.0438
0.26 0.0404
0.265 0.0448
0.27 0.049
0.275 0.0418
0.28 0.0452
0.285 0.046
0.29 0.0434
0.295 0.0434
0.3 0.0458
0.305 0.0452
0.31 0.0426
0.315 0.0426
0.32 0.0464
0.325 0.0416
0.33 0.0416
0.335 0.0508
0.34 0.0464
0.345 0.0464
0.35 0.045
0.355 0.0478
0.36 0.0424
0.365 0.0508
0.37 0.049
0.375 0.0474
0.38 0.0506
0.385 0.05
0.39 0.0486
0.395 0.0548
0.4 0.0526
0.405 0.0482
0.41 0.0454
0.415 0.0498
0.42 0.0472
0.425 0.048
0.43 0.0524
0.435 0.0506
0.44 0.0514
0.445 0.0506
0.45 0.0494
0.455 0.0494
0.46 0.0512
0.465 0.0508
0.47 0.0564
0.475 0.053
0.48 0.0514
0.485 0.0512
0.49 0.0482
0.495 0.055
0.5 0.0596
0.505 0.0572
0.51 0.0556
0.515 0.0496
0.52 0.0494
0.525 0.0576
0.53 0.0528
0.535 0.0574
0.54 0.0608
0.545 0.0516
0.55 0.0536
0.555 0.055
0.56 0.0574
0.565 0.0592
0.57 0.0538
0.575 0.0548
0.58 0.0574
0.585 0.0592
0.59 0.059
0.595 0.0524
0.6 0.0578
0.605 0.0586
0.61 0.0562
0.615 0.058
0.62 0.0512
0.625 0.0584
0.63 0.053
0.65 0.0588
0.67 0.0556
0.690000000000001 0.0628
0.710000000000001 0.0526
0.730000000000001 0.0552
0.750000000000001 0.056
0.770000000000001 0.0546
0.790000000000001 0.0536
0.810000000000001 0.0512
0.830000000000001 0.05
0.850000000000001 0.0536
0.870000000000001 0.0486
0.890000000000001 0.049
0.910000000000001 0.048
0.930000000000001 0.0428
0.950000000000001 0.0382
0.970000000000001 0.0428
0.990000000000001 0.0374
};
\end{axis}

\end{tikzpicture}}\vspace*{-0.5cm}
		\caption{Stability probability of three pairs in instance sampled from a Euclidean model.}
		\label{fig:fine3}
	\end{minipage}%
\end{figure}

\paragraph{Further Experiments.}
In \Cref{app:var}, we analyze the dependence of our results on the number of agents, observing that the $50\%$-threshold decreases with an increasing number of agents. 
Lastly, in \Cref{app:uns}, we check whether matchings that are not initially stable can have a non-negligible stability probability in case some noise is applied to the instance. In most instances, we were unable to find such matchings. 

\subsection{Stable Pairs}\label{sec:pairs}\label{sub:exp2}
We also analyze the robustness of stable pairs in our instances.
Analgous as for matchings, we define the $50\%$-threshold of a pair as the smallest value of norm-$\phi$ so that the probability of the pair being stable is below $50\%$. 
Comparing the $50\%$-thresholds of initially stable matchings to the $50\%$-thresholds of initially stable pairs, it turns out that pairs are in general much more stable than matchings. In particular, pairs oftentimes have a $50\%$-threshold above norm-$\phi=0.08$: In \Cref{fig:avg-pairs-dist}, we show the distribution of instances' average $50\%$-thresholds computed by taking the average $50\%$-threshold of pairs initially stable in the instance. We observe that for $422$ out of our $544$ instances the average $50\%$-threshold is above norm-$\phi=0.08$. 
Note that this observation is quite intuitive, as creating a single blocking pair is sufficient to make a matching unstable. In contrast, a certain pair
in this matching can still continue to be stable in other stable matchings in the instance, as oftentimes not the full matching needs to be replaced to reestablish stability.

However, there are additional differences between pairs and matchings.  
Recall that we have argued above that different stable matchings in one instances (cf.\ \Cref{fig:fine2}) are often similarly robust to changes. In contrast, the difference between pairs can be more pronounced:
As an example, in \Cref{fig:fine3} we see how the stability probability of three pairs in an instance sampled from a Euclidean model develops when adding more and more noise to the preference list. 
It is easy to think of extremely stable pairs, for instance, pairs that rank each other on the first place and are ranked last by every other agent.
Our experiments indicate that such drastic examples do appear in our dataset.
We point to \Cref{app:exp2} for a more detailed description of our experiments on stable pairs.

\section{Conclusion}
We have conducted an algorithmic and experimental study of the average-case and worst-case robustness of stable matchings, pairs, and agents. For future work, regarding the complexity part, settling the complexity of \textsc{Pair-Delete-Robustness} is the most pressing open question. Moreover, it would also be interesting to examine the parameterized complexity of our problems. For instance, it is open whether our counting problems are \#W[1]-hard when parameterized by the examined distance. 
For future experimental work, it would be interesting to run the experiments on data from other sources to further confirm our findings of the non-robustness of stable matchings in practice. 

\bibliographystyle{plainnat}

\clearpage

\appendix 

\section{Additional Material for \Cref{ch:prelims}}\label{app:prelims}

\begin{proposition}\label{thm:wc_robust_bound}
	Let $\mathcal{I}=(U,W,\mathcal{P})$ be a \textsc{Stable Marriage}-instance with $|U|=|W|=n>1$ and let $M$ be a stable matching in $\mathcal{I}$. Then, $M$ can be made unstable by performing $n-1$ swaps (this upper bound is tight).
\end{proposition}
\begin{proof}
	For contradiction, suppose that there is an instance $\mathcal{I}$ that admits a stable matching $M$ which cannot be made unstable by performing at most $d$ swaps, where $d \geq n$. Since $M$ is stable, there exists an agent $a \in A$ that does not rank $M(a)$ last. Let $a'$ be the agent that directly follows $M(a)$ in the preference list of $a$. As the preference list of $a'$ ranks $n$ agents, it takes at most $n-1$ swaps to swap $a$ in front of $M(a')$ in the preferences of $a'$. Thus, by additionally swapping $a'$ and $M(a)$ in the preference list of $a$, we can create the blocking pair $\{a,a'\}$ with at most $n$ swaps, contradicting the $d$-robustness of $M$. Tightness of the upper bound follows from the Robust instance presented in the main body.
\end{proof}

The same argument also applies to stable pairs. 

\medskip

\paragraph{Computing Stability Probabilities}
To compute stability probabilities from the answers of our counting problems, we need to compute the number of instances at deletion and swap distance $\ell$ from a given instance.
We observe that for delete, this is trivial, as there are exactly $\binom{N}{\ell}$ ways to delete $\ell$ agents from a set of $N$ agents.

For swap, we can reduce the problem to counting the number of elections at some swap distance from a given one, which was shown to be polynomial-time solvable via dynamic programming by \citet{DBLP:conf/ijcai/BoehmerBFN21} (see Section B.2 of their arxiv version).  
Note that the preference profile of men can be viewed as an election with $U$ as voters and $W$ as candidates and that the preference profile of women can be viewed as an election with $W$ as voters and $U$ as candidates. 
By summing over all possible ways to split the swap budget onto the two elections, we can compute the number of instances at swap distance $\ell$. 

\section{Additional Material for \Cref{com:dec}}
\one*
\begin{proof}
	
	\noindent \textbf{Delete.}
	We assume without loss of generality that there are at least two women and that there are at least as many men as women. 
	There is a trivial solution for $\ell \geq 2$: We delete one matched woman $w \in W$ and one matched man $u \in U$ such that $\{u,w\} \not \in M$. Any stable matching $M'$ in the resulting instance will contain the edge $\{M(w),w'\}$ for some $w' \in W \setminus \{w\}$, as there are at least as many men as women and therefore each stable matching matches all men. As a results we have $M' \not \subseteq M$. If our budget is smaller than $2$, we first check whether $M$ is stable and accept if it is not. If $\ell=1$ and $M$ is stable, for all $a\in A$, we check whether there exists an agent $a' \in A$ who prefers $M(a)$ to $M(a')$. In that case, we can delete $a$, as $\{a',M(a)\}$ blocks $M':=M \setminus \{\{a,M(a)\}\}$, so we accept. Else, we decline. 
	
	\medskip
	\noindent \textbf{Swap.}
	Let $c_a(b,b')$ denote the number of swaps needed to bring $b$ in front of $b'$ in the preferences of $a$. Our algorithm first computes $c_a(b,M(a))$ for all $a \in A$ and $b$ of opposite gender and then determines the minimum number of swaps to make a matching $M$ unstable by the following expression: 
	\begin{equation*}
	\min_{\{u,w\} \in (U \times W) \setminus M} c_u(w,M(u))+c_w(u,M(w))
	\end{equation*}
	Over all pairs outside of the matching, we minimize the cost to make one of them blocking. 
\end{proof}

\two*
\begin{proof}
	\textbf{Delete.}
	We reduce from \textsc{Independent Set}, where we are given a graph $G=(V,E)$ and an integer $k$ and we want to decide whether there is a set $V' \subseteq V$ with $|V'| = k$ such that for all $u,v \in V'$ it holds that $\{u,v\} \not \in E$. We construct our \textsc{Blocking Pairs-Delete-Robustness} instance $\mathcal{I}'$ as follows:
	For each $v \in V$, we create a man $m_v$ and a woman $w_v$. Let $\text{deg}(v)$ denote the degree of $v$ and let $u_1,...,u_{\text{deg}(v)}$ be the neighbours of $v$. The preferences of the agents are as follows:
	\begin{align*}
	m_v&: w_v \succ ... &
	w_v: m_{u_1} \succ ... \succ m_{u_{\text{deg}(v)}} \succ ...
	\end{align*}
	Furthermore, we introduce men $m'_i$ and women $w'_i$ for each $i \in \{1,...,2k+|V|\}$. For any $i$, let $u^{i}_1,...,u^{i}_{k(i)}$ be the vertices with degree $< i$. For all $i \in \{1,...,|V|\}$, the preferences of $m'_i$ and $w'_i$ are the following:
	\begin{align*}
	m'_i&: w'_i \succ ... &
	w'_i: m_{u^{i}_1} \succ ... \succ m_{u^{i}_{k(i)}} \succ m'_i \succ ...
	\end{align*}
	For the remaining $2 \cdot 2k$ agents, i.e. $m'_i$ and $w'_i$ for all $i \in \{|V|+1,...,|V|+2k\}$, all $m'_i$ rank $w'_i$ as their top choice and all $w'_i$ rank all $m_v$ before $m'_i$ before all other agents. The rest of the preferences are arbitrary. Finally, we set $\ell:= k$, $b:= k(n+2k)$ and $M:=\{m_v,w_v \mid v \in V\} \cup \{m'_i,w'_i \mid i \in \{1,...,|V|+2k\}\}$. Notice that $M$ is stable since every man is matched to his top choice.
	
	($\Rightarrow$) Assume that $G$ contains an independent set $V'$ of size $k$. We claim that after deleting $w_v$ for all $v \in V'$, $b$ pairs are blocking for $M$. For each $v \in V'$, after deleting $w_v$, $m_v$ is unassigned in $M$. Therefore, each woman $w$ that prefers $m_v$ to $M(w)$ will form a blocking pair with $m_v$. By construction, this is the case for all $w_u$ with $\{v,u\} \in E$. Notice that since the deleted women correspond to an independent set, no woman $w_u$ is deleted in our solution and therefore all $\text{deg}(v)$ women form a blocking pair. Moreover, $(n-\text{deg}(v)) + 2k$ women $w'_i$ prefer $m_v$ to $M(w'_i)$ and since we do not delete any of these women in our solution, each of them forms a blocking pair with $m_v$. Altogether, $m_v$ is involved in $\text{deg}(v)+(n-\text{deg}(v))+2k=n+2k$ blocking pairs. As we deleted $k$ women, we have $k(n+2k)$ blocking pairs in total, fulfilling our required number of $b$ blocking pairs while deleting exactly $\ell$ agents.
	
	($\Leftarrow$) Assume that in $\mathcal{I}'$, we can delete at most $\ell$ agents such that $b$ pairs are blocking for $M$. First, consider the case that some agent $a \in A \setminus \{w_v \mid v \in V\}$ is deleted. No agent $a'$ of the same gender prefers $M(a)$ to $M(a')$. Thus, the only possibility for $M(a)$ to be contained in a blocking pair is to delete the matched partner $M(b)$ of another agent $b$ such that $\{M(a),b\}$ is blocking. As we can only delete $k-1$ more agents, the deletion of $a$ can produce at most $k-1$ blocking pairs. Deleting $w_v$ for any not yet deleted $v \in V$ instead of $a$ will increase the number of blocking pairs, as  $\{w'_i,m_v\}$ is blocking for all $i \in \{|V|+1,...,2k+|V|\}$ such that $m'_i$ is not deleted (at least $k+1$). Therefore, we can assume that only agents $w_v$ for some $v \in V$ are deleted.
	Let $D$ be the set of deleted women $w_v$. Consider a deleted woman $w_v \in D$. $M(w_v)=m_v$ can form a blocking pair with $w'_{\text{deg}(v)},...,w'_{n+2k}$ and with all $w_u$ such that $\{u,v\} \in E$, but not with any other woman, since all other women $w$ prefer $M(w)$ to $m_v$ and we assumed that $M(w)$ is not deleted. It follows that each $m_v$ for which $w_v$ was deleted can be part of at most $n+2k$ blocking pairs. To achieve the total number of $k(n+2k)$ blocking pairs, each of the $k$ men $m_v$ for which we deleted $w_v$ must be part of exactly $n+2k$ blocking pairs. Therefore, each $m_v$ must form a blocking pair with every $w_u$ such that $\{u,v\} \in E$. Consequently, no such $w_u$ can be deleted. It follows that $\{u,v\} \not \in E$ for every two deleted $w_u,w_v \in D$. By definition, $\{v \mid w_v \in D\}$ is an independent set of size $|D|=k$.
	
	\medskip
	\noindent \textbf{Swap.}
	We reduce from \textsc{Clique}, where we are given a graph $G$ and an integer $k \in \NN$ and we want to decide whether $G$ contains a clique of size $k$, i.e. whether there is a $V' \subseteq V$ with $|V'|=k$ and $\{v,u\} \in E$ for all $v,u \in V'$ with $v \neq u$. \textsc{Clique} is very well known to be NP-complete. We construct our \textsc{Blocking Pairs-Swap-Robustness} instance $\mathcal{I}'$ as follows: For each vertex, we introduce a man $m_v$ and a woman $w_v$. For each edge, we introduce a man $m_e$ and a woman $w_e$. Additionally, we set $p:=nk+\binom{k}{2}$, $q:=k(p+n)+\binom{k}{2}$ and we introduce $r:=k(p+n)+(q+2)\binom{k}{2}+1$ dummy men $m^d_i$ and dummy women $w^d_i$ for each $i \in \{1,...,r\}$. 
	Let $v \in V$ be a vertex with incident edges $e_1,...,e_s$. Let $e\in E$ be an edge with endpoints $v_1$ and $v_2$. The preferences are as follows:
	\begin{align*}
	m_v&: w_v \succ w^d_1 \succ ... \succ w^d_p \succ w_{e_1} \succ ... \succ w_{e_s} \succ \\
	& \hspace*{4.5cm} w^d_{p+1} \succ ... \succ w^d_{r} \succ ...  \\
	w_v&: m_v \succ m^d_1 \succ ... m^d_r \succ ... \\
	m_e&: w_e \succ w^d_1 \succ ... \succ  w^d_r \succ ... \\  
	w_e&: m_e \succ m^d_1 \succ ... \succ m^d_q \succ  m_{v_1} \succ m_{v_2} \succ\\ 
	& \hspace*{4.5cm} m^d_{q+1} \succ ... \succ m^d_{r} \succ ...  
	\end{align*}
	For each $i \in \{1,...,k\}$, the preferences of dummy agents are:
	\begin{align*}
	m^d_i&: w^d_i \succ w^d_{i+1} \succ ... \succ w^d_r \succ w^d_1 \succ ... \succ w^d_{i-1} \succ ...  \\
	w^d_i&: m^d_i \succ m^d_{i+1} \succ ... \succ m^d_r \succ m^d_1 \succ ... \succ m^d_{i-1} \succ ...
	\end{align*}
	We set the swap budget $\ell:=r-1=k(p+n)+(q+2)\binom{k}{2}$ and the number of blocking pairs $b:=2 \binom{k}{2}$. We set the designated matching to be: 
	\begin{align*}
	M:= & \{\{m_v,w_v\} \mid v \in V\} \cup \{\{m_e,w_e\} \mid e \in E\} \\
	& \cup \{\{m^d_i,w^d_i\} \mid i \in \{1,...,r\}\}
	\end{align*}
	
	($\Rightarrow$) Assume that there is a clique $C=(\{v_1,...,v_k\},\{e_1,...,e_{\binom{k}{2}}\})$ in $G$. For each $v \in V(C)$, we swap $w_v$ down by $p+n$ positions in the preferences of $m_v$. For each $e \in E(C)$, we swap down $m_e$ by $q+2$ positions in the preferences of $w_e$. Now, for each of the $2 \binom{k}{2}$ pairs $(v,\{v,w\})$ with $v,w \in V(C), e:=\{v,w\} \in E(C)$, we have that $m_v$ prefers $w_e$ to $w_v$ and $w_e$ prefers $m_v$ to $m_e$. It follows that $\{m_v,w_e\}$ is blocking. Notice that we used exactly $k(p+n)+(q+2)\binom{k}{2}$ swaps.
	
	($\Leftarrow$) Assume that in $\mathcal{I'}$, we can perform at most $k(p+n)+(q+2) \binom{k}{2}$ swaps such that $2 \binom{k}{2}$ pairs are blocking. First notice that no blocking pair can contain a dummy agent $a$: One needs at least $r$ swaps to swap $M(a)$ behind a non-dummy agent in the preferences of $a$ and one needs $r$ swaps to make two dummy agents with different index prefer each other over their matched agents, which exceeds the given swap budget. Moreover, no blocking pair can contain $w_v$ or $m_e$ for any $v \in V$ and $e \in E$, since in their preferences, one cannot swap their matched partners behind a non-dummy agent in less than $r$ swaps and no dummy-agent is contained in a blocking pair. It follows that all blocking pairs must be $\{m_v,w_e\}$ for some $v \in V$ and $e \in E$. Let $e=\{u,u'\}$. Every $w_e$ can only form blocking pairs with $m_u$ and $m_{u'}$, since it cannot form blocking pairs with dummy agents and for any other non-dummy agent $m$, we would need at least $r$ operations to swap $m_e$ behind $m$ in the preferences of $w_e$, which exceeds our budget.
	Therefore, at least $\binom{k}{2}$ different women $w_e$ must be contained in at least one blocking pair. On the other hand, for each $w_e$, we need at least $q+1$ swaps in the preferences of $w_e$ to make $w_e$ blocking with some $m_v$. Assume for contradiction that $\binom{k}{2}+1$ different women are contained in a blocking pair. We need 
	\begin{align*}
	&(\binom{k}{2}+1)(q+1) \\
	=&\binom{k}{2}(q+1)+q+1\\
	=&\binom{k}{2}(q+1)+k(p+n)+\binom{k}{2}+1\\
	=&\binom{k}{2}(q+2)+k(p+n)+1\\
	=&r
	\end{align*}
	swaps, exceeding our budget. Therefore, exactly $\binom{k}{2}$ women $w_e$ are contained in a blocking pair. Thus, we need to spend at least $(q+1)\binom{k}{2}$ swaps in the preferences of women and there are $k(p+n)+(q+2) \binom{k}{2} - (q+1)\binom{k}{2} = k(p+n) + \binom{k}{2}$ swaps left in the preferences of men. For any $m_v$ to be contained in a blocking pair, $w_v$ must be swapped down at least $p+1$ positions in the preferences of $m_v$ (all prior swaps involve dummy women). Suppose for contradiction that at least $k+1$ different men $m_v$ are contained in a blocking pair. In total, at least $(k+1)(p+1)$ swaps are needed. We have:
	\begin{align*}
	&(k+1)(p+1)\\
	=&k(p+1)+nk+\binom{k}{2}\\
	=&k(p+n+1)+\binom{k}{2}\\
	>&k(p+n) + \binom{k}{2}
	\end{align*}
	exceeding our remaining budget. Therefore, at most $k$ men $m_v$ can be contained in a blocking pair. However, every woman $w_e$ can only be matched to men corresponding to endpoints of $e$ (remember that we need at least $r$ swaps to make any other pair blocking). Let $E^*$ be the selected edges whose corresponding women are contained in a blocking pair. Since only $k$ men $m_v$ can be contained in a blocking pair, the endpoints of the selected edges are at most $k$ vertices, i.e. $|\bigcup_{e \in E^*} e| \leq k$. By definition, $(\bigcup_{e \in E^*} e,E^*)$ is a clique of size $k$ in $G$.
\end{proof}

\three*
\begin{proof}
	\citet{DBLP:conf/sagt/BoehmerBHN20} showed that the constructive variant of this problem, where we want to perform $\ell$ swaps in order to make a pair stable, is NP-hard. They also remark that with this result, one can show NP-hardness for the destructive variant of their problem. However, in their destructive variant, we are given a pair and a budget $\ell$ and the question is whether it is possible to perform at most $\ell$ swaps so that there is one stable matching where the designated pair is not contained. In contrast, our problem requires that the pair is not contained in any matching.
	We show that their reduction idea also works for our destructive variant. 
	
	To this end, we first consider so-called \emph{Add} operations, where one is allowed to add an agent from a predefined set of agents (such agents come with preferences and their positions in the preferences of all other agents are also already known). 
	We call the respective variant of our pair robustness problem \textsc{Pair-Add-Robustness}.
	We start by showing that \textsc{Pair-Add-Robustness} can be reduced to \textsc{Pair-Swap-Robustness}. Afterwards, we show NP-hardness for \textsc{Pair-Add-Robustness}, implying NP-hardness for \textsc{Pair-Swap-Robustness}.
	\begin{lemma}\label{thm:dpa-dps}
		One can reduce \textsc{Pair-Add-Robustness} to \textsc{Pair-Swap-Robustness} in polynomial time.
	\end{lemma}
	\begin{proof}
		Let  $U_{add} \cup W_{add}$ be the set from which we can add $\ell$ agents to the instance in the given \textsc{Pair-Add-Robustness} instance $\mathcal{I}$.
		Constructing the corresponding instance of \textsc{Pair-Swap-Robustness} we start with all agents including those from $U_{add} \cup W_{add}$ in $\mathcal{I}$. Our swap budget is $\ell$.
		Moreover, for every $a \in U_{add} \cup W_{add}$, we add two agents $a
		',a''$ to our construction where $a''$ has the same gender as $a$, but $a'$ is of opposite gender. We set the preferences of the new agents to be $a': a \succ a'' \succ ...$ and $a'': a' \succ ...$ and add $a'$ as the top choice of $a$.
		
		Moreover, following ideas from   \citet{DBLP:conf/sagt/BoehmerBHN20}, we introduce dummy men $m^d_i$ and dummy women $w^d_i$ for $i \in [q:=(n-1)(l+1)]$. For any $i$, their preferences are $m^d_i: w^d_i \succ w^d{i+1} \succ ... \succ w^d_q \succ w^d_1 \succ ... \succ w^d_{i-1} \succ ...$ and for the dummy women analogous. 
		Lemma 7 by \citet{DBLP:conf/sagt/BoehmerBHN20} showed that if we can perform at most $\ell$ swaps, all stable matchings must contain $\{\{m^d_1,w^d_1\},...,\{m^d_q,w^d_q\}\}$. 
		
		Following Lemma 8 by \citet{DBLP:conf/sagt/BoehmerBHN20}, in every preference list, between any two adjacent agents, we place $\ell+1$ dummy agents that have the corresponding dummy agent of opposite gender as their top choice. However, as the only exception, we do not place these dummy agents between $a$ and $a''$ in the preference list of any $a'$ where $a \in  U_{add} \cup W_{add}$. This way, any other swap will not influence the stability of the matching.
		
		Let $\mathcal{I}$ be an instance of \textsc{Pair-Add-Robustness} where there are $\ell$ agents $a \in U_{add} \cup W_{add}$ such that after their addition to the set of agents, no stable matching contains a designated pair $\{m^{**},w^*\}$. For an instance $\mathcal{I'}$ of \textsc{Pair-Swap-Robustness} with swap budget $\ell$ constructed as described above, for every added agent $a$, we swap $a$ and $a''$ in the preferences of $a'$. Then, every stable matching contains $\{a',a''\}$ and $a$ will be matched exactly as in $\mathcal{I}$ in all stable matchings. Every $a$ where this swap is not performed must be matched to $a'$ (mutual top-choice). Notice that all dummy agents can only be matched among themselves and therefore cannot form any blocking pair. It follows that after these swaps, no stable matching contains $\{m^*,w^*\}$ (if the opposite was the case, adding the corresponding agents in $\mathcal{I}$ would create a stable matching containing the designated pair). 
		
		For the other direction, notice that it is not possible to swap any two non-dummy agents $b$ and $c$ except for $a$ and $a''$ in the preferences of $a'$ for some agent $a \in U_{add} \cup W_{add}$, as there are $\ell+1$ dummy agents between $b$ and $c$, requiring $\ell + 1$ swaps, but our budget is only $\ell$. Assume that a solution contains some swaps involving dummy agents. Then, when omitting these swaps, $\{m^*,w^*\}$ is still not contained in any stable matching (the same non-dummy pairs are blocking for all matchings). Therefore, we can assume that all swaps are between $a$ and $a''$ in the preferences of some $a'$. With the same argument as for the other direction, swapping each of these swaps in $\mathcal{I'}$ is equivalent to adding the corresponding agent $a$ in $\mathcal{I}$. 
	\end{proof}
	
	We now show that \textsc{Pair-Add-Robustness} is NP-hard and by \Cref{thm:dpa-dps} the result follows. To this end, we reduce from \textsc{Constructive-Pair-Add}, where we are given a man-woman-pair $\{m^*,w^*\}$, a set of initially not added men and women $U_{add} \cup W_{add}$, an SM instance $\mathcal{I}=(U,W,\mathcal{P})$ and a budget $\ell \in \NN$ and want to find a set $X \subseteq U_{add} \cup W_{add}$ with $|X| \leq \ell$ such that after adding all agents from $X$ to our instance $\mathcal{I}$, the pair $\{m^*,w^*\}$ is contained in a stable matching. We assume that the given instance is constructed as for the reduction from \textsc{Clique} by Theorem 1 of \citet{DBLP:conf/sagt/BoehmerBHN20}; see their paper for the full, somewhat involved construction. The general idea of their construction is as follows: Man $m^*$ can only be matched with $w^*$ if every \emph{penalizing woman} is matched to an \emph{edge man}. There are $\binom{k}{2}$ penalizing women and each edge man can only be matched to a penalizing woman if one agent corresponding to the edge and two agents corresponding to the endpoints of the edge are added. Since the budget is $\binom{k}{2}+k$, the used edges must form a clique of size $k$. Notice that since $W_{add} = \emptyset$, all agents that can be added are men.
	
	We build the \textsc{Pair-Add-Robustness} instance as follows: We add a man $m^{**}$ with preferences $m^{**}: w^* \succ ...$ . We modify the top two preferences of $w^*$ to $m^* \succ m^{**}$. 
	Finally, we set our designated pair to be $\{m^{**},w^*\}$ and our budget $\ell':=\ell$.
	
	Let $\mathcal{I}$ be an SM instance where one can add $\ell$ agents $a \in U_{add}$ such that $\{m^*,w^*\}$ is contained in a stable matching $M$. Let $X$ be the set of these added agents. We claim that adding the same agents in $\mathcal{I}'$ results in an instance where no stable matching contains $\{m^{**},w^*\}$. 
	
	Assume for contradiction that after performing the same additions in $\mathcal{I}'$, there exists a stable matching $M'$ containing $\{m^{**},w^*\}$. Then, $\{m^*,w^*\}$ must not form a blocking pair in such a matching. As $w^*$ prefers $m^*$ to $m^{**}$, we need that $m^*$ prefers $M'(m^*)$ to $w^*$. Therefore, $M'(m^*)$ has to be a penalizing woman $w^\dagger$. Since $\binom{k}{2}$ edge men were selected in $\mathcal{I}$, one edge man $m_e$ which was selected is not matched to a penalizing woman. However, as $m_e'$ and $m_e''$ are mutual-top choices, $m_e$ prefers $w^\dagger$ to its partner and vice versa. It follows that the pair $\{w^\dagger, m_e\}$ is blocking, leading to a contradiction. 
	
	For the other direction, let $\mathcal{I'}$ be an SM instance where one can add at most $\ell'$ agents such that $\{m^{**},w^*\}$ is not contained in any stable matching. 
	As $\{m^{**},w^*\}$ is not contained in any stable matching $M'$, $M'$ must contain $\{m^*,w^*\}$, because if $w^*$ is matched to an agent $m \in U \setminus \{m^*,m^{**}\}$, the pair $\{m^{**},w^*\}$ is blocking. Thus, every penalizing woman $w^\dagger$ must be matched to an agent she prefers to $m^*$. Such an agent must be an edge man $a$. In $M$, matching every penalizing woman to such an agent and matching $m^*$ to $w^*$, is then obviously stable.  
\end{proof}

\four*
\begin{proof}
	We reduce the problem to the \textsc{Constructive-Pair-Delete} problem of 
	\citet{DBLP:conf/sagt/BoehmerBHN20}, which they showed to be solvable in  $\bigO((n+m)^2)$ time. For this, let $\mathcal{I}=(U,W,\mathcal{P})$  be an SM instance with $a \in A:=U \cup W$ and let $f(\mathcal{I},a)$ be an SM instance that contains the agents of $\mathcal{I}$ and one additional agent $a'$ of the opposite gender of $a$. All other agents rank $a'$ last and $a'$ ranks $a$ first and all other agents in an arbitrary order.
	
	Consider an instance of \textsc{Agent-Delete-Robustness} where we want to exclude $a$ from a stable matching. 
	
	Assume we can delete $\ell$ agents in $\mathcal{I}$ such that $a$ is unassigned in a stable matching $M$. Then, after deleting the same agents in $f(\mathcal{I},a)$, consider the matching $M':=M \cup \{a,a'\}$. Without loss of generality assume $a \in W$ and $a' \in U$. No pair $\{m,w\}$ with $m \neq a'$ can be blocking for $M'$, as the same pair would also be blocking for $M$. No blocking pair for $M'$ can contain $a'$, as $a'$ is matched to its top choice. It follows that $M'$ is stable.
	
	Assume now for the other direction that we can delete $\ell$ agents in $\mathcal{I}'$ such that $\{a,a'\}$ is contained in a stable matching $M'$. We claim that after performing the same deletions in $\mathcal{I}$, $M:=M' \setminus \{a,a'\}$ is stable. Using the same argument as above, a blocking pair for $M$ must include $a$. Assume for contradiction that such a blocking pair $\{a,a''\}$ exists for $M$. Then, $a''$ prefers $a$ to $M(a'')=M'(a)$. However, since $a$ prefers $a''$ to $a'$, $\{a,a''\}$ would be blocking in $M'$, leading to a contradiction.
\end{proof}

\five*
\begin{proof}
	We show the hardness via a series of reductions. 
	For this, we introduce two new constructive problems. 
	First, \textsc{Constructive-Agent-Swap} which is the constructive counterpart to \textsc{Agent-Swap-Robustness} where the question is whether there is an instance at swap distance at most $\ell$ where a given agent $a$ is stable.
	Second, \textsc{Constructive-Agent-Add} which is defined analogously to \textsc{Constructive-Agent-Swap} but instead of being allowed to perform $\ell$ swaps we can add up to $\ell$ agents from a predefined set (which come with specified preference lists and positions in the other agent's lists).
	\begin{lemma}
		\label{lem:caa-cas}
		One can reduce \textsc{Constructive-Agent-Add} to \textsc{Constructive-Agent-Swap} in polynomial time.
	\end{lemma}
	\begin{proof}
		We model the \emph{Add} operations by \emph{Swap} operations analogously to \Cref{thm:dps}: For each $a \in U_{add} \cup W_{add}$, we add an agent of opposite gender $a'$ and an agent $a''$ of the same gender as $a$, with preference lists $a': a \succ a'' \succ ...$ and $a'': a' \succ ...$ and add $a'$ as the top choice of $a$. We introduce $r:=(\ell+1)(|U|+|W|)$ dummy men $m_{d}^{i}$ and the same number of dummy women $w_{d}^{i}$ with preferences as in \Cref{thm:dps}, who will prevent undesired swaps. For every adjacent pair $a_u \succ a_v$ in a preference list of the original instance such that $u$ and $v$ are not the top two choices or that the preference list belongs to an agent who is not $a'$ for some $a \in U_{add} \cup W_{add}$, we add $\ell$ dummy agents of opposite gender $m_{i\ell},...,m_{(i+1)\ell}$ (or $w_{i\ell},...,w_{(i+1)\ell}$) between $u$ and $v$, where $u$ is ranked in $i$-th place in the original preference list. For the top two choices of an agent $a'$ with $ a \in U_{add} \cup W_{add}$, we do not add these dummy agents, thus allowing swaps between these two agents to model an \emph{Add} operation.
	\end{proof}
	\begin{lemma}\label{thm:cas}
		\textsc{Constructive-Agent-Swap} is NP-complete.
	\end{lemma}
	\begin{proof}
		We prove the NP-hardness of \textsc{Constructive-Agent-Add} and then \Cref{lem:caa-cas} implies the result. In \textsc{Constructive-Pair-Add}, we are given an SM instance $\mathcal{I}$ and a budget $\ell$ and the question is whether we can add at most $\ell$ agents from a predefined set $U_{add} \cup W_{add}$ such that a designated pair $\{m^*,w^*\}$ is contained in a stable matching. It was shown to be NP-complete (see Theorem 1 by \citet{DBLP:conf/sagt/BoehmerBHN20}. See \Cref{thm:dps} for the general idea of their construction. We assume that we are given an instance of \textsc{Constructive-Pair-Add}. In the new instance $\mathcal{I'}$, we start with all agents from the given instance (including the ones from the to be added sets). 
		Moreover, we add $|U|$ many women $w'_1,...,w'_{|U|}$ with arbitrary preferences. We modify the preferences as follows: All men except $m^*$ rank $w^*$ last and rank all $w'_i$ just in front of $w^*$. The man $m^*$ ranks all $w'_i$ last. All other preferences are not changed. We set the designated agent to be $w^*$. 
		
		Assume there are $p \leq \ell$ agents such that after adding them to the instance, $\{m^*,w^*\}$ is contained in a stable matching $M$. In $\mathcal{I}'$, when adding the same agents to the instance, matching $M$ is still stable, since no woman $w'_i$ can form a blocking pair (every man $m$ is assigned since $|U|\leq|W|$ and $m$ prefers $M(m)$ to $w'_i$). Obviously, $w^*$ is assigned in $M$.
		
		For the other direction, assume there are at most $\ell$ agents such that after adding them, $w^*$ is assigned in a stable matching $M'$. Consider instance $\mathcal{I}$ after adding the same set of agents. Notice that $w^*$ cannot be matched to any man $m \in U \setminus \{m^*\}$, as some $w'_i$ would be unassigned and would form a blocking pair with $m$. Thus $\{m^*,w^*\} \in M'$. Matching $M:=M' \setminus \{\{m,w'_i\} \mid i \in [|U|]\}$ is stable in $\mathcal{I}$: If $M'$ contains any pair $\{m,w'_i\}$, then deleting it will not affect the stability in $\mathcal{I}$, because the unassigned $m$ cannot form any blocking pair. Therefore, $\{m^*,w^*\} \in M$.
	\end{proof}
	
	We show NP-hardness of \textsc{Agent-Swap-Robustness} by reducing from the constructive version \textsc{Constructive-Agent-Swap}.  For an instance $\mathcal{I}$ of the constructive version, we add a man $m^*$ with preferences $m^*: w^* \succ ...$ who will be the designated agent and $d:=|W|-|U|$ men $\tilde{m}_1,...,\tilde{m}_d$ which will ensure that $m^*$ cannot be matched to any woman who was unassigned in $\mathcal{I}$. We modify the preferences of any $w \in W \setminus \{w^*\}$ to $w: ... \succ \tilde{m}_1 \succ ... \succ \tilde{m}_d \succ m^*$ where the relation between two men $m,m' \in U$ is not changed. We modify the preferences of $w^*$ to $w^*: ... \succ m^* \succ \tilde{m}_1 \succ \tilde{m}_d$ while again not changing the relation between any two men $m,m' \in U$. Finally, we introduce $(d+1)(\ell + 1)$ dummy men and the same number of dummy women which ensure that we cannot change the preference relation of a pair involving a newly added man $m \in U':=\{m^*,\tilde{m}_1,...,\tilde{m}_d\}$. To this end, between any two neighboring agents $m,m'$ with $m \in U' \setminus U$ in any preference list, we add $\ell+1$ dummy men. All dummy women are matched last by all non-dummy men. Dummy agents prefer each other as described in \Cref{thm:dps}.
	
	Assume that we can perform at most $\ell$ swaps in $\mathcal{I}$ such that $w^*$ is assigned in a stable matching $M$. We call the resulting instance $\mathcal{I}_\sigma$. We claim that after performing the same swaps in $\mathcal{I}'$ (resulting in $\mathcal{I}_\sigma'$), $m^*$ will not be assigned in any stable matching.
	
	Assume for contradiction that there is a stable matching $M'$ in $\mathcal{I}_\sigma'$ with $\{m^*,w\} \in M'$ for some $w \in W'$. Since our instance contains at least one more man than women, some man is unassigned. A dummy man is always matched with its partner and therefore,  there is some $i \in \{1,...,d\}$ such that $\tilde{m}_i$ is unassigned (if any other man was unassigned, he would form a blocking pair with the partner of $\tilde{m}_i$). If $w \neq w^*$, $\{\tilde{m}_i,w\}$ is blocking for $M'$. For the case $w=w^*$, in order for $\{M(w^*),w^*\}$ not to be blocking, he needs to be matched to some woman $w$ he prefers to $w^*$. The same holds for any man from the original instance $\mathcal{I}$. But then, $M' \setminus \{\{m,w\} \mid m \not \in U\}$ would be stable in $\mathcal{I}_\sigma$, which is a contradiction to the Rural Hospitals Theorem (since now, $w^*$ is unassigned). It follows that $M'$ is not stable, a contradiction.\\
	For the other direction, assume that we can perform at most $\ell$ swaps such that $m^*$ is not assigned in a stable matching $M'$. As argued in \Cref{thm:dps}, any swap involving dummy agents will not change the preference order of non-dummy agents and can therefore be omitted. Moreover, $w^*$ must be matched to a man she prefers to $m^*$, thus $\{m,w^*\} \in M'$ for some $m \in U$. We claim that after performing the same swaps in $\mathcal{I}$ as in $\mathcal{I}'$ (which is possible since we omitted all other swaps), $M:=M' \setminus \{\{m,w\} \mid m \not \in U\}$  is stable in $\mathcal{I}$. 
	
	Suppose for contradiction that there is a blocking pair $\{m',w'\}$ for $M$. If $w'$ is assigned in $M$, the same pair is also blocking $M'$, since $M(m')=M'(m')$, $M(w')=M'(w')$ and the preference relations between $m',M(m'),w'$ and $M(w')$ are not altered. If $w'$ is not assigned, then, in $M'$, it has to be assigned to $\tilde{m}_i$ for some $i \in \{1,...,d\}$. Again, $\{m',w'\}$ is blocking as $w'$ prefers $m'$ to $\tilde{m}_i$. Thus, we showed that $M'$ is unstable, which is a contradiction.
\end{proof}

\section{Additional Material for \Cref{com:count}}
\subsection{Theorem 4.1} \label{app:th41}

\begin{figure*}
	\centering
	\definecolor{darkgreen}{HTML}{0D9228}
\definecolor{darkred}{HTML}{C42323}
\definecolor{lightblue}{HTML}{23AFC4}
\begin{tikzpicture}[scale=.9,auto=center,every node/.style={circle,fill=lightblue},inner sep=2pt]
  \node[label=above:$m_u$] (m1) at (2,11) {};
  \node[color=gray!35] (m1c) at (2,10.5) {};
  \node[label=above:$w_x$] (w2) at (5,11) {};
  \node[label=above:$w_y$] (w3) at (5,8) {};
  \node[color=gray!35] (w2c) at (5,10.5) {};
  \node[color=gray!35] (w3c) at (5,7.5) {};
  \node[label=right:$w_{\bar{x}}$] (w2n) at (9,10) {};
  \node[label=right:$w_{\bar{y}}$] (w3n) at (9,7) {};
  \node[color=gray!35] (w2nc) at (9,9.5) {};
  \node[color=gray!35] (w3nc) at (9,6.5) {};
  \node[label=left:$m_{\bar{u}}$] (m1n) at (-2,10) {};
  \node[color=gray!35] (m1nc) at (-2,9.5) {};
  \node[label=above:$w_u$] (w1) at (0,12) {};
  \node[color=gray!35] (w1c) at (0,11.5) {};
  \node[label=above:$m_x$] (m2) at (7,12) {};
  \node[label=right:$m_y$] (m3) at (7,9) {};
  \node[color=gray!35] (m2c) at (7,11.5) {};
  \node[color=gray!35] (m3c) at (7,8.5) {};
  \node[label=above:$m_{\bar{x}}$] (m2n) at (7,10) {};
  \node[label=above:$m_{\bar{y}}$] (m3n) at (7,7) {};
  \node[color=gray!35] (m2nc) at (7,9.5) {};
  \node[color=gray!35] (m3nc) at (7,6.5) {};
  \node[label=above:$w_{\bar{u}}$] (w1n) at (0,10) {};
  \node[color=gray!35] (w1nc) at (0,9.5) {};
  \draw[color=darkred, very thick] (m1) --  (w2);
  \draw[color=darkred, very thick] (m1) --  (w3);
  \draw[color=darkred!80] (m1c) --  (w2c);
  \draw[color=darkred!80] (m1c) --  (w3c);
  \draw[color=darkgreen, very thick] (m1) -- (w1);
  \draw[color=darkred, very thick] (m1) -- (w1n);
  \draw[color=darkgreen, very thick] (w2) -- (m2);
  \draw[color=darkred, very thick] (w2) -- (m2n);
  \draw[color=darkgreen, very thick] (w3) -- (m3);
  \draw[color=darkred, very thick] (w3) -- (m3n);
  \draw[color=darkgreen, very thick] (w1n) -- (m1n);
  \draw[color=darkgreen, very thick] (w2n) -- (m2n);
  \draw[color=darkgreen, very thick] (w3n) -- (m3n);
  \draw[color=darkgreen!80] (m1c) -- (w1c);
  \draw[color=darkred!80] (m1c) -- (w1nc);
  \draw[color=darkred!80] (m1c) -- (w1n);
  \draw[color=darkgreen!80] (w2c) -- (m2c);
  \draw[color=darkred!80] (w2c) -- (m2nc);
  \draw[color=darkred!80] (w2c) -- (m2n);
  \draw[color=darkgreen!80] (w3c) -- (m3c);
  \draw[color=darkred!80] (w3c) -- (m3nc);
  \draw[color=darkred!80] (w3c) -- (m3n);
  \draw[color=darkgreen!80] (w1nc) -- (m1nc);
  \draw[color=darkgreen!80] (w2nc) -- (m2nc);
  \draw[color=darkgreen!80] (w3nc) -- (m3nc);
  \draw[color=darkred, very thick] (m1) -- (w1nc);
  \draw[color=darkred, very thick] (w2) -- (m2nc);
  \draw[color=darkred, very thick] (w3) -- (m3nc);
\end{tikzpicture}
	\caption{Exemplary SM instance constructed from the formula $(u \vee x) \wedge (u \vee y)$. The copied part of the instance is colored gray. Edges in the designated matching are colored green, while edges that are initially blocking are colored red. The dummy agents are omitted.}
	\label{fig:th41}
\end{figure*}
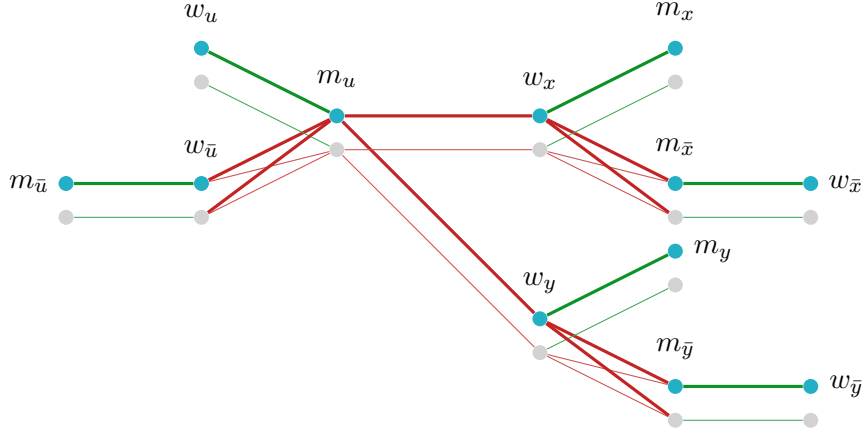
\six*
\begin{proof}
	We have already described the intuition behind the proof in the main body and give the full construction and proof of correctness here (see \Cref{fig:th41} for a visualization).
	For any literal $p$, we define the number of clauses where $p$ is contained as $c(p):=|\{c \mid c \in C, p \in c\}|$. Let $N=\max_{p \in L} c(p)$. 
	
	Let $\mathcal{I}=(V=U\cup Z,C)$ be an instance of \textsc{\#Bipartite 2-SAT With No Negations} and let $L := L_U \cup  L_Z$ be the set of all literals with $L_U := \{u, \bar{u} \mid u\in U\}$ and
	$L_Z := \{z, \bar{z} \mid z \in  Z\}$. The agents of our \textsc{\#Matching-Swap} instance $\mathcal{I}'=((U',W',\mathcal{P}), M, \ell)$ are as follows: 
	\begin{align*}
	U'&=\{m_v,m_v',m_{\bar{v}},m_{\bar{v}}' \mid v \in V\} \cup \{m_i^d \mid i \in N\}\\
	W'&=\{w_v,w_v',w_{\bar{v}},w_{\bar{v}}' \mid v \in V\} \cup \{w_i^d \mid i \in N\}
	\end{align*}
	For the preferences, consider a variable $u \in U$ and let $p_1,...,p_{c(u)} \in Z$ be all literals such that there exists some clause $c=\{u,p_i\}$ for each $i \in [c(u)]$. The preferences for the corresponding agents are as follows:
	\begin{align*}
	m_u&: w_{\bar{u}} \succ w'_{\bar{u}} \succ w_{p_1} \succ ... \succ w_{p_{c(u)}} \succ & \\
	& \hspace*{2.5cm} w_1^d \succ ... \succ w_{N-c(u)}^d \succ w_u \succ ... \\ 
	w_u&: m_u \succ ...& \\
	w_{\bar{u}}&: m_u \succ m'_u \succ m_1^d \succ ... \succ m_N^d \succ m_{\bar{u}} \succ ... &\\
	m_{\bar{u}}&: w_{\bar{u}} \succ ... &\\
	m'_u&: w_{\bar{u}} \succ w'_{\bar{u}} \succ w'_{p_1} \succ ... \succ w'_{p_{c(u)}} \succ \\
	& \hspace*{2.5cm} w_1^d \succ ... \succ w_{N-c(u)}^d \succ w'_u \succ ... \\
	w'_u&: m'_u \succ ...& \\
	w'_{\bar{u}}&: m_u \succ m'_u \succ m_1^d \succ ... \succ m_N^d \succ m'_{\bar{u}} \succ ... \\
	m'_{\bar{u}}&: w'_{\bar{u}} \succ ... 
	\end{align*}
	Now, let $z \in Z$. The preferences are very similar to agents that correspond to a vertex from $U$, but the roles of men and women are interchanged. Again, let $p_1,...,p_{c(z)} \in U$ be all literals such that there exists some clause $c=\{z,p_i\}$ for each $i \in [c(z)]$. The preferences for the corresponding agents are as follows:
	\begin{align*}
	w_z&: m_{\bar{z}} \succ m'_{\bar{z}} \succ m_{p_1} \succ ... \succ m_{p_{c(z)}} \succ\\
	& \hspace*{2.5cm} m_1^d \succ ... \succ m_{N-c(z)}^d \succ m_z \succ ... \\
	m_z&: w_z \succ ... \\
	m_{\bar{z}}&: w_z \succ w'_z \succ w_1^d \succ ... \succ w_N^d \succ w_{\bar{z}} \succ ... \\
	w_{\bar{z}}&: m_{\bar{z}} \succ ... \\
	w'_z&: m_{\bar{z}} \succ m'_{\bar{z}} \succ m'_{p_1} \succ ... \succ m'_{p_{c(z)}} \succ\\
	& \hspace*{2.5cm} m_1^d \succ ... \succ m_{N-c(z)}^d \succ m'_z \succ ... \\
	m'_z&: w'_z \succ ...\\
	m'_{\bar{z}}&: w_z \succ w'_z \succ w_1^d \succ ... \succ w_N^d \succ  w'_{\bar{z}} \succ ... \\
	w'_{\bar{z}}&: m'_{\bar{z}} \succ ... 
	\end{align*}
	
	The dummy men have preferences $m_i^d: w_i^d \succ ...$ and the dummy women have preferences $w_i^d: m_i^d \succ ...$ for $i \in [N]$. Our designated matching is 
	$
	M:=\{\{m_v,w_v\},\{m_{\bar{v}},w_{\bar{v}}\},\{m'_v,w'_v\},\{m'_{\bar{v}},w'_{\bar{v}}\} \mid v \in V\} \cup \{\{m_i^d,w_i^d\} \mid i \in [N]\}.
	$
	Notice that all agents $m_u,w_{\bar{u}},m_u',w_{\bar{u}}'$ for $u \in U$ and all agents $w_z,m_{\bar{z}},w_z',m_{\bar{z}}'$ for $z \in Z$ rank their matched partner on place $(N+3)$. Moreover, $M$ is not stable in the constructed instance: The set of blocking pairs is 
	\begin{align*}
	B:=&\{\{m_p,w_{\bar{p}}\}, \{m_p,w'_{\bar{p}}\},\{m'_p,w'_{\bar{p}}\}, \{m'_p,w_{\bar{p}}\} \mid p \in L\}\\
	\cup &\{\{m_p,w_q\},\{m'_p,w'_q\} \mid \{p,q\} \in C\}
	\end{align*} 
	
	Note that $p$ can both be a positive and negative literal. One can easily verify that these pairs are indeed blocking. We now argue why there are no other blocking pairs in $\mathcal{I}$: For a blocking pair $\{m,w\}$, it must hold that $w \succ_m M(m)$. Let $m \in U'$. For all $w \in W' \setminus \{w_i^d \mid i \in [N]\}$ that fulfill this condition, we already have that $\{m,w\} \in B$. But since $w_i^d$ is matched to her top choice for all $i \in [N]$, $\{m,w_i^d\}$ is not blocking.
	
	We set our swap budget to $\ell:=2|V|(N+2)$ and claim that the number of satisfying assignments for the given \textsc{2-SAT} formula is equal to the number of preference profiles at swap distance exactly $\ell$ from $\mathcal{I}'$ in which $M$ is stable. 
	
	To this end, we define a function $f$ which maps valid assignments of the \textsc{2-SAT} formula to valid preference profiles in $\mathcal{I}'$. Let $\mathcal{P}$ be the preference profile in $\mathcal{I}'$. For an assignment $A$, we define $f(A):=\mathcal{P}'$ as follows: Starting from $\mathcal{P}$, for each $p \in L_U$, we swap $w_p$ to the front of the preferences of $m_p$ and we swap $w_p'$ to the front of the preferences of $m_p'$. For each $p \in L_Z$, we swap $m_p$ to the front of the preferences of $w_p$ and we swap $m_p'$ to the front of the preferences of $w_p'$.
	
	We now proceed to show that $f$ describes a one-to-one correspondence between satisfying assignments and preference profiles with the desired properties, which implies that their number is equal.
	\begin{lemma}\label{lemma:injective}
		$f$ is injective.
	\end{lemma}
	\begin{proof}
		Consider two different valid assignments $A \neq B$. There is at least one literal $l$ with $l \in A$ and $l \not \in B$. Without loss of generality, let $l \in L_U$. Now, consider the two constructed instances $f(A)$ and $f(B)$. In $f(A)$, the preferences of $m_l$ are altered, while in $f(B)$ they are not. It follows that $f(A) \neq f(B)$ and thus $f$ is injective.
	\end{proof}
	\begin{lemma}\label{lemma:formula_to_prefprofile}
		Let $A$ be a satisfying assignment of the \textsc{2-SAT} formula $\mathcal{I}$. Then, $f(A)$ is a preference profile at swap distance exactly $\ell$ from $\mathcal{I}'$ where the designated matching $M$ of $\mathcal{I}'$ is stable.
	\end{lemma}
	\begin{proof}
		First notice that we only improve $M(a)$ in the preferences of any agent $a$ and thus no new blocking pairs can arise. As $A$ is a valid assignment, and thus contains one literal for each variable, for any blocking pair $\{m_p,w_{\bar{p}}\}$ with $p \in L$, we swap $M(m_p)$ to the front of $m_p$'s preferences or $M(w_{\bar{p}})$ to the front of $w_{\bar{p}}$'s preferences, resolving the blocking pair. The same holds for the blocking pairs of the duplicated part. As $A$ is a satisfying assignment, and thus contains one literal for each clause, for any blocking pair $\{m_p,w_q\}$ with $\{p,q\} \in C$, we swap $M(m_p)$ to the front of $m_p$'s preferences or $M(w_q)$ to the front of $w_q$'s preferences, resolving the blocking pair. Again, the same holds for the blocking pairs of the duplicated part. We argued above that no other pairs are blocking in the original instance and thus it follows that $M$ is stable. Notice that to swap $M(a)$ and $M(a')$ to the first place in the preferences of $a$ and $a'$, we need $2(N+2)$ swaps and since any satisfying assignment contains exactly $|V|$ literals, the overall budget is exactly matched.
	\end{proof}
	
	\begin{lemma}\label{lemma:help_swap_distance}
		Let $L:=a: b_1 \succ b_2 \succ ... \succ b_i \succ b^* \succ b_{i+1} \succ ... \succ b_n$ be a preference list. $L':=a: b^* \succ b_1 \succ b_2 \succ ... \succ b_n$ is the only preference list at swap distance $i$ to $L$ where $a$ prefers $b^*$ to $b_1$ and to $b_2$. 
	\end{lemma}
	\begin{proof}
		Since there are $i-1$ agents between $b_1$ and $b^*$, one needs at least $i$ swaps to reach a preference list where $b^*$ is preferred to $b_1$. Clearly, $L'$ has swap distance $i$ to $L$, as we only improve $b^*$ by $i$ positions. Now consider a preference list $L''$ at swap distance $i$ to $L$ where $b^*$ is preferred to $b_1$ and $b_2$. With each swap, we have to improve $b^*$ or swap down $b_1$ or $b_2$ (otherwise $b^*$ cannot be preferred to $b_1$). Assume that we swap down $b_1$ or $b_2$ in at least one step. If we swap down $b_2$, the distance between $b_1$ and $b^*$ is still $i$, but we only have $i-1$ swaps left, a contradiction. If we swap down $b_1$, then $b_2$ is improved with the same swap and the distance between $b_2$ and $b^*$ is $i$. Again, since there are only $i-1$ swaps left, this is not possible. It follows that we do not swap down $b_1$ or $b_2$ in $L''$, and thus $L''=L'$.
	\end{proof}
	
	\begin{lemma}\label{lemma:prefprofile_to_formula}
		Let $\mathcal{P}'$ be a preference profile at swap distance at most $\ell$ from $\mathcal{I}'$ such that $M$ is stable. There exists a satisfying assignment $A$ with $f(A)=\mathcal{P}'$.
	\end{lemma}
	\begin{proof}
		To resolve the blocking pairs $\{m_p,w_{\bar{p}}\}, \{m_p,w'_{\bar{p}}\},\{m'_p,w'_{\bar{p}}\}$ and $ \{m'_p,w_{\bar{p}}\}$ for any $p \in L$, we need to modify $m_p$ and its duplicated version $m_p'$ or we need to modify $w_{\bar{p}}$ and its duplicated version $w_{\bar{p}}'$, that is, swap $M(m_p)$ and $M(m_p')$ in front of $w_{\bar{p}}$ and $w_{\bar{p}}'$ in the preferences of $m_p$ and $m_p'$, respectively, or swap $M(w_{\bar{p}})$ and $M(w_{\bar{p}}')$ in front of $m_p$ and $m_p'$ in the preferences of $w_{\bar{p}}$ and $w_{\bar{p}}'$, respectively.  Clearly, we need at least $N+2$ swaps to perform such a modification in one preference list. Since we need to modify $2|V|$ preference lists and our total budget is $2|V|(N+2)$, each modification of a preference list can only take $N+2$ swaps. By \Cref{lemma:help_swap_distance}, the only way to achieve this is to swap the matched agent to the first place. Since we exhausted the complete swap budget, no additional swaps are possible.
		
		As we must modify $m_u$ or $m_{\bar{u}}$ and $w_v$ or $w_{\bar{v}}$ but can only modify at most $2|V|$ agents, the assignment $A:=\{p \in L \mid m_p \text{ or } w_p \text{ was modified}\}$ is valid. Clearly, $f(A)=P'$. Since $M$ is stable in $P'$, no pair corresponding to a clause in $\mathcal{I}$ can be blocking. It follows that for each clause $c=\{u,v\} \in C$, either $m_u$ or $w_v$ must have been modified, as otherwise $\{m_u,w_v\}$ is blocking. Therefore, for each clause $c \in C$, there is a literal $p \in A \cap c$ and by definition, $A$ is a satisfying assignment.
	\end{proof}
	\Cref{lemma:injective}, \Cref{lemma:formula_to_prefprofile} and \Cref{lemma:prefprofile_to_formula} together show that $f$ describes a one-to-one correspondence between satisfying assignments in $\mathcal{I}$ and preference profiles with swap distance $\ell$ to $\mathcal{I}'$ where $M$ is stable.
	
	The problem is contained in \#P, since we can construct a NTM that nondeterministically chooses a preference profile with swap distance exactly $\ell$ and checks whether $M$ is stable. Then, the number of accepting branches equals the number of solutions to the problem.
\end{proof}

\subsection{Theorem 4.2}

\seven*
\begin{proof}
	\noindent\textbf{Agent.}
	For the sake of completeness, we repeat the construction which we intuitively introduced and discussed in the main body.  
	Let $\mathcal{I}_{i,\ell}^G$ be a \textsc{\#Agent-Delete} instance constructed from a given graph $G$ as follows (see \Cref{fig:adconstr} for a visualization). The agents consist of \emph{edge men} $m_e$, \emph{vertex men} $m_v^{q}$, \emph{extra men} $m_p^*$ and \emph{vertex women} $w_v$. More formally, the sets of men and women are:\\
	\begin{align*}
	U=&\{m_e \mid e \in E\} \cup \{m_v^{q} \mid v \in V, q \in [N]\} \cup \{m_p^* \mid p \in [i]\}\\
	W=&\{w_v \mid v \in V\}
	\end{align*}
	where $N:=m+1$. Let $e=\{u,v\} \in E$. The preferences of the edge men are the following:
	\begin{align*}
	m_e: w_u \succ w_v \succ ...
	\end{align*}
	Let $v \in V$ and $e_1,...,e_{\text{deg}(v)}$ be the edges incident to $v$. Let $q \in [N]$ and $p \in [i]$. The preferences for the remaining agents are:
	\begin{align*}
	m_v^{q}&: w_v \succ ...\\
	w_v&: m_v^1 \succ ... \succ m_v^N \succ m_{e_1} \succ ... \\
	& \hspace*{2.2cm}\succ m_{e_{\text{deg}(v)}} \succ m_1^* \succ ... \succ m_i^* \succ ...\\
	m_p^*&: ...
	\end{align*}
	Our deletion budget is $\ell$ and will in the most cases be equal to $\ell(i,j):=j+i \cdot N$, where $j$ is the size of edge men that we want to delete. We set the designated agent to be $m_i^*$.\\
	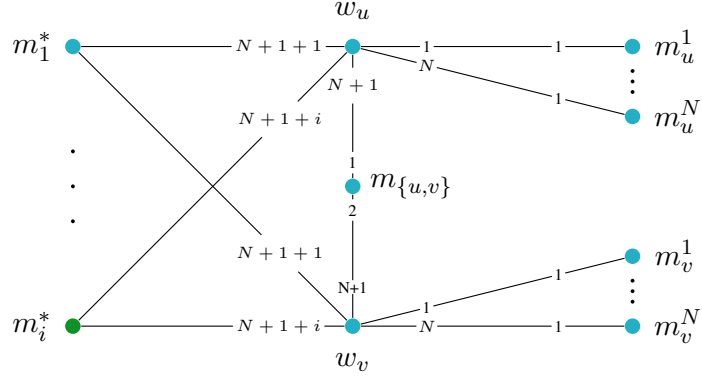
\begin{figure}
		\centering
		\resizebox{0.6\textwidth}{!}{\definecolor{darkgreen}{HTML}{0D9228}
\definecolor{darkred}{HTML}{C42323}
\definecolor{lightblue}{HTML}{23AFC4}
\begin{tikzpicture}[scale=.9,auto=center,every node/.style={circle,fill=lightblue}, inner sep = 2pt]
    \node[label=above:$w_u$] (wu) at (4,9) {};
    \node[label=below:$w_v$] (wv) at (4,5) {};
    \node[label=right:$m_u^1$] (mu1) at (8,9) {};
    \node[label=right:$m_v^1$] (mv1) at (8,6) {};
    \node[label=right:$m_u^N$] (muN) at (8,8) {};
    \node[label=right:$m_v^N$] (mvN) at (8,5) {};
    \node[label=right:$m_{\{u,v\}}$] (muv) at (4,7) {};
    \node[label=left:$m_1^*$] (ms1) at (0,9) {};
    \node[label=left:$m_i^*$,fill=darkgreen] (msi) at (0,5) {};
    \node[inner sep=0.5pt,fill=black] (1) at (8,8.65) {};
    \node[inner sep=0.5pt,fill=black] (2) at (8,8.5) {};
    \node[inner sep=0.5pt,fill=black] (3) at (8,8.35) {};
    \node[inner sep=0.5pt,fill=black] (4) at (8,5.65) {};
    \node[inner sep=0.5pt,fill=black] (5) at (8,5.5) {};
    \node[inner sep=0.5pt,fill=black] (6) at (8,5.35) {};
    \node[inner sep=0.5pt,fill=black] (7) at (0,7.5) {};
    \node[inner sep=0.5pt,fill=black] (8) at (0,7) {};
    \node[inner sep=0.5pt,fill=black] (9) at (0,6.5) {};
    \tiny
    \draw (muv) -- (wu) node [very near start, fill=white, inner sep=1pt] {1} node [near end, fill=white, inner sep=-3pt] {$N+1$};
    \draw (muv) -- (wv) node [very near start, fill=white, inner sep=1pt] {2} node [near end, fill=white, inner sep=-3pt] {N+1};
    \draw (wu) -- (mu1) node [near start, fill=white, inner sep=0pt] {1} node [near end, fill=white, inner sep=0pt] {1};
    \draw (wu) -- (muN) node [near start, fill=white, inner sep=0pt] {$N$} node [near end, fill=white, inner sep=0pt] {1};
    \draw (wv) -- (mv1) node [near start, fill=white, inner sep=0pt] {1} node [near end, fill=white, inner sep=0pt] {1};
    \draw (wv) -- (mvN) node [near start, fill=white, inner sep=0pt] {$N$} node [near end, fill=white, inner sep=0pt] {1};
    \draw (wu) -- (ms1) node [near start, fill=white, inner sep=1pt] {$N+1+1$};
    \draw (wu) -- (msi) node [near start, fill=white, inner sep=-10pt] {$N+1+i$};
    \draw (wv) -- (ms1) node [near start, fill=white, inner sep=-10pt] {$N+1+1$};
    \draw (wv) -- (msi) node [near start, fill=white, inner sep=1pt] {$N+1+i$};

\end{tikzpicture}}
		\caption{One edge and two vertices from an SM instance $\mathcal{I}_{i,\ell}^G$. The designated agent $m_i^*$ is coloured green. The edge labels denote the rank of the distant agent of the edge in the preferences of the close agent of the edge. }
		\label{fig:adconstr}
	\end{figure} 
	
	We first define two helpful notions:
	\begin{definition} 
		Consider the \textsc{\#Agent-Delete} instance $\mathcal{I}_{i,\ell(i,j)}^G$. We define $D_{i,j}$ as the number of solutions to $\mathcal{I}_{i,\ell(i,j)}^G$ that contain at least one extra man $m_p^*$ for some $p \in [i]$.
	\end{definition}
	\begin{definition}
		For a graph $G=(V,E)$, an \emph{$i$-vertex-isolating set} is an edge set $E' \subseteq E$ such that after deleting all edges in $E'$, there are exactly $i$ vertices that have no incident edges. We denote by $\mathcal{E}_i^j$ the set of all $i$-vertex-isolating sets of size $j$.
	\end{definition}
	
	Notice that for any graph $G=(V,E)$, $C \subseteq E$ is an edge cover if and only if there is no $i \in \{1,...,n\}$ such that $E \setminus C$ is an $i$-vertex-isolating set. We will now introduce the most important lemma of the proof, which will establish a connection between vertex-isolating sets and Agent-Delete solutions.
	\begin{lemma}\label{lemma:agent_delete_char}
		Let $G=(V,E)$ be a graph and $i \in [n], j \in [m]$. The following equation holds:
		\begin{equation*}
		\text{\textsc{\#AD}}(\mathcal{I}_{i,\ell(i,j)}^G)=\sum_{j'=0}^j{\sum_{i'=i}^{n}({\binom{(N+1)(n-i)}{j-j'} \cdot \binom{i'}{i}|\mathcal{E}_{i'}^{j'}|}}) + D_{i,j}
		\end{equation*}
	\end{lemma}
	\begin{proof}
		We first only consider solutions not containing any extra men. Clearly, if the number of such solutions is $\sum_{j'=0}^j{\sum_{i'=i}^{n}{\binom{(N+1)(n-i)}{j-j'} \cdot \binom{i'}{i}|\mathcal{E}_{i'}^{j'}|}}$, then adding $D_{i,j}$ will give the correct total number of solutions.
		
		As an additional restriction, we first only count solutions that satisfy the following constraints:
		\begin{enumerate}
			\item the edges corresponding to the deleted edge men form a $i'$-vertex isolating set of size $j'$, where $i' \geq i$ and $j' \leq j$.
			\item There are exactly $i$ vertices $v$ such that $m_v^q$ is deleted for all $q \in [N]$. We call such a vertex \emph{selected}. Notice that our budget $\ell=j+i \cdot N$ does not suffice to select more than $i$ vertices, since $j<N$.
			\item No woman $w_v$ is deleted such that $v$ is selected.
		\end{enumerate}
		We first show that the number of solutions that fulfil these constraints is exactly the total number of solutions as stated by the theorem minus $D_{i,j}$. Let $E'$ be an $i'$-vertex-isolating set of size $j'$ in $G$ and $i \leq i' \leq n$, $j' \leq j$. After deleting all edge men $m_e$ with $e \in E'$, the remaining deletion budget is $j-j'+i \cdot N$. Let $V' \subseteq V \setminus \bigcup_{e \in E \setminus E'}e$ with $|V'|=i$ be a subset of $i$ vertices that are isolated in $(V,E \setminus E')$. Clearly, there are $\binom{i'}{i}$ such subsets. After deleting $m_v^{q}$ for all $v \in V'$ and $q \in [N]$, there are $i$ women who prefer $m_p^*$ for any $p \in [i]$ to all other men, because we deleted all men that corresponded to any $v \in V'$ and to any edge incident to any $v \in V'$. It follows that $m_i^*$ must be assigned. Finally, we can delete $j-j'$ arbitrary vertex agents (except any agent $w_v$ with $v \in V'$) to reach the required deletion budget. There are $N \cdot (|V|-i)$ not yet deleted vertex men and $|V|-i$ vertex women that we can delete. Notice that we cannot delete more edge men since we count all solutions corresponding to a certain vertex-isolating set separately. Thus, for each $i'$-vertex-isolating set of size $j'$, there exist $\binom{(N+1)(n-i)}{j-j'} \cdot\binom{i'}{i}$ different sets of agents such that after their deletion the designated agent $m_i^*$ is assigned. It follows that there are at least $\sum_{j'=0}^j{\sum_{i'=i}^{n}{\binom{(N+1)(n-i)}{j-j'} \cdot \binom{i'}{i}|\mathcal{E}_{i'}^{j'}|}}$ many different sets that do not contain extra men and whose deletion ensure that $m_i^*$ is assigned. Clearly, all solutions that we counted until now are different, since we delete different edge men, vertex men or vertex women.
		
		Note that all solutions that we counted until now consisted of $j'$ edge men with $j' \leq j$, whose corresponding edges form a $i'$-vertex isolating set of size $j'$ with $i' \geq i$, $i \cdot N$ vertex men $m_v^q$ that all belong to the same vertices and $j-j'$ arbitrary vertex agents except $w_v$ for a selected vertex $v$. 
		
		We now argue that there are no solutions that violate the above constraints. To this end, for each violated constraint, we will show that at most $i-1$ vertex women prefer extra man to any not-deleted vertex man. From this it follows that one of the $i$ extra men is not matched to a vertex woman, and since all vertex women prefer $m_j^*$ to $m_i^*$ for any $j<i$, $m_i^*$ is not assigned.  \begin{enumerate}
			
			\item Assume that the deleted edge men do not correspond to an $i'$-vertex isolating set of size $j'$ in $G$ for some $i' \geq i$, $j'\leq j$. If $j'>j$, the remaining budget does not suffice to delete $i \cdot N$ vertex men and we are in case 2. Thus, we can assume $i'<i$, i.e. less than $i$ vertices are isolated. Thus, for each set of $i$ vertex women, at least one of them will prefer an edge man to each extra man. Assume that some extra man $m^*$ is matched to this vertex woman $w_v$ in a stable matching and let $e$ be an edge such that $v \in e$ and $m_e$ is not deleted. In order for $\{m_e,w_v\}$ not to be blocking, $m_e$ must be matched to another vertex woman $w_u$. It follows that $m_u^q$ was deleted for all $q \in [N]$ (else, $w_u$ would be matched to some $m_u^q$) and that $w_u$ is not matched to any extra man. Because of our limited budget, we cannot delete all vertex men for more than $i$ vertices. It follows that there are only $i-1$ vertex women that can be matched to extra men (since $u$ is one of the selected vertices but $w_u$ is not matched to any extra man). We can conclude that $m_i^*$ is not assigned. We can repeat the same argument for the case where $w_v$ is selected but not matched to any extra man.
			\item Assume we deleted less than $i \cdot N$ vertex men. For each set of $i$ vertex women, at least one of them will prefer a vertex man to each extra man. Thus, $m_i^*$ cannot be assigned.
			\item Assume that we deleted some $w_v$ such that $v$ was selected, i.e. all $m_v^q$ for $q \in [N]$ were deleted. Because of our limited budget, we cannot select more than $i$ vertices. Since $w_v$ is deleted, there are only $i-1$ not deleted women that were selected and therefore do not prefer a vertex man to any extra man. Again, it follows that $m_i^*$ cannot be assigned.
		\end{enumerate} 
		Thus, there are no other solutions and adding all solutions involving the deletion of extra men, we obtain $\text{\textsc{\#AD}}(\mathcal{I}_{i,\ell(i,j)}^G)=\sum_{j'=0}^j{\sum_{i'=i}^{n}{\binom{(N+1)(n-i)}{j-j'} \cdot \binom{i'}{i}|\mathcal{E}_{i'}^{j'}|}} + D_{i',j'}$.
		
	\end{proof}
	With the help of \Cref{lemma:agent_delete_char}, we now want to develop a polynomial-time algorithm that, given an oracle for \textsc{\#Agent-Delete}, computes \textsc{\#Edge Cover}. To this end, we first show that we can compute $D_{i,j}$ in polynomial time via dynamic programming.
	\begin{lemma}\label{lemma:agent_delete_d_table_poly}
		Given an oracle that solves \textsc{\#Agent-Delete}, we can compute $D_{i,j}$ in polynomial time for any $i \in [n], j \in [m]$.
	\end{lemma}
	\begin{proof}
		We will compute $D_{i,j}$ via a dynamic programming table $R_{i,j}[k]$, which will denote the number of solutions to \textsc{\#Agent-Delete} for $\mathcal{I}_{i,\ell(i,j)}^G$ where the first $k$ extra men (and no other extra men) are deleted. We claim that for each $k \in [i]$ we can compute the entries as follows:
		\begin{equation*}
		\begin{aligned}
		R_{i,j}[i-1]&=\text{\textsc{\#AD}}(\mathcal{I}_{1,j,j+iN-(i-1)}^G)\\
		R_{i,j}[k]&=\text{\textsc{\#AD}}(\mathcal{I}_{i-k,j,j+iN-k}^G)\\
		& \hspace*{0.2cm} -\sum_{r=1}^{i-1-k}{\binom{i-1-k}{r}R[k+r]}, \phantom{a} \forall k<i-1
		\end{aligned}
		\end{equation*}
		We show by induction over the number of not deleted extra men $i-k$ that $R_{i,j}[k]$ is correctly computed. For $i-k=1$ (and thus $k=i-1$), all extra men except for $m_i^*$ must be deleted (note that by definition, $m_i^*$ cannot be deleted). Therefore, we can reduce our instance by deleting all extra men except for one and subtracting $i-1$ from the budget. The resulting instance is exactly $\mathcal{I}_{1,j,j+iN-(i-1)}^G$.
		
		Now, for any $k$, let $R_{i,j}[k']$ be correctly computed for all $k' \in \{k+1,...,i-1\}$. From this we can conclude that $R_{i,j}[k]$ is correctly computed: Since we must not count the solutions where more extra men are deleted (as we only are interested in the solutions with exactly $k$ deleted extra men), we subtract these solutions from the result of our oracle call. Consider a solution for $\mathcal{I}_{i-k,j,j+iN-k}^G$ that deletes $r$ extra men. By induction hypothesis, the total number of solutions where the first $r$ extra men are deleted is $R_{i,j}[k+r]$. If a solution with $r$ deleted extra men exists, then exchanging these $r$ extra men with arbitrary different $r$ extra men is still a solution. It follows that for each solution where the first $r$ extra men were deleted, we have $\binom{i-1-k}{r}$ ways to choose the extra men. Consequently, the total number of solutions where $r$ extra men were deleted are $\binom{i-1-k}{r}R_{i,j}[k+r]$. We must subtract this number for each number $r$ of extra men in $\mathcal{I}_{i,j,j+iN-k}^G$ and thus we get $R[k]=\text{\textsc{\#AD}}(\mathcal{I}_{i-k,j,j+iN-k}^G)-\sum_{r=1}^{i-1-k}{\binom{i-1-k}{r}R[k+r]}$.
		
		Now, we can compute the total number of solutions that contain extra men, which is
		\begin{equation*}
		D_{i,j}=\sum_{k=1}^{i-1}{\binom{i-1}{k}R[k]}
		\end{equation*}
		because there are $\binom{i-1}{k}$ ways to choose $k$ out of $i-1$ extra men to delete.
	\end{proof} 
	We now move on to computing the number of $i$-vertex-isolating sets of size $j$ for any $i \in [n]$, $j \in [m]$. For this, we will fill a dynamic programming table $T[i,j]$ which will store exactly that number. We will fill the table from column by column, from $j=1$ to $j=m$, and for each column we start with $i=n$ and end with $i=1$. We will calculate $T[i,j]$ for any $i \in [n], j \in [m]$ as follows:
	\begin{align*}
	T[i,j] = & \text{\textsc{\#AD}}(\mathcal{I}_{i,\ell(i,j)}^G)-D_{i,j} \\
	& -\sum_{j'=0}^{j-1}{\sum_{i'=i}^{n}\binom{(N+1)(n-i')}{j-j'} \cdot \binom{i'}{i} \cdot T[i',j']}\\
	& -\sum_{i'=i+1}^n{\binom{i'}{i} \cdot T[i',j]}
	\end{align*}
	The following lemma states that $T[i,j]$ is equal to the number of $i$-vertex-isolating sets of size $j$.
	\begin{lemma}\label{lemma:table_is_correct}
		$T[i,j]=|\mathcal{E}_i^{j}|$ for all $i \in [n], j \in [m]$.
	\end{lemma}
	\begin{proof}
		We will show the claim by induction over $n(j+1)-i$. For $n(j+1)-i=0$, we have $i=n$ and $j=0$, and thus, using \Cref{lemma:agent_delete_char}, we obtain
		\begin{align}
		\begin{split}
		T[n,0]&=\text{\textsc{\#AD}}(\mathcal{I}_{n,\ell(i,0)}^G)-D_{n,0}\\
		&-\sum_{j'=0}^{-1}{\sum_{i'=n}^{n}\binom{(N+1)(n-i')}{j-j'} \cdot \binom{i'}{i} \cdot T[i',j']}\\
		&-\sum_{i'=n+1}^n{\binom{i'}{i} \cdot T[i',j]} \label{e1}
		\end{split}\\
		&=\binom{0}{0}\binom{n}{n}|\mathcal{E}_n^0|+D_{n,0}-D_{n,0}=|\mathcal{E}_n^0| \label{e2}
		\end{align}
		where in \Cref{e1} we use the definition of $T[i,j]$ and in \Cref{e2} we apply \Cref{lemma:agent_delete_char}.
		Now, assume that the equation is true for all $i \in [n], j \in [m]$ with $ n(j+1)-i<r$. We show that the equation is true for $i \in \{1,...,n\},j \in [m]$ such that $n(j+1)-i=r$. Using \Cref{lemma:agent_delete_char} and our induction hypothesis that $T[i',j']=\mathcal{E}_{i'}^{j'}$ for $j'<j$ and $T[i',j]=\mathcal{E}_{i'}^{j}$ for $i'>i$, we obtain
		
		{\small
			\begin{align}
			\begin{split}
			T[i,j]=&\text{\textsc{\#AD}}(\mathcal{I}_{i,\ell(i,j)}^G)-D_{i,j}\\
			-&\sum_{j'=0}^{j-1}{\sum_{i'=i}^{n}\binom{(N+1)(n-i')}{j-j'} \cdot \binom{i'}{i} \cdot T[i',j']}\\
			-&\sum_{i'=i+1}^n{\binom{i'}{i} \cdot T[i',j]} \label{e3}
			\end{split}\\
			\begin{split}
			=&\sum_{j'=0}^j{\sum_{i'=i}^n{\binom{(N+1)(n-i')}{j-j'} \cdot \binom{i'}{i} \cdot |\mathcal{E}_{i'}^{j'}|+D_{i,j}}}-D_{i,j}\\
			-&\sum_{j'=0}^{j-1}{\sum_{i'=i}^{n}\binom{(N+1)(n-i')}{j-j'} \cdot \binom{i'}{i} \cdot |\mathcal{E}_{i'}^{j'}|}\\
			& -\sum_{i'=i+1}^n{\binom{i'}{i} \cdot |\mathcal{E}_{i'}^{j'}|}  \label{e4}
			\end{split}\\
			=&\binom{(N+1)(n-i)}{j-j}\binom{i}{i}|\mathcal{E}_i^j|=|\mathcal{E}_i^j|  \label{e5}
			\end{align}}
		
		where \Cref{e3} uses the definition of $T[i,j]$ and in \Cref{e4}, we apply \Cref{lemma:agent_delete_char} and also replace $T[i',j']$ by $|\mathcal{E}_i^j|
		$ by our induction hypothesis.
	\end{proof}
	
	\begin{lemma}
		Given an oracle that solves \textsc{\#Agent-Delete}, we can compute $T[i,j]$ in polynomial time.
	\end{lemma}
	\begin{proof}
		To compute the table entry $T[i,j]$, we need to compute at most $i \cdot j$ table entries before. Each table entry can be computed in polynomial time, since, by \Cref{lemma:agent_delete_d_table_poly}, $D_{i,j}$ can be computed in polynomial time, an oracle query needs constant time and we sum over polynomially many table entries. Therefore, in total, the running time is polynomial.
	\end{proof}
	\begin{lemma}
		Let $G=(V,E)$ be a graph and let $C_k$ be the number of edge covers of size $k$ in $G$. Then, it holds that $C_k=\binom{m}{k}-\sum_{i=1}^nT[i,m-k]$.
	\end{lemma}
	\begin{proof}
		By \Cref{lemma:table_is_correct}, $\sum_{i=1}^nT[i,m-k]$ gives the number of vertex-isolating sets of size $m-k$. If $E' \subseteq E$ with $|E'|=m-k$ is not a $i'$-vertex-isolating set for some $i' \in [n]$, then clearly $E \setminus E'$ is an edge cover of size $k$. If $E' \subseteq E$ with $|E'|=m-k$ is a $i'$-vertex-isolating set for some $i' \in [n]$, then clearly $E \setminus E'$ is not an edge cover. It follows that either $E'$ is an edge cover or $E \setminus E'$ is a vertex-isolating set and the result follows.
	\end{proof}
	\begin{lemma}
		Given an oracle that solves \textsc{\#Agent-Delete}, for a given graph $G=(V,E)$ and an integer $j \in [m]$, we can compute the number of edge covers of size $j$ in $G$ in polynomial time.
	\end{lemma}
	In the end, we show membership in \#P. Consider a NTM that nondeterministically chooses $\ell$ agents (that are not $a$) and checks whether after their deletion, $a$ is assigned in a stable matching. The number of accepting branches of this NTM is equal to the number of solutions to our problem.
	
	\medskip
	\noindent\textbf{Pair.}
	We can use the hardness of \textsc{\#Agent-Delete} to show hardness for the analogous counting problem for stable pairs, i.e. \textsc{\#Pair-Delete}. We use the one-to-one correspondence between unstable agents and stable pairs as described before Observation~\ref{s:4}. Notice that the set of possible deletions is the same for both instances ($\mathcal{I}$ and $f(\mathcal{I},a)$), since we are not allowed to delete $a$ in $\mathcal{I}$ and we are not allowed to delete $a$ or $a'$ in $f(\mathcal{I},a)$. The one-to-one-correspondence between these two instances directly implies \textsc{\#Agent-Delete}$(\mathcal{I})=N-$\textsc{\#Pair-Delete}, where $N$ is the total number of possible deletions of $\ell$ agents. Thus, we can clearly compute the number of stable agent solutions of an instance in polynomial time when we are given an oracle for \textsc{\#Pair-Delete}. This implies \#P-hardness for \textsc{\#Pair-Delete}, and since we can nondeterministically choose $\ell$ agents and check whether a pair is stable after their deletion in polynomial time, the statement follows.
\end{proof}

\subsection{Approximation algorithms}
\nine*
\begin{proof}
Let $\sigma_\mathcal{I}^{k}$ be the number of preference profiles at swap distance exactly $k$ to an election or SM instance $\mathcal{I}$.
We first argue why we can uniformly sample a preference profile at swap distance $\ell$ in polynomial time. The same result for elections was shown in \citet{DBLP:conf/ijcai/BoehmerBFN21} and we will use it for the stable matching setting.
    We interpret the SM instance $\mathcal{I}=(U,W,\mathcal{P})$ as two elections $\mathcal{E}$ with the men as voters and the women as candidates and $\mathcal{F}$ with the women as voters and the men as candidates. We first compute, for all $0 \leq \ell_U \leq \ell$, the number $\sigma_\mathcal{I}^{\ell_U,\ell-\ell_U}$ of preference profiles such that the preferences of the men have swap distance exactly $\ell_U$ and the preferences of the women have swap distance exactly $\ell-\ell_U$ from $\mathcal{I}$. Note that this is possible in polynomial time since $\sigma_\mathcal{I}^{\ell_U,\ell-\ell_U}=\sigma_\mathcal{E}^{\ell_U} \cdot \sigma_\mathcal{F}^{\ell-\ell_U}$ and the values for the elections have been shown to be polynomial-time computable. We choose now $\ell_U \in \{0,...,\ell\}$ with probability $\frac{\sigma_\mathcal{I}^{\ell_U,\ell-\ell_U}}{\sum_{i=0}^\ell \sigma_\mathcal{I}^{i,\ell-i}}$, and sample an election $(U,W,\mathcal{P}_\mathcal{E})$ at swap distance $\ell_U$ from $\mathcal{E}$ and another election $(W,U,\mathcal{P}_\mathcal{F})$ at swap distance $\ell-\ell_U$ from $\mathcal{F}$. Our sampled SM instance at swap distance $\ell$ from $\mathcal{I}$ is $\mathcal{I}':=(U,W,\mathcal{P}_\mathcal{E} \cup \mathcal{P}_\mathcal{F})$. Now, it is easy to see that each SM instance at swap distance $\ell$ from $\mathcal{I}$ (where the men's preferences are at swap distance $\ell_U$) is sampled with probability
\begin{align*}
    \frac{\sigma_\mathcal{I}^{\ell_U,\ell-\ell_U}}{\sum_{i=0}^\ell \sigma_\mathcal{I}^{i,\ell-i}} \cdot \frac{1}{\sigma_\mathcal{E}^{\ell_U}} \cdot \frac{1}{\sigma_\mathcal{F}^{\ell-\ell_U}}
    =\frac{\sigma_\mathcal{I}^{\ell_U,\ell-\ell_U}}{\sigma_\mathcal{I}^\ell} \cdot \frac{1}{\sigma_\mathcal{I}^{\ell_U,\ell-\ell_U}}
    =\frac{1}{\sigma_\mathcal{I}^\ell}
\end{align*}
All steps can be computed in polynomial time. Sampling a preference profile at deletion distance $\ell$ is trivial, just delete $\ell$ agents uniformly at random. We can conclude:
\begin{lemma}\label{lem:samples}
For an SM instance $\mathcal{I}$ and an integer $\ell$, we can uniformly sample a preference profile at swap distance $\ell$ from $\mathcal{I}$.
\end{lemma}
For the randomized algorithm, we sample $s:=\lceil \frac{\ln \frac{2}{\delta}}{2 \varepsilon^2} \rceil $ preference profiles at swap distance $\ell$ from $\mathcal{I}$ and determine the ratio $p$ of preference profiles that are stable with respect to our definition of interest. Notice that checking stability for matchings, pairs and agents can all be done in polynomial time. We define random variables $X_i$ for $i \in [s]$ which are $1$ if in the $i$-th sampled preference profile, the matching, pair or agent is stable, and $0$ else. These are clearly Bernoulli random variables with parameter $p^*=\frac{S^*}{\sigma_\mathcal{I}^{\ell}}$, where $S^*$ is the exact number of preference profiles at distance $\ell$ where the matching/pair/agent is stable. Defining $X:=\sum_{i=1}^s$, we have 
\begin{align*}
\text{Pr}(|p-p^*|>\varepsilon)=
    \text{Pr}(|(X - E(X))/s| > \varepsilon)
    \leq 2 e^{-2s\varepsilon^2}\leq\delta
\end{align*}
using the Hoeffding inequality.
\end{proof}
Now, for a fixed SM instance $\mathcal{I}$, matching $M$ and integer $\ell$, we wish to approximate the number $S^*$ of preference profiles at swap distance $\ell$ from $\mathcal{I}$ such that $M$ is not stable.
Our approach will be to count the number of preference profiles where a specific pair is blocking. By summing over all pairs, we overcount the total number of preference profiles where at least one pair is blocking, but only by a polynomial factor. We first explain how to compute the number of profiles at some swap distance such that a specific pair is blocking.
\begin{proposition}
    Given four integers $i<j\leq n$ and $k$, One can determine the number of permutations $\pi$ of $n$ elements with $k$ inversions such that $\pi(i)<\pi(j)$ in polynomial time.
\end{proposition}
\begin{proof}
 We will compute the table $T_d[n,k,i]$, which counts the number of permutations of $n$ elements with $k$ inversions such that $i$ remains in front of $i+d$, by dynamic programming: $T_d[n,0,i]=1$, $T_d[n,k<0,i]=0$ and
\begin{align*}
    T_d[d+1,k,1] =& \sum_{\ell=0}^{n-2} T[d-1,k-\ell] \cdot (\ell+1)\\
    T_d[i+d,k,i>1] =& T_d[i+d,k,1]\\
    T_d[n>i+d,k,i] =&\sum_{\ell=0}^{n-1} T_d[n-1,k-\ell,i]
\end{align*}
where $T[n,k]$ is the total number of permutations of $n$ elements with $k$ inversions.
Assume we are in the first case. The two fixed elements are $i=1$ and $j=d+1$. Consider an arbitrary permutation $\pi$ with $k$ inversions such that $\pi(i)<\pi(j)$. The $\ell$ inversions where $i$ and $j$ are involved are at most $n-2$, as otherwise $\pi(j)<\pi(i)$. Clearly, the remaining $d-1$ elements must form exactly $k-\ell$ inversions. There are $\ell +1$ different permutations with a fixed permutation of the other $d-1$ elements that achieve $\ell$ inversions: The number $q$ of elements we shift $i$ to the left can be between $0$ and $\ell$ and immediately determines that $j$ must be shifted $\ell-q$ elements to the right.\\
Assume we are in the third case (the second case is symmetrical). The last element of our ordering is not one of the two fixed elements. We can thus consider the subproblem for the first $n-1$ elements and then add $j \in \{0,...,n-1\}$ inversions by shifting the $n$-th element $j$ positions to the right of our fixed permutation of the other elements. The running time is clearly polynomial and bounded by $\mathcal{O}(n^4)$, which finishes the proof.
\end{proof}
This gives exactly the number $t(n,k,i,j):=T_{j-i}[n,k,i]$ of size-$n$-preference orders at swap distance $k$ such that $i$ remains in front of $j$ for some $i<j$. Notice that we can then easily compute the number $t(n,k,j,i):=T[n,k]-T_{j-i}[n,k,i]$ of orders where $i$ and $j$ are inversed. \\
Consider a blocking pair $\{m,w\}$. The number of preference profiles at swap distance $\ell$ such that $\{m,w\}$ is blocking for matching $M$ is
\begin{align*}
    b_{m,w}:=\sum_{\ell_1+\ell_2 \leq \ell} t(n,\ell_1,rk_m(w),M(m)) \cdot t(n,\ell_2,rk_w(m),M(w))\cdot \sigma_{\mathcal{I}'}^{\ell-\ell_1-\ell_2}
\end{align*}
where $\sigma_\mathcal{I}^{k}$ is the number of profiles at swap distance $k$ from $\mathcal{I}$, $\mathcal{I}'$ is the SM instance $(U-\{m\},W-\{w\},\mathcal{P}-\{P_u,P_w\})$ and $rk_a(b)$ denotes the position of agent $b$ in the preference list of agent $a$.\\
Apart from counting, we also need to argue that we can uniformly sample from these profiles:
\begin{lemma}\label{lem:samplet}
    One can uniformly sample a preference profile at swap distance $\ell$ from $\mathcal{I}$ such that $\{m,w\}$ is blocking in polynomial time.
\end{lemma}
\begin{proof}
    We can backtrack the calculation of $b_{m,w}$ in a canonical way to obtain a uniform sampling. We first choose a pair $(\ell_1,\ell_2)$ with $\ell_1+\ell_2 \leq \ell$ with probability 
    \begin{align*}
        \frac{ t(n,\ell_1,rk_m(w),M(m)) \cdot t(n,\ell_2,rk_w(m),M(w))\cdot \sigma_{\mathcal{I}'}^{\ell-\ell_1-\ell_2}}{b_{m,w}}
    \end{align*}
    Then we sample the two preference lists of $m$ and $w$ as well as the remaining profile independently. We showed that the sampling of a preference profile without further restrictions is possible in polynomial time in \cref{lem:samples} . For the preference lists with the restriction $\pi(i)<\pi(j)$, for the third equation, we can sample $\ell \in \{0,...,n-1\}$ with probability $\frac{T_d[n-1,k-\ell,i]}{T[n,k,i]}$, sample the remaining list recursively and in the end shift the last element $\ell$ positions to the left. For the first case, we have to distinguish two cases: If $i<j$, we can sample $\ell \in \{0,...,n-2\}$ with probability $\frac{T[d-1,k-\ell]}{T_d[d+1,k,1]} \cdot (\ell+1)$, then sample a permutation of $d-1$ elements with $k-\ell$ inversions as seen in \citet{DBLP:conf/ijcai/BoehmerBFN21} and finally sample $q \in \{0,...,\ell\}$ with probability $\frac{1}{\ell+1}$ to determine how many shifts of $i$ to the right to perform. If $i>j$, we need to proceed differently: Since the two elements must form an inversion, we need at least $n-1$ swaps. We will thus instead sample $\ell \in \{n-1,...,2(n-1)-1\}$ with probability $\frac{T[d-1,k-\ell]}{T_d[d+1,k,1]} \cdot (2(n-1)-\ell)$, where $(2(n-1)-\ell)$ corresponds to the number of ways to shift $i$ to the right and $j$ to the left by $\ell$ positions in total, requiring that they form an inversion. It is now easy to see that each allowed profile will be sampled with probability $\frac{1}{b_{m,w}}$.
\end{proof}
In each profile of swap distance $\ell$ where $M$ is not stable anymore, at least one pair and at most $n^2-n$ pairs are blocking. Thus, summing over all possible blocking pairs already gives a deterministic approximation:
\begin{corollary}\label{corr:n2approx}
    For $n>1$, the algorithm that computes $b:=\sum_{\{m,w\} \in (U \times W) \setminus M} b_{m,w}$ is an $n^2-n$-approximation for $S^*$.
\end{corollary}
We now describe a FPRAS for the problem.
Fix an arbitrary ordering of all pairs $\{m,w\} \in U \times W$. Repeat the following $s$ times: Choose a pair $\{m,w\}$ with probability $\frac{b_{m,w}}{b}$ and uniformly sample a preference profile at swap distance $\ell$ where $\{m,w\}$ is blocking (which is possible by Lemma \cref{lem:samplet}). Let $X_i$ be $1$ if all pairs that preceed $\{m,w\}$ in our fixed ordering are \emph{not} blocking for $M$ in the sampled profile, and $0$ else, for $i \in [s]$. Notice that the number of pairs $(\{m,w\},\mathcal{P}')$ such that $\{m,w\}$ is the first pair to be blocking for $M$ in $\mathcal{P}'$ exactly corresponds to $S^*$ (since for each profile that we want to count, there is exactly one pair fulfilling the property). Thus, the $X_i$ are Bernoulli random variables with parameter $\frac{S^*}{b}$. Defining $X:=\sum_{x=1}^s X_i$, we output $S:=X \cdot \frac{b}{s}$. Using a Chernoff bound and the fact that $\frac{S^*}{b} \geq \frac{1}{n^2-n}\geq \frac{1}{n^2}$, which follows from \cref{corr:n2approx}, we get
\begin{align*}
&\text{Pr}(|S-S^*|>\varepsilon \cdot S^*)\\
=&\text{Pr}(|X \cdot \frac{b}{s} - S^*| >\varepsilon \cdot S^*)\\
=&\text{Pr}(|X - S^* \cdot \frac{s}{b}| >\varepsilon \cdot S^* \cdot \frac{s}{b})\\
        =&\text{Pr}(|X - E(X)| >\varepsilon \cdot E(X))\\
    \leq& 2 e^{-s \cdot (S^*/b)\varepsilon^2/3}\\
    \leq& 2 e^{-s \varepsilon^2/3n^2}\\
    =&\frac{1}{4}
\end{align*}
for $s=\frac{3n^2 \cdot \ln 8 }{\varepsilon^2}$ repeats, which is polynomial. So, our approximate solution is within factor $\varepsilon$ of the exact solution with probability $\geq \frac{3}{4}$. Additionally, the running time for each repeat is determined by the time required for the sampling and for the stability checking, which is both quadratic. The running time to initially create all tables is bounded by $\mathcal{O}(n^4)$. So the running time of the algorithm is $\mathcal{O}(\frac{n^4}{\varepsilon^2})$, which is polynomial. We get:
\ten*
\section{Additional Material for \Cref{sub:exp}}

\subsection{Mallows Noise Model} \label{app:Mal}
In the main body, we showed that computing the exact stable matching probability for a specific number of swaps is hard. In fact, the hardness result for \textsc{\#Matching-Swap} carries over to the Mallows model. We show this in a very similar fashion to the analogous proof for the election setting by \citet{DBLP:journals/sncs/BaumeisterH23}. Formally, given a stable marriage instance $\mathcal{I}$, a matching $M$ and a rational number $\phi \in [0,1]$, the \textsc{Mallows Stable Matching Probability} problem asks for computing the probability that $M$ is stable in an instance generated by the Mallows model with dispersion parameter $\phi$ from $\mathcal{I}$.
\begin{lemma}
	One can Turing-reduce \textsc{\#Matching-Swap} to \textsc{Mallows Stable Matching Probability} in polynomial time.
\end{lemma}
\begin{proof}
	Assume we are given an SM instance $\mathcal{I}=(U,W,\mathcal{P})$, a matching $M$ and a swap budget $r$. Our new \textsc{Mallows Stable Matching Probability} instance consists of the exact same SM instance $\mathcal{I}$ and matching $M$. We set the dispersion parameter $\phi:=\frac{1}{(n!)^{2n}}$. Consider an arbitrary preference profile $\mathcal{P}'$ with swap distance $j$ to $\mathcal{P}$. We have 
	\begin{align*}
	\mathcal{D}_\text{Mallows}^{P,\phi}(P')=&\frac{1}{Z}\frac{1}{(n!)^{2n \cdot j}}=(n!)^{2n} \cdot \frac{1}{Z}\frac{1}{(n!)^{2n \cdot (j+1)}}\\
	& >\sum_{P'':\kappa(P,P')>j}\mathcal{D}_\text{Mallows}^{P,\phi}(P'')
	\end{align*}
	i.e. the probability assigned to $P'$ is larger than the summed probabilities assigned to all preference profiles with a larger swap distance. This is because there can be at most $n!$ different preference lists per agent and thus the total number of preference profiles is bounded by $(n!)^{2n}$. Therefore, we can decompose the \textsc{Mallows Stable Matching Probability} answer $\Phi$ as follows. Clearly, the probability that $M$ is stable is $N_j \cdot (\frac{1}{z}\phi^j)+y$, where $y<(\frac{1}{z}\phi^j)$ and $N_j$ is the number of preference profiles at swap distance $j$ where $M$ is stable. $y$ resembles the probability that $M$ is stable in some preference profile with higher swap distance. We set $\Phi_0=\Phi$ and compute for $0\leq j \leq \frac{2n^2 \cdot (n-1)}{2}-1$:
	\begin{equation*}
	N_j=\lfloor \Phi_j/((\frac{1}{n!^{2n}})^j/Z) \rfloor \text{  and  } \Phi_{j+1}=\Phi_j-N_j \cdot ((\frac{1}{n!^{2n}})^j/Z)
	\end{equation*}
	Then, $N_j$ gives the number of preference profiles at swap distance $j$ from $\mathcal{I}$ where $M$ is stable.
\end{proof}

\subsection{Dataset} \label{app:data}
We now present the statistical cultures. We use the cultures from \citet{DBLP:journals/corr/abs-2208-04041} and the Mallows-Robust one.
\begin{itemize}
	\item \textbf{Impartial Culture (IC)}: Each preference list $\succ_a$ is sampled uniformly at random from the set of all agents of opposite gender.
	\item \textbf{2 Group-Impartial Culture (2-IC)}: We randomly partition $U$ and $W$ into $U_1 \uplus U_2=U$ and $W_1 \uplus W_2 = W$ with $\frac{U_1}{U}=\frac{W_1}{W}=p$ for some parameter $p \in [0,0.5]$. Each man $m \in M_i$ prefers all women $w \in W_i$ over all women of the other group (and vice versa). The preferences within the groups are drawn uniformly at random from the set of all agents of opposite gender for each agent. Our data set contains 20 instances for each $p \in \{0.25,0.5\}$.
	\item \textbf{1-Dimensional Euclidean (1D)}: For each agent $a \in A$, we sample a point $p_a \in [0,1]$ uniformly at random. Agent $a$ prefers an agent $a'$ to another agent $a''$ if $|p_a-p_{a'}|<|p_a-p_{a''}|$.
	\item \textbf{2-Dimensional Euclidean (2D)}: For each agent $a \in A$, we sample a point $p_a \in [0,1]^2$ uniformly at random. Agent $a$ prefers an agent $a'$ to another agent $a''$ if the Euclidean distance between $a$ and $a'$ is smaller than the Euclidean distance between $a$ and $a''$.
	\item \textbf{2-Dimensional Reverse-Euclidean (Rev-Euc)}: In addition to the setting in 2D, we partition $U$ and $W$ into two sets each ($U_1,U_2$ and $W_1,W_2$). An agent $a\in U_1 \cup W_1$ ranks the other agents as described in 2D. An agent $b \in U_2 \cup W_2$ ranks the other agents in opposite order, i.e. they prefer agents with a greater Euclidean distance. We are given a parameter $p \in [0,1]$ that describes the percentage of agents in $U_2$ and $W_2$, respectively. Our data set contains 20 instances for each $p \in \{0.05,0.15,0.25\}$.
	\item \textbf{Fame-Euclidean (Fame-Euc)}: Given a parameter $f \in [0,1]$, for each agent, we sample random points $p_a \in [0,1]^2$ (as for the 2D culture) and $f_a \in [0,f]$. An agent $a$ evaluates another agent of opposite gender $b$ analogously to the 2d culture, but subtracts $f_a$ from this evaluation. The points $f_a$ correspond to the fame of an agent, which influences their rank in the preferences of the other agents. Our data set contains 20 instances for each $f \in \{0.2,0.4\}$.
	\item \textbf{Expectations-Euclidean (Ex-Euc)}: We first sample a point $p_a$ for each agent $a \in A$ as for the 2D culture. We then generate a second point $q_a$ for each agent using a 2-dimensional Gaussian distribution with mean $p_a$ and standard deviation $\sigma$. Each agent $a$ ranks the agents of opposite gender $b$ increasingly by the distance between $p_a$ and $q_b$. The points $p_a$ correspond to where the agents would like their partner to be located, while $q_a$ is their real location. Our data set contains 20 instances for each $\sigma \in \{0.2,0.4\}$.
	\begin{figure}
		\centering
		\includegraphics[width=0.7\columnwidth]{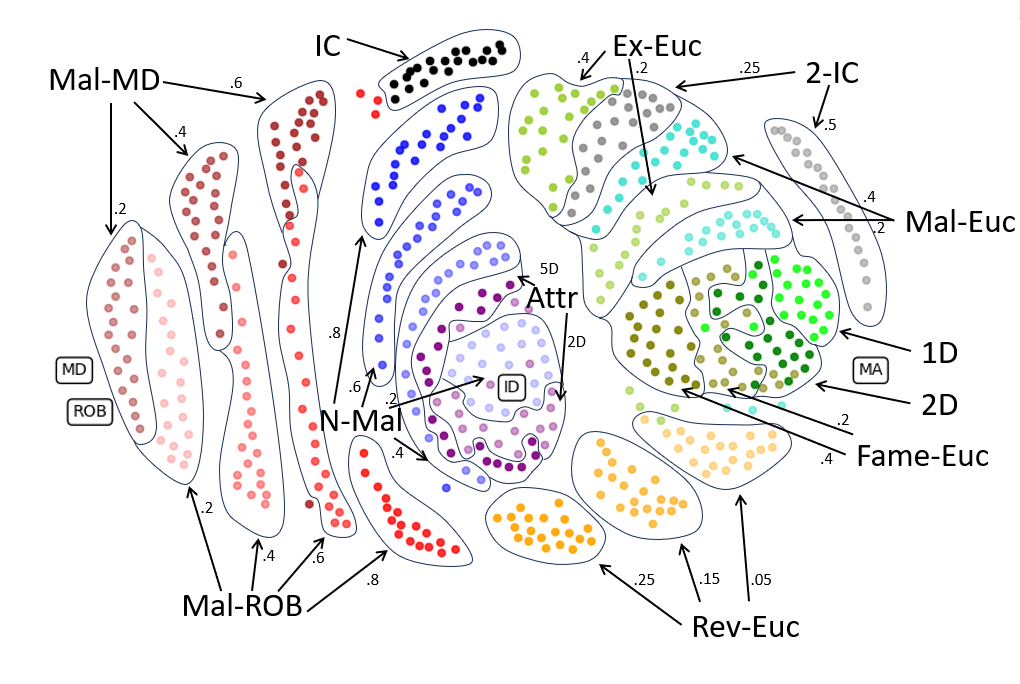}
		\caption{Map of SM instances, each point is colored according to the culture from which the instance was sampled. Instances of the same culture that have a higher parameter (e.g. norm-$\phi$) have the same color, but are more transparent.}
		\label{fig:map-cul}
	\end{figure}
	\item \textbf{Attributes (Attr)} Given a dimension $d \in \NN$, we sample two vectors for each agent $p^{a},w^{a} \in [0,1]^d$ uniformly at random. Agent $a$ ranks other agents $b$ decreasingly by $\sum_{i \in [d]}{w_i^{a} \cdot p_i^{b}}$. Our data set contains 20 instances for each $d \in \{2,5\}$.
	\item \textbf{Norm-Mallows (N-Mal) }: Starting from a random ordering of men and women, the preference list of each agent is an ordering drawn from the Mallows model with dispersion parameter norm-$\phi$ from the initial ordering of the agents of opposite gender. Our data set contains 20 instances for each norm-$\phi$ such that norm-$\phi \in \{0.2,0.4,0.6,0.8\}$.
	\item \textbf{Mallows-Euclidean (Mal-Euc)}: We create the preferences for all agents as described for the 2D-culture. From this instance, we draw a new instance according to the Mallows model with dispersion parameter norm-$\phi$. Our data set contains 20 instances for each norm-$\phi \in \{0.2,0.4\}$.
	\item \textbf{Mallows-Asymmetric (Mal-MD)}: The procedure is analogous to the Mallows-Euclidean culture. Instead of an instance from the 2D culture, we start with the MD instance. Our data set contains 20 instances for each norm-$\phi \in \{0.2,0.4,0.6\}$.
	\item \textbf{Mallows-Robust (Mal-ROB)}: The procedure is analogous to the Mallows-Euclidean culture. Instead of an instance from the 2D culture, we start with the Robust instance. We added this culture to our data set in the expectation that its instances will be among the most robust. Our data set contains 20 instances for each norm-$\phi \in \{0.2,0.4,0.6,0.8\}$.
\end{itemize}

\begin{figure}
	\centering
	\includegraphics[width=0.7\columnwidth]{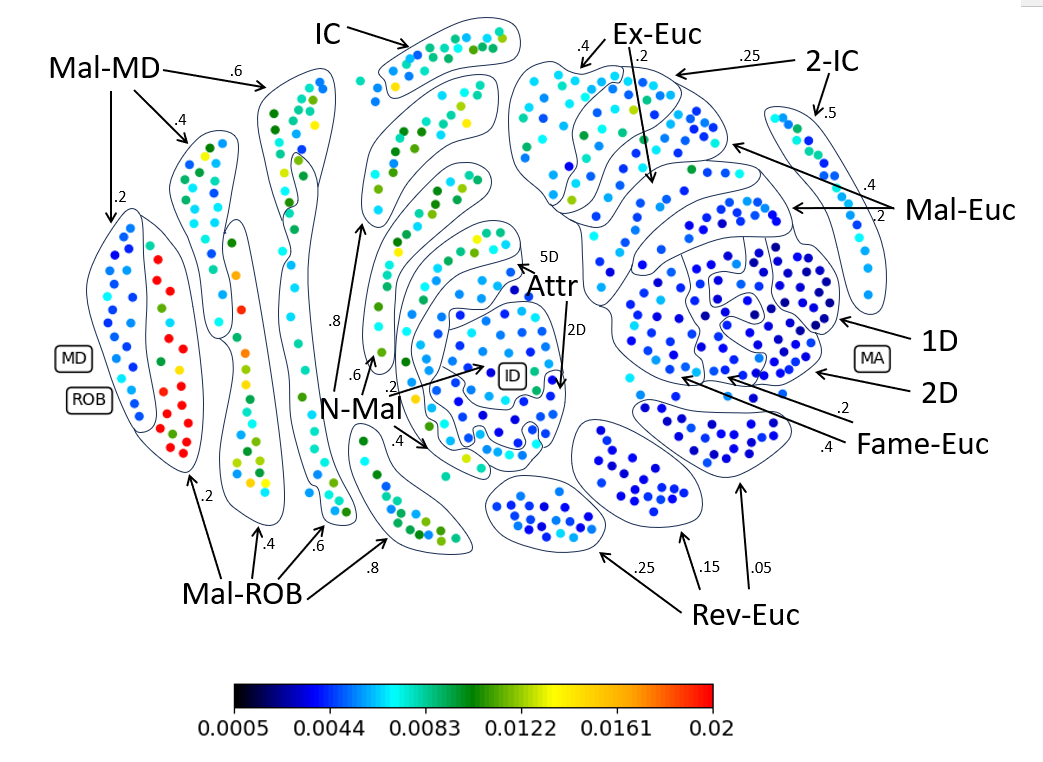}
	\caption{50\%-threshold for the man-optimal stable matchings of all instances of the data set. The scale is logarithmic and capped at norm-$\phi=0.02$.}
	\label{fig:map-fif}
\end{figure}

\paragraph{Map.}
To better understand which cultures generate robust stable matchings and which instances do not, we use the map of stable marriage instances introduced by \citet{DBLP:journals/corr/abs-2208-04041}. Each point on this map corresponds to one instance in our data set. The distance between two points models how different the two instances are according to the so-called \emph{mutual attraction distance}, introduced by \citet{DBLP:journals/corr/abs-2208-04041}. Each agent $a$ is associated with a vector where the $i$-th entry is the position in which $a$ appears in the preferences of the agent that $a$ ranks in position $i$. We then match agents from the two instances we want to compare such that the $\ell_1$ distance between the associated vectors of matched agents is minimized.  In the map, the labels of the different cultures are connected to the corresponding instances. In \Cref{fig:map-cul}, we see the map for our data set where each culture has a different color. 

\subsection{Additional Material for \Cref{sub:exp1}: Men-Optimal Matching} \label{app:men}

We now take a closer look how the average-case robustness of men-optimal matchings depends on the model from which the instance was sampled. 
For this, we show in \Cref{fig:map-fif} a map of the 50\%-threshold of the man-optimal stable matchings of all instances.	
Examining this map, we can observe that most Euclidean instances have a very low average-case robustness. Instances from the IC culture or close to IC are quite robust. The ROB instance and the instances from the Mal-ROB culture have the highest average-case robustness. The ROB instance has a 50\%-threshold close to 0.3, which implies that one needs almost 200 random swaps in each preference list to make the matching unstable. This is 100 times more swaps than for standard instances from our dataset. 

One can observe that more structured cultures tend to be less robust than cultures that involve more ``randomness''. For example, most Euclidean instances have a very low 50\%-threshold while the uniformly at random generated IC instances are surprisingly robust. Intuitively, in a structured instance, there are natural candidates for blocking pairs that can very easily become blocking because of their mutual liking. For example, consider a 2D instance where the points of one man $m$ and two women $w,w'$ were sampled very close to each other. Suppose that $m$ and $w$ are mutual top-choices and thus $\{m,w\}$ is contained in the initial stable matching. It could be that $m$ ranks $w'$ second and $w'$ ranks $m$ first. In that case, the pair $\{m,w'\}$ has a high probability of becoming blocking.

\begin{figure*}[t!]
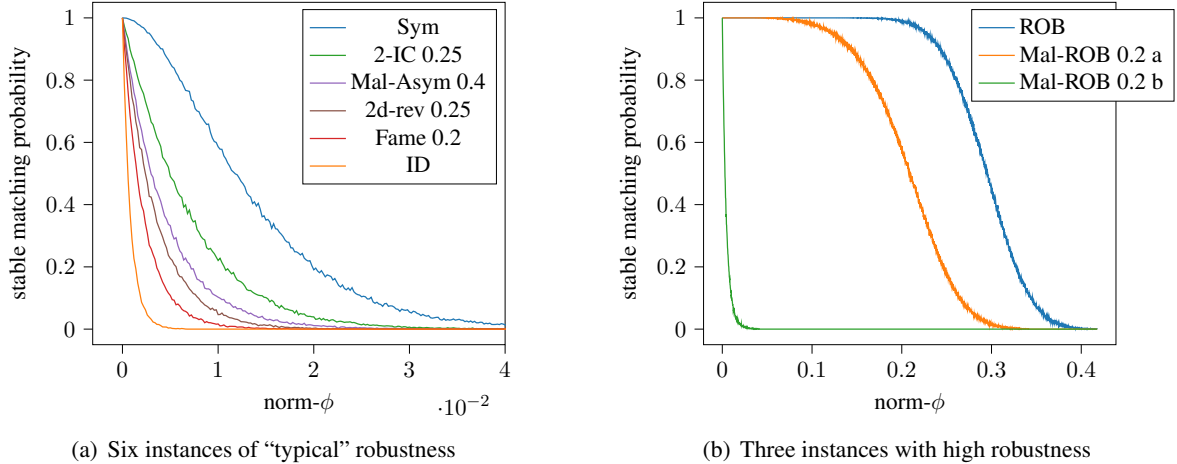

	\centering
	\begin{minipage}[t]{0.49\linewidth}
		\centering
		\subfigure[Six instances of ``typical'' robustness]{\input{plot_data/devplots1.tex}}
		\label{fig:manopt-plots-1}
	\end{minipage}
	\hfill
	\begin{minipage}[t]{0.49\linewidth}
		\centering
		\subfigure[Three instances with high robustness]{\input{plot_data/devplots2.tex}}
		\label{fig:manopt-plots-2}
	\end{minipage}
	\caption{Average-case robustness for nine exemplary instances.}
	\label{fig:manopt-plots}
\end{figure*} 

\begin{figure*}[t]
	\begin{align}
	\sum_{i \in [n]} x[i,j]&\leq 1 & \forall j \in [n] \label{e6}\\
	\sum_{j \in [n]} x[i,j]&\leq 1 & \forall i \in [n] \label{e7}\\
	\sum_{i,j \in [n]: \rk_{i^*}(j)\leq \rk_{i^*}(j^*) \vee \rk_{j^*}(i)\leq \rk_{j^*}(i^*)} x[i,j] &\geq 1 &\forall i^*,j^* \in [n] \label{e8}\\
	\begin{split}
	\sum_{k=1}^{\rk_i(j)-1}{k \cdot x[i,p_i(\rk_i(j)-k)]}  \label{e9}\\
	+\sum_{k=1}^{\rk_j(i)-1}{k \cdot x[p_j(\rk_j(i)-k)]}+2n \cdot x[i,j] &= y[i,j] 
	\end{split}
	&\forall i,j \in [n]\\
	k+1-y[i,j]-k \cdot \sum_{\ell=1}^{k-1}z[i,j,\ell] &\leq z[i,j,k] & \forall i,j \in [n], k \in [d]  \label{e10}\\
	x[i,j] &\in \{0,1\} &\forall i,j \in [n]  \label{e11}\\
	\min \sum_{k=1}^{d}n^{d-k} \cdot \sum_{i,j \in [n]} z[i,j,k]  \label{e12}
	\end{align}
	\caption{ILP for computing the robust matching.}\label{ILP}
\end{figure*}
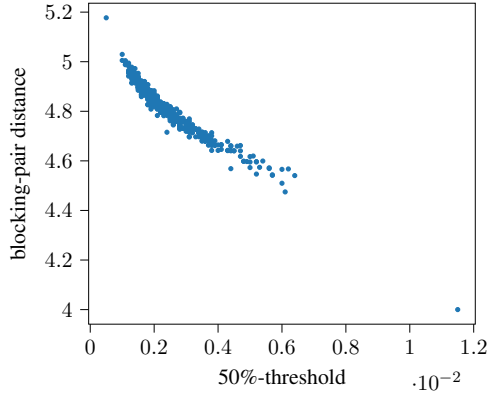
\begin{figure}[t]
	\centering
\begin{tikzpicture}[scale=0.75]

\definecolor{color0}{rgb}{0.12156862745098,0.466666666666667,0.705882352941177}

\begin{axis}[
tick align=outside,
tick pos=left,
x grid style={white!69.0196078431373!black},
xlabel={50\%-threshold},
xmin=-5e-05, xmax=0.01205,
xtick style={color=black},
y grid style={white!69.0196078431373!black},
ylabel={blocking-pair distance},
ymin=3.9412535158, ymax=5.2359260862,
ytick style={color=black}
]
\addplot [semithick, color0, mark=*, mark size=1, mark options={solid}, only marks]
table {%
0.001 5.005164261
0.0115 4.000102269
0.0005 5.177077333
0.0023 4.795266396
0.0037 4.68117209
0.0038 4.680858693
0.0034 4.697417045
0.0052 4.595813107
0.0048 4.598356842
0.0027 4.759106951
0.0022 4.81562148
0.0025 4.793923977
0.0034 4.714525556
0.0027 4.780995167
0.0037 4.680852302
0.0028 4.74483653
0.0031 4.769426744
0.0038 4.678035614
0.0062 4.56742626
0.0037 4.683327055
0.0033 4.697067832
0.0038 4.680152013
0.0034 4.713290967
0.0028 4.75762845
0.0023 4.815647554
0.0022 4.796688926
0.0029 4.731546381
0.0027 4.783671199
0.0022 4.828564797
0.0037 4.683031434
0.0019 4.845809828
0.0022 4.818375535
0.0028 4.796043918
0.0023 4.781986709
0.0022 4.817750349
0.0033 4.727973251
0.0031 4.727971837
0.0028 4.742652524
0.0037 4.700135561
0.0024 4.797565984
0.0052 4.546717538
0.0054 4.599338718
0.0041 4.646175517
0.002 4.837396227
0.0026 4.757063169
0.0026 4.794660099
0.0023 4.806029216
0.0025 4.769458298
0.003 4.744897682
0.0035 4.716741026
0.0025 4.816709132
0.0018 4.827595935
0.0027 4.770464484
0.0022 4.795771931
0.0026 4.771965851
0.0022 4.817349138
0.0035 4.699114109
0.0032 4.745976059
0.0018 4.826158515
0.0038 4.71366962
0.0025 4.771690453
0.0032 4.728581727
0.0031 4.71585258
0.0029 4.772490684
0.0022 4.818837491
0.0027 4.758698308
0.0026 4.772010874
0.0018 4.85715077
0.0022 4.827648142
0.0038 4.700835058
0.0026 4.759501088
0.0024 4.785115858
0.0023 4.82841802
0.0029 4.732364677
0.0023 4.807154287
0.0019 4.8383918
0.0025 4.796690285
0.0023 4.817162395
0.0037 4.700861039
0.0019 4.837694647
0.0023 4.807763024
0.0029 4.769736378
0.0023 4.796481938
0.003 4.757327
0.0031 4.697748982
0.0034 4.728535815
0.005 4.572897951
0.0025 4.806212668
0.0037 4.663170451
0.0056 4.569154304
0.0044 4.568513144
0.003 4.758354922
0.0036 4.713874769
0.0047 4.661937875
0.0039 4.680895215
0.0064 4.540633147
0.0028 4.742891829
0.0028 4.756548828
0.0037 4.68271393
0.0039 4.661592763
0.0024 4.793677023
0.0025 4.781005194
0.0057 4.543614075
0.0045 4.640438426
0.004 4.664284337
0.0028 4.745124799
0.0043 4.64245829
0.0035 4.695370289
0.0044 4.640797003
0.0051 4.619501944
0.0047 4.640420911
0.0052 4.59675512
0.006 4.509782663
0.0032 4.713298252
0.004 4.642026294
0.0031 4.729155009
0.0034 4.715429712
0.0028 4.757548723
0.0035 4.713892659
0.0049 4.597330224
0.0038 4.662352722
0.0035 4.69708766
0.0038 4.642094429
0.0033 4.72877842
0.0053 4.573962285
0.0034 4.713931115
0.0043 4.67868965
0.0031 4.729709123
0.0045 4.640787663
0.0043 4.641247077
0.0024 4.715455369
0.0028 4.769926626
0.0048 4.598315642
0.0035 4.679422736
0.0031 4.729962601
0.0035 4.71264381
0.0035 4.696767059
0.0034 4.714212871
0.0061 4.474770789
0.0047 4.618121413
0.0028 4.72759182
0.005 4.595285751
0.0028 4.793423591
0.0013 4.977872693
0.0015 4.953195586
0.0011 4.98796391
0.0014 4.929022584
0.0014 4.941415324
0.0012 4.953806499
0.0011 4.988674979
0.0017 4.891559447
0.0012 4.982948426
0.0013 4.959993438
0.0015 4.915699185
0.0011 5.000003551
0.0012 4.994392384
0.0012 4.983712392
0.0013 4.934589894
0.0012 4.989012154
0.001 5.029604919
0.0012 4.9592796
0.0011 5.00413771
0.0015 4.914827749
0.0015 4.915133307
0.0015 4.884739421
0.002 4.849720129
0.0014 4.954816222
0.0014 4.914253419
0.0013 4.947360369
0.0015 4.914836131
0.0019 4.847888078
0.0014 4.948123449
0.0016 4.891747418
0.0012 4.964825944
0.0018 4.884433763
0.0014 4.920360669
0.0022 4.821107238
0.0016 4.898943271
0.0015 4.936329223
0.0013 4.947245216
0.0018 4.85768971
0.0012 4.941767347
0.0015 4.935662553
0.0013 4.954318162
0.0021 4.849359712
0.0014 4.948241481
0.0014 4.935657618
0.0016 4.921347596
0.0013 4.929459544
0.0016 4.92246592
0.0017 4.876960158
0.0013 4.960609787
0.0014 4.928744894
0.0016 4.906924842
0.0013 4.914296907
0.0014 4.972382341
0.0015 4.941775384
0.0019 4.866466508
0.0019 4.875799834
0.0015 4.929030748
0.0013 4.954340506
0.0016 4.921191671
0.0021 4.861786247
0.0018 4.891109236
0.0017 4.89986164
0.0015 4.898490541
0.0016 4.858950035
0.0019 4.883290589
0.0016 4.883301514
0.0016 4.90657477
0.0018 4.882670409
0.0016 4.8664556
0.0016 4.883488366
0.0018 4.84880197
0.0017 4.875606873
0.002 4.829806462
0.0017 4.891718755
0.0018 4.865426493
0.0015 4.906146715
0.0014 4.92188623
0.0019 4.867137361
0.0014 4.94086722
0.002 4.837671236
0.0023 4.797350206
0.0023 4.817395507
0.0017 4.89262286
0.0023 4.829774367
0.002 4.848946388
0.0018 4.865037752
0.0023 4.807997596
0.0019 4.819027916
0.0026 4.806488714
0.0015 4.899261295
0.0015 4.897859718
0.0018 4.873427774
0.002 4.839363454
0.002 4.838207085
0.002 4.857518518
0.0016 4.866607202
0.0017 4.905978653
0.0016 4.89786294
0.0028 4.748249624
0.0023 4.796947649
0.0017 4.905683323
0.0015 4.906889472
0.0017 4.892964519
0.0014 4.926959688
0.0015 4.906878035
0.0016 4.906128871
0.0016 4.89061373
0.0019 4.866103995
0.0024 4.809995486
0.0018 4.86524266
0.0016 4.882980501
0.0018 4.907445125
0.0023 4.831381568
0.0018 4.883919543
0.0021 4.859276847
0.0019 4.867480937
0.0019 4.883629218
0.0017 4.882658035
0.0018 4.874776803
0.0017 4.92033264
0.002 4.857490013
0.002 4.858399199
0.002 4.873760012
0.0027 4.809551907
0.0014 4.946469457
0.0027 4.785811384
0.0025 4.797133596
0.0018 4.891721053
0.0017 4.897124478
0.0015 4.918510771
0.0016 4.87559601
0.0015 4.906580479
0.0019 4.882970866
0.0018 4.875569604
0.0022 4.819219213
0.0019 4.866813684
0.0019 4.882786159
0.0022 4.818205273
0.0028 4.795341118
0.0018 4.865941049
0.0024 4.817320309
0.0029 4.759956622
0.0016 4.898155391
0.0018 4.880394351
0.003 4.719200023
0.0026 4.758707758
0.0025 4.806270794
0.0018 4.874415678
0.0015 4.89786624
0.0027 4.772753372
0.002 4.837454894
0.0021 4.809517266
0.0022 4.829012786
0.0021 4.848953849
0.0031 4.745175345
0.0019 4.857679295
0.002 4.828773008
0.0021 4.82883747
0.0041 4.665490277
0.0023 4.807540792
0.0021 4.816921996
0.0028 4.743718867
0.0038 4.682302524
0.0028 4.743723774
0.0028 4.745484038
0.0016 4.881814033
0.0033 4.729740429
0.0028 4.758918308
0.0024 4.795753281
0.0023 4.817336773
0.0026 4.771456542
0.0027 4.770945531
0.003 4.75755444
0.0027 4.783659871
0.0035 4.713611683
0.0028 4.770393512
0.0027 4.772429992
0.0026 4.771216321
0.0031 4.759673281
0.0029 4.757317382
0.0019 4.8652315
0.0014 4.941253299
0.0024 4.807983383
0.0017 4.891217297
0.0018 4.875748794
0.0021 4.848035671
0.0019 4.839218109
0.0021 4.808896629
0.0013 4.945580549
0.0017 4.875749072
0.0025 4.761552376
0.0022 4.828801617
0.0016 4.919924102
0.0021 4.83878828
0.0019 4.866610519
0.0017 4.865030126
0.0021 4.858558626
0.0024 4.785380748
0.0022 4.837794003
0.0017 4.866079407
0.0022 4.795585989
0.0019 4.856588142
0.0018 4.864555205
0.0027 4.770709933
0.0029 4.759950861
0.0027 4.783898581
0.0023 4.828796393
0.0032 4.731762183
0.0029 4.730937182
0.0022 4.828392227
0.0019 4.838226088
0.0026 4.784169304
0.0039 4.665499366
0.0019 4.855513671
0.002 4.848038922
0.0019 4.863796789
0.0019 4.82658992
0.0022 4.829961106
0.0023 4.793993202
0.0019 4.826944887
0.0022 4.80761124
0.0023 4.805182018
0.0025 4.817980815
0.002 4.849175316
0.0018 4.847876894
0.0016 4.913279744
0.002 4.882981362
0.0023 4.80892793
0.002 4.818035146
0.0021 4.81927398
0.003 4.732148894
0.0024 4.808253905
0.0019 4.864357765
0.0022 4.797819605
0.0024 4.82798567
0.0029 4.760249565
0.0019 4.808668679
0.0026 4.759244866
0.0021 4.847331073
0.0021 4.818388862
0.0041 4.664398091
0.0028 4.73205988
0.0034 4.71454507
0.0032 4.742573694
0.0029 4.770228383
0.0027 4.759426861
0.0044 4.658943628
0.0029 4.759421848
0.0021 4.837832808
0.0038 4.68231555
0.0025 4.784359554
0.0028 4.732334696
0.0021 4.782916736
0.0057 4.541778703
0.0035 4.682331604
0.0022 4.819219047
0.0029 4.759725745
0.0038 4.681556373
0.0026 4.744818771
0.0031 4.696399689
0.005 4.617253476
0.0031 4.743449869
0.0034 4.713643913
0.0036 4.678345114
0.0035 4.697756016
0.0032 4.729078274
0.0023 4.806862311
0.0036 4.696458554
0.0033 4.727033939
0.006 4.565787556
0.0044 4.6612379
0.0044 4.641637222
0.003 4.72971546
0.0023 4.792987882
0.0026 4.769159594
0.002 4.837191387
0.0036 4.695383228
0.0046 4.659300119
0.0056 4.571886541
0.0025 4.793264047
0.002 4.849558703
0.0022 4.819062226
0.0013 4.914582579
0.002 4.867498606
0.0017 4.88363776
0.0018 4.874777642
0.0016 4.908346272
0.002 4.83924928
0.0018 4.883340737
0.0021 4.819708954
0.0015 4.892374631
0.0018 4.88430963
0.0018 4.875450597
0.0015 4.892834334
0.0021 4.830435755
0.0016 4.883984597
0.0022 4.810198332
0.002 4.858283345
0.0015 4.93527388
0.0018 4.907046272
0.0027 4.774775078
0.0027 4.761265623
0.0031 4.732362785
0.0026 4.747118668
0.0024 4.785081899
0.0024 4.784621855
0.0027 4.772474031
0.0023 4.807567557
0.0026 4.79597735
0.0025 4.808438513
0.0021 4.828201545
0.0027 4.783669782
0.0023 4.816076931
0.002 4.838798762
0.003 4.746596574
0.0022 4.836281834
0.0031 4.716753172
0.0025 4.806426798
0.0027 4.782963805
0.0025 4.818430145
};
\end{axis}

\end{tikzpicture}
	\caption{Correlation between blocking pair proximity and the 50\%-threshold for all instances from the data set except for the instances from the robust cultures (whose values are very extreme and would disturb the plot).}\label{fig:heur-fif-corr}
\end{figure}

To obtain more information than only the 50\%-threshold, we examine how the stable matching probability changes when we increase the norm-$\phi$ value. To this end, we have chosen nine exemplary instances presented in \Cref{fig:manopt-plots}. They are chosen to represent an instance with average robustness as well as instances with extreme robustness values. Since the range of three of the instances greatly differs from the other six instances, we depict these three more robust instances in \Cref{fig:manopt-plots}(b).

\begin{figure*}[t]
	\begin{minipage}[t]{0.45\textwidth}
		\centering
			\subfigure[Difference between man-opt and summed-rank min stable matchings. All values $\geq 0.01$ are depicted with the same color.\label{fig:map-man-sum}]{\includegraphics[width=\columnwidth]{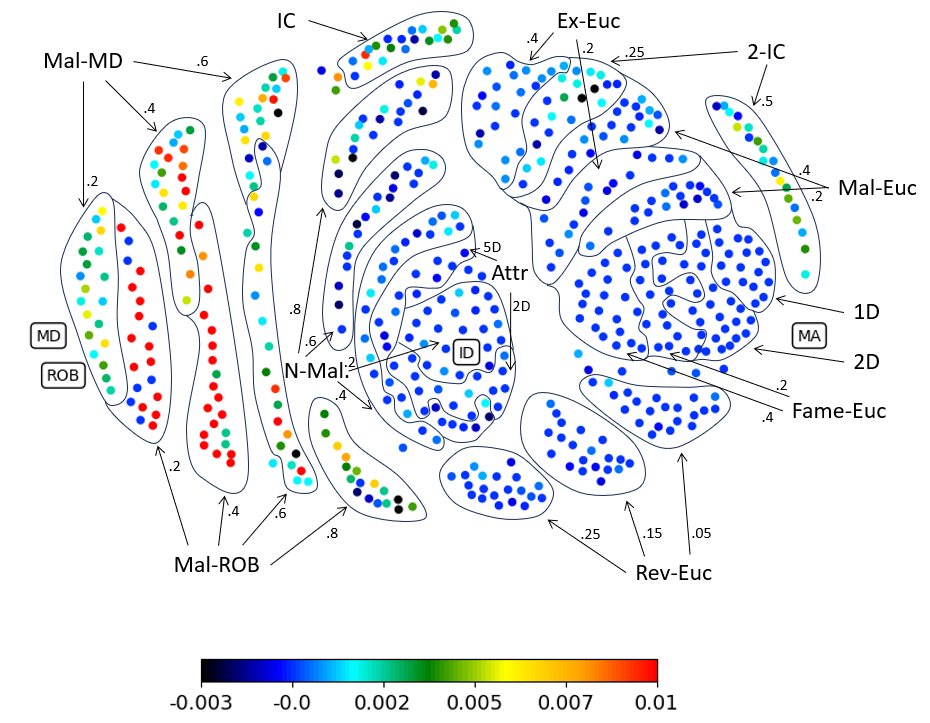}}
	\end{minipage}
	\hfill
	\begin{minipage}[t]{0.45\textwidth}
		\centering
			\subfigure[Difference between summed-rank-min and robust stable matchings.]{\includegraphics[width=\columnwidth]{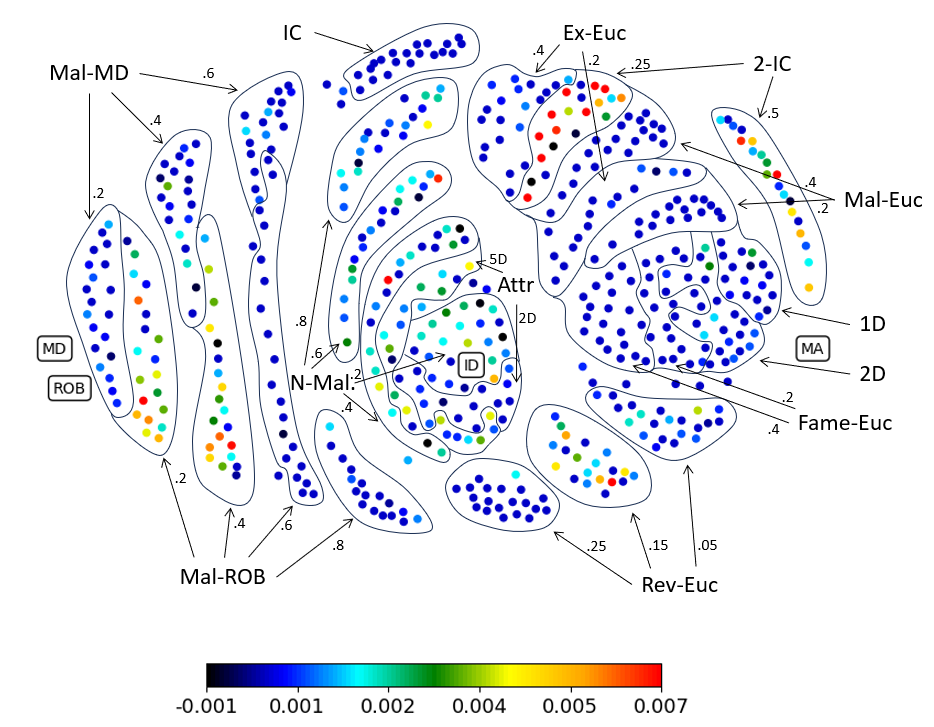}}
		\label{fig:map-sum-rob}
	\end{minipage}
	\caption{Difference between the 50\%-threshold values of (a) the man-optimal and summed-rank-minimal stable matching and (b) the summed-rank-minimal stable matching and the robust stable matching for all instances of the data set.}\label{fig:mapTog}
\end{figure*}

First, we observe that the stable matching probability decreases monotonically with increasing norm-$\phi$. Particularly intriguing are the two Mal-ROB instances with norm-$\phi=0.2$ from \Cref{fig:manopt-plots}(b): 
We observed two different kinds of instances: Very robust instances (marked red in the map) and not very robust instances. In \Cref{fig:manopt-plots}(b), we can see the development of the stable matching probability for the ROB instance and for two Mal-ROB instances with norm-$\phi=0.2$. Clearly, while the robustness of the orange instance is close to the ROB instance, the green instance is far less robust, despite having the same distance from the ROB instance as the orange instance. The intuitive explanation for this large difference is that, in the orange instance, the man-optimal stable matching is still the same as for the blue ROB instance, while for the green instance, the man-optimal stable matching is different and thus has a far worse worst-case robustness - and thereby also a far worse average-case robustness. 

\subsection{Additional Material for \Cref{sub:exp1}: Robustness Heuristic} \label{app:rob}
In \Cref{fig:heur-fif-corr}, we show a correlation plot between the $50\%$-threshold and the blocking pair proximity, where each point corresponds to one instance.
The observed correlation is quite strong. 

\subsection{Additional Material for \Cref{sub:exp1}: Other Stable Matchings} \label{app:oth}

In the following, we give a formulation of computing the robust matching (as defined in the main body) as an Integer Linear Program (ILP), where we will have that $U=W=[n]$ and  $x[i,j]$ is set to $1$ if man $i$ is matched to woman $j$ and to $0$ else. Variable $y[i,j]$ is the blocking distance of pair $\{i,j\}$. It is set to some large value (in our case $2n$) if $i$ and $j$ are matched. Finally, $z[i,j,\ell]$ is set to $1$ if the blocking distance of the pair $\{i,j\}$ is exactly $\ell$. Variable $z$ is needed to count the number of pairs at a certain blocking distance. See \Cref{ILP} for the constraints.

\Cref{e6,e7} guarantee that the result is a matching. \Cref{e8} guarantees that the matching is stable. \Cref{e9} sets $y[i,j]$ to the distance of the pair $\{i,j\}$ to being blocking. \Cref{e10} ensures that $z[i,j,\ell]$ is set to 1 if the distance of the pair $\{i,j\}$ from being blocking is exactly $\ell$. Finally, in \Cref{e12}, we minimize the blocking pair proximity of the matching. Obviously, we can omit the logarithm and still the same matching minimizes the expression, since the logarithm is monotonic.

\begin{figure*}[t]
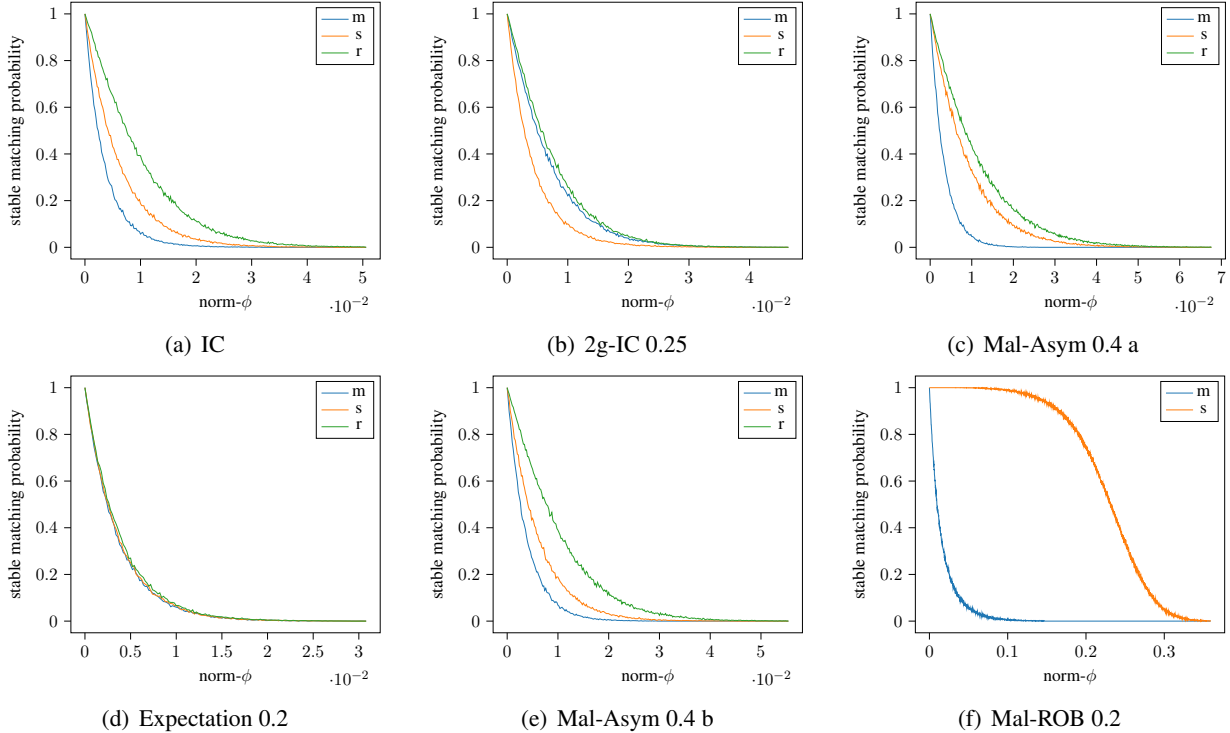

	\centering
	\begin{minipage}[t]{0.32\textwidth}
		\centering
		\subfigure[IC]{\input{plot_data/devplots3.tex}}
		\label{fig:dev3}
	\end{minipage}
	\hfill
	\begin{minipage}[t]{0.32\textwidth}
		\centering
		\subfigure[2g-IC 0.25]{\input{plot_data/devplots4.tex}}
		\label{fig:dev4}
	\end{minipage}
	\hfill
	\begin{minipage}[t]{0.32\textwidth}
		\centering
		\subfigure[Mal-Asym 0.4 a]{\input{plot_data/devplots5.tex}}
		\label{fig:dev5}
	\end{minipage}
	\hfill
	\begin{minipage}[t]{0.32\textwidth}
		\centering
		\subfigure[Expectation 0.2]{
\begin{tikzpicture}[scale=0.6]

\definecolor{color0}{rgb}{0.12156862745098,0.466666666666667,0.705882352941177}
\definecolor{color1}{rgb}{1,0.498039215686275,0.0549019607843137}
\definecolor{color2}{rgb}{0.172549019607843,0.627450980392157,0.172549019607843}

\begin{axis}[
tick align=outside,
tick pos=left,
x grid style={white!69.0196078431373!black},
xlabel={norm-$\phi$},
xmin=-0.00154, xmax=0.0323399999999999,
xtick style={color=black},
y grid style={white!69.0196078431373!black},
ylabel={stable matching probability},
ymin=-0.05, ymax=1.05,
ytick style={color=black}
]
\addplot [semithick, color0]
table {%
0 1
0.0002 0.941
0.0004 0.8952
0.0006 0.844
0.0008 0.7966
0.001 0.745
0.0012 0.7074
0.0014 0.6702
0.0016 0.6422
0.0018 0.6014
0.002 0.5556
0.0022 0.5372
0.0024 0.5064
0.0026 0.4678
0.0028 0.4566
0.003 0.4306
0.0032 0.4106
0.0034 0.3702
0.0036 0.3602
0.0038 0.328
0.004 0.3174
0.0042 0.2936
0.0044 0.287
0.0046 0.262
0.0048 0.262
0.005 0.24
0.0052 0.221
0.0054 0.2256
0.0056 0.2046
0.0058 0.1942
0.006 0.1734
0.0062 0.1764
0.0064 0.1554
0.0066 0.1486
0.0068 0.1488
0.007 0.1318
0.0072 0.1306
0.0074 0.1198
0.00759999999999999 0.1176
0.00779999999999999 0.1102
0.00799999999999999 0.1058
0.0082 0.0982
0.0084 0.0912
0.0086 0.0838
0.0088 0.0826
0.009 0.0788
0.0092 0.0732
0.0094 0.0616
0.0096 0.066
0.0098 0.0624
0.01 0.0596
0.0102 0.0576
0.0104 0.0514
0.0106 0.054
0.0108 0.0448
0.011 0.0412
0.0112 0.0422
0.0114 0.039
0.0116 0.0408
0.0118 0.0344
0.012 0.0348
0.0122 0.0322
0.0124 0.032
0.0126 0.029
0.0128 0.0278
0.013 0.0256
0.0132 0.024
0.0134 0.0238
0.0136 0.0224
0.0138 0.0168
0.014 0.0186
0.0142 0.0152
0.0144 0.0188
0.0146 0.0146
0.0148 0.0174
0.015 0.0124
0.0152 0.0162
0.0154 0.0112
0.0156 0.0126
0.0158 0.0108
0.016 0.011
0.0162 0.0114
0.0164 0.0116
0.0166 0.0084
0.0168 0.0062
0.017 0.0074
0.0172 0.0082
0.0174 0.0082
0.0176 0.0086
0.0178 0.0044
0.018 0.0064
0.0182 0.005
0.0184 0.0044
0.0186 0.0048
0.0188 0.0056
0.019 0.0058
0.0192 0.0072
0.0194 0.005
0.0196 0.0048
0.0198 0.0024
0.02 0.0024
0.0202 0.0042
0.0204 0.0048
0.0206 0.0032
0.0208 0.0032
0.021 0.0026
0.0212 0.0034
0.0214 0.0024
0.0216 0.002
0.0218 0.003
0.022 0.0012
0.0222 0.002
0.0224 0.0014
0.0226 0.0012
0.0228 0.0024
0.023 0.0014
0.0232 0.0026
0.0234 0.0008
0.0236 0.0014
0.0238 0.0014
0.024 0.0016
0.0242 0.0012
0.0244 0.0008
0.0246 0.0008
0.0248 0.0008
0.025 0.002
0.0252 0.0016
0.0254 0.0006
0.0256 0.0006
0.0258 0.0008
0.026 0.0008
0.0262 0.0008
0.0264 0.0004
0.0266 0.0004
0.0267999999999999 0.001
0.0269999999999999 0.001
0.0271999999999999 0.0008
0.0273999999999999 0.0002
0.0275999999999999 0.0006
0.0277999999999999 0.0002
0.0279999999999999 0.0004
0.0281999999999999 0.0004
0.0283999999999999 0.0004
0.0285999999999999 0.0002
0.0287999999999999 0.0002
0.0289999999999999 0.0004
0.0291999999999999 0.0002
0.0293999999999999 0.0004
0.0295999999999999 0.0002
0.0297999999999999 0.0002
0.0299999999999999 0.0004
0.0301999999999999 0.0006
0.0303999999999999 0.0004
0.0305999999999999 0.0002
0.0307999999999999 0
};
\addplot [semithick, color1]
table {%
0 1
0.0002 0.9498
0.0004 0.8922
0.0006 0.8378
0.0008 0.8062
0.001 0.7662
0.0012 0.7138
0.0014 0.6836
0.0016 0.641
0.0018 0.6108
0.002 0.5764
0.0022 0.5382
0.0024 0.5152
0.0026 0.4934
0.0028 0.4636
0.003 0.4354
0.0032 0.4104
0.0034 0.395
0.0036 0.3706
0.0038 0.3386
0.004 0.332
0.0042 0.3044
0.0044 0.2862
0.0046 0.2778
0.0048 0.2602
0.005 0.251
0.0052 0.2338
0.0054 0.2322
0.0056 0.2012
0.0058 0.1912
0.006 0.1886
0.0062 0.1794
0.0064 0.1634
0.0066 0.1586
0.0068 0.1528
0.007 0.147
0.0072 0.1304
0.0074 0.1156
0.00759999999999999 0.1134
0.00779999999999999 0.1104
0.00799999999999999 0.106
0.0082 0.1074
0.0084 0.0958
0.0086 0.0948
0.0088 0.0836
0.009 0.0834
0.0092 0.0772
0.0094 0.0754
0.0096 0.0742
0.0098 0.066
0.01 0.0674
0.0102 0.0618
0.0104 0.0542
0.0106 0.048
0.0108 0.0516
0.011 0.0478
0.0112 0.0448
0.0114 0.0378
0.0116 0.0426
0.0118 0.037
0.012 0.03
0.0122 0.0376
0.0124 0.0318
0.0126 0.0266
0.0128 0.0248
0.013 0.0254
0.0132 0.0248
0.0134 0.0218
0.0136 0.0218
0.0138 0.0192
0.014 0.0212
0.0142 0.0174
0.0144 0.0182
0.0146 0.014
0.0148 0.0126
0.015 0.015
0.0152 0.0152
0.0154 0.0132
0.0156 0.0128
0.0158 0.0128
0.016 0.0092
0.0162 0.0118
0.0164 0.0124
0.0166 0.0086
0.0168 0.0086
0.017 0.007
0.0172 0.0064
0.0174 0.009
0.0176 0.0066
0.0178 0.0056
0.018 0.0088
0.0182 0.0054
0.0184 0.005
0.0186 0.0058
0.0188 0.0074
0.019 0.004
0.0192 0.0056
0.0194 0.0036
0.0196 0.0034
0.0198 0.0056
0.02 0.003
0.0202 0.0044
0.0204 0.0036
0.0206 0.0026
0.0208 0.0036
0.021 0.0038
0.0212 0.0012
0.0214 0.0036
0.0216 0.003
0.0218 0.002
0.022 0.0024
0.0222 0.0018
0.0224 0.0014
0.0226 0.0022
0.0228 0.0012
0.023 0.003
0.0232 0.0006
0.0234 0.0014
0.0236 0.0024
0.0238 0.0014
0.024 0.001
0.0242 0.002
0.0244 0.0012
0.0246 0.0016
0.0248 0.0004
0.025 0.0012
0.0252 0.001
0.0254 0.0016
0.0256 0.0008
0.0258 0.0008
0.026 0.0004
0.0262 0.0006
0.0264 0.0006
0.0266 0.0004
0.0267999999999999 0.0006
0.0269999999999999 0.0008
0.0271999999999999 0.0004
0.0273999999999999 0.0002
0.0275999999999999 0.0006
0.0277999999999999 0.0008
0.0279999999999999 0.0004
0.0281999999999999 0.0004
0.0283999999999999 0.0002
0.0285999999999999 0
0.0287999999999999 0
0.0289999999999999 0
0.0291999999999999 0
0.0293999999999999 0
0.0295999999999999 0
0.0297999999999999 0
0.0299999999999999 0
0.0301999999999999 0
0.0303999999999999 0
0.0305999999999999 0
0.0307999999999999 0
};
\addplot [semithick, color2]
table {%
0 1
0.0002 0.9478
0.0004 0.904
0.0006 0.8556
0.0008 0.8112
0.001 0.7662
0.0012 0.7312
0.0014 0.682
0.0016 0.6662
0.0018 0.639
0.002 0.5916
0.0022 0.5688
0.0024 0.525
0.0026 0.508
0.0028 0.4758
0.003 0.4584
0.0032 0.4344
0.0034 0.4176
0.0036 0.397
0.0038 0.3734
0.004 0.3536
0.0042 0.328
0.0044 0.3154
0.0046 0.2976
0.0048 0.2684
0.005 0.27
0.0052 0.2334
0.0054 0.2382
0.0056 0.2216
0.0058 0.214
0.006 0.2004
0.0062 0.1936
0.0064 0.187
0.0066 0.1718
0.0068 0.163
0.007 0.1498
0.0072 0.1566
0.0074 0.1342
0.00759999999999999 0.1308
0.00779999999999999 0.125
0.00799999999999999 0.1166
0.0082 0.1214
0.0084 0.1008
0.0086 0.1064
0.0088 0.0988
0.009 0.0942
0.0092 0.0826
0.0094 0.08
0.0096 0.0802
0.0098 0.0776
0.01 0.0656
0.0102 0.0672
0.0104 0.0644
0.0106 0.0598
0.0108 0.056
0.011 0.0554
0.0112 0.0482
0.0114 0.0486
0.0116 0.041
0.0118 0.0486
0.012 0.0338
0.0122 0.039
0.0124 0.0342
0.0126 0.0318
0.0128 0.028
0.013 0.0286
0.0132 0.031
0.0134 0.0246
0.0136 0.0274
0.0138 0.022
0.014 0.0226
0.0142 0.0214
0.0144 0.0198
0.0146 0.0194
0.0148 0.0186
0.015 0.0188
0.0152 0.0158
0.0154 0.0172
0.0156 0.0158
0.0158 0.0134
0.016 0.0128
0.0162 0.0158
0.0164 0.0112
0.0166 0.0116
0.0168 0.011
0.017 0.01
0.0172 0.0136
0.0174 0.009
0.0176 0.0092
0.0178 0.01
0.018 0.0074
0.0182 0.0072
0.0184 0.0064
0.0186 0.006
0.0188 0.006
0.019 0.0056
0.0192 0.0082
0.0194 0.0054
0.0196 0.006
0.0198 0.005
0.02 0.0044
0.0202 0.0036
0.0204 0.0042
0.0206 0.0044
0.0208 0.0042
0.021 0.0038
0.0212 0.0046
0.0214 0.0028
0.0216 0.0016
0.0218 0.0038
0.022 0.0042
0.0222 0.0028
0.0224 0.0018
0.0226 0.002
0.0228 0.0034
0.023 0.0022
0.0232 0.0016
0.0234 0.0004
0.0236 0.002
0.0238 0.0022
0.024 0.0022
0.0242 0.002
0.0244 0.0008
0.0246 0.0006
0.0248 0.0022
0.025 0.0012
0.0252 0.0012
0.0254 0.0014
0.0256 0.0018
0.0258 0.0012
0.026 0.0004
0.0262 0.001
0.0264 0.0014
0.0266 0.0004
0.0267999999999999 0.0008
0.0269999999999999 0.0008
0.0271999999999999 0.0004
0.0273999999999999 0.0008
0.0275999999999999 0.0004
0.0277999999999999 0.0006
0.0279999999999999 0.0006
0.0281999999999999 0.0004
0.0283999999999999 0.0004
0.0285999999999999 0.0006
0.0287999999999999 0.0004
0.0289999999999999 0.0004
0.0291999999999999 0
0.0293999999999999 0
0.0295999999999999 0
0.0297999999999999 0
0.0299999999999999 0
0.0301999999999999 0
0.0303999999999999 0
0.0305999999999999 0
0.0307999999999999 0
};
\addlegendentry{m}
\addlegendentry{s}
\addlegendentry{r}
\end{axis}

\end{tikzpicture}}
		\label{fig:dev6}
	\end{minipage}
	\hfill
	\begin{minipage}[t]{0.32\textwidth}
		\centering
		\subfigure[Mal-Asym 0.4 b]{\input{plot_data/devplots7.tex}}
		\label{fig:dev7}
	\end{minipage}
	\hfill
	\begin{minipage}[t]{0.32\textwidth}
		\centering
		\subfigure[Mal-ROB 0.2]{\input{plot_data/devplots8.tex}}
		\label{fig:dev8}
	\end{minipage}
	\caption{Average-case robustness for the man-optimal (m), summed-rank minimizing (s) and robust (r) stable matchings in six exemplary instances. The cultures from which the models have been sampled are in the caption.}
	\label{fig:combdevplots}
\end{figure*}

Expanding on our discussion from the main body, we now compare the robustness of the three matchings.
In \Cref{fig:mapTog}(a), we see that the difference between the summed-rank minimal stable matching and the man-optimal stable matching is particularly large for the Mal-ROB, Mal-MD, IC and 2-IC instances. As we will see later, these are exactly the cultures that produce instances with many stable pairs (see \Cref{fig:map_stablepairs}). This suggests that in the above cultures, the summed-rank min matching can differ more significantly from the man-optimal stable matching and thus also have a significantly higher robustness. In total, we observed that there are 121 instances where the 50\%-threshold of the summed-rank-min matching is at least 0.001 greater than the man-optimal stable matching. For 222 instances, the two matchings are completely identical. There are only very few (7) instances where the man-optimal matching performs better than the summed-rank minimal matching. An example can be seen in \Cref{fig:combdevplots}(b). On the other hand, the summed-rank-min matching can also be extremely more robust, as seen in \Cref{fig:combdevplots}(f). In this Mal-ROB instance, the summed-rank-min matching corresponds to the stable matching of the ROB instance, while the man-optimal matching differs.

We now turn to the robust stable matching, which never performs worse than the man-optimal matching and also outperforms the summed-rank minimal matching in many instances. While usually the difference to the summed-rank minimal matching is rather small, \Cref{fig:combdevplots}(a,e) show that the robust stable matching can still be clearly more robust than the other matchings. As seen in \Cref{fig:mapTog}(b), the difference between summed-rank-min and robust stable matchings is particularly large for instances from the Mal-ROB, 2-IC, Rev-Euc, N-Mal and Attr cultures. We observed that there are 77 instances where the 50\%-threshold of the robust stable matching is at least 0.001 greater than the 50\%-threshold of the summed-rank-min stable matching (so one could argue that for around 14\% of the instances, the robust stable matching is significantly more robust than the summed-rank-min matching), while for 284 instances, the two matchings are exactly identical. For all remaining instances, the two matchings have a very similar robustness. 

As mentioned above, for some instances, two or all three matchings coincide. The instances we chose to present in \Cref{fig:combdevplots} have distinct man-optimal, summed-rank-min and robust matchings (except for \Cref{fig:dev8}, where the summed-rank-min and robust matching coincide) but we aim to represent all instances regarding their robustness (that is, for each instance in the data set, there is some exemplary instance that has a similar robustness). For some, mostly Euclidean, instances, the matchings are the same or they differ but their robustness is very similar. An example is shown in \Cref{fig:combdevplots}(d). But, generally, for instances where the robustness differs, the robust and summed-rank-min matchings are mostly closer to each other than to the man-optimal matching (for example in \Cref{fig:combdevplots}(c)).

\subsection{Additional Material for \Cref{sub:exp1}: Varying the Number of Agents} \label{app:var}
Now, we want to examine the impact of the instance size on the average-case robustness. Recall that in this chapter we assume $n=m$, that is, the number of men and women is equal. Therefore, the instance size only depends on the number of men $n$. Intuitively, there are two opposing factors that influence how and whether the robustness of matchings changes when increasing the number of agents: On the one hand, when increasing the instance size, we can perform more ``useless'' swaps, that 
is, swaps in the preference list $\succ_a$ that do not involve $M(a)$. This would make the instance more robust. On the other hand, by adding more agents, we obtain more potential blocking pairs, making the instance less robust.

We observed that these two effects almost cancle out: \Cref{fig:corr-inst} shows the number of swaps needed per agent to reach the 50\%-threshold (i.e. norm-$\phi \cdot \frac{n \cdot (n-1)}{4}$, where norm-$\phi$ is the 50\%-threshold) for 300 IC instances and 300 2D instances of different size. The 50\%-threshold itself rapidly decreases with increasing instance size, since the total number of possible swaps grows quadratically. In \Cref{fig:corr-inst}(a), one can observe that larger instances lead to a higher robustness, but the robustness grows very slowly. In \Cref{fig:corr-inst}(b), far less swaps are required (which also matches the results of \Cref{fig:map-fif}) and the robustness does not change significantly when adding more agents. 

\begin{figure*}[t]
	\centering
	\begin{minipage}[t]{0.49\linewidth}
		\subfigure[IC culture]{
\begin{tikzpicture}[scale=0.8]

\definecolor{color0}{rgb}{0.12156862745098,0.466666666666667,0.705882352941177}

\begin{axis}[
tick align=outside,
tick pos=left,
x grid style={white!69.0196078431373!black},
xlabel={number of men},
xmin=-3.95, xmax=104.95,
xtick style={color=black},
y grid style={white!69.0196078431373!black},
ylabel={av. swaps per agent to reach 50\%-threshold},
ymin=-0.2751875, ymax=5.7789375,
ytick style={color=black}
]
\addplot [semithick, color0, mark=*, mark size=1, mark options={solid}, only marks]
table {%
1 0
2 0.5
3 0.3475
4 0.945
5 0.825
6 1.875
7 1.05
8 1.2425
9 1.12
10 0.8325
11 0.95
12 3.63
13 1.47
14 0.92625
15 2.1175
16 0.95625
17 1.88
18 2.38
19 1.1925
20 1.71
21 1.75
22 0.84
23 1.815
24 1.2075
25 2.34
26 2.625
27 1.235
28 1.9575
29 2.31
30 1.34125
31 1.875
32 3.875
33 1.4
34 1.52625
35 1.7425
36 2.31875
37 3.96
38 1.665
39 1.4725
40 2.29125
41 3
42 1.07625
43 1.68
44 2.9025
45 2.42
46 2.25
47 4.255
48 1.58625
49 1.68
50 2.94
51 2.375
52 5.355
53 2.34
54 2.84875
55 2.2275
56 2.3375
57 2.03
58 2.28
59 1.595
60 3.245
61 2.325
62 2.135
63 1.86
64 2.52
65 3.12
66 3.25
67 1.7325
68 2.1775
69 3.57
70 2.32875
71 2.275
72 2.57375
73 2.25
74 1.825
75 3.0525
76 2.0625
77 3.8
78 1.82875
79 3.4125
80 3.65375
81 4
82 4.455
83 2.255
84 3.32
85 2.625
86 4.56875
87 3.7625
88 2.8275
89 3.08
90 2.11375
91 1.8
92 2.38875
93 4.6
94 4.06875
95 2.115
96 3.8
97 2.64
98 1.94
99 3.5525
100 3.83625
1 0
2 0.5
3 0.4925
4 0.37875
5 0.62
6 1.525
7 0.57
8 1.015
9 1.24
10 1.78875
11 2.7375
12 2.0625
13 1.08
14 1.7225
15 1.645
16 1.35
17 2.96
18 2.44375
19 1.17
20 1.49625
21 0.825
22 3.6225
23 1.3475
24 2.04125
25 2.1
26 1.28125
27 1.3975
28 2.2275
29 3.185
30 1.55875
31 2.7
32 1.97625
33 1.88
34 2.145
35 2.6775
36 2.1
37 1.575
38 5.50375
39 2.8975
40 1.60875
41 2.45
42 1.48625
43 2.205
44 2.63375
45 1.98
46 4.5
47 1.495
48 2.2325
49 2.04
50 2.45
51 1.625
52 2.805
53 3.25
54 3.64375
55 2.295
56 3.09375
57 2.17
58 1.78125
59 3.335
60 2.8025
61 1.725
62 2.44
63 2.4025
64 3.07125
65 3.04
66 2.1125
67 4.3725
68 1.8425
69 2.975
70 4.83
71 2.1
72 2.21875
73 2.07
74 2.7375
75 2.775
76 3.5625
77 4.75
78 3.08
79 2.145
80 2.46875
81 2
82 2.2275
83 3.485
84 5.1875
85 3.57
86 3.71875
87 1.72
88 2.3925
89 2.42
90 4.005
91 2.5875
92 2.16125
93 3.22
94 5.3475
95 3.6425
96 2.85
97 3.48
98 2.425
99 3.0625
100 2.35125
1 0
2 0.49875
3 1.0325
4 0.9375
5 0.575
6 0.73125
7 0.9525
8 1.61875
9 1.18
10 1.665
11 1.6125
12 0.9075
13 1.095
14 1.18625
15 1.8725
16 1.3125
17 1.62
18 1.21125
19 0.99
20 2.2325
21 2.475
22 1.155
23 1.4025
24 1.12125
25 2.16
26 1.375
27 2.405
28 3.1725
29 1.645
30 1.74
31 1.875
32 2.4025
33 4.16
34 2.2275
35 2.55
36 2.84375
37 1.395
38 1.3875
39 2.3275
40 4.5825
41 2.05
42 2.9725
43 1.8375
44 2.0425
45 1.485
46 1.6875
47 2.645
48 2.29125
49 2.1
50 2.3275
51 2.125
52 1.6575
53 3.64
54 1.325
55 2.835
56 1.1
57 2.03
58 1.63875
59 1.74
60 3.31875
61 2.7
62 2.82125
63 2.48
64 2.59875
65 2.4
66 2.275
67 2.7225
68 2.8475
69 2.72
70 2.9325
71 3.15
72 2.39625
73 2.97
74 4.015
75 2.405
76 2.8125
77 4.37
78 2.98375
79 3.705
80 1.975
81 2.3
82 3.4425
83 3.3825
84 2.59375
85 2.205
86 2.44375
87 2.4725
88 4.02375
89 3.96
90 2.8925
91 4.05
92 3.07125
93 3.335
94 2.67375
95 4.465
96 1.78125
97 2.52
98 4.24375
99 3.185
100 2.475
};
\end{axis}

\end{tikzpicture}}
		\label{fig:corr-inst-1}
	\end{minipage}
	\hfill
	\begin{minipage}[t]{0.49\linewidth}
		\subfigure[2D culture]{
\begin{tikzpicture}[scale=0.8]

\definecolor{color0}{rgb}{0.12156862745098,0.466666666666667,0.705882352941177}

\begin{axis}[
tick align=outside,
tick pos=left,
x grid style={white!69.0196078431373!black},
xlabel={number of men},
xmin=-3.95, xmax=104.95,
xtick style={color=black},
y grid style={white!69.0196078431373!black},
ylabel={av. swaps per agent to reach 50\%-threshold},
ymin=-0.0982125, ymax=2.0624625,
ytick style={color=black}
]
\addplot [semithick, color0, mark=*, mark size=1, mark options={solid}, only marks]
table {%
1 0
1 0
1 0
2 0.357999999999995
2 0.4885
2 0.492375
3 0.94025
3 0.471
3 0.42025
4 1.137375
4 0.422625
4 0.785249999999996
5 0.861
5 0.7815
5 1.2445
6 1.364375
6 0.960624999999999
6 0.938125
7 0.411
7 0.63975
7 0.81825
8 0.872375
8 0.4375
8 0.86275
9 0.852
9 0.563
9 0.845
10 0.7065
10 0.617625
10 1.063125
11 0.86875
11 0.86125
11 0.8625
12 0.521125
12 1.027125
12 0.897875
13 0.721500000000002
13 0.694500000000002
13 0.7545
14 0.77675
14 0.567125
14 0.93275
15 1.12875
15 0.8225
15 0.5145
16 1.145625
16 1.130625
16 0.815625
17 0.701999999999998
17 0.962
17 1.042
18 0.8245
18 1.187875
18 0.898875
19 1.96425
19 0.681750000000001
19 0.66825
20 0.9405
20 0.893
20 1.1115
21 0.927499999999999
21 1.235
21 0.93
22 0.90825
22 0.693
22 0.937125
23 0.858
23 0.9405
23 0.979
24 0.94875
24 0.698625000000002
24 1.0925
25 0.882
25 0.936
25 0.591
26 0.728125
26 0.687500000000001
26 0.940625
27 0.6825
27 0.67925
27 0.74425
28 0.84375
28 1.002375
28 0.958500000000001
29 1.099
29 0.938
29 0.7945
30 0.822875
30 0.80475
30 1.06575
31 0.92625
31 0.87
31 0.8775
32 0.732375
32 0.8525
32 0.957125
33 0.788
33 0.736
33 0.876
34 1.14675
34 0.71775
34 1.117875
35 1.10075
35 1.071
35 0.63325
36 0.93625
36 1.010625
36 0.93625
37 0.7335
37 0.7695
37 0.9585
38 0.7955
38 0.999
38 0.744625
39 1.045
39 0.89775
39 0.85025
40 0.901875
40 0.862875
40 0.9165
41 1.01
41 0.985
41 0.91
42 1.15825
42 1.040375
42 1.05575
43 0.787499999999999
43 1.12875
43 1.50675
44 0.757875
44 1.1395
44 0.89225
45 1.034
45 0.572
45 0.8525
46 0.894375
46 0.770625
46 0.68625
47 0.8165
47 1.10975
47 1.081
48 1.122125
48 0.97525
48 0.7755
49 1.338
49 0.966
49 0.828
50 0.8085
50 1.329125
50 1.194375
51 0.95
51 0.95
51 0.8
52 0.949875
52 0.822375
52 1.16025
53 1.1505
53 0.858
53 0.715
54 0.7155
54 0.848
54 0.907625
55 0.80325
55 1.053
55 0.77625
56 0.78375
56 1.134375
56 0.99
57 0.875
57 0.742
57 0.798
58 0.81225
58 1.118625
58 0.69825
59 0.96425
59 0.841
59 0.84825
60 1.069375
60 0.877625
60 1.025125
61 0.87
61 1.05
61 0.9
62 1.159
62 0.86925
62 0.976
63 1.2555
63 0.99975
63 0.98425
64 0.984375
64 1.1025
64 0.945
65 0.808
65 0.88
65 0.832
66 0.918125
66 0.788125
66 0.853125
67 0.81675
67 0.792
67 1.0395
68 0.98825
68 0.979875
68 0.804
69 0.8585
69 0.9775
69 1.1305
70 0.836625
70 0.8625
70 1.060875
71 1.015
71 0.95375
71 0.8925
72 0.949625
72 0.94075
72 0.887499999999999
73 1.125
73 0.774
73 0.981
74 0.976375
74 0.757375
74 0.903375
75 0.69375
75 0.8695
75 1.2025
76 0.99375
76 0.84375
76 1.059375
77 1.007
77 1.064
77 0.874
78 0.837375
78 0.952875
78 1.164625
79 0.90675
79 1.014
79 0.80925
80 0.977625
80 0.88875
80 0.957875
81 1
81 0.93
81 0.94
82 0.860625
82 1.053
82 1.0125
83 0.85075
83 0.89175
83 0.738
84 0.76775
84 0.923374999999999
84 0.93375
85 0.8925
85 0.924
85 0.7455
86 1.051875
86 0.87125
86 0.85
87 0.903
87 0.999749999999999
87 0.9245
88 0.8265
88 0.76125
88 1.10925
89 0.968
89 0.748
89 0.891
90 0.86775
90 0.86775
90 1.012375
91 1.0575
91 1.0125
91 0.9
92 0.819
92 0.966875
92 1.02375
93 0.966
93 0.7935
93 0.9545
94 0.8835
94 0.7905
94 0.848625
95 0.752
95 0.81075
95 0.85775
96 0.914375
96 1.128125
96 0.771875
97 0.792
97 0.948
97 0.828
98 0.97
98 0.860875
98 0.848750000000001
99 1.1025
99 0.833
99 1.2005
100 1.225125
100 0.9405
100 0.928125
};
\end{axis}

\end{tikzpicture}}
		\label{fig:corr-inst-2}
	\end{minipage}
	\caption{Correlation of instance size and average-case robustness of men-optimal matching for 300 instances (three instances of size $n=m=i$ for each $i \in [100]$) of two different cultures}
	\label{fig:corr-inst}
\end{figure*}
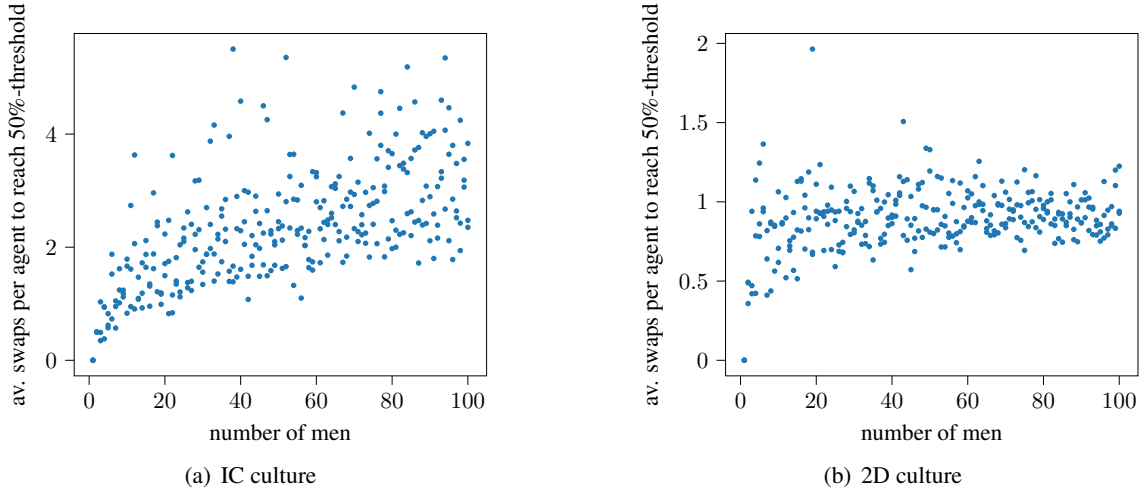

\subsection{Additional Material for \Cref{sub:exp1}: Blocking Pairs}\label{app:bp}

In this section, instead of only considering stable and unstable matchings, we count the number of blocking pairs for a matching when adding random noise. The justification is that while an instance may have a low stable matching probability, the average number of blocking pairs in the perturbed instance may still be relatively low. In the main body, we showed that this problem is already hard for the worst-case robustness and thus we cannot efficiently compute the average number of blocking pairs for a specific number of swaps. Therefore, we again use the Mallows model to determine the average number of blocking pairs. 

\Cref{fig:dev-bps} shows the development of the average number of blocking pairs for increasing norm-$\phi$ values for five exemplary instances of the data set. The instances not depicted in the plot behave quite similarly to these five selected instances. Contrary to the stable matching probability, and as can be seen in \Cref{fig:dev-bps}, the number of blocking pairs increases almost linearly with increasing norm-$\phi$ (at least for small norm-$\phi$ values). Therefore, we use the average number of blocking pairs for norm-$\phi=0.1$ as our measure instead of looking for a specific value (e.g. 50\%, as we did for the stability setting).
We observe that the number of blocking pairs and the 50\% stability threshold have a quite strong correlation. However, for very (similarly) unrobust instances, the number of blocking pairs after some random noise can significantly differ.

In \Cref{fig:corr-bps-fif}, we can see that, at least for the man-optimal matchings, two instances that have a very similar 50\%-stability-threshold can have a very different average number of blocking pairs. Therefore, it is sensible to examine which instances have a particularly low number of average blocking pairs. \Cref{fig:map-avgbps} shows the map of our instances colored according to their average number of blocking pairs. Notice that a small value now corresponds to a robust instance. 
While the robust instances remain similar to the stability setting, we can now better differentiate between non-robust instances. We observe that Euclidean cultures that do not involve perturbing single preference lists (like 1D, 2D, Attr and Rev-Euc) have the most blocking pairs. The instances have less blocking pairs when they are closer to the Mal-ROB and IC instances.
\begin{figure}[t]
	\centering
		\includegraphics[width=0.7\columnwidth]{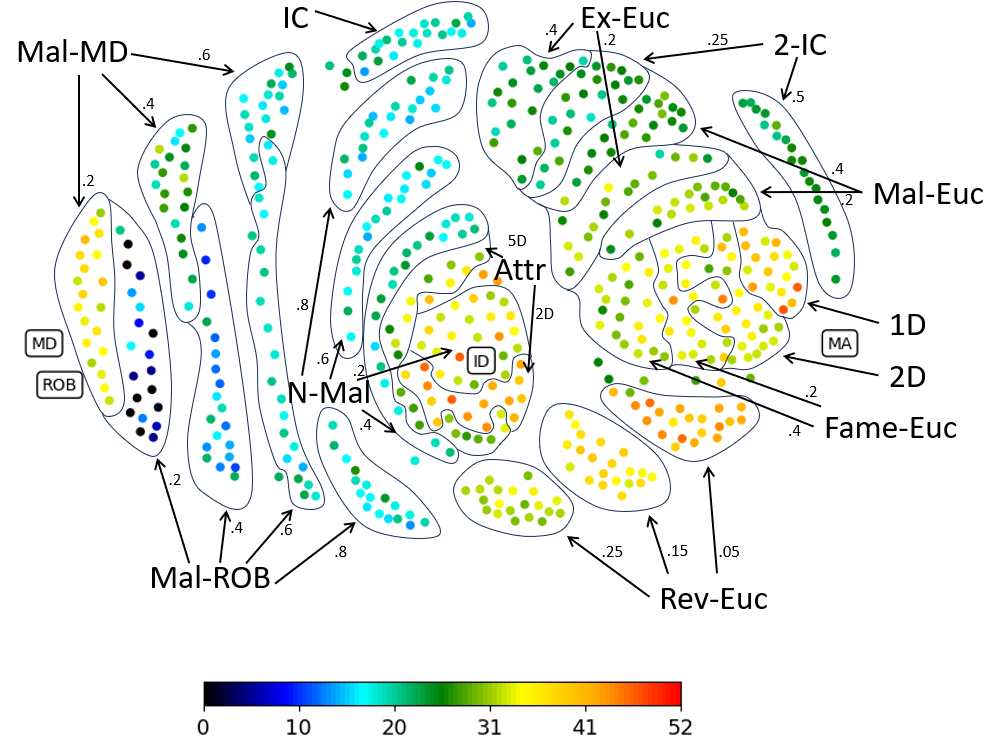}
		\caption{Average number of blocking pairs  of men-optimal matching for norm-$\phi=0.1$ for each instance of our data set.}
		\label{fig:map-avgbps}
\end{figure}

\subsection{Additional Material for \Cref{sub:exp1}: Unstable Matchings}\label{app:uns}
We want to shortly examine the stable matching probability of matchings that are initially not stable. As there is a very large number of unstable matchings, we restrict ourselves to one specific (possibly unstable) matching, which is defined analogously to the summed-rank min stable matching, but the stability criterion is omitted, i.e. we consider the matching maximizing social welfare, regardless of whether it is stable or not. More formally, we define this matching as follows:
\begin{definition}
	Let $\mathcal{I}=(U,W,\mathcal{P})$ be an SM instance. The \emph{minimum weight matching} is the (not necessarily stable) matching $M$ that minimizes $\sum_{a \in A}\rk_a(M(a))$.
\end{definition}
Examining the robustness of the minimum weight matching for all instances of the data set, we find that most matchings that are not stable in the initial preferences have a very low stable matching probability for any norm-$\phi$ value.

For any instance except for the Mal-ROB culture and the MA and ROB instances, it contained at least 4 blocking pairs which resulted in having a stable matching probability of 0\% for every norm-$\phi$ value. However, in \Cref{fig:dev15}, we present one special mallows-robust instance with norm-$\phi=0.2$ where the min-weight matching initially contains only one blocking pair. In this case, the matching becomes even more robust than the initially stable matchings. Notice that this is an outlier and for usual instances that are not closely related to the ROB instance, the minimum-weight matching is very unrobust.
\begin{figure}[t]
	\begin{minipage}[t]{.49\textwidth}
		\centering
		\input plot_data/devplots15.tex
		\caption{Average-Case Robustness for the man-optimal and summed-rank minimal stable matchings as well as for the min-weight matching for Mal-ROB 0.25.}
		\label{fig:dev15}
	\end{minipage}
	\hfill
	\begin{minipage}[t]{.49\textwidth}
		\centering
\begin{tikzpicture}[scale=0.8]

\definecolor{color0}{rgb}{0.12156862745098,0.466666666666667,0.705882352941177}
\definecolor{color1}{rgb}{1,0.498039215686275,0.0549019607843137}
\definecolor{color2}{rgb}{0.172549019607843,0.627450980392157,0.172549019607843}
\definecolor{color3}{rgb}{0.83921568627451,0.152941176470588,0.156862745098039}
\definecolor{color4}{rgb}{0.580392156862745,0.403921568627451,0.741176470588235}

\begin{axis}[
tick align=outside,
tick pos=left,
x grid style={white!69.0196078431373!black},
xlabel={norm-\(\displaystyle \phi\)},
xmin=-0.0125, xmax=0.2625,
xtick style={color=black},
y grid style={white!69.0196078431373!black},
ylabel={average number of blocking pairs},
ymin=-6.485, ymax=136.185,
ytick style={color=black},
legend style={at={(0,0.98)},xshift=0.2cm,anchor=north west,nodes=right}
]
\addplot [semithick, color0]
table {%
0 0
0.025 14.18
0.05 26.03
0.075 39.41
0.1 52.15
0.125 62.93
0.15 75.22
0.175 90.84
0.2 103.91
0.225 115.79
0.25 129.7
};
\addplot [semithick, color4]
table {%
0 0
0.025 12.74
0.05 23.24
0.075 33.02
0.1 40.97
0.125 48.29
0.15 54.46
0.175 64.07
0.2 70.43
0.225 76.84
0.25 86.64
};

\addplot [semithick, color3]
table {%
0 0
0.025 6
0.05 11.88
0.075 16.63
0.1 21.82
0.125 29.48
0.15 36.04
0.175 43.03
0.2 49.7
0.225 61.47
0.25 70.83
};
\addplot [semithick, color2]
table {%
0 0
0.025 4.63
0.05 9.31
0.075 14.09
0.1 18.27
0.125 23.86
0.15 30.61
0.175 35.3
0.2 41.36
0.225 47.21
0.25 54.3
};
\addplot [semithick, color1]
table {%
0 0
0.025 1.37
0.05 2.24
0.075 3.85
0.1 5.11
0.125 6.3
0.15 7.4
0.175 8.07
0.2 9.8
0.225 11.57
0.25 13.47
};
\addlegendentry{Asym}
\addlegendentry{1d}
\addlegendentry{N-Mal 0.4}
\addlegendentry{IC}
\addlegendentry{Mal-ROB 0.2}
\end{axis}

\end{tikzpicture}
		\caption{Average number of blocking pairs for the man-optimal matchings of five exemplary instances, for norm-$\phi$-values from 0 to 0.25.}
		\label{fig:dev-bps}
	\end{minipage}
	
\end{figure}

\subsection{Additional Material for \Cref{sub:exp2}} \label{app:exp2}
In this section, we analyze the robustness of stable pairs. Analogously to the matching setting, for each SM instance $\mathcal{I}$ and stable pair $\{m,w\}$, we define the 50\%-stability threshold as the smallest norm-$\phi$-value such that the probability of $\{m,w\}$ being stable in an instance drawn from the Mallows model with this dispersion parameter from $\mathcal{I}$ is below 50\%. To compute all stable pairs of an instance, we use the algorithm by \citet{gusfield1987three}. Clearly, each instance contains at least $n$ stable pairs. In \Cref{fig:map_stablepairs}, one can see the number of stable pairs for each instance in our data set.
Especially the MD and Mal-MD have many stable pairs. The distribution of the 50\%-threshold of the stable pairs of all instances can be seen in \Cref{fig:distpairs}. Notably, stable pairs with a 50\%-threshold above 0.1 are by no means an exception. In fact, as we will see later with the average stable pair robustness, an average stable pair often achieves a robustness greater than 0.1 and the low robustness values seen in \Cref{fig:distpairs} are mostly due to the many and unrobust stable pairs of the MD and Mal-MD instances.
\begin{figure*}[t]
	\begin{minipage}[t]{0.45\linewidth}
		\centering
		\includegraphics[width=\textwidth]{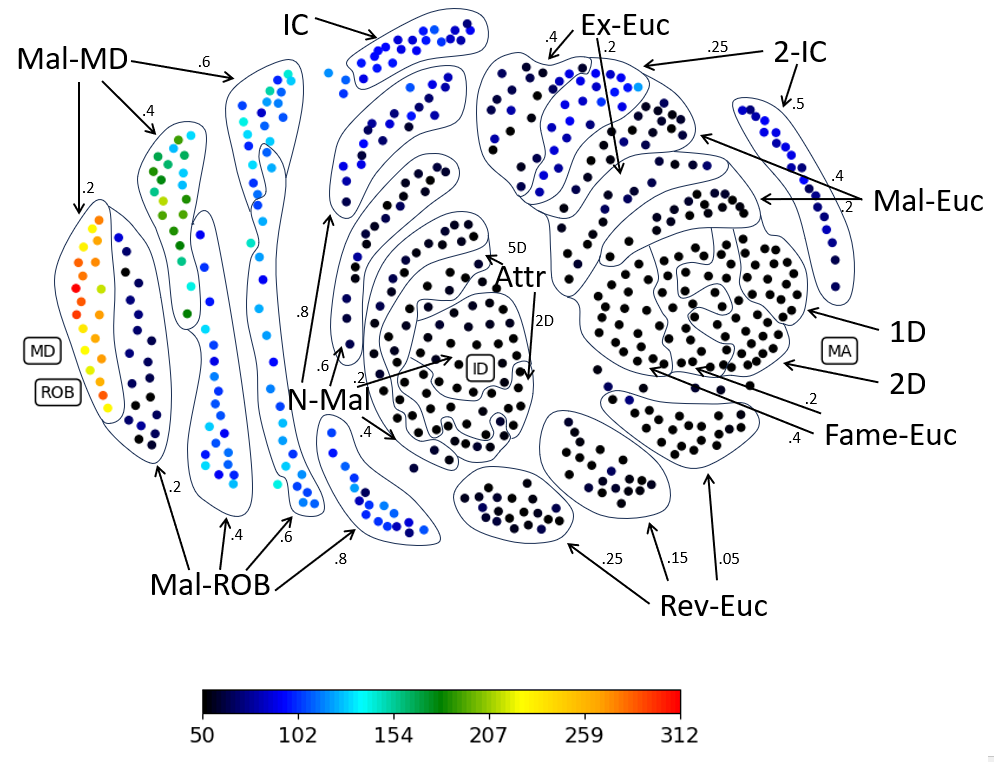}
		\caption{Number of stable pairs in the initial instance.}
		\label{fig:map_stablepairs}
	\end{minipage}\hfill
	\begin{minipage}[t]{0.45\linewidth}
\begin{tikzpicture}[scale=0.8]

\definecolor{color0}{rgb}{0.12156862745098,0.466666666666667,0.705882352941177}

\begin{axis}[
tick align=outside,
tick pos=left,
x grid style={white!69.0196078431373!black},
xlabel={50\%-threshold},
xmin=-0.05445, xmax=0.6,
xtick style={color=black},
y grid style={white!69.0196078431373!black},
ylabel={number of stable pairs},
ymin=0, ymax=7254.45,
ytick style={color=black}
]
\draw[draw=none,fill=color0] (axis cs:-0.0045,0) rectangle (axis cs:0.0045,3698);
\draw[draw=none,fill=color0] (axis cs:0.0055,0) rectangle (axis cs:0.0145,6909);
\draw[draw=none,fill=color0] (axis cs:0.0155,0) rectangle (axis cs:0.0245,3893);
\draw[draw=none,fill=color0] (axis cs:0.0255,0) rectangle (axis cs:0.0345,2907);
\draw[draw=none,fill=color0] (axis cs:0.0355,0) rectangle (axis cs:0.0445,2397);
\draw[draw=none,fill=color0] (axis cs:0.0455,0) rectangle (axis cs:0.0545,2065);
\draw[draw=none,fill=color0] (axis cs:0.0555,0) rectangle (axis cs:0.0645,1895);
\draw[draw=none,fill=color0] (axis cs:0.0655,0) rectangle (axis cs:0.0745,1747);
\draw[draw=none,fill=color0] (axis cs:0.0755,0) rectangle (axis cs:0.0845,1628);
\draw[draw=none,fill=color0] (axis cs:0.0855,0) rectangle (axis cs:0.0945,1424);
\draw[draw=none,fill=color0] (axis cs:0.0955,0) rectangle (axis cs:0.1045,1289);
\draw[draw=none,fill=color0] (axis cs:0.1055,0) rectangle (axis cs:0.1145,1150);
\draw[draw=none,fill=color0] (axis cs:0.1155,0) rectangle (axis cs:0.1245,1084);
\draw[draw=none,fill=color0] (axis cs:0.1255,0) rectangle (axis cs:0.1345,1025);
\draw[draw=none,fill=color0] (axis cs:0.1355,0) rectangle (axis cs:0.1445,865);
\draw[draw=none,fill=color0] (axis cs:0.1455,0) rectangle (axis cs:0.1545,811);
\draw[draw=none,fill=color0] (axis cs:0.1555,0) rectangle (axis cs:0.1645,776);
\draw[draw=none,fill=color0] (axis cs:0.1655,0) rectangle (axis cs:0.1745,687);
\draw[draw=none,fill=color0] (axis cs:0.1755,0) rectangle (axis cs:0.1845,610);
\draw[draw=none,fill=color0] (axis cs:0.1855,0) rectangle (axis cs:0.1945,569);
\draw[draw=none,fill=color0] (axis cs:0.1955,0) rectangle (axis cs:0.2045,575);
\draw[draw=none,fill=color0] (axis cs:0.2055,0) rectangle (axis cs:0.2145,533);
\draw[draw=none,fill=color0] (axis cs:0.2155,0) rectangle (axis cs:0.2245,468);
\draw[draw=none,fill=color0] (axis cs:0.2255,0) rectangle (axis cs:0.2345,483);
\draw[draw=none,fill=color0] (axis cs:0.2355,0) rectangle (axis cs:0.2445,412);
\draw[draw=none,fill=color0] (axis cs:0.2455,0) rectangle (axis cs:0.2545,402);
\draw[draw=none,fill=color0] (axis cs:0.2555,0) rectangle (axis cs:0.2645,399);
\draw[draw=none,fill=color0] (axis cs:0.2655,0) rectangle (axis cs:0.2745,399);
\draw[draw=none,fill=color0] (axis cs:0.2755,0) rectangle (axis cs:0.2845,387);
\draw[draw=none,fill=color0] (axis cs:0.2855,0) rectangle (axis cs:0.2945,319);
\draw[draw=none,fill=color0] (axis cs:0.2955,0) rectangle (axis cs:0.3045,266);
\draw[draw=none,fill=color0] (axis cs:0.3055,0) rectangle (axis cs:0.3145,284);
\draw[draw=none,fill=color0] (axis cs:0.3155,0) rectangle (axis cs:0.3245,273);
\draw[draw=none,fill=color0] (axis cs:0.3255,0) rectangle (axis cs:0.3345,218);
\draw[draw=none,fill=color0] (axis cs:0.3355,0) rectangle (axis cs:0.3445,220);
\draw[draw=none,fill=color0] (axis cs:0.3455,0) rectangle (axis cs:0.3545,192);
\draw[draw=none,fill=color0] (axis cs:0.3555,0) rectangle (axis cs:0.3645,198);
\draw[draw=none,fill=color0] (axis cs:0.3655,0) rectangle (axis cs:0.3745,147);
\draw[draw=none,fill=color0] (axis cs:0.3755,0) rectangle (axis cs:0.3845,152);
\draw[draw=none,fill=color0] (axis cs:0.3855,0) rectangle (axis cs:0.3945,152);
\draw[draw=none,fill=color0] (axis cs:0.3955,0) rectangle (axis cs:0.4045,149);
\draw[draw=none,fill=color0] (axis cs:0.4055,0) rectangle (axis cs:0.4145,131);
\draw[draw=none,fill=color0] (axis cs:0.4155,0) rectangle (axis cs:0.4245,92);
\draw[draw=none,fill=color0] (axis cs:0.4255,0) rectangle (axis cs:0.4345,102);
\draw[draw=none,fill=color0] (axis cs:0.4355,0) rectangle (axis cs:0.4445,104);
\draw[draw=none,fill=color0] (axis cs:0.4455,0) rectangle (axis cs:0.4545,106);
\draw[draw=none,fill=color0] (axis cs:0.4555,0) rectangle (axis cs:0.4645,104);
\draw[draw=none,fill=color0] (axis cs:0.4655,0) rectangle (axis cs:0.4745,123);
\draw[draw=none,fill=color0] (axis cs:0.4755,0) rectangle (axis cs:0.4845,130);
\draw[draw=none,fill=color0] (axis cs:0.4855,0) rectangle (axis cs:0.4945,133);
\draw[draw=none,fill=color0] (axis cs:0.4955,0) rectangle (axis cs:0.5045,142);
\draw[draw=none,fill=color0] (axis cs:0.5055,0) rectangle (axis cs:0.5145,148);
\draw[draw=none,fill=color0] (axis cs:0.5155,0) rectangle (axis cs:0.5245,94);
\draw[draw=none,fill=color0] (axis cs:0.5255,0) rectangle (axis cs:0.5345,75);
\draw[draw=none,fill=color0] (axis cs:0.5355,0) rectangle (axis cs:0.5445,56);
\draw[draw=none,fill=color0] (axis cs:0.5455,0) rectangle (axis cs:0.5545,24);
\draw[draw=none,fill=color0] (axis cs:0.5555,0) rectangle (axis cs:0.5645,11);
\draw[draw=none,fill=color0] (axis cs:0.5655,0) rectangle (axis cs:0.5745,1);
\draw[draw=none,fill=color0] (axis cs:0.5755,0) rectangle (axis cs:0.5845,3);
\draw[draw=none,fill=color0] (axis cs:0.5855,0) rectangle (axis cs:0.5945,1);
\draw[draw=none,fill=color0] (axis cs:0.5955,0) rectangle (axis cs:0.6045,0);
\draw[draw=none,fill=color0] (axis cs:0.6055,0) rectangle (axis cs:0.6145,0);
\draw[draw=none,fill=color0] (axis cs:0.6155,0) rectangle (axis cs:0.6245,0);
\draw[draw=none,fill=color0] (axis cs:0.6255,0) rectangle (axis cs:0.6345,0);
\draw[draw=none,fill=color0] (axis cs:0.6355,0) rectangle (axis cs:0.6445,0);
\draw[draw=none,fill=color0] (axis cs:0.6455,0) rectangle (axis cs:0.6545,0);
\draw[draw=none,fill=color0] (axis cs:0.6555,0) rectangle (axis cs:0.6645,0);
\draw[draw=none,fill=color0] (axis cs:0.6655,0) rectangle (axis cs:0.6745,0);
\draw[draw=none,fill=color0] (axis cs:0.6755,0) rectangle (axis cs:0.6845,0);
\draw[draw=none,fill=color0] (axis cs:0.6855,0) rectangle (axis cs:0.6945,0);
\draw[draw=none,fill=color0] (axis cs:0.6955,0) rectangle (axis cs:0.7045,0);
\draw[draw=none,fill=color0] (axis cs:0.7055,0) rectangle (axis cs:0.7145,0);
\draw[draw=none,fill=color0] (axis cs:0.7155,0) rectangle (axis cs:0.7245,0);
\draw[draw=none,fill=color0] (axis cs:0.7255,0) rectangle (axis cs:0.7345,0);
\draw[draw=none,fill=color0] (axis cs:0.7355,0) rectangle (axis cs:0.7445,0);
\draw[draw=none,fill=color0] (axis cs:0.7455,0) rectangle (axis cs:0.7545,0);
\draw[draw=none,fill=color0] (axis cs:0.7555,0) rectangle (axis cs:0.7645,0);
\draw[draw=none,fill=color0] (axis cs:0.7655,0) rectangle (axis cs:0.7745,0);
\draw[draw=none,fill=color0] (axis cs:0.7755,0) rectangle (axis cs:0.7845,0);
\draw[draw=none,fill=color0] (axis cs:0.7855,0) rectangle (axis cs:0.7945,0);
\draw[draw=none,fill=color0] (axis cs:0.7955,0) rectangle (axis cs:0.8045,0);
\draw[draw=none,fill=color0] (axis cs:0.8055,0) rectangle (axis cs:0.8145,0);
\draw[draw=none,fill=color0] (axis cs:0.8155,0) rectangle (axis cs:0.8245,0);
\draw[draw=none,fill=color0] (axis cs:0.8255,0) rectangle (axis cs:0.8345,0);
\draw[draw=none,fill=color0] (axis cs:0.8355,0) rectangle (axis cs:0.8445,0);
\draw[draw=none,fill=color0] (axis cs:0.8455,0) rectangle (axis cs:0.8545,0);
\draw[draw=none,fill=color0] (axis cs:0.8555,0) rectangle (axis cs:0.8645,0);
\draw[draw=none,fill=color0] (axis cs:0.8655,0) rectangle (axis cs:0.8745,0);
\draw[draw=none,fill=color0] (axis cs:0.8755,0) rectangle (axis cs:0.8845,0);
\draw[draw=none,fill=color0] (axis cs:0.8855,0) rectangle (axis cs:0.8945,0);
\draw[draw=none,fill=color0] (axis cs:0.8955,0) rectangle (axis cs:0.9045,0);
\draw[draw=none,fill=color0] (axis cs:0.9055,0) rectangle (axis cs:0.9145,0);
\draw[draw=none,fill=color0] (axis cs:0.9155,0) rectangle (axis cs:0.9245,0);
\draw[draw=none,fill=color0] (axis cs:0.9255,0) rectangle (axis cs:0.9345,0);
\draw[draw=none,fill=color0] (axis cs:0.9355,0) rectangle (axis cs:0.9445,0);
\draw[draw=none,fill=color0] (axis cs:0.9455,0) rectangle (axis cs:0.9545,0);
\draw[draw=none,fill=color0] (axis cs:0.9555,0) rectangle (axis cs:0.9645,0);
\draw[draw=none,fill=color0] (axis cs:0.9655,0) rectangle (axis cs:0.9745,0);
\draw[draw=none,fill=color0] (axis cs:0.9755,0) rectangle (axis cs:0.9845,0);
\draw[draw=none,fill=color0] (axis cs:0.9855,0) rectangle (axis cs:0.9945,0);
\end{axis}

\end{tikzpicture}
		\caption{Distribution of the 50\%-thresholds of all stable pairs in our data set. }
		\label{fig:distpairs}
	\end{minipage}
\end{figure*}
\subsubsection{Easy Robustness measures}
As for the matching setting, we want to find a measure that has a high correlation with the 50\%-stability-threshold. Again, we cannot hope for an exact measure due to related hardness result. Furthermore, the stable pair robustness does not only depend on the distance of some specific pairs to being blocking (as in the matching case), but it also depends on the stability of other pairs, making it more unlikely to find an easy measure with a strong correlation. 

For each woman $w'$ that $m$ prefers to $w$, we will add a penalty depending on the rank of $m$ in $w'$ (namely $n-\rk_{w'}(m)$). If $w'$ ranks $m$ high, $\{m,w'\}$ has a good chance to be blocking in a matching containing $\{m,w\}$.

\begin{definition}
	Let $\mathcal{I}=(U,W,\mathcal{P})$ be an SM instance with $|U|=|W|=n$ and let $\{m,w\}$ be a pair with $m \in U$ and $w \in W$. The \emph{blocking score} of $\{m,w\}$ is 
	\begin{align*}
	\mathcal{B}_\mathcal{I}(m,w)=&\sum_{i=1}^{\rk_m(w)-1}{n-\rk_{p_m(i)}(m)}\\
	&+\sum_{i=1}^{\rk_w(m)-1}{n-\rk_{p_w(i)}(w)}
	\end{align*}
\end{definition}

Note that if $m$ and $w$ are mutual top-choices, the blocking score is 0. As for the blocking pair proximity of a matching, small values correspond to a high robustness.

This measure does not achieve a correlation as strong as the blocking pair proximity in the matching setting, the Pearson Correlation Coefficient being around -0.54. However, it can still be a good way to compare the pair robustness of stable pairs of the same instance. \Cref{fig:corr-pairheur-fif} shows the correlation for the stable pairs of three exemplary instances. For two of them, we can observe a somewhat strong correlation, while for the Norm-Mal instance, the blocking score is not a good indicator of the 50\%-threshold.

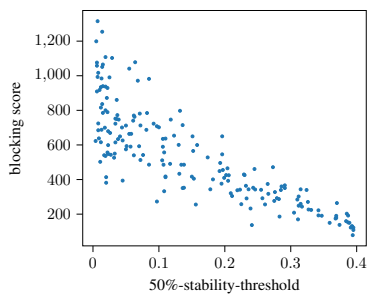
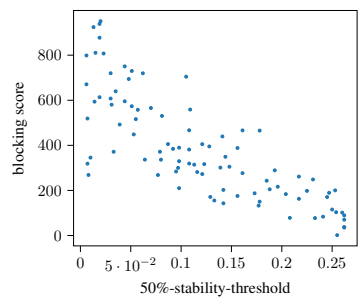
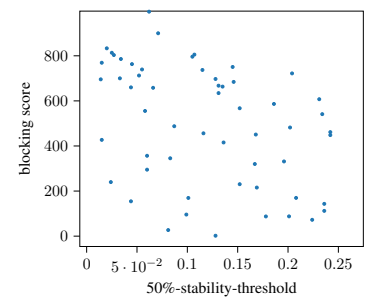
\begin{figure*}[t]
	\begin{minipage}[t]{0.32 \linewidth}
		\centering
		\subfigure[Mal-Asym 0.4]{
\begin{tikzpicture}[scale=0.55]

\definecolor{color0}{rgb}{0.12156862745098,0.466666666666667,0.705882352941177}

\begin{axis}[
tick align=outside,
tick pos=left,
x grid style={white!69.0196078431373!black},
xlabel={50\%-stability-threshold},
xmin=-0.01555, xmax=0.41455,
xtick style={color=black},
y grid style={white!69.0196078431373!black},
ylabel={blocking score},
ymin=16.632050096, ymax=1377.142800464,
ytick style={color=black}
]
\addplot [semithick, color0, mark=*, mark size=1, mark options={solid}, only marks]
table {%
0.034 614.73198976
0.368 174.55489776
0.138 485.81868816
0.106 635.06624112
0.0970000000000001 707.9492304
0.061 740.32966528
0.182 400.8170776
0.0720000000000001 512.90399024
0.153 405.04439312
0.23 425.69286576
0.0760000000000001 541.86968592
0.022 929.25964816
0.0850000000000001 982.022488
0.225 257.56962544
0.206 424.78126912
0.022 789.58360448
0.239 232.57744016
0.02 381.47043648
0.387 191.75055376
0.202 425.87900368
0.291 366.43864768
0.0920000000000001 723.00911312
0.0720000000000001 712.71471424
0.11 413.56300512
0.019 536.42044
0.023 547.8303272
0.025 990.5898264
0.315 343.47162544
0.393 130.47475568
0.317 243.44151984
0.178 343.50447984
0.007 1315.30140272
0.312 248.77106016
0.136 350.68052672
0.064 1078.07882848
0.245 460.7911328
0.369 190.36104976
0.132 797.43509008
0.266 317.41237968
0.279 292.31172096
0.037 772.49756048
0.015 1064.85882496
0.106 589.55151248
0.062 539.98366592
0.196 545.66271776
0.0820000000000001 784.52185168
0.265 374.44643632
0.291 350.82745904
0.219 338.06212544
0.342 222.12351376
0.15 416.87396672
0.019 800.14750896
0.394 78.47344784
0.013 984.30643088
0.158 599.56145488
0.041 649.26691696
0.331 224.31115344
0.014 1051.05059984
0.05 712.03118064
0.232 363.89970096
0.05 574.69600832
0.174 527.9587896
0.312 302.0550112
0.11 636.81764592
0.193 495.87477776
0.275 276.41063408
0.276 327.49154352
0.0810000000000001 687.29822144
0.359 149.67644368
0.393 128.1889264
0.105 511.0525952
0.123 653.81160768
0.282 287.29369968
0.021 871.7191312
0.039 746.48492816
0.389 149.76771568
0.034 722.9535104
0.31 284.2946376
0.285 339.35837952
0.038 633.66731104
0.013 687.77032224
0.385 202.60971632
0.023 679.33416992
0.044 700.8759768
0.029 1102.76466736
0.014 1254.41162864
0.055 1041.08042608
0.327 226.76374992
0.156 255.50593776
0.347 190.75885984
0.017 540.43946848
0.131 433.304572
0.151 649.634336
0.241 136.72444672
0.283 185.88555584
0.033 549.83124048
0.071 781.21435984
0.152 565.23570256
0.243 352.5965288
0.211 304.68960976
0.342 192.67144864
0.019 1107.82880176
0.058 664.25773088
0.311 169.7320624
0.045 393.85873152
0.02 413.77695776
0.13 600.3682296
0.395 121.9737712
0.199 409.75310144
0.006 1076.09611216
0.009 637.96518656
0.236 343.633664
0.036 749.88727344
0.273 471.37789648
0.387 152.22776368
0.008 684.90425968
0.197 455.28388944
0.004 622.72448032
0.015 939.60170176
0.205 392.79334432
0.324 339.73477472
0.036 860.27600448
0.018 699.72072736
0.062 769.99312064
0.019 937.2761784
0.01 917.7272184
0.134 485.61752016
0.108 332.35501568
0.233 290.93926176
0.118 746.53331152
0.006 909.15787008
0.223 341.56830304
0.138 352.63329728
0.027 543.20394688
0.107 485.47301696
0.032 526.27520064
0.023 598.80004256
0.117 490.79512416
0.07 592.47138032
0.229 429.470152
0.326 270.67150176
0.011 502.4346392
0.286 360.25737216
0.136 714.60537808
0.255 340.74235088
0.055 663.83158576
0.196 649.79118384
0.015 785.29251328
0.006 1057.86074512
0.051 593.36883392
0.057 590.30251184
0.195 376.05444768
0.192 448.91804416
0.063 762.62003488
0.395 108.37338112
0.246 342.5146136
0.261 256.08459856
0.374 137.9674272
0.315 259.41734368
0.108 556.77626896
0.026 669.36386032
0.007 993.08271728
0.005 1198.95868112
0.305 208.53380832
0.391 122.50350416
0.0970000000000001 272.9051304
0.014 836.19672944
0.008 1016.85190624
0.068 971.3506888
0.045 625.9612584
0.0850000000000001 486.09858896
0.204 464.95353088
0.1 701.84059648
0.012 816.03035728
0.209 320.84318768
0.11 417.82045792
0.008 724.23599056
0.012 934.23401648
0.384 159.53192432
0.369 178.40370848
0.022 556.34389776
0.012 615.61966848
0.257 291.82798336
0.37 264.09144208
0.034 752.0587784
};
\end{axis}

\end{tikzpicture}}
	\end{minipage}
	\hfill
	\begin{minipage}[t]{0.32 \linewidth}
		\centering
			\subfigure[2g-IC 0.5]{
\begin{tikzpicture}[scale=0.55]

\definecolor{color0}{rgb}{0.12156862745098,0.466666666666667,0.705882352941177}

\begin{axis}[
tick align=outside,
tick pos=left,
x grid style={white!69.0196078431373!black},
xlabel={50\%-stability-threshold},
xmin=-0.00680000000000001, xmax=0.2748,
xtick style={color=black},
y grid style={white!69.0196078431373!black},
ylabel={blocking score},
ymin=-45.890201976, ymax=998.183314136,
ytick style={color=black}
]
\addplot [semithick, color0, mark=*, mark size=1, mark options={solid}, only marks]
table {%
0.245 170.63134624
0.051 573.48725216
0.02 950.72542704
0.044 750.32470544
0.26 101.67097072
0.231 248.63252448
0.0980000000000001 210.09108496
0.108 381.4277712
0.0800000000000001 336.51068336
0.064 336.73419712
0.062 719.65135296
0.035 639.71555728
0.156 175.5707096
0.008 268.65112592
0.007 318.77143024
0.254 103.3939168
0.177 132.741924
0.113 314.4079864
0.121 272.59819584
0.142 142.92937264
0.178 150.29341072
0.053 448.47135664
0.0920000000000001 384.54709872
0.144 349.41187904
0.193 289.236596
0.0870000000000001 405.55051888
0.185 243.08953184
0.139 301.2526368
0.031 580.3878176
0.116 281.51150976
0.108 466.89117968
0.057 557.65257552
0.25 115.30745344
0.033 371.52996288
0.006 670.68821248
0.013 924.41597392
0.0770000000000001 268.43598912
0.241 83.71022496
0.161 276.7262944
0.007 518.98886288
0.039 492.44383888
0.055 516.30294048
0.161 466.3310616
0.023 807.42942576
0.014 593.80619472
0.0980000000000001 389.71332736
0.109 558.83097232
0.148 305.15620608
0.019 613.46786112
0.178 465.47582704
0.133 155.39019184
0.129 171.16622
0.253 200.31527616
0.019 877.50947088
0.006 798.63181232
0.044 595.3918344
0.048 694.27536864
0.128 395.19259632
0.247 189.53424576
0.0810000000000001 530.61056528
0.173 187.3462672
0.262 35.55051616
0.141 439.46136432
0.123 316.75734752
0.217 162.87184096
0.01 345.5765088
0.204 183.7265096
0.03 719.88932144
0.051 729.93219376
0.208 78.42790064
0.262 37.62879504
0.0790000000000001 371.70982784
0.156 388.26973904
0.03 607.773004
0.0950000000000001 283.97544944
0.233 77.44975664
0.262 89.36713968
0.019 938.65270544
0.0980000000000001 329.65222784
0.142 202.39157408
0.262 70.304996
0.225 197.95259536
0.015 810.51342128
0.105 704.62529136
0.217 261.52345872
0.255 1.56768512
0.0970000000000001 300.19467776
0.121 405.28904736
0.07 565.6893544
0.196 216.56694128
0.108 319.16731792
0.188 204.90650864
};
\end{axis}

\end{tikzpicture}}
	\end{minipage}
	\hfill
	\begin{minipage}[t]{0.32 \linewidth}
		\centering
			\subfigure[Norm-Mallows 0.4]{
\begin{tikzpicture}[scale=0.55]

\definecolor{color0}{rgb}{0.12156862745098,0.466666666666667,0.705882352941177}

\begin{axis}[
tick align=outside,
tick pos=left,
x grid style={white!69.0196078431373!black},
xlabel={50\%-stability-threshold},
xmin=-0.00680000000000001, xmax=0.2748,
xtick style={color=black},
y grid style={white!69.0196078431373!black},
ylabel={blocking score},
ymin=-45.890201976, ymax=998.183314136,
ytick style={color=black}
]
\addplot [semithick, color0, mark=*, mark size=1, mark options={solid}, only marks]
table {%
0.06 294.95398176
0.152 230.45816176
0.128 1.87791696
0.0990000000000001 95.75225328
0.101 169.47190208
0.0830000000000001 345.41356192
0.06 356.54866752
0.044 154.5122616
0.208 169.70686096
0.178 87.43261312
0.201 87.74619856
0.055 739.35014304
0.116 456.18927936
0.196 331.380592
0.024 239.78629024
0.236 143.3061432
0.052 712.50929104
0.0810000000000001 26.99220384
0.224 72.51269952
0.136 415.54684944
0.045 763.11515888
0.062 995.81292464
0.033 700.01648944
0.169 215.54963872
0.234 541.17478768
0.071 899.91791472
0.202 481.95391392
0.066 658.07401136
0.167 320.02693648
0.236 112.1629368
0.204 721.9199752
0.058 555.35752352
0.231 607.42549936
0.168 450.708532
0.145 750.5198632
0.105 795.99402096
0.135 663.35883184
0.03 1003.5298384
0.013 1130.78535024
0.152 567.3302856
0.044 659.92898288
0.034 785.81618752
0.186 586.48657824
0.131 634.55305312
0.146 684.22973888
0.128 697.16450752
0.242 461.51221216
0.0870000000000001 487.57372896
0.115 737.19328928
0.131 667.47261792
0.014 695.58147696
0.025 813.4428672
0.015 769.20202416
0.242 448.22485872
0.02 833.18217104
0.015 427.1077872
0.027 803.33322608
0.107 805.27540352
};
\end{axis}

\end{tikzpicture}}
	\end{minipage}
	\caption{Correlation between the 50\%-threshold and the blocking score for all stable pairs of three instances. The cultures from which the instances have been sampled can be found in the caption.}
	\label{fig:corr-pairheur-fif}
\end{figure*}

Based on these stable pairs, for each instance, we examine the following values:
\begin{itemize}
	\item the \textbf{most robust stable pair} is the stable pair with the highest 50\%-stability threshold.
	\item the \textbf{least robust stable pair} is the stable pair with the lowest 50\%-stability threshold.
	\item the \textbf{average stable pair robustness} is the average over the 50\%-stability thresholds of all stable pairs.
	\item the \textbf{robustness variance} is the variance over the 50\%-stability thresholds of all stable pairs.
\end{itemize}
In the following subsections, we will closer examine these indicators for the instances of our data set.
In \Cref{fig:devpairs}, we chose six instances that are representative for the data set regarding the robustness of their stable pairs. For each instance, we depict the most and least robust stable pair as well as some interesting stable pairs in between. We will use these plots to draw attention to some interesting observations regarding our pair robustness indicators. 

\begin{figure*}[t]
	\begin{minipage}[t]{0.3\linewidth}
		\centering
		\subfigure[Mal-Asym 0.6]{
\begin{tikzpicture}[scale=0.6]

\definecolor{color0}{rgb}{0.12156862745098,0.466666666666667,0.705882352941177}
\definecolor{color1}{rgb}{1,0.498039215686275,0.0549019607843137}
\definecolor{color2}{rgb}{0.172549019607843,0.627450980392157,0.172549019607843}

\begin{axis}[
tick align=outside,
tick pos=left,
x grid style={white!69.0196078431373!black},
xlabel={norm-$\phi$},
xmin=-0.04925, xmax=1.03425,
xtick style={color=black},
y grid style={white!69.0196078431373!black},
ylabel={stable pair probability},
ymin=-0.02879, ymax=1.04899,
ytick style={color=black}
]

\addplot [semithick, color0]
table {%
0 1
0.005 1
0.01 1
0.015 1
0.02 1
0.025 1
0.03 0.9988
0.035 0.9986
0.04 0.9992
0.045 0.998
0.05 0.9982
0.055 0.9972
0.06 0.9962
0.065 0.9964
0.07 0.9964
0.075 0.995
0.08 0.993
0.085 0.991
0.09 0.9936
0.095 0.992
0.1 0.9896
0.105 0.992
0.11 0.9882
0.115 0.9856
0.12 0.9832
0.125 0.9812
0.13 0.979
0.135 0.9796
0.14 0.9762
0.145 0.971
0.15 0.969
0.155 0.9666
0.16 0.9594
0.165 0.9524
0.17 0.951
0.175 0.949
0.18 0.9404
0.185 0.9384
0.19 0.9374
0.195 0.9302
0.2 0.9258 
0.205 0.9218
0.21 0.9124
0.215 0.9032
0.22 0.9028
0.225 0.8876
0.23 0.8794
0.235 0.872
0.24 0.8782
0.245 0.8674
0.25 0.855
0.255 0.8384
0.26 0.8414
0.265 0.8198
0.27 0.8164
0.275 0.8112
0.28 0.809
0.285 0.7896
0.29 0.7932
0.295 0.7718
0.3 0.7662
0.305 0.7652
0.31 0.7456
0.315 0.7272
0.32 0.7244
0.325 0.7162
0.33 0.699
0.335 0.697
0.34 0.6732
0.345 0.6802
0.35 0.6668
0.355 0.6658
0.36 0.645
0.365 0.639
0.37 0.6234
0.375 0.621
0.38 0.6058
0.385 0.5996
0.39 0.5934
0.395 0.5824
0.4 0.574
0.405 0.553
0.41 0.537
0.415 0.5464
0.42 0.53
0.425 0.5126
0.43 0.5058
0.435 0.5034
0.44 0.4968
0.445 0.4928
0.45 0.4882
0.455 0.4514
0.46 0.4418
0.465 0.4512
0.47 0.4388
0.475 0.4258
0.48 0.4274
0.485 0.414
0.49 0.4122
0.495 0.399
0.5 0.3902
0.505 0.3766
0.51 0.3766
0.515 0.3664
0.52 0.3448
0.525 0.3562
0.53 0.3464
0.535 0.3312
0.54 0.3294
0.545 0.331
0.55 0.3088
0.555 0.3086
0.56 0.3084
0.565 0.288
0.57 0.2892
0.575 0.2666
0.58 0.2782
0.585 0.2684
0.59 0.263
0.595 0.2642
0.6 0.2494
0.605 0.2424
0.625 0.2358
0.645 0.2024
0.665 0.184
0.685 0.192
0.705000000000001 0.1546
0.725000000000001 0.1402
0.745000000000001 0.1338
0.765000000000001 0.1158
0.785000000000001 0.1148
0.805000000000001 0.097
0.825000000000001 0.0926
0.845000000000001 0.0818
0.865000000000001 0.07
0.885000000000001 0.0722
0.905000000000001 0.0592
0.925000000000001 0.0536
0.945000000000001 0.0436
0.965000000000001 0.0422
0.985000000000001 0.0388
};
\addplot [semithick, color2]
table {%
0 1
0.005 0.84
0.01 0.7478
0.015 0.7028
0.02 0.6692
0.025 0.665
0.03 0.6478
0.035 0.6298
0.04 0.6382
0.045 0.6286
0.05 0.621
0.055 0.6128
0.06 0.6146
0.065 0.614
0.07 0.6068
0.075 0.5974
0.08 0.5974
0.085 0.6038
0.09 0.5854
0.095 0.5756
0.1 0.5902
0.105 0.5896
0.11 0.5826
0.115 0.5674
0.12 0.5724
0.125 0.5702
0.13 0.5696
0.135 0.5554
0.14 0.562
0.145 0.5618
0.15 0.5562
0.155 0.5586
0.16 0.5516
0.165 0.5552
0.17 0.5472
0.175 0.546
0.18 0.5368
0.185 0.5388
0.19 0.5284
0.195 0.5288
0.2 0.5226 
0.205 0.507
0.21 0.5144
0.215 0.5044
0.22 0.508
0.225 0.502
0.23 0.4986
0.235 0.4904
0.24 0.5018
0.245 0.4908
0.25 0.4754
0.255 0.4806
0.26 0.4808
0.265 0.4726
0.27 0.47
0.275 0.4462
0.28 0.4568
0.285 0.4492
0.29 0.4536
0.295 0.452
0.3 0.4342
0.305 0.4274
0.31 0.4342
0.315 0.4126
0.32 0.4204
0.325 0.4038
0.33 0.3886
0.335 0.4074
0.34 0.4008
0.345 0.3942
0.35 0.3928
0.355 0.3792
0.36 0.3702
0.365 0.381
0.37 0.3594
0.375 0.3666
0.38 0.3702
0.385 0.3558
0.39 0.3374
0.395 0.3382
0.4 0.3418
0.405 0.3384
0.41 0.336
0.415 0.3258
0.42 0.3228
0.425 0.3132
0.43 0.3074
0.435 0.3
0.44 0.29
0.445 0.2974
0.45 0.291
0.455 0.2966
0.46 0.281
0.465 0.2872
0.47 0.2796
0.475 0.2624
0.48 0.255
0.485 0.2642
0.49 0.2596
0.495 0.2526
0.5 0.2568
0.505 0.2482
0.51 0.2356
0.515 0.2428
0.52 0.2328
0.525 0.2458
0.53 0.2246
0.535 0.225
0.54 0.2214
0.545 0.2258
0.55 0.2216
0.555 0.2108
0.56 0.2056
0.565 0.2072
0.57 0.2074
0.575 0.1888
0.58 0.1918
0.585 0.188
0.59 0.1874
0.595 0.1864
0.6 0.1882
0.605 0.1706
0.625 0.1548
0.645 0.1484
0.665 0.1538
0.685 0.135
0.705000000000001 0.127
0.725000000000001 0.1112
0.745000000000001 0.102
0.765000000000001 0.096
0.785000000000001 0.1024
0.805000000000001 0.0816
0.825000000000001 0.0842
0.845000000000001 0.0666
0.865000000000001 0.0668
0.885000000000001 0.0562
0.905000000000001 0.06
0.925000000000001 0.055
0.945000000000001 0.0464
0.965000000000001 0.0458
0.985000000000001 0.039
};
\addplot [semithick, color1]
table {%
0 1
0.005 0.8432
0.01 0.7294
0.015 0.6202
0.02 0.5328
0.025 0.4678
0.03 0.42
0.035 0.3696
0.04 0.328
0.045 0.3006
0.05 0.2694
0.055 0.2362
0.06 0.222
0.065 0.1906
0.07 0.1826
0.075 0.1558
0.08 0.1516
0.085 0.1326
0.09 0.1236
0.095 0.112
0.1 0.0978
0.105 0.1034
0.11 0.09
0.115 0.0886
0.12 0.0762
0.125 0.0748
0.13 0.0714
0.135 0.0672
0.14 0.0622
0.145 0.0652
0.15 0.0686
0.155 0.058
0.16 0.0524
0.165 0.0464
0.17 0.0486
0.175 0.0502
0.18 0.0432
0.185 0.0454
0.19 0.0358
0.195 0.0442
0.2 0.0412
0.205 0.0394
0.21 0.0366
0.215 0.0404
0.22 0.0346
0.225 0.037
0.23 0.0336
0.235 0.0352
0.24 0.0358
0.245 0.031
0.25 0.033
0.255 0.0314
0.26 0.0284
0.265 0.0282
0.27 0.0306
0.275 0.028
0.28 0.0284
0.285 0.0264
0.29 0.0324
0.295 0.024
0.3 0.0304
0.305 0.0306
0.31 0.0258
0.315 0.027
0.32 0.0244
0.325 0.0214
0.33 0.0254
0.335 0.027
0.34 0.0268
0.345 0.0252
0.35 0.0222
0.355 0.025
0.36 0.0278
0.365 0.0238
0.37 0.0226
0.375 0.0266
0.38 0.0282
0.385 0.0262
0.39 0.0248
0.395 0.0256
0.4 0.0242
0.405 0.0254
0.41 0.025
0.415 0.023
0.42 0.0234
0.425 0.029
0.43 0.0252
0.435 0.0252
0.44 0.0288
0.445 0.0248
0.45 0.0202
0.455 0.025
0.46 0.0288
0.465 0.027
0.47 0.0226
0.475 0.0222
0.48 0.0248
0.485 0.0234
0.49 0.0256
0.495 0.027
0.5 0.028
0.505 0.0306
0.51 0.0274
0.515 0.0294
0.52 0.0252
0.525 0.0246
0.53 0.0322
0.535 0.0232
0.54 0.0246
0.545 0.0266
0.55 0.0214
0.555 0.0286
0.56 0.0286
0.565 0.0306
0.57 0.0274
0.575 0.0294
0.58 0.0322
0.585 0.0246
0.59 0.0278
0.595 0.0322
0.6 0.0264
0.605 0.0252
0.625 0.0298
0.645 0.0328
0.665 0.034
0.685 0.0334
0.705000000000001 0.0334
0.725000000000001 0.0394
0.745000000000001 0.0344
0.765000000000001 0.0314
0.785000000000001 0.037
0.805000000000001 0.035
0.825000000000001 0.0332
0.845000000000001 0.037
0.865000000000001 0.0318
0.885000000000001 0.0384
0.905000000000001 0.0374
0.925000000000001 0.0342
0.945000000000001 0.0346
0.965000000000001 0.0428
0.985000000000001 0.0382
};

\addlegendentry{A (78)}
\addlegendentry{B (249)}
\addlegendentry{C (1189)}
\end{axis}

\end{tikzpicture}}
		\label{fig:dev11}
	\end{minipage}
	\hfill
	\begin{minipage}[t]{0.3\linewidth}
		\centering
		\subfigure[Fame-Euc 0.4]{
\begin{tikzpicture}[scale=0.6]

\definecolor{color0}{rgb}{0.12156862745098,0.466666666666667,0.705882352941177}
\definecolor{color1}{rgb}{1,0.498039215686275,0.0549019607843137}
\definecolor{color2}{rgb}{0.172549019607843,0.627450980392157,0.172549019607843}

\begin{axis}[
tick align=outside,
tick pos=left,
x grid style={white!69.0196078431373!black},
xlabel={norm-$\phi$},
xmin=-0.04925, xmax=1.03425,
xtick style={color=black},
y grid style={white!69.0196078431373!black},
ylabel={stable pair probability},
ymin=-0.0332, ymax=1.0492,
ytick style={color=black}
]

\addplot [semithick, color0]
table {%
0 1
0.005 0.9862
0.01 0.9588
0.015 0.929
0.02 0.8976
0.025 0.8658
0.03 0.8298
0.035 0.8188
0.04 0.7974
0.045 0.7646
0.05 0.7554
0.055 0.7364
0.06 0.7282
0.065 0.7196
0.07 0.6936
0.075 0.6914
0.08 0.6814
0.085 0.6808
0.09 0.6648
0.095 0.642
0.1 0.648
0.105 0.6464
0.11 0.6482
0.115 0.6434
0.12 0.6352
0.125 0.6244
0.13 0.6214
0.135 0.6098
0.14 0.6052
0.145 0.5958
0.15 0.5916
0.155 0.5864
0.16 0.5884
0.165 0.5808
0.17 0.5896
0.175 0.5658
0.18 0.5796
0.185 0.5592
0.19 0.5618
0.195 0.5466
0.2 0.5398
0.205 0.539
0.21 0.5448
0.215 0.5292
0.22 0.513
0.225 0.5194
0.23 0.508
0.235 0.5214
0.24 0.492
0.245 0.5058
0.25 0.5002
0.255 0.494
0.26 0.488
0.265 0.4918
0.27 0.4756
0.275 0.467
0.28 0.4648
0.285 0.4614
0.29 0.4518
0.295 0.445
0.3 0.4334
0.305 0.4292
0.31 0.4424
0.315 0.4326
0.32 0.4294
0.325 0.4112
0.33 0.407
0.335 0.4148
0.34 0.4016
0.345 0.4092
0.35 0.3962
0.355 0.3924
0.36 0.3908
0.365 0.3876
0.37 0.3774
0.375 0.3832
0.38 0.364
0.385 0.3776
0.39 0.372
0.395 0.361
0.4 0.3578
0.405 0.3542
0.41 0.3528
0.415 0.3464
0.42 0.3484
0.425 0.3398
0.43 0.3512
0.435 0.3286
0.44 0.3276
0.445 0.3288
0.45 0.329
0.455 0.3296
0.46 0.3184
0.465 0.3258
0.47 0.3144
0.475 0.3256
0.48 0.3192
0.485 0.2976
0.49 0.3088
0.495 0.312
0.5 0.2954
0.505 0.2962
0.51 0.2884
0.515 0.2934
0.52 0.2912
0.525 0.283
0.53 0.2824
0.535 0.266
0.54 0.276
0.545 0.2582
0.55 0.2614
0.555 0.268
0.56 0.2644
0.565 0.2626
0.57 0.2608
0.575 0.263
0.58 0.2522
0.585 0.248
0.605 0.2298
0.625 0.227
0.645 0.2122
0.665 0.2092
0.685 0.1976
0.705000000000001 0.1828
0.725000000000001 0.1822
0.745000000000001 0.1658
0.765000000000001 0.153
0.785000000000001 0.134
0.805000000000001 0.1368
0.825000000000001 0.1136
0.845000000000001 0.1044
0.865000000000001 0.0996
0.885000000000001 0.0818
0.905000000000001 0.0782
0.925000000000001 0.067
0.945000000000001 0.0588
0.965000000000001 0.0504
0.985000000000001 0.0478
};
\addplot [semithick, color1]
table {%
0 1
0.005 0.9946
0.01 0.9814
0.015 0.968
0.02 0.9576
0.025 0.9346
0.03 0.9214
0.035 0.9106
0.04 0.8992
0.045 0.8786
0.05 0.8696
0.055 0.8556
0.06 0.849
0.065 0.829
0.07 0.8118
0.075 0.7898
0.08 0.7784
0.085 0.7742
0.09 0.7512
0.095 0.7344
0.1 0.7224
0.105 0.7022
0.11 0.6828
0.115 0.6672
0.12 0.6558
0.125 0.6446
0.13 0.6256
0.135 0.6154
0.14 0.5898
0.145 0.572
0.15 0.571
0.155 0.5486
0.16 0.5376
0.165 0.523
0.17 0.517
0.175 0.5152
0.18 0.4852
0.185 0.469
0.19 0.4668
0.195 0.4528
0.2 0.439 
0.205 0.4386
0.21 0.4296
0.215 0.413
0.22 0.399
0.225 0.3958
0.23 0.3886
0.235 0.3858
0.24 0.3688
0.245 0.3744
0.25 0.3506
0.255 0.3424
0.26 0.3466
0.265 0.3412
0.27 0.3198
0.275 0.3036
0.28 0.3186
0.285 0.3078
0.29 0.2986
0.295 0.297
0.3 0.2982
0.305 0.2796
0.31 0.2786
0.315 0.2708
0.32 0.2732
0.325 0.2682
0.33 0.259
0.335 0.2472
0.34 0.2482
0.345 0.2366
0.35 0.2386
0.355 0.2344
0.36 0.2266
0.365 0.2278
0.37 0.2288
0.375 0.2226
0.38 0.2166
0.385 0.2128
0.39 0.2022
0.395 0.201
0.4 0.1908
0.405 0.1898
0.41 0.2004
0.415 0.1842
0.42 0.1844
0.425 0.1816
0.43 0.1794
0.435 0.178
0.44 0.1796
0.445 0.1664
0.45 0.1566
0.455 0.1742
0.46 0.1654
0.465 0.1634
0.47 0.1624
0.475 0.1592
0.48 0.1642
0.485 0.1628
0.49 0.1524
0.495 0.162
0.5 0.1582
0.505 0.146
0.51 0.1492
0.515 0.1432
0.52 0.1472
0.525 0.14
0.53 0.1404
0.535 0.1382
0.54 0.138
0.545 0.142
0.55 0.139
0.555 0.1212
0.56 0.1228
0.565 0.1324
0.57 0.1326
0.575 0.1318
0.58 0.1282
0.585 0.12
0.605 0.119
0.625 0.1068
0.645 0.1106
0.665 0.0976
0.685 0.103
0.705000000000001 0.099
0.725000000000001 0.0992
0.745000000000001 0.0876
0.765000000000001 0.0852
0.785000000000001 0.0756
0.805000000000001 0.0726
0.825000000000001 0.0654
0.845000000000001 0.061
0.865000000000001 0.0646
0.885000000000001 0.0562
0.905000000000001 0.0516
0.925000000000001 0.0518
0.945000000000001 0.0474
0.965000000000001 0.0388
0.985000000000001 0.041
};
\addplot [semithick, color2]
table {%
0 1
0.005 0.8638
0.01 0.7292
0.015 0.59
0.02 0.4946
0.025 0.4182
0.03 0.3394
0.035 0.2838
0.04 0.2508
0.045 0.2206
0.05 0.193
0.055 0.1842
0.06 0.1814
0.065 0.1628
0.07 0.1498
0.075 0.1456
0.08 0.1448
0.085 0.1414
0.09 0.1218
0.095 0.1258
0.1 0.111
0.105 0.116
0.11 0.1042
0.115 0.1018
0.12 0.1044
0.125 0.0992
0.13 0.1048
0.135 0.1006
0.14 0.0844
0.145 0.09
0.15 0.0922
0.155 0.086
0.16 0.0804
0.165 0.075
0.17 0.0744
0.175 0.0838
0.18 0.0714
0.185 0.0782
0.19 0.0742
0.195 0.0748
0.2 0.0656
0.205 0.0672
0.21 0.0728
0.215 0.0584
0.22 0.0664
0.225 0.064
0.23 0.0564
0.235 0.0626
0.24 0.0628
0.245 0.056
0.25 0.0592
0.255 0.0606
0.26 0.0568
0.265 0.0618
0.27 0.0608
0.275 0.0598
0.28 0.0606
0.285 0.0556
0.29 0.0486
0.295 0.054
0.3 0.0602
0.305 0.055
0.31 0.0518
0.315 0.0518
0.32 0.053
0.325 0.0528
0.33 0.0552
0.335 0.0524
0.34 0.0506
0.345 0.0466
0.35 0.0532
0.355 0.0512
0.36 0.0508
0.365 0.0536
0.37 0.0494
0.375 0.0502
0.38 0.0524
0.385 0.0538
0.39 0.0446
0.395 0.0496
0.4 0.0484
0.405 0.0474
0.41 0.0484
0.415 0.0448
0.42 0.047
0.425 0.0456
0.43 0.049
0.435 0.0504
0.44 0.0414
0.445 0.0456
0.45 0.0472
0.455 0.0444
0.46 0.048
0.465 0.051
0.47 0.0404
0.475 0.0478
0.48 0.044
0.485 0.0406
0.49 0.0456
0.495 0.0418
0.5 0.0422
0.505 0.0448
0.51 0.041
0.515 0.0442
0.52 0.0444
0.525 0.0392
0.53 0.0418
0.535 0.037
0.54 0.0442
0.545 0.045
0.55 0.041
0.555 0.04
0.56 0.042
0.565 0.0436
0.57 0.043
0.575 0.0412
0.58 0.036
0.585 0.0412
0.605 0.0388
0.625 0.043
0.645 0.0396
0.665 0.0336
0.685 0.0394
0.705000000000001 0.0352
0.725000000000001 0.0348
0.745000000000001 0.0384
0.765000000000001 0.034
0.785000000000001 0.0418
0.805000000000001 0.0372
0.825000000000001 0.0434
0.845000000000001 0.0382
0.865000000000001 0.0364
0.885000000000001 0.0354
0.905000000000001 0.0388
0.925000000000001 0.0404
0.945000000000001 0.039
0.965000000000001 0.0368
0.985000000000001 0.0394
};

\addlegendentry{A (121)}
\addlegendentry{B (448)}
\addlegendentry{C (1035)}
\end{axis}

\end{tikzpicture}}
		\label{fig:dev14}
	\end{minipage}
	\hfill
	\begin{minipage}[t]{0.3\linewidth}
		\centering
		\subfigure[Rev-Euc 0.05 a]{
\begin{tikzpicture}[scale=0.6]

\definecolor{color0}{rgb}{0.12156862745098,0.466666666666667,0.705882352941177}
\definecolor{color1}{rgb}{1,0.498039215686275,0.0549019607843137}
\definecolor{color2}{rgb}{0.172549019607843,0.627450980392157,0.172549019607843}

\begin{axis}[
tick align=outside,
tick pos=left,
x grid style={white!69.0196078431373!black},
xlabel={norm-$\phi$},
xmin=-0.04975, xmax=1.04475,
xtick style={color=black},
y grid style={white!69.0196078431373!black},
ylabel={stable pair probability},
ymin=-0.02963, ymax=1.04903,
ytick style={color=black}
]

\addplot [semithick, color0]
table {%
0 1
0.005 0.9462
0.01 0.9006
0.015 0.861
0.02 0.819
0.025 0.7896
0.03 0.7616
0.035 0.732
0.04 0.7132
0.045 0.6962
0.05 0.6676
0.055 0.6436
0.06 0.6286
0.065 0.6216
0.07 0.5908
0.075 0.5798
0.08 0.5794
0.085 0.5426
0.09 0.5382
0.095 0.5084
0.1 0.5194
0.105 0.496
0.11 0.474
0.115 0.4622
0.12 0.4544
0.125 0.4284
0.13 0.4378
0.135 0.4218
0.14 0.4178
0.145 0.4074
0.15 0.3872
0.155 0.395
0.16 0.3752
0.165 0.3612
0.17 0.3624
0.175 0.36
0.18 0.3526
0.185 0.3334
0.19 0.3294
0.195 0.3284
0.2 0.3204 
0.205 0.3044
0.21 0.2878
0.215 0.311
0.22 0.2944
0.225 0.287
0.23 0.2818
0.235 0.27
0.24 0.2654
0.245 0.2692
0.25 0.2548
0.255 0.2482
0.275 0.2202
0.295 0.2196
0.315 0.1986
0.335 0.1898
0.355 0.1774
0.375 0.1606
0.395 0.1638
0.415 0.1424
0.435 0.1308
0.455 0.1292
0.475 0.1204
0.495 0.1184
0.515 0.1114
0.535 0.107
0.555 0.101
0.575 0.1044
0.595 0.087
0.615 0.0834
0.635 0.0852
0.655 0.076
0.675 0.08
0.695000000000001 0.0748
0.715000000000001 0.0736
0.735000000000001 0.0674
0.755000000000001 0.0614
0.775000000000001 0.062
0.795000000000001 0.0588
0.815000000000001 0.058
0.835000000000001 0.053
0.855000000000001 0.052
0.875000000000001 0.048
0.895000000000001 0.0492
0.915000000000001 0.0468
0.935000000000001 0.0438
0.955000000000001 0.0392
0.975000000000001 0.0426
0.995000000000001 0.032
};

\addplot [semithick, color1]
table {%
0 1
0.005 0.7962
0.01 0.6748
0.015 0.5798
0.02 0.5212
0.025 0.4762
0.03 0.4376
0.035 0.4182
0.04 0.3994
0.045 0.3744
0.05 0.3582
0.055 0.352
0.06 0.345
0.065 0.328
0.07 0.3278
0.075 0.3174
0.08 0.3148
0.085 0.2988
0.09 0.2926
0.095 0.2858
0.1 0.2676
0.105 0.2702
0.11 0.2702
0.115 0.269
0.12 0.2444
0.125 0.2476
0.13 0.2324
0.135 0.2322
0.14 0.242
0.145 0.2342
0.15 0.2178
0.155 0.206
0.16 0.2106
0.165 0.2096
0.17 0.2006
0.175 0.1982
0.18 0.1836
0.185 0.2046
0.19 0.1836
0.195 0.1734
0.2 0.174 
0.205 0.1578
0.21 0.1606
0.215 0.1588
0.22 0.1564
0.225 0.151
0.23 0.1476
0.235 0.1398
0.24 0.139
0.245 0.1366
0.25 0.134
0.255 0.1382
0.275 0.1228
0.295 0.11
0.315 0.1018
0.335 0.106
0.355 0.0916
0.375 0.084
0.395 0.0884
0.415 0.0786
0.435 0.0752
0.455 0.0828
0.475 0.076
0.495 0.069
0.515 0.075
0.535 0.0704
0.555 0.0694
0.575 0.0688
0.595 0.0614
0.615 0.0706
0.635 0.0712
0.655 0.072
0.675 0.0628
0.695000000000001 0.0662
0.715000000000001 0.0622
0.735000000000001 0.0546
0.755000000000001 0.0484
0.775000000000001 0.0602
0.795000000000001 0.0518
0.815000000000001 0.0504
0.835000000000001 0.0562
0.855000000000001 0.0478
0.875000000000001 0.0468
0.895000000000001 0.0492
0.915000000000001 0.05
0.935000000000001 0.0412
0.955000000000001 0.0422
0.975000000000001 0.038
0.995000000000001 0.0334
};
\addplot [semithick, color2]
table {%
0 1
0.005 0.6674
0.01 0.4882
0.015 0.361
0.02 0.2762
0.025 0.2258
0.03 0.205
0.035 0.1824
0.04 0.1622
0.045 0.1296
0.05 0.1336
0.055 0.127
0.06 0.114
0.065 0.1086
0.07 0.0994
0.075 0.0894
0.08 0.0886
0.085 0.0854
0.09 0.082
0.095 0.08
0.1 0.0686
0.105 0.0628
0.11 0.0662
0.115 0.0658
0.12 0.064
0.125 0.0598
0.13 0.0548
0.135 0.06
0.14 0.0504
0.145 0.0574
0.15 0.0462
0.155 0.0538
0.16 0.0494
0.165 0.0418
0.17 0.0498
0.175 0.0428
0.18 0.041
0.185 0.0428
0.19 0.0406
0.195 0.0498
0.2 0.0446
0.205 0.0362
0.21 0.0362
0.215 0.0424
0.22 0.0328
0.225 0.0372
0.23 0.0392
0.235 0.0354
0.24 0.0366
0.245 0.0336
0.25 0.0362
0.255 0.031
0.275 0.035
0.295 0.0356
0.315 0.0322
0.335 0.0354
0.355 0.0366
0.375 0.0342
0.395 0.0386
0.415 0.0356
0.435 0.0364
0.455 0.0384
0.475 0.0416
0.495 0.0414
0.515 0.0378
0.535 0.04
0.555 0.0362
0.575 0.0484
0.595 0.046
0.615 0.0478
0.635 0.0508
0.655 0.0488
0.675 0.0516
0.695000000000001 0.042
0.715000000000001 0.0484
0.735000000000001 0.0546
0.755000000000001 0.0484
0.775000000000001 0.0516
0.795000000000001 0.0478
0.815000000000001 0.0558
0.835000000000001 0.0486
0.855000000000001 0.0482
0.875000000000001 0.0516
0.895000000000001 0.0458
0.915000000000001 0.047
0.935000000000001 0.0454
0.955000000000001 0.0382
0.975000000000001 0.0392
0.995000000000001 0.0386
};
\addlegendentry{A (735)}
\addlegendentry{B (1186)}
\addlegendentry{C (1426)}
\end{axis}

\end{tikzpicture}}
		\label{fig:dev13}
	\end{minipage}
	\hfill
	\begin{minipage}[t]{0.3\linewidth}
		\centering
		\subfigure[Rev-Euc 0.05 b]{
\begin{tikzpicture}[scale=0.6]

\definecolor{color0}{rgb}{0.12156862745098,0.466666666666667,0.705882352941177}
\definecolor{color1}{rgb}{1,0.498039215686275,0.0549019607843137}
\definecolor{color2}{rgb}{0.172549019607843,0.627450980392157,0.172549019607843}

\begin{axis}[
tick align=outside,
tick pos=left,
x grid style={white!69.0196078431373!black},
xlabel={norm-$\phi$},
xmin=-0.0495, xmax=1.0395,
xtick style={color=black},
y grid style={white!69.0196078431373!black},
ylabel={stable pair probability},
ymin=-0.029, ymax=1.049,
ytick style={color=black}
]

\addplot [semithick, color0]
table {%
0 1
0.005 1
0.01 0.9998
0.015 0.9988
0.02 0.9968
0.025 0.9908
0.03 0.984
0.035 0.974
0.04 0.9672
0.045 0.96
0.05 0.952
0.055 0.947
0.06 0.9392
0.065 0.937
0.07 0.9216
0.075 0.923
0.08 0.9144
0.085 0.8996
0.09 0.891
0.095 0.884
0.1 0.885
0.105 0.8784
0.11 0.8652
0.115 0.866
0.12 0.8592
0.125 0.861
0.13 0.837
0.135 0.8432
0.14 0.847
0.145 0.8374
0.15 0.8176
0.155 0.8224
0.16 0.816
0.165 0.8124
0.17 0.7962
0.175 0.7892
0.18 0.7864
0.185 0.7674
0.19 0.7844
0.195 0.7602
0.2 0.7522 
0.205 0.741
0.21 0.7372
0.215 0.7262
0.22 0.714
0.225 0.7202
0.23 0.7008
0.235 0.68
0.24 0.6892
0.245 0.678
0.25 0.6774
0.255 0.6626
0.26 0.6614
0.265 0.6524
0.27 0.6292
0.275 0.6288
0.28 0.6244
0.285 0.6106
0.29 0.6108
0.295 0.5954
0.3 0.6008
0.305 0.5794
0.31 0.5642
0.315 0.5622
0.32 0.5704
0.325 0.562
0.33 0.542
0.335 0.5382
0.34 0.5392
0.345 0.5206
0.35 0.5224
0.355 0.5144
0.36 0.509
0.365 0.5086
0.37 0.491
0.375 0.4848
0.38 0.478
0.385 0.4648
0.39 0.4678
0.395 0.4686
0.4 0.4574
0.405 0.4372
0.41 0.4422
0.415 0.4356
0.42 0.4296
0.425 0.4322
0.43 0.4232
0.435 0.4334
0.44 0.415
0.445 0.4014
0.45 0.401
0.455 0.3904
0.46 0.38
0.465 0.368
0.47 0.3612
0.475 0.369
0.48 0.365
0.485 0.3688
0.49 0.356
0.495 0.3476
0.5 0.3512
0.505 0.3462
0.51 0.3232
0.515 0.3384
0.52 0.3244
0.525 0.3174
0.53 0.3064
0.535 0.3048
0.54 0.3166
0.545 0.3062
0.55 0.2884
0.555 0.2898
0.56 0.2832
0.565 0.283
0.57 0.2862
0.575 0.2762
0.58 0.275
0.585 0.262
0.59 0.2496
0.595 0.2514
0.6 0.2538
0.605 0.2518
0.61 0.2528
0.615 0.234
0.62 0.2098
0.625 0.2296
0.63 0.2246
0.65 0.2122
0.67 0.199
0.690000000000001 0.1788
0.710000000000001 0.1534
0.730000000000001 0.1458
0.750000000000001 0.1414
0.770000000000001 0.13
0.790000000000001 0.1066
0.810000000000001 0.0998
0.830000000000001 0.0918
0.850000000000001 0.0826
0.870000000000001 0.0762
0.890000000000001 0.0684
0.910000000000001 0.058
0.930000000000001 0.0514
0.950000000000001 0.0498
0.970000000000001 0.0454
0.990000000000001 0.04
};
\addplot [semithick, color1]
table {%
0 1
0.005 0.8746
0.01 0.7554
0.015 0.6654
0.02 0.5964
0.025 0.5286
0.03 0.4864
0.035 0.446
0.04 0.4322
0.045 0.4016
0.05 0.3966
0.055 0.3862
0.06 0.3802
0.065 0.3774
0.07 0.3758
0.075 0.3678
0.08 0.3528
0.085 0.3544
0.09 0.3452
0.095 0.3476
0.1 0.3484
0.105 0.349
0.11 0.3328
0.115 0.3156
0.12 0.3192
0.125 0.3142
0.13 0.3094
0.135 0.3108
0.14 0.3084
0.145 0.2894
0.15 0.292
0.155 0.2802
0.16 0.2914
0.165 0.2836
0.17 0.2738
0.175 0.255
0.18 0.2688
0.185 0.259
0.19 0.2538
0.195 0.245
0.2 0.2368
0.205 0.245
0.21 0.229
0.215 0.223
0.22 0.2264
0.225 0.2238
0.23 0.2166
0.235 0.2064
0.24 0.1938
0.245 0.1944
0.25 0.2
0.255 0.2006
0.26 0.1832
0.265 0.1782
0.27 0.1672
0.275 0.1702
0.28 0.1656
0.285 0.171
0.29 0.1714
0.295 0.1598
0.3 0.159
0.305 0.145
0.31 0.1484
0.315 0.1494
0.32 0.1342
0.325 0.1442
0.33 0.1356
0.335 0.1406
0.34 0.135
0.345 0.1394
0.35 0.1296
0.355 0.1216
0.36 0.1186
0.365 0.115
0.37 0.108
0.375 0.114
0.38 0.1132
0.385 0.1028
0.39 0.096
0.395 0.097
0.4 0.1114
0.405 0.099
0.41 0.1004
0.415 0.0982
0.42 0.0972
0.425 0.0948
0.43 0.0932
0.435 0.0922
0.44 0.0904
0.445 0.0956
0.45 0.0918
0.455 0.08
0.46 0.0924
0.465 0.086
0.47 0.0894
0.475 0.0824
0.48 0.0836
0.485 0.0836
0.49 0.0842
0.495 0.0756
0.5 0.0802
0.505 0.0806
0.51 0.0764
0.515 0.0746
0.52 0.0752
0.525 0.0746
0.53 0.0696
0.535 0.0774
0.54 0.075
0.545 0.0702
0.55 0.065
0.555 0.0698
0.56 0.0686
0.565 0.0666
0.57 0.0676
0.575 0.0684
0.58 0.0654
0.585 0.06
0.59 0.068
0.595 0.0592
0.6 0.0606
0.605 0.0742
0.61 0.0612
0.615 0.0606
0.62 0.072
0.625 0.063
0.63 0.0628
0.65 0.0574
0.67 0.057
0.690000000000001 0.0572
0.710000000000001 0.0548
0.730000000000001 0.0586
0.750000000000001 0.0564
0.770000000000001 0.05
0.790000000000001 0.053
0.810000000000001 0.0482
0.830000000000001 0.0556
0.850000000000001 0.0464
0.870000000000001 0.0454
0.890000000000001 0.047
0.910000000000001 0.0452
0.930000000000001 0.0442
0.950000000000001 0.0414
0.970000000000001 0.0362
0.990000000000001 0.0402
};
\addplot [semithick, color2]
table {%
0 1
0.005 0.778
0.01 0.585
0.015 0.46
0.02 0.3678
0.025 0.2932
0.03 0.242
0.035 0.191
0.04 0.1612
0.045 0.14
0.05 0.1316
0.055 0.1098
0.06 0.0992
0.065 0.0972
0.07 0.0828
0.075 0.0798
0.08 0.0702
0.085 0.0716
0.09 0.0622
0.095 0.061
0.1 0.0626
0.105 0.0584
0.11 0.0548
0.115 0.0572
0.12 0.052
0.125 0.0526
0.13 0.0524
0.135 0.0494
0.14 0.0474
0.145 0.055
0.15 0.048
0.155 0.047
0.16 0.0516
0.165 0.0444
0.17 0.046
0.175 0.0462
0.18 0.0472
0.185 0.0478
0.19 0.0462
0.195 0.047
0.2 0.0454
0.205 0.0426
0.21 0.0466
0.215 0.0474
0.22 0.0418
0.225 0.049
0.23 0.0458
0.235 0.042
0.24 0.0424
0.245 0.0414
0.25 0.0468
0.255 0.0438
0.26 0.0404
0.265 0.0448
0.27 0.049
0.275 0.0418
0.28 0.0452
0.285 0.046
0.29 0.0434
0.295 0.0434
0.3 0.0458
0.305 0.0452
0.31 0.0426
0.315 0.0426
0.32 0.0464
0.325 0.0416
0.33 0.0416
0.335 0.0508
0.34 0.0464
0.345 0.0464
0.35 0.045
0.355 0.0478
0.36 0.0424
0.365 0.0508
0.37 0.049
0.375 0.0474
0.38 0.0506
0.385 0.05
0.39 0.0486
0.395 0.0548
0.4 0.0526
0.405 0.0482
0.41 0.0454
0.415 0.0498
0.42 0.0472
0.425 0.048
0.43 0.0524
0.435 0.0506
0.44 0.0514
0.445 0.0506
0.45 0.0494
0.455 0.0494
0.46 0.0512
0.465 0.0508
0.47 0.0564
0.475 0.053
0.48 0.0514
0.485 0.0512
0.49 0.0482
0.495 0.055
0.5 0.0596
0.505 0.0572
0.51 0.0556
0.515 0.0496
0.52 0.0494
0.525 0.0576
0.53 0.0528
0.535 0.0574
0.54 0.0608
0.545 0.0516
0.55 0.0536
0.555 0.055
0.56 0.0574
0.565 0.0592
0.57 0.0538
0.575 0.0548
0.58 0.0574
0.585 0.0592
0.59 0.059
0.595 0.0524
0.6 0.0578
0.605 0.0586
0.61 0.0562
0.615 0.058
0.62 0.0512
0.625 0.0584
0.63 0.053
0.65 0.0588
0.67 0.0556
0.690000000000001 0.0628
0.710000000000001 0.0526
0.730000000000001 0.0552
0.750000000000001 0.056
0.770000000000001 0.0546
0.790000000000001 0.0536
0.810000000000001 0.0512
0.830000000000001 0.05
0.850000000000001 0.0536
0.870000000000001 0.0486
0.890000000000001 0.049
0.910000000000001 0.048
0.930000000000001 0.0428
0.950000000000001 0.0382
0.970000000000001 0.0428
0.990000000000001 0.0374
};
\addlegendentry{A (620)}
\addlegendentry{B (705)}
\addlegendentry{C (1630)}
\end{axis}

\end{tikzpicture}}
		\label{fig:dev10}
	\end{minipage}
	\hfill
	\begin{minipage}[t]{0.3\linewidth}
		\centering
		\subfigure[ROB]{
\begin{tikzpicture}[scale=0.6]

\definecolor{color0}{rgb}{0.12156862745098,0.466666666666667,0.705882352941177}
\definecolor{color1}{rgb}{1,0.498039215686275,0.0549019607843137}
\definecolor{color2}{rgb}{0.172549019607843,0.627450980392157,0.172549019607843}

\begin{axis}[
tick align=outside,
tick pos=left,
x grid style={white!69.0196078431373!black},
xlabel={norm-$\phi$},
xmin=-0.04925, xmax=1.03425,
xtick style={color=black},
y grid style={white!69.0196078431373!black},
ylabel={stable pair probability},
ymin=-0.01241, ymax=1.04821,
ytick style={color=black}
]

\addplot [semithick, color2]
table {%
0 1
0.005 1
0.01 1
0.015 1
0.02 1
0.025 1
0.03 1
0.035 1
0.04 1
0.045 1
0.05 1
0.055 1
0.06 1
0.065 1
0.07 1
0.075 1
0.08 1
0.085 1
0.09 1
0.095 1
0.1 1
0.105 1
0.11 1
0.115 1
0.12 1
0.125 0.9998
0.13 1
0.135 0.9996
0.14 0.9998
0.145 0.999
0.15 0.9984
0.155 0.998
0.16 0.9992
0.165 0.998
0.17 0.9976
0.175 0.9986
0.18 0.9974
0.185 0.9978
0.19 0.9982
0.195 0.995
0.2 0.993
0.205 0.9922
0.21 0.994
0.215 0.9914
0.22 0.992
0.225 0.991
0.23 0.9876
0.235 0.9864
0.24 0.9878
0.245 0.9878
0.25 0.9812
0.255 0.9816
0.26 0.9822
0.265 0.9806
0.27 0.9742
0.275 0.9738
0.28 0.967
0.285 0.962
0.29 0.9542
0.295 0.9584
0.3 0.9572
0.305 0.9444
0.31 0.94
0.315 0.9424
0.32 0.9368
0.325 0.9382
0.33 0.9372
0.335 0.9172
0.34 0.9078
0.345 0.9002
0.35 0.9006
0.355 0.8882
0.36 0.8892
0.365 0.8764
0.37 0.8766
0.375 0.8602
0.38 0.8396
0.385 0.8368
0.39 0.8292
0.395 0.816
0.4 0.806
0.405 0.7974
0.41 0.7908
0.415 0.7792
0.42 0.7744
0.425 0.7486
0.43 0.7464
0.435 0.7318
0.44 0.7164
0.445 0.708
0.45 0.6934
0.455 0.6822
0.46 0.673
0.465 0.6576
0.47 0.6546
0.475 0.6394
0.48 0.6188
0.485 0.6062
0.49 0.5972
0.495 0.5996
0.5 0.5774
0.505 0.579
0.51 0.5556
0.515 0.5404
0.52 0.5364
0.525 0.523
0.53 0.5144
0.535 0.4806
0.54 0.4756
0.545 0.475
0.55 0.461
0.555 0.4518
0.56 0.4366
0.565 0.4266
0.57 0.4198
0.575 0.418
0.58 0.402
0.585 0.4084
0.59 0.382
0.595 0.3688
0.6 0.3676
0.605 0.3596
0.61 0.3406
0.615 0.3282
0.62 0.315
0.625 0.3146
0.63 0.3068
0.635 0.3094
0.64 0.2942
0.645 0.2926
0.65 0.2762
0.655 0.2712
0.66 0.2652
0.665 0.2568
0.67 0.2576
0.675 0.248
0.68 0.2324
0.685 0.2232
0.705000000000001 0.2096
0.725000000000001 0.187
0.745000000000001 0.1716
0.765000000000001 0.1414
0.785000000000001 0.1274
0.805000000000001 0.1214
0.825000000000001 0.1022
0.845000000000001 0.0944
0.865000000000001 0.0878
0.885000000000001 0.0724
0.905000000000001 0.067
0.925000000000001 0.065
0.945000000000001 0.0538
0.965000000000001 0.0462
0.985000000000001 0.0416
};
\addplot [semithick, color1]
table {%
0 1
0.005 1
0.01 1
0.015 1
0.02 1
0.025 1
0.03 1
0.035 1
0.04 1
0.045 1
0.05 1
0.055 1
0.06 1
0.065 1
0.07 1
0.075 1
0.08 1
0.085 1
0.09 1
0.095 1
0.1 1
0.105 1
0.11 1
0.115 0.9998
0.12 1
0.125 0.9996
0.13 0.9998
0.135 0.9996
0.14 0.999
0.145 0.999
0.15 0.999
0.155 0.9988
0.16 0.9976
0.165 0.9988
0.17 0.9974
0.175 0.9972
0.18 0.9966
0.185 0.9952
0.19 0.9946
0.195 0.9938
0.2 0.9944
0.205 0.9902
0.21 0.9934
0.215 0.9888
0.22 0.9896
0.225 0.9894
0.23 0.9894
0.235 0.984
0.24 0.983
0.245 0.9844
0.25 0.9802
0.255 0.979
0.26 0.9742
0.265 0.9746
0.27 0.9686
0.275 0.969
0.28 0.9666
0.285 0.9626
0.29 0.9572
0.295 0.9514
0.3 0.9552
0.305 0.9504
0.31 0.9448
0.315 0.939
0.32 0.9276
0.325 0.9288
0.33 0.9218
0.335 0.9136
0.34 0.9022
0.345 0.8966
0.35 0.899
0.355 0.8842
0.36 0.8852
0.365 0.872
0.37 0.865
0.375 0.8454
0.38 0.8456
0.385 0.8356
0.39 0.8216
0.395 0.8082
0.4 0.8048
0.405 0.7862
0.41 0.7832
0.415 0.7798
0.42 0.7582
0.425 0.7456
0.43 0.7366
0.435 0.7292
0.44 0.712
0.445 0.7022
0.45 0.686
0.455 0.6736
0.46 0.6672
0.465 0.647
0.47 0.649
0.475 0.6316
0.48 0.6158
0.485 0.6086
0.49 0.588
0.495 0.5852
0.5 0.5738
0.505 0.5548
0.51 0.5556
0.515 0.5378
0.52 0.5202
0.525 0.5096
0.53 0.4964
0.535 0.4996
0.54 0.4694
0.545 0.469
0.55 0.4542
0.555 0.448
0.56 0.4232
0.565 0.4332
0.57 0.4214
0.575 0.4092
0.58 0.3974
0.585 0.391
0.59 0.3718
0.595 0.3612
0.6 0.3464
0.605 0.3548
0.61 0.3556
0.615 0.3244
0.62 0.3234
0.625 0.3252
0.63 0.3134
0.635 0.292
0.64 0.2936
0.645 0.277
0.65 0.277
0.655 0.274
0.66 0.2672
0.665 0.258
0.67 0.2396
0.675 0.2386
0.68 0.2268
0.685 0.2244
0.705000000000001 0.2
0.725000000000001 0.1848
0.745000000000001 0.165
0.765000000000001 0.1372
0.785000000000001 0.1306
0.805000000000001 0.118
0.825000000000001 0.1016
0.845000000000001 0.0902
0.865000000000001 0.0852
0.885000000000001 0.075
0.905000000000001 0.0652
0.925000000000001 0.0552
0.945000000000001 0.0496
0.965000000000001 0.046
0.985000000000001 0.04
};

\addplot [semithick, color0]
table {%
0 1
0.005 1
0.01 1
0.015 1
0.02 1
0.025 1
0.03 1
0.035 1
0.04 1
0.045 1
0.05 1
0.055 1
0.06 1
0.065 1
0.07 1
0.075 1
0.08 1
0.085 1
0.09 1
0.095 1
0.1 1
0.105 1
0.11 1
0.115 1
0.12 0.9998
0.125 0.9998
0.13 0.9994
0.135 0.9998
0.14 0.9986
0.145 0.9994
0.15 0.998
0.155 0.9976
0.16 0.9976
0.165 0.997
0.17 0.9968
0.175 0.9954
0.18 0.996
0.185 0.992
0.19 0.994
0.195 0.99
0.2 0.988
0.205 0.9894
0.21 0.992
0.215 0.9876
0.22 0.987
0.225 0.9862
0.23 0.9838
0.235 0.9824
0.24 0.9814
0.245 0.9772
0.25 0.9768
0.255 0.9734
0.26 0.9726
0.265 0.9678
0.27 0.9624
0.275 0.9614
0.28 0.957
0.285 0.9592
0.29 0.949
0.295 0.9434
0.3 0.9454
0.305 0.9384
0.31 0.9386
0.315 0.93
0.32 0.9252
0.325 0.921
0.33 0.921
0.335 0.916
0.34 0.9002
0.345 0.9022
0.35 0.8968
0.355 0.8784
0.36 0.8702
0.365 0.8672
0.37 0.8668
0.375 0.8512
0.38 0.8532
0.385 0.8264
0.39 0.8128
0.395 0.7978
0.4 0.7916
0.405 0.7856
0.41 0.7862
0.415 0.7534
0.42 0.753
0.425 0.7448
0.43 0.7498
0.435 0.7342
0.44 0.7108
0.445 0.692
0.45 0.6936
0.455 0.6668
0.46 0.6592
0.465 0.6498
0.47 0.639
0.475 0.6276
0.48 0.6152
0.485 0.6032
0.49 0.588
0.495 0.587
0.5 0.5676
0.505 0.5664
0.51 0.5448
0.515 0.5346
0.52 0.5042
0.525 0.5046
0.53 0.5
0.535 0.4904
0.54 0.4726
0.545 0.4682
0.55 0.462
0.555 0.4502
0.56 0.425
0.565 0.426
0.57 0.4136
0.575 0.3968
0.58 0.4042
0.585 0.3922
0.59 0.3652
0.595 0.3622
0.6 0.351
0.605 0.3582
0.61 0.336
0.615 0.3276
0.62 0.3316
0.625 0.3162
0.63 0.2962
0.635 0.2896
0.64 0.2832
0.645 0.2732
0.65 0.268
0.655 0.2646
0.66 0.257
0.665 0.2516
0.67 0.2426
0.675 0.233
0.68 0.2246
0.685 0.2274
0.705000000000001 0.2054
0.725000000000001 0.1968
0.745000000000001 0.1614
0.765000000000001 0.1436
0.785000000000001 0.1174
0.805000000000001 0.1108
0.825000000000001 0.109
0.845000000000001 0.094
0.865000000000001 0.082
0.885000000000001 0.077
0.905000000000001 0.0614
0.925000000000001 0.0618
0.945000000000001 0.0532
0.965000000000001 0.041
0.985000000000001 0.0386
};
\addlegendentry{A (0)}
\addlegendentry{B (0)}
\addlegendentry{C (0)}
\end{axis}

\end{tikzpicture}}
		\label{fig:dev9}
	\end{minipage}
	\hfill
	\begin{minipage}[t]{0.3\linewidth}
		\centering
		\subfigure[Mal-ROB 0.2]{\input{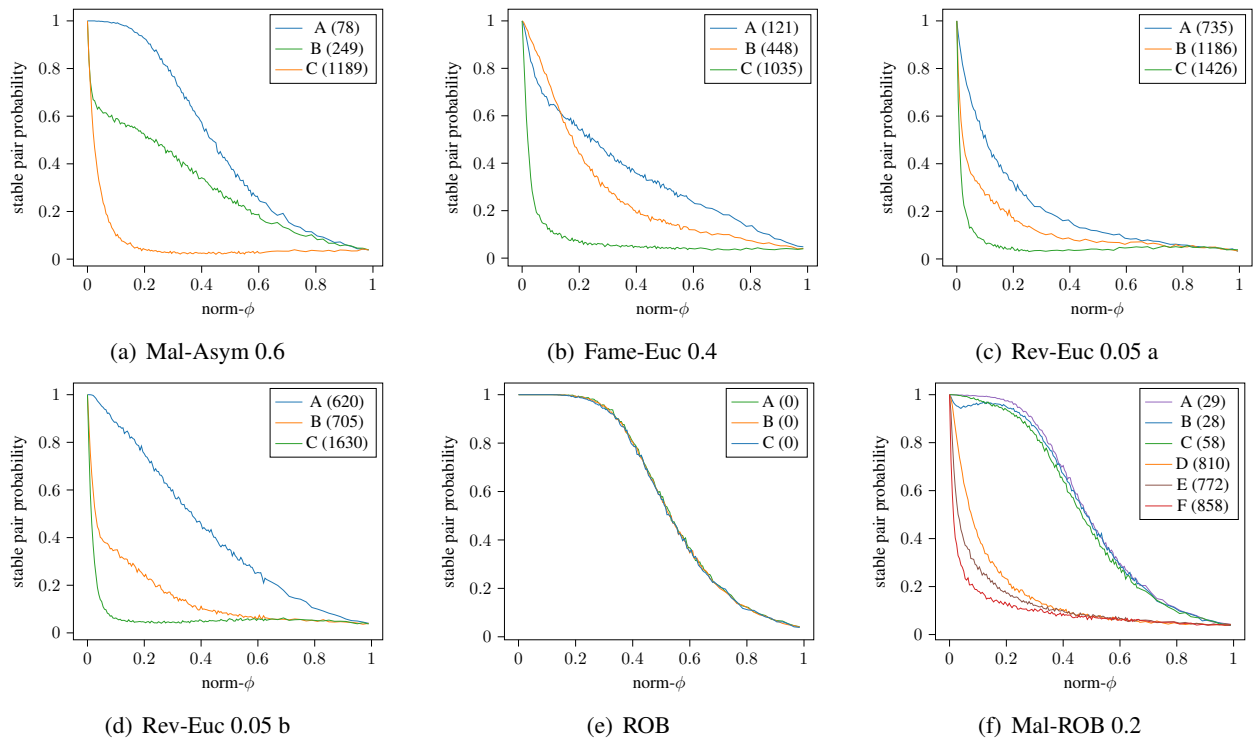}}
		\label{fig:dev12}
	\end{minipage}
	\hfill
	\caption{Average-case robustness of some exemplary stable pairs for six instances. The culture from which each instance has been sampled can be found in the caption. For each pair, its blocking score is written in brackets.}
	\label{fig:devpairs}
\end{figure*}

\subsubsection{Average stable pair robustness}
We plotted the distribution of the average stable pair robustness for all instances of the data set in \Cref{fig:avg-pairs-dist}. We can directly observe the following
The average stable pair robustness is much higher than the matching robustness of an instance. On average, around 50 random swaps per preference list are needed to make a pair unstable. 

This finding matches our intuition, since, while it is enough to create one blocking pair to make a matching unstable, a certain pair in this matching can still be stable in some other stable matching. Thus, it is much harder for a matching to be stable than for a single pair.

As in the matching setting, there are some instances from the Mal-Rob culture that are particularly robust. However, the difference to the ``normal'' instances is smaller in the pairs setting. For example, consider the stable pair robustness of the Mal-Asym instance depicted in \Cref{fig:devpairs}(a). The figure shows the most robust and least robust stable pair as well as a pair with an average robustness. The average stable pair robustness for this instance is 0.15. However, the stable pair probability decreases slowly for most stable pairs: For a norm-$\phi$ value of 0.4, the blue and green pairs have still a quite high stable pair probability. 

Still, the 50\%-stability-threshold is not a bad indicator of the average stable pair robustness, as can be seen in \Cref{fig:corr-avgpair-fif}. The Pearson Correlation Coefficient is 0.65.
\begin{figure}[t]
	\centering
\begin{tikzpicture}[scale=0.8]

\definecolor{color0}{rgb}{0.12156862745098,0.466666666666667,0.705882352941177}

\begin{axis}[
tick align=outside,
tick pos=left,
x grid style={white!69.0196078431373!black},
xlabel={average stable pair robustness},
xmin=-0.017118, xmax=0.571558,
xtick style={color=black},
y grid style={white!69.0196078431373!black},
ylabel={average man-opt matching robustness},
ymin=0, ymax=0.02,
ytick style={color=black}
]
\addplot [semithick, color0, mark=*, mark size=1, mark options={solid}, only marks]
table {%
0.00964 0.001
0.1308 0.0115
0.0138 0.0005
0.5448 0.2926
0.367910448 0.0162
0.362753623 0.2036
0.363970588 0.0062
0.493 0.2257
0.382121212 0.0453
0.356376812 0.0088
0.4944 0.2351
0.29 0.0035
0.320694444 0.0041
0.395079365 0.0135
0.321688312 0.0048
0.491 0.2265
0.336164384 0.0028
0.421724138 0.0171
0.352173913 0.2133
0.343287671 0.0049
0.362352941 0.0087
0.4992 0.2341
0.392380952 0.2185
0.365 0.0083
0.226373626 0.0052
0.214536082 0.0082
0.195728155 0.004
0.199357798 0.0033
0.181709402 0.0063
0.201568627 0.0071
0.181559633 0.0031
0.198113208 0.0075
0.238510638 0.0045
0.158778626 0.0038
0.204272727 0.0041
0.172936508 0.0029
0.158372093 0.0026
0.194554455 0.0058
0.215108696 0.0065
0.173716814 0.0053
0.220326087 0.0051
0.21 0.0054
0.21287037 0.0042
0.162255639 0.0025
0.157659574 0.0028
0.132903226 0.0039
0.14254717 0.0028
0.137118644 0.005
0.164736842 0.0034
0.144628099 0.0035
0.143538462 0.0024
0.14912 0.0028
0.143738318 0.0033
0.141982759 0.0047
0.127913043 0.0026
0.142936508 0.0033
0.142704918 0.0033
0.152641509 0.0032
0.144545455 0.0027
0.147837838 0.0046
0.160380952 0.0034
0.163823529 0.0032
0.13944 0.0041
0.136785714 0.0051
0.16030303 0.0035
0.117314815 0.0037
0.119459459 0.0051
0.151386139 0.0039
0.131810345 0.0026
0.124271845 0.0039
0.153780488 0.0024
0.162328767 0.0051
0.140722892 0.0024
0.135057471 0.0047
0.136936937 0.0043
0.134387755 0.003
0.115081967 0.0035
0.116826923 0.0022
0.136074766 0.0043
0.138461538 0.004
0.123474576 0.0034
0.144615385 0.0043
0.117596154 0.0023
0.124732143 0.0024
0.128247423 0.0023
0.138636364 0.0037
0.149506173 0.0038
0.125263158 0.0034
0.134565217 0.0052
0.113846154 0.0048
0.112685185 0.0027
0.142083333 0.0022
0.132268041 0.0025
0.163194444 0.0034
0.135684211 0.0027
0.127575758 0.0037
0.138953488 0.0028
0.111734694 0.0031
0.118085106 0.0038
0.129120879 0.0062
0.131555556 0.0037
0.130824742 0.0033
0.133255814 0.0038
0.133977273 0.0034
0.09606383 0.0028
0.099494949 0.0023
0.104047619 0.0022
0.100352941 0.0029
0.079836066 0.0027
0.111648352 0.0022
0.117361111 0.0037
0.095222222 0.0019
0.101558442 0.0022
0.107956989 0.0028
0.104222222 0.0023
0.099047619 0.0022
0.109772727 0.0033
0.102021277 0.0031
0.116666667 0.0028
0.095050505 0.0037
0.109873418 0.0024
0.114303797 0.0052
0.086699029 0.0054
0.094216867 0.0041
0.096527778 0.002
0.120923077 0.0026
0.104594595 0.0026
0.097375 0.0023
0.107058824 0.0025
0.098795181 0.003
0.108539326 0.0035
0.094246575 0.0025
0.084725275 0.0018
0.106621622 0.0027
0.094444444 0.0022
0.114626866 0.0026
0.089625 0.0022
0.097 0.0035
0.098315789 0.0032
0.101573034 0.0018
0.100108696 0.0038
0.112112676 0.0025
0.104606742 0.0032
0.101904762 0.0031
0.054333333 0.0029
0.050892857 0.0022
0.072 0.0027
0.061730769 0.0026
0.049649123 0.0018
0.055964912 0.0022
0.072962963 0.0038
0.0716 0.0026
0.064038462 0.0024
0.058461538 0.0023
0.0698 0.0029
0.064230769 0.0023
0.054385965 0.0019
0.0618 0.0025
0.056923077 0.0023
0.062407407 0.0037
0.046226415 0.0019
0.057894737 0.0023
0.0646 0.0029
0.060377358 0.0023
0.09245283 0.003
0.095614035 0.0031
0.098703704 0.0034
0.120555556 0.005
0.096 0.0025
0.114642857 0.0037
0.088103448 0.0056
0.120925926 0.0044
0.07828125 0.003
0.105535714 0.0036
0.1172 0.0047
0.095471698 0.0039
0.133 0.0064
0.0994 0.0028
0.083770492 0.0028
0.084642857 0.0037
0.096491228 0.0039
0.094920635 0.0024
0.1175 0.0025
0.101333333 0.0057
0.148333333 0.0045
0.111428571 0.004
0.1105 0.0028
0.126964286 0.0043
0.1564 0.0035
0.110322581 0.0044
0.120882353 0.0051
0.110597015 0.0047
0.1516 0.0052
0.159056604 0.006
0.102337662 0.0032
0.120175439 0.004
0.078196721 0.0031
0.11 0.0034
0.105846154 0.0028
0.129180328 0.0035
0.10171875 0.0049
0.121886792 0.0038
0.13557377 0.0035
0.113898305 0.0038
0.133611111 0.0033
0.143043478 0.0053
0.11 0.0034
0.139342105 0.0043
0.112921348 0.0031
0.12137931 0.0045
0.120210526 0.0043
0.14442623 0.0024
0.10266055 0.0028
0.122142857 0.0048
0.118571429 0.0035
0.13752809 0.0031
0.132962963 0.0035
0.108080808 0.0035
0.131641791 0.0034
0.119746835 0.0061
0.115555556 0.0047
0.154 0.0028
0.125316456 0.005
0.129538462 0.0028
0.065 0.0013
0.0676 0.0015
0.055 0.0011
0.0628 0.0014
0.0572 0.0014
0.0752 0.0012
0.0562 0.0011
0.0792 0.0017
0.0592 0.0012
0.0598 0.0013
0.0668 0.0015
0.052 0.0011
0.063 0.0012
0.0602 0.0012
0.069 0.0013
0.053 0.0012
0.0396 0.001
0.0696 0.0012
0.0564 0.0011
0.0708 0.0015
0.0598 0.0015
0.0682 0.0015
0.0714 0.002
0.0604 0.0014
0.0754 0.0014
0.0628 0.0013
0.0664 0.0015
0.072 0.0019
0.065 0.0014
0.0602 0.0016
0.0618 0.0012
0.0642 0.0018
0.061 0.0014
0.0804 0.0022
0.077 0.0016
0.0586 0.0015
0.0572 0.0013
0.067 0.0018
0.0668 0.0012
0.0626 0.0015
0.046111111 0.0013
0.0476 0.0021
0.04 0.0014
0.0418 0.0014
0.045 0.0016
0.030961538 0.0013
0.0482 0.0016
0.043 0.0017
0.0298 0.0013
0.0428 0.0014
0.0458 0.0016
0.0384 0.0013
0.0312 0.0014
0.0412 0.0015
0.0516 0.0019
0.050555556 0.0019
0.046 0.0015
0.021730769 0.0013
0.053846154 0.0016
0.0416 0.0021
0.057090909 0.0018
0.0586 0.0017
0.061034483 0.0015
0.059473684 0.0016
0.059454545 0.0019
0.048548387 0.0016
0.051698113 0.0016
0.059122807 0.0018
0.0628 0.0016
0.046818182 0.0016
0.073962264 0.0018
0.047307692 0.0017
0.083333333 0.002
0.063962264 0.0017
0.0836 0.0018
0.0684 0.0015
0.049807692 0.0014
0.065166667 0.0019
0.0578 0.0014
0.072 0.002
0.078545455 0.0023
0.0866 0.0023
0.058333333 0.0017
0.077758621 0.0023
0.068484848 0.002
0.081607143 0.0018
0.074423077 0.0023
0.072666667 0.0019
0.08265625 0.0026
0.0726 0.0015
0.0672 0.0015
0.072982456 0.0018
0.066296296 0.002
0.073508772 0.002
0.085614035 0.002
0.07328125 0.0016
0.063606557 0.0017
0.072307692 0.0016
0.065090909 0.0028
0.0882 0.0023
0.0512 0.0017
0.0532 0.0015
0.0534 0.0017
0.055 0.0014
0.0558 0.0015
0.0568 0.0016
0.0608 0.0016
0.0594 0.0019
0.0578 0.0024
0.0646 0.0018
0.0614 0.0016
0.0544 0.0018
0.062 0.0023
0.0594 0.0018
0.058 0.0021
0.0512 0.0019
0.0554 0.0019
0.0536 0.0017
0.06 0.0018
0.0548 0.0017
0.0552 0.002
0.054 0.002
0.0542 0.002
0.061 0.0027
0.0474 0.0014
0.0752 0.0027
0.07 0.0025
0.048 0.0018
0.0544 0.0017
0.058 0.0015
0.049 0.0016
0.0584 0.0015
0.0528 0.0019
0.0486 0.0018
0.0564 0.0022
0.054 0.0019
0.0608 0.0019
0.0766 0.0022
0.0654 0.0028
0.053 0.0018
0.084477612 0.0024
0.086111111 0.0029
0.07890411 0.0016
0.076545455 0.0018
0.092105263 0.003
0.084423077 0.0026
0.072121212 0.0025
0.06296875 0.0018
0.073793103 0.0015
0.075384615 0.0027
0.079454545 0.002
0.07 0.0021
0.072542373 0.0022
0.070178571 0.0021
0.078857143 0.0031
0.060961538 0.0019
0.058333333 0.002
0.058135593 0.0021
0.101803279 0.0041
0.077222222 0.0023
0.08825 0.0021
0.119152542 0.0028
0.088625 0.0038
0.121029412 0.0028
0.079180328 0.0028
0.084615385 0.0016
0.0921875 0.0033
0.083859649 0.0028
0.092923077 0.0024
0.094230769 0.0023
0.088947368 0.0026
0.095777778 0.0027
0.095961538 0.003
0.084146341 0.0027
0.11 0.0035
0.119310345 0.0028
0.106111111 0.0027
0.09875 0.0026
0.090588235 0.0031
0.105769231 0.0029
0.071153846 0.0019
0.0525 0.0014
0.081153846 0.0024
0.063076923 0.0017
0.053928571 0.0018
0.070769231 0.0021
0.060689655 0.0019
0.0706 0.0021
0.049230769 0.0013
0.064821429 0.0017
0.0725 0.0025
0.077 0.0022
0.05625 0.0016
0.0575 0.0021
0.059259259 0.0019
0.079423077 0.0017
0.060555556 0.0021
0.069322034 0.0024
0.068703704 0.0022
0.078888889 0.0017
0.083214286 0.0022
0.057941176 0.0019
0.054310345 0.0018
0.068035714 0.0027
0.085 0.0029
0.068666667 0.0027
0.1012 0.0023
0.088387097 0.0032
0.075972222 0.0029
0.071071429 0.0022
0.058125 0.0019
0.076190476 0.0026
0.090185185 0.0039
0.088545455 0.0019
0.077272727 0.002
0.073793103 0.0019
0.0832 0.0019
0.067719298 0.0022
0.096851852 0.0023
0.088846154 0.0019
0.075814978 0.0022
0.085035971 0.0023
0.081386861 0.0025
0.080069686 0.002
0.074473684 0.0018
0.076969697 0.0016
0.070280374 0.002
0.073614035 0.0023
0.067537879 0.002
0.080071429 0.0021
0.075480769 0.003
0.076071429 0.0024
0.072169492 0.0019
0.074595745 0.0022
0.080186567 0.0024
0.085727273 0.0029
0.076888889 0.0019
0.069655172 0.0026
0.068432056 0.0021
0.094801762 0.0021
0.126481481 0.0041
0.114553991 0.0028
0.1303125 0.0034
0.127675676 0.0032
0.11247191 0.0029
0.122060606 0.0027
0.112565445 0.0044
0.121497326 0.0029
0.11984252 0.0021
0.118216561 0.0038
0.121363636 0.0025
0.122820513 0.0028
0.121111111 0.0021
0.126666667 0.0057
0.138607595 0.0035
0.116627219 0.0022
0.117731959 0.0029
0.11155914 0.0038
0.129300699 0.0026
0.121 0.0031
0.149454545 0.005
0.125479452 0.0031
0.111633987 0.0034
0.159534884 0.0036
0.150347826 0.0035
0.123484848 0.0032
0.105724138 0.0023
0.142868852 0.0036
0.136831683 0.0033
0.133703704 0.006
0.130857143 0.0044
0.13212766 0.0044
0.126826923 0.003
0.123538462 0.0023
0.152410714 0.0026
0.12140625 0.002
0.1667 0.0036
0.15245614 0.0046
0.131908397 0.0056
0.129507042 0.0025
0.0372 0.002
0.056153846 0.0022
0.033584906 0.0013
0.041 0.002
0.0416 0.0017
0.045 0.0018
0.035 0.0016
0.042 0.002
0.0488 0.0018
0.043653846 0.0021
0.032307692 0.0015
0.039642857 0.0018
0.041153846 0.0018
0.041272727 0.0015
0.0454 0.0021
0.038076923 0.0016
0.0498 0.0022
0.045769231 0.002
0.036 0.0015
0.0344 0.0018
0.060769231 0.0027
0.066153846 0.0027
0.06 0.0031
0.0676 0.0026
0.0668 0.0024
0.062075472 0.0024
0.086346154 0.0027
0.064 0.0023
0.0616 0.0026
0.054107143 0.0025
0.055714286 0.0021
0.0798 0.0027
0.062727273 0.0023
0.056101695 0.002
0.059056604 0.003
0.066551724 0.0022
0.066346154 0.0031
0.069230769 0.0025
0.070384615 0.0027
0.054464286 0.0025
};
\end{axis}

\end{tikzpicture}
		\caption{Correlation between the 50\%-stability threshold and the average stable pair robustness for the instances of our data set. Instances with an extremely high robustness are left out.}
		\label{fig:corr-avgpair-fif}
\end{figure}
\begin{figure*}[t]
	\begin{minipage}[t]{0.49\linewidth}
		\centering
		\includegraphics[width=\textwidth]{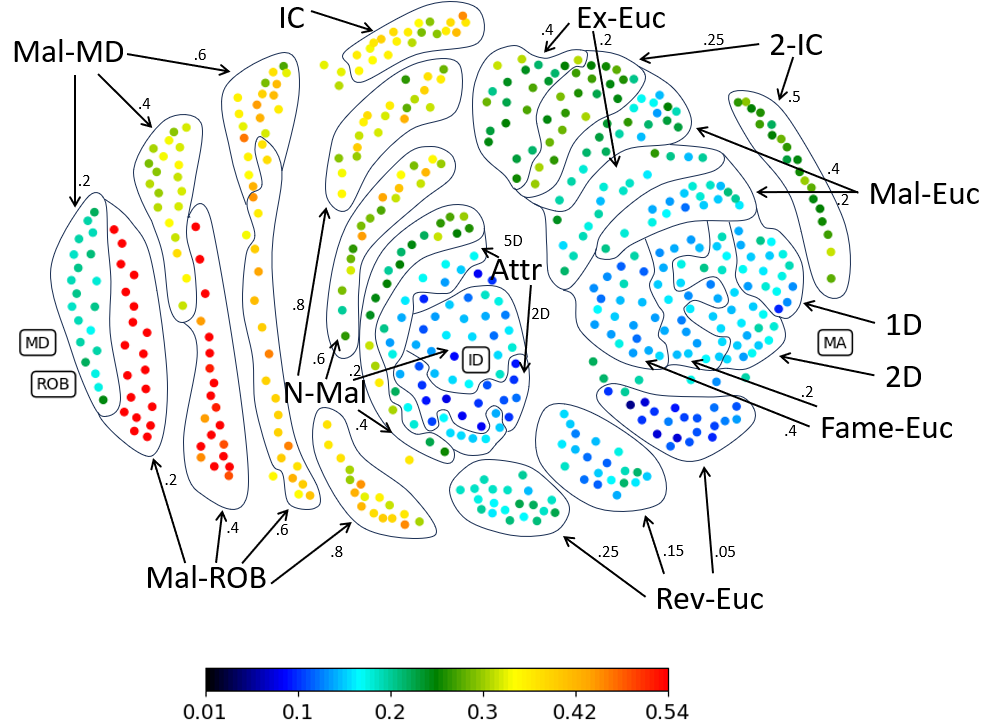}
		\caption{Average stable pair robustness for all instances of the data set.}
		\label{fig:avg-rob-pairs}
	\end{minipage}
	\hfill
	\begin{minipage}[t]{0.49\linewidth}
		\centering
		\includegraphics[width=\textwidth]{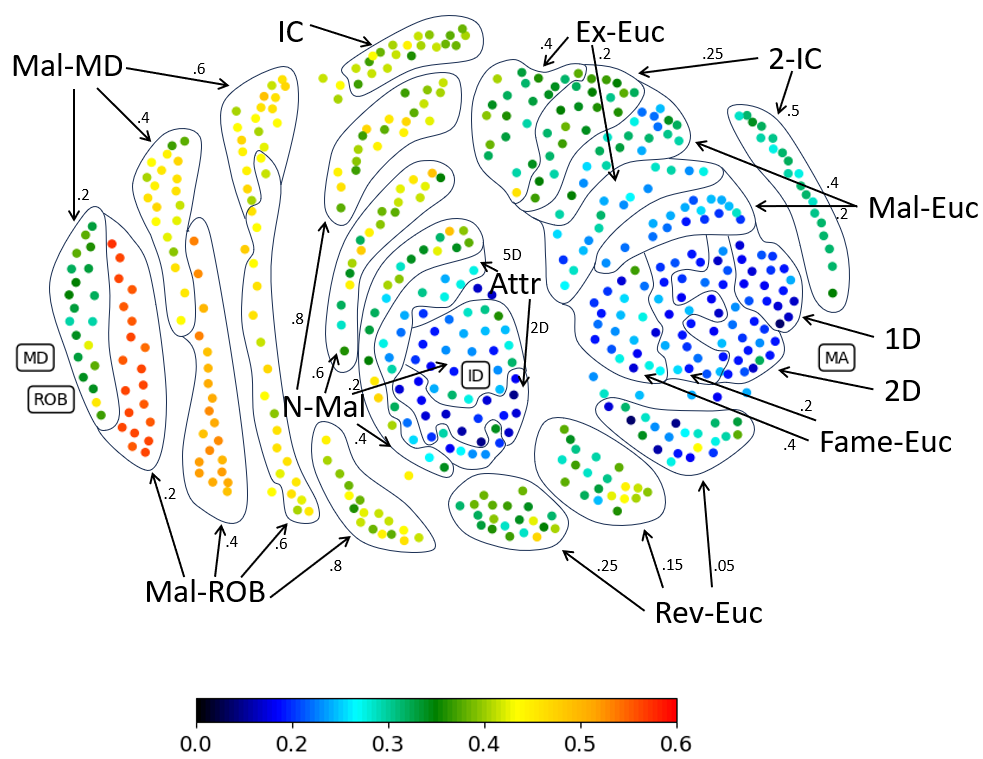}
		\caption{50\%-threshold of the most robust pair for each instance of the data set.}
		\label{fig:map-max-pairs}
	\end{minipage}
\end{figure*}

\subsubsection{Most robust stable pair}
First, one should note that it is not quite clear how to intuitively define the most robust stable pair. While we chose to look for the stable pair with the highest 50\%-threshold, \Cref{fig:devpairs}(b) shows that there could be other legitimate definitions that would yield other results. For our definition, the blue pair is the most robust, but if one is interested in smaller changes/norm-$\phi$ values, one could perfectly argue that the orange pair is the most robust. Generally, stable pairs do not behave as monotonously as stable matchings. In fact, the blue stable pair in \Cref{fig:devpairs}(f), after having a stable pair probability of 0.94 at norm-$\phi=0.04$, increases to 0.97 and only drops below 0.94 again at a norm-$\phi$ value greater 0.2. 

We now turn to \Cref{fig:map-max-pairs}, which depicts the most robust stable pair for all instances of the data set. One can observe that the most robust stable pair greatly differs from instance to instance. Unsurprisingly, pairs in the mallows-robust culture tend to be more robust. Generally, the right part of the map has small values for the most robust stable pair (see \Cref{fig:devpairs}(c) for an example), while the most robust stable pair for the instances on the left of the map has a high 50\%-threshold. An instance has usually a worse most robust stable pair if it is close to the ID or MA instance, as can be seen in \Cref{fig:map-max-pairs}.

The correlation to the matching-50\%-threshold is fairly low (the Pearson Correlation Coefficient being 0.3), among other things because of (extreme) outliers (see \Cref{fig:corr-maxpair-fif}). Thus, finding a robust stable pair in an instance does not imply that the instance admits robust stable matchings.

\begin{figure*}[t]
	\begin{minipage}[t]{0.49\linewidth}
		\centering
\begin{tikzpicture}[scale=0.8]

\definecolor{color0}{rgb}{0.12156862745098,0.466666666666667,0.705882352941177}

\begin{axis}[
tick align=outside,
tick pos=left,
x grid style={white!69.0196078431373!black},
xlabel={most stable pair robustness},
xmin=0.023, xmax=0.617,
xtick style={color=black},
y grid style={white!69.0196078431373!black},
ylabel={average man-opt matching robustness},
ymin=0, ymax=0.02,
ytick style={color=black}
]
\addplot [semithick, color0, mark=*, mark size=1, mark options={solid}, only marks]
table {%
0.05 0.001
0.18 0.0115
0.05 0.0005
0.59 0.2926
0.54 0.0162
0.55 0.2036
0.55 0.0062
0.54 0.2257
0.54 0.0453
0.55 0.0088
0.55 0.2351
0.56 0.0035
0.54 0.0041
0.55 0.0135
0.55 0.0048
0.54 0.2265
0.54 0.0028
0.55 0.0171
0.55 0.2133
0.54 0.0049
0.55 0.0087
0.55 0.2341
0.53 0.2185
0.55 0.0083
0.48 0.0052
0.48 0.0082
0.49 0.004
0.51 0.0033
0.49 0.0063
0.51 0.0071
0.49 0.0031
0.47 0.0075
0.52 0.0045
0.52 0.0038
0.5 0.0041
0.47 0.0029
0.47 0.0026
0.49 0.0058
0.49 0.0065
0.48 0.0053
0.51 0.0051
0.5 0.0054
0.51 0.0042
0.5 0.0025
0.44 0.0028
0.41 0.0039
0.37 0.0028
0.43 0.005
0.45 0.0034
0.39 0.0035
0.4 0.0024
0.44 0.0028
0.4 0.0033
0.41 0.0047
0.38 0.0026
0.44 0.0033
0.44 0.0033
0.44 0.0032
0.41 0.0027
0.43 0.0046
0.44 0.0034
0.44 0.0032
0.39 0.0041
0.41 0.0051
0.36 0.0035
0.39 0.0037
0.39 0.0051
0.39 0.0039
0.39 0.0026
0.4 0.0039
0.33 0.0024
0.44 0.0051
0.35 0.0024
0.41 0.0047
0.37 0.0043
0.38 0.003
0.41 0.0035
0.36 0.0022
0.42 0.0043
0.36 0.004
0.39 0.0034
0.42 0.0043
0.38 0.0023
0.41 0.0024
0.39 0.0023
0.38 0.0037
0.36 0.0038
0.34 0.0034
0.38 0.0052
0.39 0.0048
0.37 0.0027
0.36 0.0022
0.39 0.0025
0.36 0.0034
0.41 0.0027
0.39 0.0037
0.35 0.0028
0.38 0.0031
0.33 0.0038
0.34 0.0062
0.39 0.0037
0.39 0.0033
0.36 0.0038
0.42 0.0034
0.34 0.0028
0.35 0.0023
0.32 0.0022
0.31 0.0029
0.29 0.0027
0.32 0.0022
0.31 0.0037
0.33 0.0019
0.26 0.0022
0.32 0.0028
0.29 0.0023
0.34 0.0022
0.28 0.0033
0.36 0.0031
0.32 0.0028
0.27 0.0037
0.34 0.0024
0.34 0.0052
0.33 0.0054
0.28 0.0041
0.28 0.002
0.28 0.0026
0.28 0.0026
0.31 0.0023
0.32 0.0025
0.28 0.003
0.28 0.0035
0.27 0.0025
0.27 0.0018
0.25 0.0027
0.24 0.0022
0.28 0.0026
0.26 0.0022
0.24 0.0035
0.28 0.0032
0.29 0.0018
0.29 0.0038
0.28 0.0025
0.26 0.0032
0.25 0.0031
0.15 0.0029
0.21 0.0022
0.27 0.0027
0.28 0.0026
0.25 0.0018
0.21 0.0022
0.21 0.0038
0.28 0.0026
0.19 0.0024
0.2 0.0023
0.33 0.0029
0.24 0.0023
0.14 0.0019
0.19 0.0025
0.16 0.0023
0.19 0.0037
0.15 0.0019
0.2 0.0023
0.24 0.0029
0.18 0.0023
0.3 0.003
0.31 0.0031
0.31 0.0034
0.4 0.005
0.37 0.0025
0.37 0.0037
0.23 0.0056
0.32 0.0044
0.38 0.003
0.47 0.0036
0.38 0.0047
0.31 0.0039
0.45 0.0064
0.35 0.0028
0.25 0.0028
0.31 0.0037
0.22 0.0039
0.29 0.0024
0.43 0.0025
0.28 0.0057
0.39 0.0045
0.42 0.004
0.43 0.0028
0.4 0.0043
0.36 0.0035
0.29 0.0044
0.35 0.0051
0.29 0.0047
0.4 0.0052
0.46 0.006
0.38 0.0032
0.39 0.004
0.25 0.0031
0.31 0.0034
0.37 0.0028
0.47 0.0035
0.33 0.0049
0.42 0.0038
0.42 0.0035
0.32 0.0038
0.43 0.0033
0.39 0.0053
0.39 0.0034
0.45 0.0043
0.42 0.0031
0.35 0.0045
0.38 0.0043
0.37 0.0024
0.34 0.0028
0.42 0.0048
0.35 0.0035
0.37 0.0031
0.38 0.0035
0.45 0.0035
0.34 0.0034
0.4 0.0061
0.34 0.0047
0.42 0.0028
0.44 0.005
0.35 0.0028
0.16 0.0013
0.14 0.0015
0.17 0.0011
0.14 0.0014
0.14 0.0014
0.2 0.0012
0.13 0.0011
0.17 0.0017
0.12 0.0012
0.13 0.0013
0.13 0.0015
0.12 0.0011
0.14 0.0012
0.16 0.0012
0.15 0.0013
0.12 0.0012
0.09 0.001
0.15 0.0012
0.12 0.0011
0.15 0.0015
0.21 0.0015
0.16 0.0015
0.15 0.002
0.15 0.0014
0.19 0.0014
0.16 0.0013
0.15 0.0015
0.15 0.0019
0.17 0.0014
0.14 0.0016
0.19 0.0012
0.15 0.0018
0.14 0.0014
0.18 0.0022
0.19 0.0016
0.18 0.0015
0.15 0.0013
0.12 0.0018
0.16 0.0012
0.21 0.0015
0.3 0.0013
0.24 0.0021
0.16 0.0014
0.13 0.0014
0.41 0.0016
0.15 0.0013
0.28 0.0016
0.23 0.0017
0.11 0.0013
0.24 0.0014
0.25 0.0016
0.1 0.0013
0.25 0.0014
0.31 0.0015
0.19 0.0019
0.26 0.0019
0.14 0.0015
0.09 0.0013
0.35 0.0016
0.19 0.0021
0.33 0.0018
0.26 0.0017
0.35 0.0015
0.3 0.0016
0.31 0.0019
0.29 0.0016
0.27 0.0016
0.37 0.0018
0.25 0.0016
0.27 0.0016
0.42 0.0018
0.21 0.0017
0.41 0.002
0.38 0.0017
0.39 0.0018
0.35 0.0015
0.25 0.0014
0.4 0.0019
0.34 0.0014
0.35 0.002
0.36 0.0023
0.31 0.0023
0.31 0.0017
0.3 0.0023
0.29 0.002
0.34 0.0018
0.36 0.0023
0.35 0.0019
0.45 0.0026
0.34 0.0015
0.25 0.0015
0.36 0.0018
0.37 0.002
0.28 0.002
0.4 0.002
0.28 0.0016
0.33 0.0017
0.32 0.0016
0.25 0.0028
0.38 0.0023
0.14 0.0017
0.21 0.0015
0.16 0.0017
0.15 0.0014
0.2 0.0015
0.17 0.0016
0.15 0.0016
0.29 0.0019
0.15 0.0024
0.17 0.0018
0.15 0.0016
0.18 0.0018
0.15 0.0023
0.15 0.0018
0.22 0.0021
0.15 0.0019
0.15 0.0019
0.17 0.0017
0.21 0.0018
0.16 0.0017
0.17 0.002
0.13 0.002
0.13 0.002
0.23 0.0027
0.13 0.0014
0.22 0.0027
0.24 0.0025
0.18 0.0018
0.15 0.0017
0.21 0.0015
0.14 0.0016
0.24 0.0015
0.16 0.0019
0.15 0.0018
0.24 0.0022
0.21 0.0019
0.19 0.0019
0.34 0.0022
0.18 0.0028
0.16 0.0018
0.25 0.0024
0.19 0.0029
0.28 0.0016
0.21 0.0018
0.22 0.003
0.23 0.0026
0.19 0.0025
0.19 0.0018
0.25 0.0015
0.21 0.0027
0.26 0.002
0.19 0.0021
0.23 0.0022
0.2 0.0021
0.24 0.0031
0.16 0.0019
0.28 0.002
0.19 0.0021
0.25 0.0041
0.26 0.0023
0.33 0.0021
0.35 0.0028
0.36 0.0038
0.38 0.0028
0.31 0.0028
0.31 0.0016
0.28 0.0033
0.33 0.0028
0.29 0.0024
0.29 0.0023
0.26 0.0026
0.28 0.0027
0.3 0.003
0.31 0.0027
0.37 0.0035
0.3 0.0028
0.26 0.0027
0.44 0.0026
0.31 0.0031
0.3 0.0029
0.22 0.0019
0.14 0.0014
0.21 0.0024
0.16 0.0017
0.2 0.0018
0.18 0.0021
0.2 0.0019
0.17 0.0021
0.15 0.0013
0.18 0.0017
0.21 0.0025
0.2 0.0022
0.2 0.0016
0.2 0.0021
0.16 0.0019
0.19 0.0017
0.16 0.0021
0.2 0.0024
0.18 0.0022
0.25 0.0017
0.28 0.0022
0.21 0.0019
0.18 0.0018
0.18 0.0027
0.26 0.0029
0.24 0.0027
0.28 0.0023
0.26 0.0032
0.3 0.0029
0.19 0.0022
0.2 0.0019
0.32 0.0026
0.25 0.0039
0.29 0.0019
0.26 0.002
0.29 0.0019
0.22 0.0019
0.18 0.0022
0.31 0.0023
0.27 0.0019
0.35 0.0022
0.3 0.0023
0.3 0.0025
0.3 0.002
0.33 0.0018
0.35 0.0016
0.29 0.002
0.29 0.0023
0.4 0.002
0.32 0.0021
0.32 0.003
0.3 0.0024
0.26 0.0019
0.31 0.0022
0.26 0.0024
0.31 0.0029
0.35 0.0019
0.31 0.0026
0.43 0.0021
0.34 0.0021
0.43 0.0041
0.45 0.0028
0.44 0.0034
0.4 0.0032
0.42 0.0029
0.45 0.0027
0.34 0.0044
0.39 0.0029
0.44 0.0021
0.44 0.0038
0.35 0.0025
0.4 0.0028
0.37 0.0021
0.42 0.0057
0.44 0.0035
0.41 0.0022
0.41 0.0029
0.38 0.0038
0.43 0.0026
0.41 0.0031
0.44 0.005
0.46 0.0031
0.47 0.0034
0.41 0.0036
0.44 0.0035
0.45 0.0032
0.44 0.0023
0.44 0.0036
0.39 0.0033
0.41 0.006
0.38 0.0044
0.41 0.0044
0.45 0.003
0.45 0.0023
0.38 0.0026
0.4 0.002
0.42 0.0036
0.38 0.0046
0.42 0.0056
0.41 0.0025
0.13 0.002
0.13 0.0022
0.12 0.0013
0.12 0.002
0.14 0.0017
0.13 0.0018
0.11 0.0016
0.12 0.002
0.15 0.0018
0.16 0.0021
0.12 0.0015
0.16 0.0018
0.1 0.0018
0.16 0.0015
0.15 0.0021
0.14 0.0016
0.13 0.0022
0.14 0.002
0.15 0.0015
0.1 0.0018
0.24 0.0027
0.26 0.0027
0.19 0.0031
0.2 0.0026
0.27 0.0024
0.14 0.0024
0.29 0.0027
0.2 0.0023
0.23 0.0026
0.19 0.0025
0.18 0.0021
0.24 0.0027
0.19 0.0023
0.31 0.002
0.16 0.003
0.24 0.0022
0.18 0.0031
0.26 0.0025
0.22 0.0027
0.24 0.0025
};
\end{axis}

\end{tikzpicture}
		\caption{Correlation between the 50\%-stability threshold and the most robust stable pair robustness for the instances of our data set. Instances with an extremely high robustness are left out.}
		\label{fig:corr-maxpair-fif}
	\end{minipage}
	\hfill
	\begin{minipage}[t]{0.49\linewidth}
		\centering
		\includegraphics[width=\textwidth]{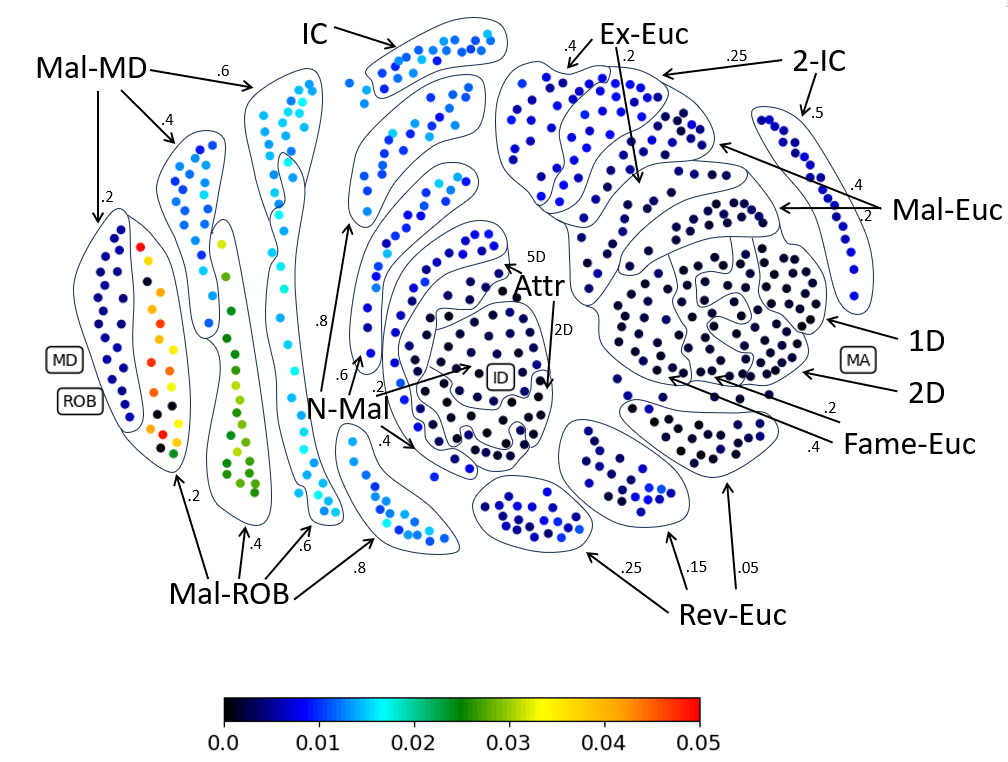}
		\caption{Robustness variance for each instance of the data set.}
		\label{fig:map-var-pairs}
	\end{minipage}
	
\end{figure*}

\subsubsection{Least robust stable pair}
Most instances have stable pairs that are very unrobust. Thus, the least robust stable pair is not a very meaningful measure. For the same reason, we do not examine the difference between the most and least robust stable pair, as it mostly almost corresponds to the value of the most robust pair. Only four non-special instances have a least robust stable pair whose threshold is significantly larger than zero, and all of them are from the Mal-ROB culture.

\subsubsection{Robustness variance}
With this measure, we try to quantify by how much the robustness of two different stable pairs can differ. As can be seen in \Cref{fig:map-var-pairs}, the left side of the map again has clearly greater values. Even though the ROB instance itself has a very low robustness variance (see \Cref{fig:devpairs}(e)), most Mal-ROB instances have a very high robustness variance. More precisely, the stable pairs of the Mal-ROB instances that are also stable in the ROB instance are very robst, while all other stable pairs are not robust. This is illustrated in \Cref{fig:dev12}. For instances of other cultures, we often encounter single stable pairs that are much more robust than any other stable pair of the instance, as in \Cref{fig:devpairs}(d). This results in a greater robustness variance and also shows that it is sensible to regard the average and most robust stable pairs separately. 

One should also notice that the map of the robustness variance is very similar to the map of the number of stable pairs in the initial instance, suggesting that if an instance has only few stable pairs, they are similarly robust.

\subsubsection{Unstable Pairs}
We now want to examine the development of the stable pair probability of unstable pairs. While clearly any unstable pair has a stable pair probabilty of 0 for norm-$\phi=0$, we expect some unstable pairs to have a much higher stable pair probability (for some norm-$\phi$ values) than other stable pairs (e.g. because of a higher mutual agreement of the agents). Contrary to unstable matchings, we observe that such unstable pairs with high stable pair probability exist and are not uncommon:
Pairs that are initially unstable can have a quite high stable pair probability for small norm-$\phi$ values. 

\Cref{fig:unstable_pairs} shows some exemplary unstable pairs for two instances. We show the unstable pairs with the highest, lowest and an average stable pair probability. For most instances, there is some unstable pair that has a stable pair probability around 0.5 for norm-$\phi$ values close to 0.2. However, the stable pair probabilty of most unstable pairs does not increase significantly and behaves similar to the blue pairs in \Cref{fig:unstable_pairs}(a,b).

\begin{figure*}[t]
	\begin{minipage}[t]{0.45\linewidth}
		\centering
		\subfigure[IC]{
\begin{tikzpicture}[scale=0.8]

\definecolor{color0}{rgb}{0.12156862745098,0.466666666666667,0.705882352941177}
\definecolor{color1}{rgb}{1,0.498039215686275,0.0549019607843137}
\definecolor{color2}{rgb}{0.172549019607843,0.627450980392157,0.172549019607843}

\begin{axis}[
tick align=outside,
tick pos=left,
x grid style={white!69.0196078431373!black},
xmin=-0.05, xmax=1.05,
xtick style={color=black},
y grid style={white!69.0196078431373!black},
ymin=-0.02775, ymax=1.02,
ytick style={color=black},
xlabel={norm-$\phi$},
ylabel={stable pair probability}
]
\addplot [semithick, color0]
table {%
0 0
0.01 0
0.02 0
0.03 0
0.04 0
0.05 0
0.06 0
0.07 0
0.08 0
0.09 0
0.1 0
0.11 0
0.12 0
0.13 0
0.14 0.002
0.15 0.002
0.16 0
0.17 0
0.18 0.001
0.19 0.002
0.2 0.001
0.21 0.003
0.22 0.002
0.23 0.005
0.24 0.001
0.25 0.007
0.26 0.003
0.27 0.003
0.28 0.004
0.29 0.009
0.3 0.005
0.31 0.002
0.32 0.01
0.33 0.005
0.34 0.012
0.35 0.012
0.36 0.006
0.37 0.012
0.38 0.016
0.39 0.017
0.4 0.01
0.41 0.017
0.42 0.012
0.43 0.02
0.44 0.022
0.45 0.015
0.46 0.024
0.47 0.03
0.48 0.029
0.49 0.028
0.5 0.024
0.51 0.023
0.52 0.027
0.53 0.026
0.54 0.026
0.55 0.027
0.56 0.019
0.57 0.038
0.58 0.031
0.59 0.041
0.6 0.036
0.61 0.026
0.62 0.026
0.63 0.037
0.64 0.033
0.65 0.034
0.66 0.045
0.67 0.028
0.68 0.033
0.69 0.042
0.7 0.044
0.71 0.03
0.72 0.041
0.73 0.047
0.74 0.042
0.75 0.056
0.76 0.034
0.77 0.041
0.78 0.043
0.79 0.04
0.8 0.043
0.81 0.037
0.820000000000001 0.034
0.830000000000001 0.038
0.840000000000001 0.049
0.850000000000001 0.038
0.860000000000001 0.059
0.870000000000001 0.037
0.880000000000001 0.045
0.890000000000001 0.035
0.900000000000001 0.05
0.910000000000001 0.039
0.920000000000001 0.034
0.930000000000001 0.047
0.940000000000001 0.037
0.950000000000001 0.04
0.960000000000001 0.037
0.970000000000001 0.042
0.980000000000001 0.051
0.990000000000001 0.056
1 0.042
};
\addplot [semithick, color1]
table {%
0 0
0.01 0.037
0.02 0.114
0.03 0.13
0.04 0.158
0.05 0.165
0.06 0.196
0.07 0.218
0.08 0.233
0.09 0.236
0.1 0.26
0.11 0.254
0.12 0.281
0.13 0.252
0.14 0.249
0.15 0.288
0.16 0.281
0.17 0.281
0.18 0.299
0.19 0.287
0.2 0.287
0.21 0.273
0.22 0.273
0.23 0.283
0.24 0.265
0.25 0.243
0.26 0.286
0.27 0.249
0.28 0.29
0.29 0.261
0.3 0.241
0.31 0.24
0.32 0.264
0.33 0.241
0.34 0.237
0.35 0.236
0.36 0.238
0.37 0.225
0.38 0.241
0.39 0.209
0.4 0.213
0.41 0.223
0.42 0.195
0.43 0.222
0.44 0.183
0.45 0.208
0.46 0.178
0.47 0.212
0.48 0.196
0.49 0.183
0.5 0.202
0.51 0.194
0.52 0.187
0.53 0.168
0.54 0.181
0.55 0.169
0.56 0.15
0.57 0.141
0.58 0.161
0.59 0.148
0.6 0.141
0.61 0.151
0.62 0.159
0.63 0.146
0.64 0.136
0.65 0.126
0.66 0.122
0.67 0.141
0.68 0.13
0.69 0.128
0.7 0.098
0.71 0.106
0.72 0.105
0.73 0.096
0.74 0.082
0.75 0.105
0.76 0.085
0.77 0.101
0.78 0.081
0.79 0.095
0.8 0.078
0.81 0.091
0.820000000000001 0.071
0.830000000000001 0.073
0.840000000000001 0.078
0.850000000000001 0.069
0.860000000000001 0.052
0.870000000000001 0.071
0.880000000000001 0.06
0.890000000000001 0.036
0.900000000000001 0.048
0.910000000000001 0.054
0.920000000000001 0.054
0.930000000000001 0.056
0.940000000000001 0.059
0.950000000000001 0.047
0.960000000000001 0.041
0.970000000000001 0.041
0.980000000000001 0.039
0.990000000000001 0.03
1 0.057
};
\addplot [semithick, color2]
table {%
0 0
0.01 0.109
0.02 0.227
0.03 0.308
0.04 0.353
0.05 0.433
0.06 0.453
0.07 0.48
0.08 0.511
0.09 0.5
0.1 0.525
0.11 0.498
0.12 0.538
0.13 0.506
0.14 0.517
0.15 0.551
0.16 0.526
0.17 0.555
0.18 0.475
0.19 0.5
0.2 0.494
0.21 0.48
0.22 0.45
0.23 0.485
0.24 0.445
0.25 0.496
0.26 0.427
0.27 0.403
0.28 0.423
0.29 0.425
0.3 0.418
0.31 0.413
0.32 0.439
0.33 0.382
0.34 0.416
0.35 0.384
0.36 0.356
0.37 0.343
0.38 0.343
0.39 0.328
0.4 0.335
0.41 0.348
0.42 0.33
0.43 0.297
0.44 0.295
0.45 0.287
0.46 0.282
0.47 0.283
0.48 0.267
0.49 0.285
0.5 0.265
0.51 0.23
0.52 0.231
0.53 0.233
0.54 0.213
0.55 0.239
0.56 0.191
0.57 0.206
0.58 0.189
0.59 0.195
0.6 0.189
0.61 0.186
0.62 0.2
0.63 0.18
0.64 0.184
0.65 0.168
0.66 0.165
0.67 0.139
0.68 0.144
0.69 0.128
0.7 0.131
0.71 0.129
0.72 0.12
0.73 0.152
0.74 0.119
0.75 0.117
0.76 0.11
0.77 0.094
0.78 0.09
0.79 0.104
0.8 0.085
0.81 0.078
0.820000000000001 0.078
0.830000000000001 0.08
0.840000000000001 0.066
0.850000000000001 0.08
0.860000000000001 0.073
0.870000000000001 0.072
0.880000000000001 0.061
0.890000000000001 0.052
0.900000000000001 0.063
0.910000000000001 0.069
0.920000000000001 0.061
0.930000000000001 0.06
0.940000000000001 0.05
0.950000000000001 0.047
0.960000000000001 0.05
0.970000000000001 0.047
0.980000000000001 0.047
0.990000000000001 0.036
1 0.036
};
\end{axis}

\end{tikzpicture}}
		\label{fig:dev16}
	\end{minipage}
	\hfill
	\begin{minipage}[t]{0.45\linewidth}
		\centering
		\subfigure[2d-rev 0.05]{
\begin{tikzpicture}[scale=0.8]

\definecolor{color0}{rgb}{0.12156862745098,0.466666666666667,0.705882352941177}
\definecolor{color1}{rgb}{1,0.498039215686275,0.0549019607843137}
\definecolor{color2}{rgb}{0.172549019607843,0.627450980392157,0.172549019607843}

\begin{axis}[
tick align=outside,
tick pos=left,
x grid style={white!69.0196078431373!black},
xmin=-0.05, xmax=1.05,
xtick style={color=black},
y grid style={white!69.0196078431373!black},
ymin=-0.0166, ymax=1.02,
ytick style={color=black},
xlabel={norm-$\phi$},
ylabel={stable pair probability}
]
\addplot [semithick, color0]
table {%
0 0
0.01 0
0.02 0
0.03 0
0.04 0
0.05 0
0.06 0
0.07 0
0.08 0
0.09 0
0.1 0
0.11 0
0.12 0
0.13 0
0.14 0
0.15 0
0.16 0.001
0.17 0
0.18 0
0.19 0
0.2 0
0.21 0
0.22 0.001
0.23 0.001
0.24 0.002
0.25 0
0.26 0
0.27 0.001
0.28 0.001
0.29 0
0.3 0.003
0.31 0.002
0.32 0
0.33 0.002
0.34 0.003
0.35 0.003
0.36 0
0.37 0.003
0.38 0.002
0.39 0.005
0.4 0.007
0.41 0.003
0.42 0.003
0.43 0.004
0.44 0.005
0.45 0.005
0.46 0.008
0.47 0.012
0.48 0.01
0.49 0.012
0.5 0.007
0.51 0.01
0.52 0.007
0.53 0.014
0.54 0.014
0.55 0.022
0.56 0.017
0.57 0.017
0.58 0.016
0.59 0.017
0.6 0.024
0.61 0.032
0.62 0.031
0.63 0.031
0.64 0.03
0.65 0.024
0.66 0.024
0.67 0.025
0.68 0.021
0.69 0.034
0.7 0.034
0.71 0.037
0.72 0.032
0.73 0.035
0.74 0.035
0.75 0.042
0.76 0.035
0.77 0.036
0.78 0.038
0.79 0.048
0.8 0.048
0.81 0.04
0.820000000000001 0.054
0.830000000000001 0.037
0.840000000000001 0.059
0.850000000000001 0.046
0.860000000000001 0.038
0.870000000000001 0.045
0.880000000000001 0.045
0.890000000000001 0.05
0.900000000000001 0.049
0.910000000000001 0.05
0.920000000000001 0.041
0.930000000000001 0.042
0.940000000000001 0.069
0.950000000000001 0.03
0.960000000000001 0.035
0.970000000000001 0.044
0.980000000000001 0.04
0.990000000000001 0.052
1 0.04
};
\addplot [semithick, color1]
table {%
0 0
0.01 0
0.02 0.008
0.03 0.014
0.04 0.024
0.05 0.031
0.06 0.037
0.07 0.061
0.08 0.061
0.09 0.073
0.1 0.077
0.11 0.108
0.12 0.105
0.13 0.108
0.14 0.114
0.15 0.124
0.16 0.126
0.17 0.125
0.18 0.122
0.19 0.109
0.2 0.117
0.21 0.127
0.22 0.101
0.23 0.125
0.24 0.115
0.25 0.117
0.26 0.103
0.27 0.113
0.28 0.105
0.29 0.122
0.3 0.099
0.31 0.091
0.32 0.098
0.33 0.088
0.34 0.063
0.35 0.079
0.36 0.092
0.37 0.071
0.38 0.081
0.39 0.083
0.4 0.06
0.41 0.067
0.42 0.058
0.43 0.066
0.44 0.062
0.45 0.078
0.46 0.049
0.47 0.054
0.48 0.043
0.49 0.044
0.5 0.053
0.51 0.041
0.52 0.055
0.53 0.035
0.54 0.039
0.55 0.042
0.56 0.033
0.57 0.036
0.58 0.036
0.59 0.029
0.6 0.032
0.61 0.042
0.62 0.033
0.63 0.026
0.64 0.036
0.65 0.026
0.66 0.036
0.67 0.034
0.68 0.032
0.69 0.034
0.7 0.033
0.71 0.024
0.72 0.025
0.73 0.028
0.74 0.025
0.75 0.044
0.76 0.029
0.77 0.031
0.78 0.037
0.79 0.036
0.8 0.042
0.81 0.029
0.820000000000001 0.038
0.830000000000001 0.033
0.840000000000001 0.046
0.850000000000001 0.024
0.860000000000001 0.027
0.870000000000001 0.039
0.880000000000001 0.039
0.890000000000001 0.043
0.900000000000001 0.038
0.910000000000001 0.037
0.920000000000001 0.039
0.930000000000001 0.047
0.940000000000001 0.023
0.950000000000001 0.043
0.960000000000001 0.043
0.970000000000001 0.032
0.980000000000001 0.031
0.990000000000001 0.036
1 0.044
};
\addplot [semithick, color2]
table {%
0 0
0.01 0.153
0.02 0.244
0.03 0.283
0.04 0.256
0.05 0.286
0.06 0.312
0.07 0.271
0.08 0.293
0.09 0.287
0.1 0.304
0.11 0.303
0.12 0.299
0.13 0.331
0.14 0.319
0.15 0.321
0.16 0.307
0.17 0.299
0.18 0.332
0.19 0.309
0.2 0.306
0.21 0.317
0.22 0.312
0.23 0.311
0.24 0.294
0.25 0.262
0.26 0.296
0.27 0.273
0.28 0.309
0.29 0.287
0.3 0.277
0.31 0.269
0.32 0.268
0.33 0.265
0.34 0.264
0.35 0.241
0.36 0.232
0.37 0.228
0.38 0.218
0.39 0.253
0.4 0.191
0.41 0.245
0.42 0.198
0.43 0.206
0.44 0.173
0.45 0.193
0.46 0.189
0.47 0.184
0.48 0.193
0.49 0.164
0.5 0.169
0.51 0.158
0.52 0.14
0.53 0.166
0.54 0.131
0.55 0.167
0.56 0.13
0.57 0.138
0.58 0.123
0.59 0.136
0.6 0.119
0.61 0.114
0.62 0.118
0.63 0.121
0.64 0.101
0.65 0.1
0.66 0.094
0.67 0.107
0.68 0.102
0.69 0.108
0.7 0.096
0.71 0.086
0.72 0.081
0.73 0.09
0.74 0.079
0.75 0.078
0.76 0.087
0.77 0.065
0.78 0.061
0.79 0.074
0.8 0.078
0.81 0.067
0.820000000000001 0.065
0.830000000000001 0.057
0.840000000000001 0.051
0.850000000000001 0.059
0.860000000000001 0.075
0.870000000000001 0.053
0.880000000000001 0.054
0.890000000000001 0.048
0.900000000000001 0.062
0.910000000000001 0.046
0.920000000000001 0.036
0.930000000000001 0.042
0.940000000000001 0.036
0.950000000000001 0.042
0.960000000000001 0.046
0.970000000000001 0.037
0.980000000000001 0.034
0.990000000000001 0.041
1 0.029
};
\end{axis}

\end{tikzpicture}}
		\label{fig:dev17}
	\end{minipage}
	\caption{Average-case robustness of some exemplary initially unstable pairs for two different instances. Each line corresponds to one stable pair.}\label{fig:unstable_pairs}
\end{figure*}
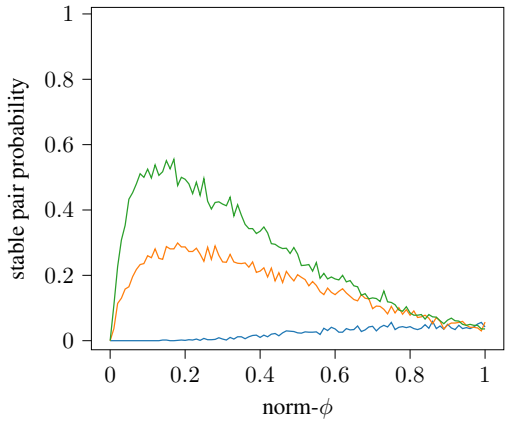
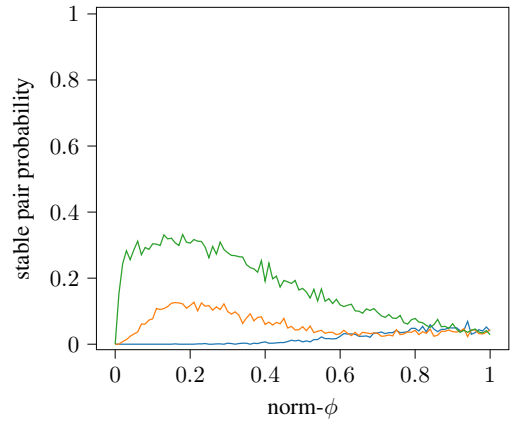

\begin{thebibliography}{49}
\providecommand{\natexlab}[1]{#1}
\providecommand{\url}[1]{\texttt{#1}}
\expandafter\ifx\csname urlstyle\endcsname\relax
  \providecommand{\doi}[1]{doi: #1}\else
  \providecommand{\doi}{doi: \begingroup \urlstyle{rm}\Url}\fi

\bibitem[Abdulkadiro{\u{g}}lu et~al.(2005)Abdulkadiro{\u{g}}lu, Pathak, Roth,
  and S{\"o}nmez]{abdulkadirouglu2005boston}
Atila Abdulkadiro{\u{g}}lu, Parag~A Pathak, Alvin~E Roth, and Tayfun
  S{\"o}nmez.
\newblock The boston public school match.
\newblock \emph{Am. Econ. Rev.}, 95\penalty0 (2):\penalty0 368--371, 2005.

\bibitem[Akbarpour et~al.(2020)Akbarpour, Li, and Gharan]{ALO20}
Mohammad Akbarpour, Shengwu Li, and Shayan~Oveis Gharan.
\newblock Thickness and information in dynamic matching markets.
\newblock \emph{J. Political Econ.}, 128\penalty0 (3):\penalty0 783--815, 2020.

\bibitem[Aziz et~al.(2020)Aziz, Bir{\'{o}}, Gaspers, de~Haan, Mattei, and
  Rastegari]{DBLP:journals/algorithmica/AzizBGHMR20}
Haris Aziz, P{\'{e}}ter Bir{\'{o}}, Serge Gaspers, Ronald de~Haan, Nicholas
  Mattei, and Baharak Rastegari.
\newblock Stable matching with uncertain linear preferences.
\newblock \emph{Algorithmica}, 82\penalty0 (5):\penalty0 1410--1433, 2020.

\bibitem[Aziz et~al.(2022)Aziz, Bir{\'{o}}, Fleiner, Gaspers, de~Haan, Mattei,
  and Rastegari]{DBLP:journals/tcs/0001BFGHMR22}
Haris Aziz, P{\'{e}}ter Bir{\'{o}}, Tam{\'{a}}s Fleiner, Serge Gaspers, Ronald
  de~Haan, Nicholas Mattei, and Baharak Rastegari.
\newblock Stable matching with uncertain pairwise preferences.
\newblock \emph{Theor. Comput. Sci.}, 909:\penalty0 1--11, 2022.

\bibitem[Baccara et~al.(2020)Baccara, Lee, and Yariv]{baccara2020optimal}
Mariagiovanna Baccara, SangMok Lee, and Leeat Yariv.
\newblock Optimal dynamic matching.
\newblock \emph{Theor. Econ.}, 15\penalty0 (3):\penalty0 1221--1278, 2020.

\bibitem[Baumeister and Hogrebe(2023)]{DBLP:journals/sncs/BaumeisterH23}
Dorothea Baumeister and Tobias Hogrebe.
\newblock On the complexity of predicting election outcomes and estimating
  their robustness.
\newblock \emph{{SN} Comput. Sci.}, 4\penalty0 (4):\penalty0 362, 2023.

\bibitem[B{\'{e}}rczi et~al.(2022)B{\'{e}}rczi, Cs{\'{a}}ji, and
  Kir{\'{a}}ly]{DBLP:journals/corr/abs-2204-13485}
Krist{\'{o}}f B{\'{e}}rczi, Gergely Cs{\'{a}}ji, and Tam{\'{a}}s Kir{\'{a}}ly.
\newblock Manipulating the outcome of stable matching and roommates problems.
\newblock \emph{CoRR}, abs/2204.13485, 2022.

\bibitem[Bhattacharya et~al.(2015)Bhattacharya, Hoefer, Huang, Kavitha, and
  Wagner]{DBLP:conf/icalp/BhattacharyaHHK15}
Sayan Bhattacharya, Martin Hoefer, Chien{-}Chung Huang, Telikepalli Kavitha,
  and Lisa Wagner.
\newblock Maintaining near-popular matchings.
\newblock In \emph{Proceedings of the 42nd International Colloquium on
  Automata, Languages, and Programming ({ICALP} '15)}, pages 504--515.
  Springer, 2015.

\bibitem[Boehmer and Niedermeier(2021)]{DBLP:conf/atal/BoehmerN21}
Niclas Boehmer and Rolf Niedermeier.
\newblock Broadening the research agenda for computational social choice:
  Multiple preference profiles and multiple solutions.
\newblock In \emph{Proceedings of the 20th International Conference on
  Autonomous Agents and Multiagent Systems ({AAMAS} '21)}, pages 1--5. {ACM},
  2021.

\bibitem[Boehmer et~al.(2021{\natexlab{a}})Boehmer, Bredereck, Faliszewski, and
  Niedermeier]{DBLP:conf/ijcai/BoehmerBFN21}
Niclas Boehmer, Robert Bredereck, Piotr Faliszewski, and Rolf Niedermeier.
\newblock Winner robustness via swap- and shift-bribery: Parameterized counting
  complexity and experiments.
\newblock In \emph{Proceedings of the Thirtieth International Joint Conference
  on Artificial Intelligence ({IJCAI} '21)}, pages 52--58. ijcai.org,
  2021{\natexlab{a}}.

\bibitem[Boehmer et~al.(2021{\natexlab{b}})Boehmer, Bredereck, Faliszewski,
  Niedermeier, and Szufa]{DBLP:journals/corr/abs-2105-07815}
Niclas Boehmer, Robert Bredereck, Piotr Faliszewski, Rolf Niedermeier, and
  Stanislaw Szufa.
\newblock Putting a compass on the map of elections.
\newblock In \emph{Proceedings of the 30th International Joint Conference on
  Artificial Intelligence ({IJCAI} '21)}, pages 59--65. ijcai.org,
  2021{\natexlab{b}}.

\bibitem[Boehmer et~al.(2021{\natexlab{c}})Boehmer, Bredereck, Heeger, and
  Niedermeier]{DBLP:conf/sagt/BoehmerBHN20}
Niclas Boehmer, Robert Bredereck, Klaus Heeger, and Rolf Niedermeier.
\newblock Bribery and control in stable marriage.
\newblock \emph{J. Artif. Intell. Res.}, 71:\penalty0 993--1048,
  2021{\natexlab{c}}.

\bibitem[Boehmer et~al.(2022{\natexlab{a}})Boehmer, Bredereck, Faliszewski, and
  Niedermeier]{DBLP:conf/eaamo/BoehmerBFN22}
Niclas Boehmer, Robert Bredereck, Piotr Faliszewski, and Rolf Niedermeier.
\newblock A quantitative and qualitative analysis of the robustness of
  (real-world) election winners.
\newblock In \emph{Proceedings of the Second ACM Conference on Equity and
  Access in Algorithms, Mechanisms, and Optimization ({EAAMO} '22)}, pages
  7:1--7:10. {ACM}, 2022{\natexlab{a}}.

\bibitem[Boehmer et~al.(2022{\natexlab{b}})Boehmer, Heeger, and
  Niedermeier]{uschanged}
Niclas Boehmer, Klaus Heeger, and Rolf Niedermeier.
\newblock Theory of and experiments on minimally invasive stability
  preservation in changing two-sided matching markets.
\newblock In \emph{Proceedings of the Thirty-Sixth {AAAI} Conference on
  Artificial Intelligence ({AAAI} '22)}, pages 4851--4858. {AAAI} Press,
  2022{\natexlab{b}}.

\bibitem[Boehmer et~al.(2023{\natexlab{a}})Boehmer, Bredereck, Knop, and
  Luo]{DBLP:journals/aamas/BoehmerBKL23}
Niclas Boehmer, Robert Bredereck, Dusan Knop, and Junjie Luo.
\newblock Fine-grained view on bribery for group identification.
\newblock \emph{Auton. Agents Multi Agent Syst.}, 37\penalty0 (1):\penalty0 21,
  2023{\natexlab{a}}.

\bibitem[Boehmer et~al.(2023{\natexlab{b}})Boehmer, Faliszewski, and
  Kraiczy]{icml}
Niclas Boehmer, Piotr Faliszewski, and Sonja Kraiczy.
\newblock Properties of the {M}allows model depending on the number of
  alternatives: A warning for an experimentalist.
\newblock In \emph{Proceedings of the Fortieth International Conference on
  Machine Learning ({ICML '23})}, pages 2689--2711. {PMLR}, 2023{\natexlab{b}}.

\bibitem[Boehmer et~al.(2023{\natexlab{c}})Boehmer, Heeger, and
  Szufa]{DBLP:journals/corr/abs-2208-04041}
Niclas Boehmer, Klaus Heeger, and Stanislaw Szufa.
\newblock A map of diverse synthetic stable roommates instances.
\newblock In \emph{Proceedings of the 2023 International Conference on
  Autonomous Agents and Multiagent Systems ({AAMAS} '23)}, pages 1003--1011.
  {ACM}, 2023{\natexlab{c}}.

\bibitem[Bredereck et~al.(2020)Bredereck, Chen, Knop, Luo, and
  Niedermeier]{DBLP:conf/aaai/BredereckCKLN20}
Robert Bredereck, Jiehua Chen, Du\v{s}an Knop, Junjie Luo, and Rolf
  Niedermeier.
\newblock Adapting stable matchings to evolving preferences.
\newblock In \emph{Proceedings of the Thirty-Fourth {AAAI} Conference on
  Artificial Intelligence ({AAAI} '20)}, pages 1830--1837. {AAAI} Press, 2020.

\bibitem[Brill et~al.(2022)Brill, Schmidt{-}Kraepelin, and
  Suksompong]{DBLP:journals/ai/BrillSS22}
Markus Brill, Ulrike Schmidt{-}Kraepelin, and Warut Suksompong.
\newblock Margin of victory for tournament solutions.
\newblock \emph{Artif. Intell.}, 302:\penalty0 103600, 2022.

\bibitem[Bubley and Dyer(1997)]{bubley1997graph}
Russ Bubley and Martin~E. Dyer.
\newblock Graph orientations with no sink and an approximation for a hard case
  of {\#}sat.
\newblock In \emph{Proceedings of the Eighth Annual {ACM-SIAM} Symposium on
  Discrete Algorithms ({SODA} '97)}, pages 248--257. {ACM/SIAM}, 1997.

\bibitem[Chen et~al.(2021)Chen, Skowron, and
  Sorge]{DBLP:journals/teco/ChenSS21}
Jiehua Chen, Piotr Skowron, and Manuel Sorge.
\newblock Matchings under preferences: Strength of stability and tradeoffs.
\newblock \emph{{ACM} Trans. Economics and Comput.}, 9\penalty0 (4):\penalty0
  20:1--20:55, 2021.

\bibitem[Damiano and Lam(2005)]{DBLP:journals/geb/DamianoL05}
Ettore Damiano and Ricky Lam.
\newblock Stability in dynamic matching markets.
\newblock \emph{Games Econ. Behav.}, 52\penalty0 (1):\penalty0 34--53, 2005.

\bibitem[D{\"{o}}ring and Peters(2023)]{DBLP:conf/atal/Doring023}
Michelle D{\"{o}}ring and Jannik Peters.
\newblock Margin of victory for weighted tournament solutions.
\newblock In \emph{Proceedings of the 2023 International Conference on
  Autonomous Agents and Multiagent Systems ({AAMAS} '23)}, pages 1716--1724.
  {ACM}, 2023.

\bibitem[Doval(2022)]{DBLP:journals/corr/abs-1906.11391}
Laura Doval.
\newblock Dynamically stable matching.
\newblock \emph{Theor. Econ.}, 17\penalty0 (2):\penalty0 687--724, 2022.

\bibitem[Eiben et~al.(2023)Eiben, Gutin, Neary, Rambaud, Wahlstr{\"{o}}m, and
  Yeo]{DBLP:journals/tcs/EibenGNRWY23}
Eduard Eiben, Gregory~Z. Gutin, Philip~R. Neary, Cl{\'{e}}ment Rambaud, Magnus
  Wahlstr{\"{o}}m, and Anders Yeo.
\newblock Preference swaps for the stable matching problem.
\newblock \emph{Theor. Comput. Sci.}, 940:\penalty0 222--230, 2023.

\bibitem[Faliszewski and Rothe(2016)]{DBLP:reference/choice/FaliszewskiR16}
Piotr Faliszewski and J{\"{o}}rg Rothe.
\newblock Control and bribery in voting.
\newblock In Felix Brandt, Vincent Conitzer, Ulle Endriss, J{\'{e}}r{\^{o}}me
  Lang, and Ariel~D. Procaccia, editors, \emph{Handbook of Computational Social
  Choice}, pages 146--168. Cambridge University Press, 2016.

\bibitem[Faliszewski et~al.(2009)Faliszewski, Hemaspaandra, and
  Hemaspaandra]{DBLP:journals/jair/FaliszewskiHH09}
Piotr Faliszewski, Edith Hemaspaandra, and Lane~A. Hemaspaandra.
\newblock How hard is bribery in elections?
\newblock \emph{J. Artif. Intell. Res.}, 35:\penalty0 485--532, 2009.

\bibitem[Feigenbaum et~al.(2020)Feigenbaum, Kanoria, Lo, and
  Sethuraman]{Feigenbaum17}
Itai Feigenbaum, Yash Kanoria, Irene Lo, and Jay Sethuraman.
\newblock Dynamic matching in school choice: Efficient seat reallocation after
  late cancellations.
\newblock \emph{Manag. Sci.}, 66\penalty0 (11):\penalty0 5341--5361, 2020.

\bibitem[Gajulapalli et~al.(2020)Gajulapalli, Liu, Mai, and
  Vazirani]{DBLP:conf/fsttcs/GajulapalliLMV20}
Karthik Gajulapalli, James~A. Liu, Tung Mai, and Vijay~V. Vazirani.
\newblock Stability-preserving, time-efficient mechanisms for school choice in
  two rounds.
\newblock In \emph{Proceedings of the 40th {IARCS} Annual Conference on
  Foundations of Software Technology and Theoretical Computer Science ({FSTTCS}
  '20)}, pages 21:1--21:15. Schloss Dagstuhl, 2020.

\bibitem[Gale and Shapley(1962)]{gale1962college}
David Gale and Lloyd~S Shapley.
\newblock College admissions and the stability of marriage.
\newblock \emph{The American Mathematical Monthly}, 69\penalty0 (1):\penalty0
  9--15, 1962.

\bibitem[Gangam et~al.(2022)Gangam, Mai, Raju, and
  Vazirani]{DBLP:conf/fsttcs/GangamMRV22}
Rohith~Reddy Gangam, Tung Mai, Nitya Raju, and Vijay~V. Vazirani.
\newblock A structural and algorithmic study of stable matching lattices of
  "nearby" instances, with applications.
\newblock In \emph{Proceedings of the 42nd {IARCS} Annual Conference on
  Foundations of Software Technology and Theoretical Computer Science ({FSTTCS}
  '22)}, pages 19:1--19:20. Schloss Dagstuhl, 2022.

\bibitem[Gangam et~al.(2023)Gangam, Mai, Raju, and
  Vazirani]{DBLP:journals/corr/abs-2304-02590}
Rohith~Reddy Gangam, Tung Mai, Nitya Raju, and Vijay~V. Vazirani.
\newblock The stable matching lattice under changed preferences, and associated
  algorithms.
\newblock \emph{CoRR}, abs/2304.02590, 2023.

\bibitem[Genc et~al.(2017{\natexlab{a}})Genc, Siala, O'Sullivan, and
  Simonin]{DBLP:conf/aaai/Genc0OS17}
Begum Genc, Mohamed Siala, Barry O'Sullivan, and Gilles Simonin.
\newblock Robust stable marriage.
\newblock In \emph{Proceedings of the Thirty-First {AAAI} Conference on
  Artificial Intelligence ({AAAI} '17)}, pages 4925--4926. {AAAI} Press,
  2017{\natexlab{a}}.

\bibitem[Genc et~al.(2017{\natexlab{b}})Genc, Siala, O'Sullivan, and
  Simonin]{DBLP:conf/ijcai/Genc0OS17}
Begum Genc, Mohamed Siala, Barry O'Sullivan, and Gilles Simonin.
\newblock Finding robust solutions to stable marriage.
\newblock In \emph{Proceedings of the Twenty-Sixth International Joint
  Conference on Artificial Intelligence ({IJCAI} '17)}, pages 631--637.
  ijcai.org, 2017{\natexlab{b}}.

\bibitem[Genc et~al.(2019)Genc, Siala, Simonin, and
  O'Sullivan]{DBLP:journals/tcs/GencSSO19}
Begum Genc, Mohamed Siala, Gilles Simonin, and Barry O'Sullivan.
\newblock Complexity study for the robust stable marriage problem.
\newblock \emph{Theor. Comput. Sci.}, 775:\penalty0 76--92, 2019.

\bibitem[Gusfield(1987)]{gusfield1987three}
Dan Gusfield.
\newblock Three fast algorithms for four problems in stable marriage.
\newblock \emph{SIAM J. Comput.}, 16\penalty0 (1):\penalty0 111--128, 1987.

\bibitem[Hitsch et~al.(2010)Hitsch, Horta{\c{c}}su, and
  Ariely]{hitsch2010matching}
G{\"u}nter~J Hitsch, Al~Horta{\c{c}}su, and Dan Ariely.
\newblock Matching and sorting in online dating.
\newblock \emph{Am. Econ. Rev.}, 100\penalty0 (1):\penalty0 130--163, 2010.

\bibitem[Kanade et~al.(2016)Kanade, Leonardos, and
  Magniez]{DBLP:conf/approx/KanadeLM16}
Varun Kanade, Nikos Leonardos, and Fr{\'{e}}d{\'{e}}ric Magniez.
\newblock Stable matching with evolving preferences.
\newblock In \emph{Proceedings of Approximation, Randomization, and
  Combinatorial Optimization ({APPROX/RANDOM} '16)}, pages 36:1--36:13. Schloss
  Dagstuhl, 2016.

\bibitem[Kennes et~al.(2011)Kennes, Monte, Tumennasan,
  et~al.]{kennes2011daycare}
John Kennes, Daniel Monte, Norovsambuu Tumennasan, et~al.
\newblock \emph{The daycare assignment problem}.
\newblock Department of Economics and Business Economics, Aarhus BSS, 2011.

\bibitem[Knuth(1976)]{knuth1976marriages}
Donald~Ervin Knuth.
\newblock Marriages stables.
\newblock \emph{Technical report}, 1976.

\bibitem[Liu(2021)]{DBLP:journals/corr/abs-2007-03794}
Ce~Liu.
\newblock Stability in repeated matching markets.
\newblock \emph{CoRR}, abs/2007.03794v2, 2021.
\newblock URL \url{https://arxiv.org/abs/2007.03794}.

\bibitem[Mai and Vazirani(2018)]{DBLP:conf/esa/MaiV18}
Tung Mai and Vijay~V. Vazirani.
\newblock Finding stable matchings that are robust to errors in the input.
\newblock In \emph{Proceedings of the 26th Annual European Symposium on
  Algorithms ({ESA} '18)}, pages 60:1--60:11. Schloss Dagstuhl, 2018.

\bibitem[Mallows(1957)]{mal:j:mallows}
C.~Mallows.
\newblock Non-null ranking models.
\newblock \emph{Biometrica}, 44:\penalty0 114--130, 1957.

\bibitem[Manlove(2013)]{DBLP:books/ws/Manlove13}
David~F. Manlove.
\newblock \emph{Algorithmics of Matching Under Preferences}, volume~2 of
  \emph{Series on Theoretical Computer Science}.
\newblock WorldScientific, 2013.

\bibitem[Provan and Ball(1983)]{provan1983complexity}
J~Scott Provan and Michael~O Ball.
\newblock The complexity of counting cuts and of computing the probability that
  a graph is connected.
\newblock \emph{SIAM J. Comput.}, 12\penalty0 (4):\penalty0 777--788, 1983.

\bibitem[Roth(1986)]{roth1986allocation}
Alvin~E Roth.
\newblock On the allocation of residents to rural hospitals: a general property
  of two-sided matching markets.
\newblock \emph{Econometrica}, 54\penalty0 (2):\penalty0 425--427, 1986.

\bibitem[Schober et~al.(2018)Schober, Boer, and
  Schwarte]{schober2018correlation}
Patrick Schober, Christa Boer, and Lothar~A Schwarte.
\newblock Correlation coefficients: appropriate use and interpretation.
\newblock \emph{Anesth. Analg.}, 126\penalty0 (5):\penalty0 1763--1768, 2018.

\bibitem[Shiryaev et~al.(2013)Shiryaev, Yu, and
  Elkind]{DBLP:conf/atal/ShiryaevYE13}
Dmitry Shiryaev, Lan Yu, and Edith Elkind.
\newblock On elections with robust winners.
\newblock In \emph{Proceedings of the 2013 International Conference on
  Autonomous Agents and Multi-Agent Systems ({AAMAS} '13)}, pages 415--422.
  {IFAAMAS}, 2013.

\bibitem[{Vande Vate}(1989)]{VANDEVATE1989147}
John~H. {Vande Vate}.
\newblock Linear programming brings marital bliss.
\newblock \emph{Oper. Res. Lett.}, 8\penalty0 (3):\penalty0 147--153, 1989.

\end{thebibliography}
\end{document}